\documentclass[11pt]{amsart}
\usepackage{fullpage}
\usepackage{amsmath, amssymb, amsthm}
\usepackage{amsthm}
\usepackage{hyperref}
\hypersetup{
    colorlinks=true,
    linkcolor=blue,
    filecolor=magenta,      
    urlcolor=cyan,
    pdftitle={Overleaf Example},
    pdfpagemode=FullScreen,
    }

\usepackage{xcolor} 
\usepackage{cleveref}
\usepackage{aliascnt} 
\usepackage{MnSymbol}
\usepackage{mathrsfs}
\usepackage{bbm}
\usepackage{esint} 
\usepackage{soul} 
\usepackage{scalerel}
\usepackage{stackengine,wasysym}
\usepackage{graphicx}
\usepackage{subcaption}
\usepackage{tikz}
\usetikzlibrary {arrows.meta,bending,positioning}
\usetikzlibrary{decorations.markings}

\usepackage[style=numeric,sorting=nyt, url = false, doi = false, eprint = false, isbn = false,giveninits=true, uniquename=init]{biblatex}
\usepackage{pdfpages}
\usepackage{csquotes}
\nocite{*}

\def\ep{\ensuremath\epsilon}
\def\l{\ensuremath\left}
\def\r{\ensuremath\right}
\def\g{\ensuremath\gamma}

\newcommand\numberthis{\addtocounter{equation}{1}\tag{\theequation}}
\newtheorem{thm}{Theorem}[section]
\crefname{thm}{Theorem}{Theorems}
\newtheorem*{theorem*}{Theorem}

\newaliascnt{definition}{thm}
\newtheorem{definition}[definition]{Definition}
\aliascntresetthe{definition}
\crefname{definition}{Definition}{Definitions}

\newaliascnt{lem}{thm}
\newtheorem{lem}[lem]{Lemma}
\aliascntresetthe{lem}
\crefname{lem}{Lemma}{Lemmas}

\newaliascnt{prop}{thm}
\newtheorem{prop}[prop]{Proposition}
\aliascntresetthe{prop}
\crefname{prop}{Proposition}{Propositions}

\newaliascnt{cor}{thm}
\newtheorem{cor}[cor]{Corollary}
\aliascntresetthe{cor}
\crefname{cor}{Corollary}{Corollaries}

\newaliascnt{rmk}{thm}
\newtheorem{rmk}[rmk]{Remark}
\aliascntresetthe{rmk}
\crefname{rmk}{Remark}{Remarks}

\crefname{figure}{Figure}{Figures} 

\numberwithin{equation}{section}

\newcommand{\hssh}{H_{\mathbb{Z}}}
\newcommand{\q}{q_*}
\newcommand{\qzl}{q_{*,\lambda}}
\newcommand{\qzk}{q_{*,k}}

\newcommand{\sgn}{\text{sgn}}

\newcommand{\Hbulk}{H_{_{\rm bulk}}}
\newcommand{\Hedge}{H_{_{\rm edge}}}

\title{Dispersive decay bounds for the SSH model on the half-line}
\author{Remy Kassem, Amir Sagiv, Michael I. Weinstein}

\begin{document}
\begin{abstract}
We study the Schr\"odinger flow for the SSH model, a class of self-adjoint discrete dimer lattice Hamiltonians on the half-line. Using oscillatory integral techniques, we prove dispersive time-decay estimates, which quantify the spreading of energy throughout the lattice for a localized initial condition. Furthermore, we determine precise dependence of the constants in the decay rates on the parameters of the Hamiltonian. The analysis is complicated by the fact that as a consequence of the boundary condition, the expression for the propagator contains oscillatory integrals with nonintegrable singularities.   
\end{abstract} 
\maketitle

\setcounter{tocdepth}{1}
\tableofcontents

\section{Introduction}\label{sec: intro}
Consider a one-dimensional bipartite lattice defined as the union of two interpenetrating copies of the discrete lattice $\Lambda = \mathbb{Z}$ or $\mathbb{N}_0 := \mathbb{N}\cup\{0\}$; see \cref{fig: SSH Models a,fig: SSH Models b}. We may think of this lattice as consisting of cells, each containing   two vertices or ``atoms'', labeled $A$ and $B$; each $A$ vertex has two nearest neighbor $B$ vertices and each $B$ vertex has two nearest neighbor $A$ vertices. For every $n \in \Lambda$, we assign an amplitude $\psi_n = [\psi_n^A, \psi_n^B]^T \in \mathbb{C}^2$ to the $n$'th cell. A vector of such amplitudes, $\psi \in \ell^2(\Lambda ;\mathbb{C}^2)$, is called a {\it wave function} and characterizes the state of the system.

We define the {\it bulk SSH Hamiltonian}, $\Hbulk \psi$, to be the following translation invariant nearest neighbor Hermitian difference operator acting in the Hilbert space $ \ell^2(\mathbb{Z}) = \ell^2(\mathbb{Z}; \mathbb{C}^2)$: 
\begin{align*}
    (\Hbulk \psi)_n =
    \begin{bmatrix}
    \g_1 \psi_n^B + \g_2 \psi_{n-1}^B \\ \g_1 \psi_n^A + \overline{\g_2} \psi_{n+1}^A
    \end{bmatrix}, \quad n \in \mathbb{Z}\, .
\end{align*}
The coupling coefficients $\gamma_1$ (in-cell) and $\gamma_2$ (out-of-cell) are also called ``hopping" coefficients; $\g_2$, is allowed to be complex-valued and $\gamma_1$ is real-valued. The bipartite lattice and its associated Hamiltonian $\Hbulk$ are often called  the SSH lattice and bulk SSH Hamiltonian  \cite{OriginalSSH,Asboth_2016}. 

In this paper, we focus on a truncated  lattice (\cref{fig: SSH Models b}) and refer to the $A$--site in the terminal ($n=0$) cell as the \textit{edge} of the terminated SSH lattice. The edge Hamiltonian, $\Hedge$, is the corresponding self-adjoint operator acting on $ \ell^2(\mathbb{N}_0) = \ell^2(\mathbb{N}_0; \mathbb{C}^2)$:
\begin{align*}
    (\Hedge \psi)_{n} &= \begin{bmatrix} 
    \g_1 \psi^B_n + \g_2 \psi^B_{n-1} \\
    \g_1 \psi^A_n + \overline{\g_2 }\psi^A_{n+1}\end{bmatrix}, \quad \forall n \geq 0 \\
    \psi_{-1} &:= 0 \, .
\end{align*}

The spectrum of $\Hbulk$ acting on $\ell^2(\mathbb{Z})$ consists of two disjoint bands of continuous spectrum whenever $\g_1 \neq |\g_2|$, see \cref{fig: SSH Models c}. The finite energy range between these intervals is called the {\it spectral gap}. By Weyl's Theorem \cite[Chapter~XIII.4~Lemma~3]{RS4}, $\sigma_{ess}(\Hedge) = \sigma(\Hbulk) $. In the special case when the inter-cell and intra-cell hopping coefficients have equal magnitude, the model is equivalent to a (scaled) discrete Laplacian. Furthermore, in this case, the spectral bands meet at a point and the spectral gap closes. 

A property of the SSH bulk and edge Hamiltonians related to topological wave phenomena, see \cref{rmk: Topological Wave Phenonmena}, arises by allowing the hopping coefficients to vary. When $|\g_2| > \g_1$, $\Hedge$ has a single discrete eigenvalue in its spectrum in addition to its two intervals of continuous spectrum, corresponding to an \textit{edge state}, localized near the edge and exponentially decaying into the bulk. In contrast, if $|\g_2| < \g_1$ then the full spectrum of $\Hedge$ is equal to that of $\Hbulk$; see \cref{fig: SSH Models d}.

Hence when studying the general time dynamics $t\mapsto \exp(-it\Hedge )$, it is natural to orthogonally decompose our Hilbert space as $\ell^2(\mathbb{N}_0) =\{{\rm edge\  state}\}\oplus \ell_{ac}^2(\mathbb{N}_0)$, where $\mathscr{P}_{ac}=\mathscr{P}_{ac}(\Hedge)$ is the orthogonal projection onto the continuous spectral part of $\Hedge$ and $\mathscr{P}_{ac}\ \ell^2(\mathbb{N}_0)=\ell_{ac}^2(\mathbb{N}_0)$. 

This article is a detailed study of the time evolution governed by the Schr\"odinger flow 
\begin{align*}
    i \partial_t \psi = \Hedge \psi, \quad \psi \in \ell^2( \mathbb{N}_0)\, ,
\end{align*}
in norms which register dispersive time-decay. Our main time decay results are stated precisely in Theorem \ref{thm: Main thm 1}. Here, we provide a rough statement. Let
\begin{align*}
    \ell_{\sigma}^1 (\mathbb{N}_0;\mathbb{C}^2) = \{ f = (f_n)_{n \in \mathbb{N}_0}   \textit{ such that } \sum_{n \in \mathbb{N}_0} (1+n)^{\sigma} |f_n|_{\mathbb{C}^2} < \infty \} \, .
\end{align*}

\begin{theorem*}[Informal version of \cref{thm: Main thm 1}]
 Assume that $\Hedge$ has a spectral gap and that $f$ is orthogonal to any edge state of $\Hedge$. Then, there exists a constant $c(\Hedge)$ such that for all $n \in \mathbb{N}_0$ and $t > 0$, 
\begin{align*}
    \l|\l[e^{-it\Hedge} f\r]_n \r| &\leq c(\Hedge) \|f\|_{\ell_1^1} \langle t \rangle^{-1/3}\, , \\
    \l|\l[e^{-it\Hedge} f\r]_n \r| &\leq c(\Hedge) \|f\|_{\ell_2^1} \l[ \langle t \rangle^{-1/2}+  n \langle t \rangle^{-1}\r]\, .
\end{align*}
\end{theorem*}

\begin{figure}[htbp]
\centering

\resizebox{6in}{!}{%
  \begin{minipage}{\textwidth}
  \centering
  \begin{subfigure}{0.49\textwidth}
    \centering
    \includegraphics[width=\linewidth]{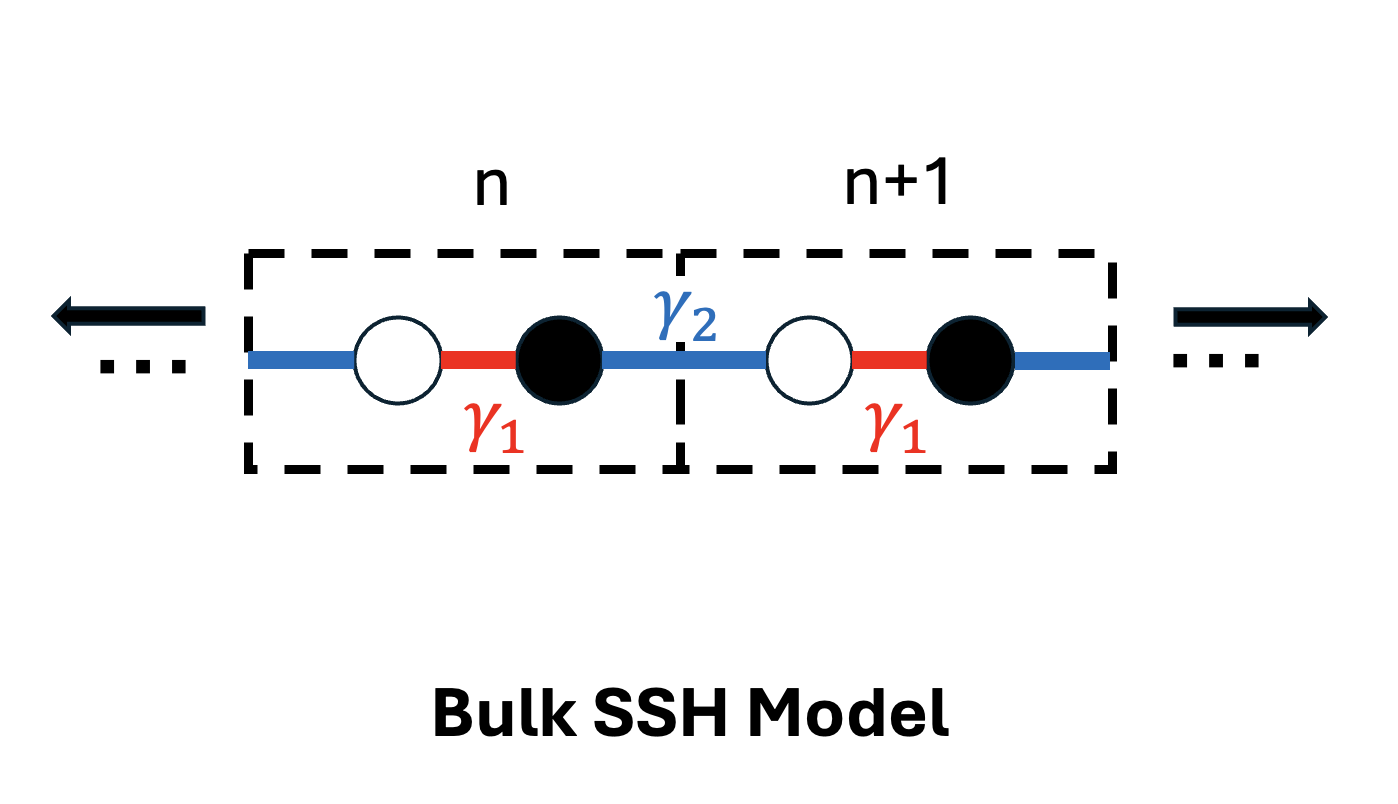}
    \caption{}\label{fig: SSH Models a}
  \end{subfigure}
  \begin{subfigure}{0.49\textwidth}
    \centering
    \includegraphics[width=\linewidth]{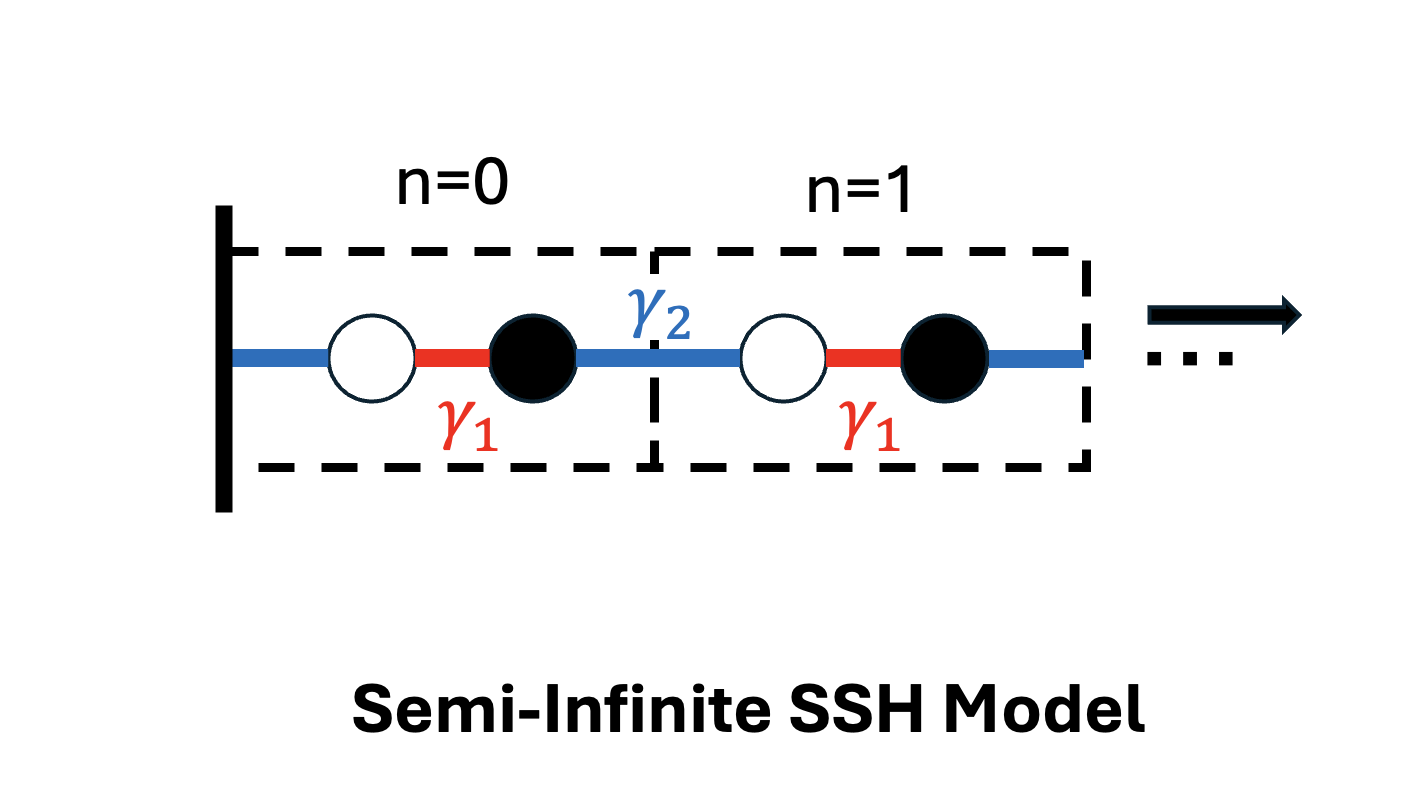}
    \caption{}\label{fig: SSH Models b}
  \end{subfigure}
  \begin{subfigure}{0.49\textwidth}
    \centering
    \includegraphics[width=\linewidth]{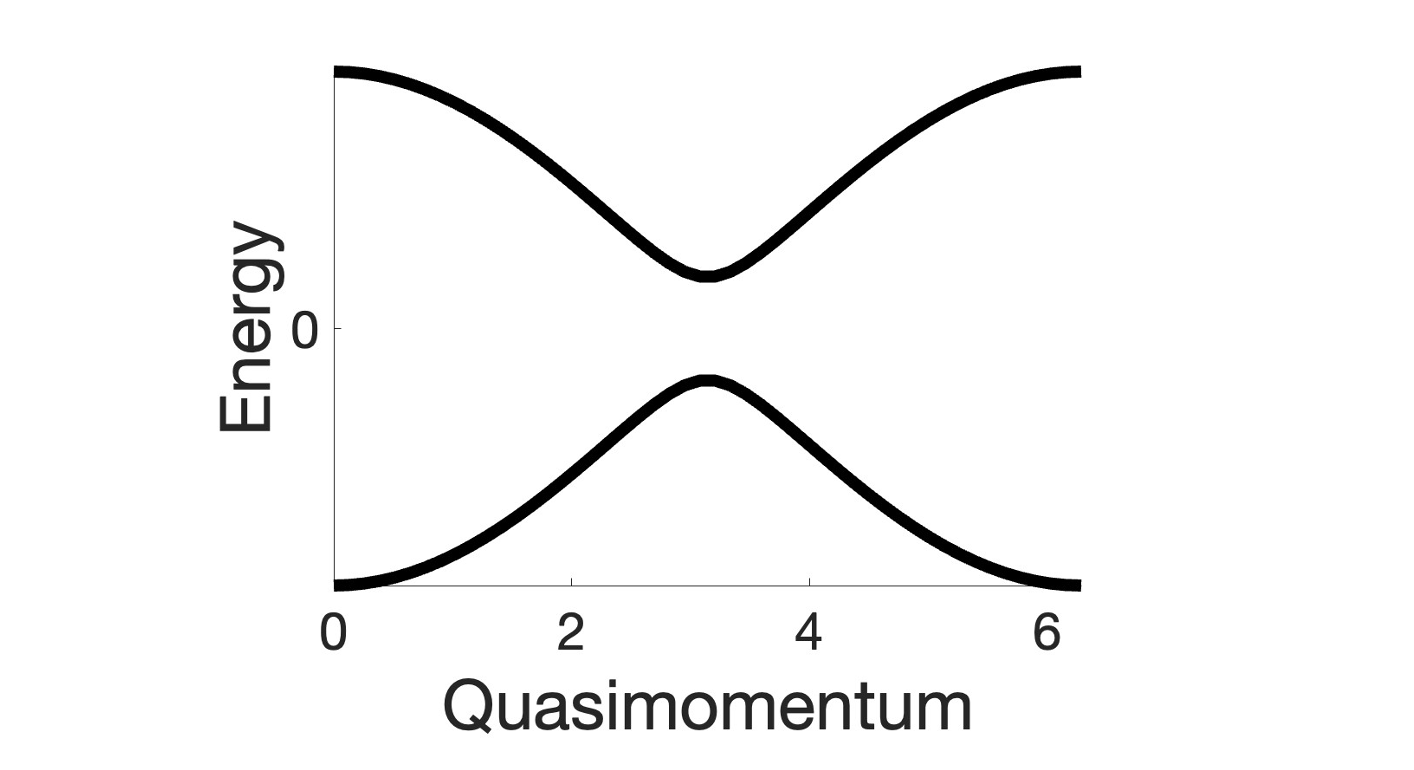}
    \caption{}\label{fig: SSH Models c}
  \end{subfigure}
  \begin{subfigure}{0.49\textwidth}
    \centering
    \includegraphics[width=\linewidth]{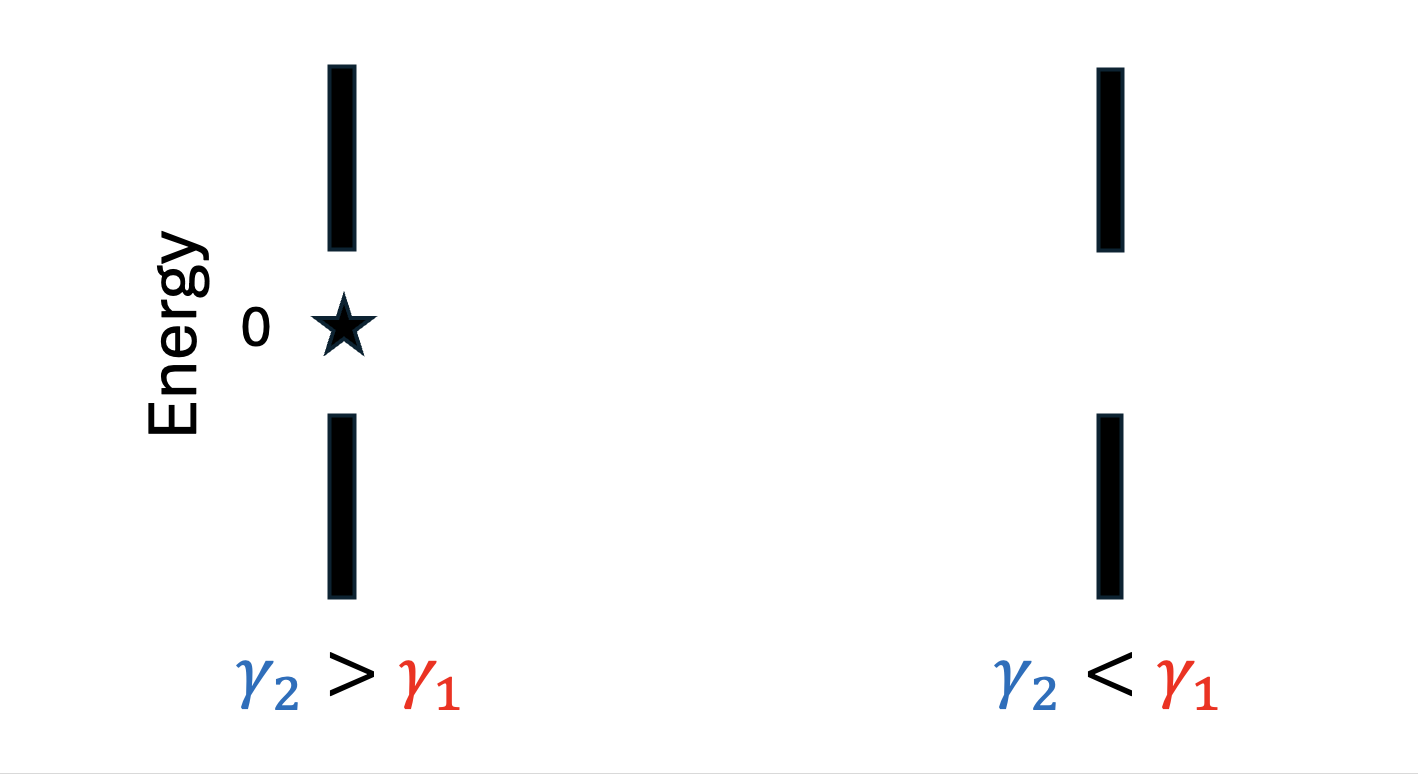}
    \caption{}\label{fig: SSH Models d}
  \end{subfigure}
  \end{minipage}
}
\caption{(A) The $n$'th and $(n+1)$'th cells of the Bulk SSH model. (B) The $n =0$ and $n=1$ cells of the Semi-Infinite SSH model. (C) Spectrum of Bulk SSH model. (D) Spectrum of Semi-Infinite SSH model possess a single discrete eigenvalue when the out-of-cell hopping coefficient is larger than the in-cell hopping coefficient. }
\label{fig: SSH Models}
\end{figure}

\subsection{Outline}
Section \ref{sec: Discrete Hamiltonians} discusses the bulk and edge SSH Hamiltonians, their spectrum, and dispersive decay estimates for the bulk SSH operator.
Section \ref{sec: Main Results} states our main results, dispersive decay estimates for the edge SSH operator, see \cref{thm: Main thm 1}, \cref{cor: Main thm}, and \cref{cor: Main thm interpolation}, as well as the definitions and notation that will be used throughout the paper. A detailed overview of the proof of \cref{thm: Main thm 1} is provided in Section \ref{sec: roadmap}.

We obtain the estimates of \cref{thm: Main thm 1} in two stages. First, we compute an explicit expression for $e^{-it\Hedge}$. Section \ref{sec: Computing the Resolvent} is devoted to computing the resolvent $(\Hedge - z I)^{-1}$, which is equal to $(\Hbulk - z I)^{-1}$ plus a correction due to the boundary. In Section \ref{sec: Computing the Propagator}, we use the Dunford integral to compute the propagator $\exp\l(-it\Hedge \r)$, see \cref{thm: propagator of SSH}. As with the resolvent, the propagator involves bulk and boundary terms.

The second stage is focused on obtaining a dispersive decay estimate of the propagator. The bulk propagator, $e^{-i \Hbulk t}$, is straightforward to estimate, but the correction due to the boundary requires more detailed harmonic analysis. In Section \ref{sec: Representative Oscillatory Integrals} we identify three representative oscillatory integrals for the contribution of the boundary to the propagator, see equations \eqref{eq: Type I def}--\eqref{eq: Type III def}. The first two are estimated in Section \ref{sec: Representative Oscillatory Integrals}. A significant technical challenge comes from the last oscillatory integral, equation \eqref{eq: Type III def}, because it has a \textit{nonintegrable singularity}. 
We proceed by breaking this integral into two components. The  first component is estimated via contour integration in Section \ref{sec: Type IIIa Decay Rates}. Sections \ref{sec: Type IIIb Uniform Decay} and \ref{sec: Type IIIb Local Decay} obtain spatially uniform and spatially weighted decay rate estimates, respectively, of the second component of the nonintegrable oscillatory integral.

\subsection{Why study the SSH model?}

In condensed matter physics, the SSH model is a paradigm of one-dimensional topological insulators. Although the model is quite simple, it exhibits many important properties which are central to topological insulators, e.g., chiral symmetry, bulk-boundary correspondence, and topological invariants \cite{Fruchart_2013, Asboth_2016}. The SSH model has also drawn interest in other areas of physics, such as acoustic waveguides \cite{Coutant_2021},  lasing in ring resonator arrays \cite{Rechtsman_2018}, and ultracold atoms \cite{Cooper_2019}.

The Hamiltonians $\Hbulk$ and $\Hedge$  can also be realized as a tight binding limit, describing the low-energy spectrum of continuum Schr\"odinger operators with one-dimensional periodic dimer potentials \cite{MW135}. Furthermore, an operator of type $\Hedge$ appears as the limiting tight-binding operator for the low energy spectrum of graphene, terminated along a zigzag edge  in the strong binding regime \cite{MW127}; see also \cite{SW-advances:22}.

\subsection{Relevant Earlier Work}
Dispersive decay concerns the spatial spread  and attenuation to zero of solutions as time advances. It is a phenomena governed by the continuous spectrum of an underlying Hamiltonian. Estimates on the rate of dispersive decay have been studied extensively in many PDEs, e.g., the Dirac equation \cite{Egorova_2016,Erdogan_2018,Erdogan_2019,Erdogan_2021,MW145} and the Schr\"odinger equation with a spatially decaying potential \cite{Soffer_1991,Weder_2003,Schlag_2005}, to name a few.

When $V(x)$ is a smooth periodic potential on $\mathbb{R}$, the continuous spectrum of $-\partial_x^2 + V$ is the union of intervals (``bands"). Dispersive estimates for such systems were obtained by  \cite{Cai_2006,Cuccagna_2008}. In \cite{Cai_2006}, Cai showed that solutions to the Schr\"odinger equation with a Weierstrass $\wp$-function potential (such real-valued periodic potentials possess only one spectral gap \cite{Hochstadt_1965}) generically decay like $t^{-1/3}$ in the $L^1 \rightarrow L^{\infty}$ operator norm.

The study of dispersive decay estimates for discrete systems, however, is more recent. Schultz \cite{Schultz_1998} studied the wave equation on $\mathbb{Z}^d$ for $d=2,3$. Many papers have focused on the discrete Schr{\"o}dinger with a decaying potential, see for instance \cite{Stefanov_2005,Komech_2006,Pelinovsky_2008,Cuccagna_2009,Egorova_2015}. More relevant to this work are discrete systems with periodic potentials. A key property shared with continuum periodic systems is discrete translation invariance. Very recently, Damanik, Fillman, and Young \cite{Damanik_2025} obtained a $t^{-1/3}$ dispersive $\ell^1 \to \ell^{\infty}$ bound for periodic nearest-neighbors operators (Jacobi operators) on $\mathbb{Z}$.

To the best of our knowledge, this article, along with the authors' \cite{MW154}, are the first works to consider dispersive decay estimates on the
discrete half line. We note the work of Ito and Jensen~\cite{Jensen_2017}, which, while it does not obtain dispersive decay estimates, does study the spectral theory of the discrete Laplacian plus a potential on $\mathbb{N}$. In the paper, they computed expansions of the resolvent at the thresholds. Similar expansions of the resolvent have been used to obtain dispersive estimates for the Schr\"odinger equation on $\mathbb{Z}$, see for instance \cite{Komech_2006,Pelinovsky_2008}. 

In \cite{MW154}, the authors obtained dispersive decay estimates of the Schr\"odinger time-evolution generated by periodic Jacobi operators acting on $\ell^2(\mathbb{N})$ . Jacobi operators are a family of real, self-adjoint, tridiagonal, infinite matrices with positive off-diagonal entries \cite{Teschl2000_Jacobi,Simon2010_SzegoTheroem,Lukic2022_SpectralTheory}. Since the SSH Hamiltonian is a particular type of Jacobi operator, the results of \cite{MW154} apply to the context of this article as well. We comment here on some of the distinctions between the  papers. 
\begin{enumerate}
    \item The Jacobi operator structure of the Hamiltonian is crucial to the argument of \cite{MW154}, as the authors use Jacobi operator theory \cite{Lukic2022_SpectralTheory} to construct an expression for the time evolution operator. In contrast, the approach of this article is built on a more general method of computing the spectral measure via a limiting sequence of Dunford integrals \cite[Definition~5.2.11] {Buhler_Functional}. We expect this approach to be applicable to non-nearest neighbor models such as those considered in \cite{Hirsch_1986,Hirsch_1987,Santos_1992}, which are \emph{not} Jacobi operators. 
    \item In contrast with \cite{MW154}, the representation of the propagator in the current work makes the role played by the edge in the time-evolution very explicit. \cref{thm: FL resolvent} shows that the contribution of the edge to the resolvent is nonlocal in frequency space and \cref{thm: propagator of SSH} expresses the propagator the $\Hedge$ as the sum of $\exp(-i\Hbulk t)$ plus a correction term.
    \item In \cref{thm: Main thm 1} and its corollaries, we provide the precise dependence of the constants in the decay rates on the hopping coefficients $\g_1$ and $\g_2$, which was not considered in \cite{MW154}. The precise dependence of the constants allows one to uniformly control the decay of the solution for a varied range of hopping coefficients. Such control will be useful when considering 2D problems such as graphene with a zig-zag edge, see item (2) in Section \ref{subsec: Future Directions}. 
\end{enumerate}

\subsection{Future directions and open problems}\label{subsec: Future Directions}

\begin{enumerate}
    \item \textbf{Radiation damping in Floquet systems}: Quantum systems driven by a time-periodic forcing term are of great interest in  condensed matter physics and the metamaterials community in such materials because the time-periodic forcing allows for another dimension in which one may tune the system \cite{OkaFloquet, AluFloquet}. These systems are known as \textit{Floquet systems}, and their evolution is generically described by 
    \begin{align*}
        i \psi_t = (H_0 + \ep H_1(t))\psi \, , 
    \end{align*}
    where $H_1(t)$ is a self-adjoint periodic operator. The characteristic behavior of such Floquet systems is that localized modes of the unforced system (e.g. $\ep = 0$) resonantly couple to the continuous spectrum. This coupling causes a radiative transfer of energy away from the localized modes, which are said to be \textit{metastable} \cite{MW42,Costin_2002,MW138}. This is an example of a broader phenomenon known as radiation damping or ionization. 
    Estimating the dispersive decay of solutions to the Schr\"odinger equation is a key step towards understanding ionization phenomena on intermediate and long time-scales. 
    \item \textbf{Two-dimensional discrete models (e.g. graphene)}: A natural extension of our analysis on SSH Hamiltonians with an edge is to that of sharply terminated honeycomb lattice structures such as graphene with a zig-zag edge \cite{MW127}. Taking a Floquet-Bloch transform of graphene with a zig-zag edge along the direction of translation invariance yields a Hamiltonian whose fibers at each quasimomenta are SSH Hamiltonians with complex hopping coefficients. Dispersive decay estimates of the bulk honeycomb lattice have been recently studied by \cite{Hong_2025}, without using this connection to the SSH model.
\end{enumerate}

 \subsection{Notation and Conventions}
 \begin{enumerate}
     \item $\mathbb{N}_0=\{0,1,2,\dots\} = \mathbb{N} \cup \{0\}$.
     \item We say $A\lesssim B$ if there exists some universal constant $c$ such that $A \leq cB$. 
     \item Let $\Lambda = \mathbb{Z}$ or $ \mathbb{N}_0$. For $1 \leq p < \infty$, define $ \ell_{\sigma}^p(\Lambda) = \ell_{\sigma}^p(\Lambda; \mathbb{C}^d) = \{f : \|f\|_{\ell_{\sigma}^p} < \infty \}$ and $\ell_{\sigma}^{\infty}(\Lambda) = \ell_{\sigma}^{\infty}(\Lambda; \mathbb{C}^d) = \{f : \|f\|_{\ell_{\sigma}^{\infty}} < \infty \}$ 
     where
     \begin{align*}
         \|f\|_{\ell^p_{\sigma}} \equiv \l[\sum\limits_{n\in \Lambda} \left| (1+|n|)^{\sigma} f_n\right|_{\mathbb{C}^d}^p \r]^{1/p} \, , \qquad \qquad \|f\|_{\ell^{\infty}_{\sigma}} \equiv \sup\limits_{n \in \Lambda} \l| (1+|n|)^{\sigma}f_n\r|_{\mathbb{C}^d} \, .
     \end{align*}
    When the context is clear, we write $|f_n|$ instead of  $|f
     _n|_{\mathbb{C}^d}$.
     \item  $W^{k,p} = \{f : \|f\|_{W^{k,p}} < \infty \}$ where $\| f \|_{W^{k,p}} = \l( \sum_{|\alpha \leq k} \int |\partial_{\alpha} f|^p \r)^{1/p}$.
     \item $\fint_a^b f(x)\, dx$ denotes the Cauchy principal value.
     \item $\hat{f}(q)$ is the discrete Fourier transform on $\ell^2(\mathbb{Z})$.
     \item $\tilde{f}(q)$ is the discrete Fourier-Laplace transform on $\ell^2(\mathbb{N}_0)$; see \eqref{def: FL transform}.
     \item $\g_1$, $\g_2$ are hopping coefficients; $\g_+= |\g_2| + \g_1 $ and $\g_- = \l| |\g_2| - \g_1 \r| $.
     \item $H=H_{_{\rm edge}} : \ell^2(\mathbb{N}_0;\mathbb{C}^2)\to \ell^2(\mathbb{N}_0;\mathbb{C}^2)$ is the SSH Hamiltonian on the half-line; see \cref{def: Edge SSH}. 
     \item We abbreviate the spectral projection onto the absolutely continuous subspace of $H$ as $\mathscr{P}_{ac}=\mathscr{P}_{ac}(H)$. Furthermore let $\mathscr{P}_{\pm} = \mathscr{P}_{\pm}(H)$ denote the projections onto the positive and negative part of the absolutely continuous spectral subspace of $H$ respectively. Then $ \mathscr{P}_{ac} = \mathscr{P}_+ + \mathscr{P}_- $.
 \end{enumerate}

 \subsection*{Acknowledgments} The authors thank Rahul Sengottuvel for a very helpful discussion concerning the use of ChatGPT-5 Pro, which the authors used in the proof of \cref{lem: inf of k'''}. MIW and RK  were supported in part by NSF grants: DMS-1908657, DMS-1937254 and DMS-2510769, and Simons Foundation Math+X Investigator Award \#376319 [MIW]. AS was supported in part by NSF Grant No.\ DMS-2508811. AS and MIW were supported in part by BSF Grant No.\ 2022254. Part of this research was carried out during the 2023-24 academic year, when M.I.W. was a Visiting Member in the School of Mathematics, Institute of Advanced Study, Princeton, supported by the Charles Simonyi Endowment, and a Visiting Fellow in the Department of Mathematics at Princeton University.

\section{Discrete Hamiltonians: $\Hbulk$ and $\Hedge$}\label{sec: Discrete Hamiltonians}

\subsection{Analysis of the bulk SSH Hamiltonian}
Since bulk SSH Hamiltonian is diagonalized by the discrete Fourier transform, studying the time dynamics of $t\mapsto \exp(-it \hssh )$ is quite simple. 

\begin{definition}\label{def: Bulk SSH}
Choose constants $\g_1>0$ and $\g_2\in\mathbb C$. For any $\psi\in \ell^2(\mathbb Z;\mathbb{C}^2)$, define the bulk SSH Hamiltonian $\hssh : \ell^2(\mathbb{Z}; \mathbb{C}^2) \rightarrow  \ell^2(\mathbb{Z}; \mathbb{C}^2)$  by 
\begin{align*}
    (\hssh \psi)_n =
    \begin{bmatrix}
    \g_1 \psi_n^B + \g_2 \psi_{n-1}^B \\ \g_1 \psi_n^A + \overline{\g_2} \psi_{n+1}^A
    \end{bmatrix}, \quad n \in \mathbb{Z}\, .
\end{align*}
\end{definition}

The coefficients $\g_1$ and $\g_2$ are called {\it hopping coefficients}.
Define 
\begin{align}
    &h: \mathbb{R}/2 \pi \mathbb{Z} \to \mathbb{C}; \qquad h(q) := \g_1 + \g_2 e^{-iq} \label{def: h} \\
    &k: \mathbb{R}/2 \pi \mathbb{Z} \to \mathbb{C}; \qquad k(q) := \sqrt{\g_1^2 + |\g_2|^2 +2 \g_1 |\g_2|  \cos(q)} \label{def: k} \\
    &\varphi := \arctan\l(\frac{\Im \g_2}{\Re \g_2}\r) \in (-\pi/2,\pi/2) \, .   \label{def: varphi}
\end{align}

\bigskip

Through elementary trigonometric identities, we check that \begin{align}
    e^{i \varphi} &= \frac{\g_2}{|\g_2|} \label{eq: e^(i varphi)} \\
    k^2(q - \varphi) &= h(q) \overline{h(\overline{q})} \, .  \label{eq: k shifted}
\end{align}
The discrete Fourier transform of $(H_{\mathbb{Z}} - z I) \psi = f$ yields \begin{align}\label{eq: DFT of bulk operator}
    \begin{bmatrix}
        -z & h(q) \\
        \overline{h( \overline{q})} & -z 
    \end{bmatrix} \hat{\psi}(q) = \hat{f}(q) \text{ for all } q \in [-\pi, \pi]. 
\end{align}
Inverting the matrix, we see that the discrete Fourier transform of the bulk resolvent is given by
\begin{align}\label{eq: bulk resolvent}
    \l[\hat{R}_{bulk}(z) \hat{f}\, \r](q) = \frac{1}{k^2(q - \varphi) - z^2} \begin{bmatrix}
        z & h(q) \\
        \overline{h( \overline{q})} & z 
    \end{bmatrix} \hat{f}(q) \, .
\end{align}

Taking the discrete Fourier transform of the time-dependent Schr\"odinger equation 
\begin{align*}
    i \frac{d}{dt} \psi (t) &= (\hssh \psi)(t) \\
    \psi(0) &= f \in \ell^2 \l(\mathbb{Z}; \mathbb{C}^2\r)
\end{align*}
yields 
\begin{align*}
    i \frac{d}{dt} \hat{\psi}(q, t) &= \begin{bmatrix}
        0 & h(q) \\
        \overline{h( \overline{q})} & 0 
    \end{bmatrix} \hat{\psi}(q,t) =: \hat{H}_{\mathbb{Z}}  \hat{\psi}(q, t) \\
    \hat{\psi}(q, t) &= e^{-i \hat{H}_{\mathbb{Z}} t} \hat{\psi}(q,0) \, .
\end{align*}
The matrix $\hat{H}_{\mathbb{Z}}$
has eigenpairs $\l( \pm k(q - \varphi) , v_{\pm}(q - \varphi)\r)$. Looking at the range of $\pm k(q - \varphi)$, we see that the spectrum of $\hssh$ is \begin{align*}
    \sigma( \hssh)  = [-  \g_+, - \g_- ] \cup [\g_-, \g_+  ]\, ,
\end{align*}
where we define \begin{align}\label{eq: gamma pm}
    \g_+ := |\g_2| + \g_1 \quad \text{and} \quad \g_- := ||\g_2| - \g_1| \, .
\end{align}

After diagonalizing the matrix and taking the inverse discrete Fourier transform, the solution to the time-dependent Schr\"odinger equation is \begin{align}\label{eq: bulk SSH solution}
    \psi_n (t) &= \sum_{+/-} \frac{1}{2\pi} \int \limits_{-\pi}^{\pi} e^{i( nq \mp k(q-\varphi)t) } \langle v_{\pm}(q - \varphi), \hat{f} (q) \rangle v_{\pm}(q - \varphi)  \, dq  \, .
\end{align}
Having expressed the $n^{th}$ entry of the solution as an oscillatory integral, \cref{thm: free dispersive} follows from \cref{cor: application of VDC}.

\begin{thm}\label{thm: free dispersive}
Suppose that $\g_1, |\g_2| > 0$ and that $\g_1 \neq |\g_2|$. Then for all $n \in \mathbb{Z}$, 
\begin{align}\label{eq: free dispersive}
    \l|[e^{-i \hssh t} f]_n \r| \lesssim  \begin{cases}
    \l(1+ \min\{ \g_1, |\g_2|\}^{-1/2}\r) \langle t \rangle ^{-1/3} \|f\|_{\ell^1} \\
    \l(1+\min\{\g_1, |\g_2|\}^{-1}\r) \l[\langle t\rangle^{-1/2}  + n\langle t \rangle^{-1} \r] \|f\|_{\ell_1^1} \, . 
    \end{cases}
\end{align}
\end{thm}

\subsection{SSH with an Edge}
As noted, this paper considers the SSH model with an edge.
\begin{definition}\label{def: Edge SSH}
We refer to the Hamiltonian associated to the SSH model with Dirichlet boundary  conditions at the edge as the SSH Hamiltonian with an edge. It is given by $H : \ell^2(\mathbb{N}_0; \mathbb{C}^2) \rightarrow \ell^2(\mathbb{N}_0; \mathbb{C}^2) $, 
\begin{align*}
    (H \psi)_{n} &= \begin{bmatrix} 
    \g_1 \psi^B_n + \g_2 \psi^B_{n-1} \\
    \g_1 \psi^A_n + \overline{\g_2 }\psi^A_{n+1}\end{bmatrix}, \quad \forall n \geq 0 \\
    \psi_{-1} &:= 0 \, .
\end{align*}
    
\end{definition}

\begin{rmk}
Equivalently, we could view the SSH Hamiltonian with an edge as acting on the subspace of $\ell^2(\mathbb{Z})$, \begin{align*}
    \{ f\in \ell^2( \mathbb{Z}) : f_n \equiv 0 \text{ for all } n \leq -1 \} \, . 
\end{align*}
\end{rmk}

In the case where $|\g_2| > \g_1$, $H$ has a zero-energy mode $\phi_*$ given by \begin{equation}\label{eq: zero energy mode}\begin{split}
    (\phi_*)_n &= \begin{bmatrix}
    \l(- \g_1 / \, \overline{\g_2} \r)^n \\ 0
    \end{bmatrix}\, , \quad \forall n \geq 0\, , \\
    (\phi_*)_n &\equiv 0 \,,  \quad \forall n \leq -1. 
\end{split}
\end{equation}
Since $\phi_*$ is localized near the edge and exponentially decays into the bulk, it is often referred to as an \textit{edge state}. It can be easily seen from the expression for $\phi_*$ that  $\phi_* \in \ell^2(\mathbb{N}_0; \mathbb{C}^2)$ if and only if  $|\g_2| > \g_1$. In the case where $|\g_2| < \g_1$, $H$ has no eigenvalues, see \cref{prop: no embedded eigenvalues}. Thus the spectrum of $H$ is given by 
\begin{equation}
\label{eq: H spectrum}
\begin{split}
\sigma(H) &= \begin{cases}
    \sigma_{ess}(H) \cup \{0\},  & \text{ for } |\g_2| > \g_1 \\
     \sigma_{ess}(H),  & \text{ for } |\g_2| <  \g_1 
\end{cases}  \\
    \text{where }\sigma_{ess}(H) &=  [-  \g_+, - \g_- ] \cup [\g_-, \g_+  ]\, .
\end{split}
\end{equation}

\begin{rmk}\label{rmk: Topological Wave Phenonmena}
This behavior of the SSH edge Hamiltonian--the existence of an edge state when $|\g_2| > \g_1$, but not for $|\g_2| < \g_1$--is paradigmatic of the field of topological wave phenomena. A discrete-valued topological index (in this case a winding number called the {\it Zak phase}) characterizes these two distinct ``topological phases" \cite{Asboth_2016}. This winding number is the number of times the curve traced out by $\g_1 + \g_2 e^{-iq}$ goes around the origin. Equivalently, one can consider the curve traced out by the vector whose components are the coefficients of $\hat{H}_{\mathbb{Z}}$ written out in the basis of the Pauli matrices. The chiral symmetry of the SSH model ensures that the curve is restricted to a two-dimensional plane in $\mathbb{R}^3$, which makes the associated winding number well-defined.

The winding number for the SSH model is either 1, when $|\g_2| > \g_1$, or 0, when $|\g_2| < \g_1$. When $|\g_2| > \g_1$, the SSH model is said to be topologically nontrivial and when  $|\g_2| <\g_1$ the model is topologically trivial. These phases are robust in the sense that one cannot continuously deform an SSH Hamiltonian from a topologically trivial to a topologically nontrivial phase (or vice versa) without closing the spectral gap (equivalently setting $\g_1 = |\g_2|$).

The relevance of the topological phase of $\Hbulk$ to its associated edge Hamiltonian $\Hedge$ is expressed through the \textit{bulk-edge correspondence}, which says that the number of edge states is equal to the winding number. For more details on topological characteristics of the SSH model, see \cite{Asboth_2016,MW135}. 
\end{rmk}

\section{Main Results}\label{sec: Main Results}

Let $\mathscr{P}_{ac}=\mathscr{P}_{ac}(H)$ denote the projection onto the continuous spectral subspace of $H$. 
Our main result is the following theorem about time-decay estimates of $e^{-iHt}\mathscr{P}_{ac} f$ when the essential spectrum of $H$ is gapped, which corresponds to the case where the hopping coefficients satisfy $\g_1 \neq |\g_2|$. We denote the spectral gap radius by $\g_- \equiv \Big||\g_2| - \g_1\Big|$, see \eqref{eq: gamma pm}.

\begin{thm}\label{thm: Main thm 1}
Assume $\gamma_->0$.
Let $f \in \ell_1^1( \mathbb{N}_0; \mathbb{C}^2)$.  Then, for all $n\in \mathbb{N}_0$, 
\begin{align*}
    \l|\l[e^{-iHt} \mathscr{P}_{ac}  f \r]_n \r| &\lesssim \l(1 +  \min\{ \g_1, |\g_2| \}^{-2/3} + \g_-^{-1/3} \r) \langle t \rangle^{-1/3} \|f\|_{\ell_1^1} \\ 
    \l|\l[e^{-iHt} \mathscr{P}_{ac} f \r]_n \r| &\lesssim
    \l(1 +  \min\{ \g_1, |\g_2| \}^{-1} + \g_-^{-1/2} \r) \l[ \log\l(\sqrt{2+ t^2} \r) \langle t \rangle^{-1/2} + n\langle t \rangle^{-1}\r] \|f\|_{\ell_1^1} \, .
\end{align*}
Moreover, if $f \in \ell_2^1( \mathbb{N}_0; \mathbb{C}^2)$, then
\begin{align*}
    \Big|\l[e^{-iHt} \mathscr{P}_{ac} f \r]_n \Big| \lesssim
    \l(1+\min\{\g_1, |\g_2|\}^{-1} + \g_-^{-1/2}\r) \l[\langle t\rangle^{-1/2}  + n\langle t \rangle^{-1} \r] \|f\|_{\ell_2^1} \, .
\end{align*}

\end{thm}

In this article we present a detailed proof of Theorem \ref{thm: Main thm 1}. An overview of the proof with references to  specific sections is presented in Section \ref{sec: roadmap}.

\begin{rmk}
By \cref{prop: no embedded eigenvalues}, $H$ has no nonzero eigenvalues. It is easy to check that when $|\g_2| < \g_1$ the spectrum of $H$ has no bound states at all. Hence in this case, $\mathscr{P}_{ac} = Id$.
\end{rmk}

\begin{cor}\label{cor: Main thm}
If $\g_- >0$, then 
\begin{align*}
    \l\| e^{-iHt} \mathscr{P}_{ac}   \r\|_{\ell_1^1 \rightarrow \ell^{\infty}} &\lesssim \l(1 + \min\{ \g_1, |\g_2| \}^{-2/3} + \g_-^{-1/3} \r)  \langle t \rangle^{-1/3}  \\ 
    \l\| e^{-iHt} \mathscr{P}_{ac}   \r\|_{\ell_1^1 \rightarrow \ell_{-1}^{\infty}} &\lesssim
    \l(1 + \min\{\g_1, |\g_2|\}^{-1} + \g_-^{-1/2} \r)  \log\l(\sqrt{2+ t^2} \r) \langle t \rangle^{-1/2}  \\
    \l\| e^{-iHt} \mathscr{P}_{ac}   \r\|_{\ell_2^1 \rightarrow \ell_{-1}^{\infty}}  &\lesssim
    \l(1 + \min\{\g_1, |\g_2|\}^{-1} + \g_-^{-1/2} \r) \langle t \rangle^{-1/2}
    \, .
\end{align*}
\end{cor}

Using Stein's interpolation theorem \cite{Stein_1956} to interpolate between
\begin{align*}
    \l\|e^{-iHt} \mathscr{P}_{ac} \r\|_{\ell^2 \rightarrow \ell^2} = 1
\end{align*} 
and the estimates of \cref{cor: Main thm} 
yields the following bounds. 
\begin{cor}\label{cor: Main thm interpolation}
For all $1 < p < 2$, let $q$ be its conjugate exponent satisfying $1/p + 1/q = 1$. Let $r = 2/p -1$.
Then, \begin{align*}
    \l\| e^{-iHt} \mathscr{P}_{ac} \r\|_{\ell_r^p \rightarrow \ell^q} &\lesssim \l(1 + \min\{ \g_1, |\g_2| \}^{-2/3} + \g_-^{-1/3} \r)^{r}  \langle t \rangle^{-\frac{r}{3} }  \\
    \l\| e^{-iHt} \mathscr{P}_{ac} \r\|_{\ell_r^p \rightarrow \ell_{-r}^q} &\lesssim  \l(1 + \min\{\g_1, |\g_2|\}^{-1} + \g_-^{-1/2} \r)^{r}  \l(\log\l(\sqrt{2+ t^2} \r) \langle t \rangle^{-1/2} \r)^{r} \\
     \l\| e^{-iHt} \mathscr{P}_{ac} \r\|_{\ell_{2r}^p \rightarrow \ell_{-r}^q} &\lesssim  \l(1 + \min\{\g_1, |\g_2|\}^{-1} + \g_-^{-1/2} \r)^{r}  \langle t \rangle^{-\frac{r}{2}} 
    \, .
\end{align*}
\end{cor}

\begin{rmk}
In addition to such norm bounds, \cref{thm: Main thm 1} gives detailed information of the decay rate of the propagator along certain trajectories. For example, an observer moving along a curve $n = t^{1/2}$ sees an overall decay rate of $t^{-1/2}$ from the third estimate.
\end{rmk}

\begin{rmk}

Let $\mathscr{P}_{\pm}(H)$ denote the projections onto the positive and negative part of the continuous spectral subspace of $H$ respectively. Note that  $\mathscr{P}_{+}(H) + \mathscr{P}_{-}(H) = \mathscr{P}_{ac}(H)$. It is well-known that the SSH Hamiltonian has chiral symmetry \cite[Section~1.4]{Asboth_2016}, i.e., there exists a Hermitian and unitary operator $\Gamma$ such that $\Gamma H \Gamma = - H$. Consequently, the propagator acting on data spectrally supported on the negative part of the essential spectrum, $e^{-iHt} \mathscr{P}_-(H) f $, obeys the same decay rate estimates as $e^{-iHt} \mathscr{P}_+(H) f $. Hence in order to prove \cref{thm: Main thm 1}, it suffices to estimate  $e^{-iHt} \mathscr{P}_+(H) f $. 

\end{rmk}

\subsection{Notation and Definitions}

\begin{definition}\label{def: FL transform}
The Fourier-Laplace transform of a vector $\psi \in \ell^2 (\mathbb{N}_0; \mathbb{C}^2)$  and its inverse are given by \begin{equation}
    \Tilde{\psi}(q) = \sum_{m \geq 0} e^{-imq} \psi_m, \, \qquad 
    \psi_n = \frac{1}{2\pi} \int_{- \pi}^{\pi} e^{inq} \Tilde{\psi}(q)\, dq \, .
\end{equation}
\end{definition}

\begin{definition}\label{def: right-shift operator}
Let $S$ denote the right-shift operator $S : \ell^2(\mathbb{N}_0; \mathbb{C}) \rightarrow  \ell^2(\mathbb{N}_0; \mathbb{C}) $ defined by 
\begin{align*}
    S(x_0, x_1, x_2, \dots) := (0, x_0, x_1, \dots)  \, .
\end{align*}
Note that the right-shift operator has the following property \begin{align}\label{eq: S shift}
    e^{-iq} \tilde{f}(q) = \widetilde{S f}(q)\, .
\end{align}

\end{definition}

\begin{definition}
For $n = 0, 1$ and $z \in \mathbb{C}\backslash \sigma_{ess}(H)$, define
\begin{align}
    J(n;z) &:=  \frac{1}{2\pi} \int\limits_{-\pi}^{\pi}  \frac{e^{inq}}{ k^2(q - \varphi) - z^2} \,dq\, \label{def: J integral} \\
     K(n;z) &:=  \frac{1}{2\pi} \int\limits_{-\pi}^{\pi}  \frac{e^{inq} h(q)}{k^2(q - \varphi) - z^2} \, dq \, .\label{def: K integral}
\end{align}
Note that $ K(n;z) = \g_1 J(n;z) + \g_2 J(n-1;z)$. 
\end{definition}

\begin{definition}\label{def: U}
For $z \in \mathbb{C}\backslash \sigma_{ess}(H)$, define $U(\cdot\; ; z) \in L^2\l( S^1; \mathbb{C}^{2 \times 2} \r)$ via
\begin{align}
    U(q; z) = \frac{\overline{\g_2} e^{iq}}{z^2 -  k^2(q- \varphi)} \begin{bmatrix}
          h(q) & 0 \\ z & 0
    \end{bmatrix} \, .
\end{align}
\end{definition}

\begin{definition}\label{def: V}
$V$ is the averaging operator, 
\begin{align}
\begin{split}
    V: L^2( S^1; \mathbb{C}^2) &\rightarrow  \mathbb{C}^2 \\
    \begin{bmatrix}
         x \\ y
    \end{bmatrix} &\mapsto  \begin{bmatrix}
         \frac{1}{2\pi}  \int_{-\pi}^{\pi} x(q) \, dq  \\[3pt]  \frac{1}{2\pi}  \int_{-\pi}^{\pi} y(q) \, dq
    \end{bmatrix}\, .
\end{split}
\end{align}
\end{definition}

\begin{definition}\label{def: Gamma curve}
Let $\Gamma = \Gamma_1 \cup \Gamma_2 \cup \Gamma_3 \cup \Gamma_4 $,  where the oriented curves $\Gamma_i$ are given by 
\begin{align*}
    \Gamma_1 &= (-\pi - i \infty, -\pi)\, ,  & \Gamma_2 &= [-\pi, 0]  \, ,\\
    \Gamma_3 &= [0, \pi] \, , & \Gamma_4 &= (\pi, \pi - i \infty) \, .
\end{align*}
Furthermore, let \begin{align*}
    \Gamma_c = ( - i \infty, 0)\, . 
\end{align*}
See Figure \ref{fig: 2 panel} for a visualization. 
\end{definition}

\begin{definition}\label{def: q}
Let $D$ be the lower half-strip \begin{align*}
    D : = \{ x + iy \mid -\pi < x < \pi , y < 0\} \, . 
\end{align*}
Viewed as a set, $\Gamma_4 := \{\pi + iy \mid y < 0\}$ is the right boundary of $D$, see Figure \ref{fig: 2 panel}. At times which will be clear from context, we will view $\Gamma_4$ as an oriented curve.
Define $\q(\omega)$ to be the unique solution in $D \cup \Gamma_4$ of the equation \begin{align*}
    k^2(q)  = \omega, \quad \omega \in \mathbb{C}\backslash [\g_-^2 ,\g_+^2 ] \, . 
\end{align*} 

\begin{rmk}
    The existence and uniqueness of $\q(\omega)$ is proven in Proposition \ref{prop: limit of q}. 
\end{rmk}
\end{definition}

\section{Overview of the Proof of the Main Theorem, \cref{thm: Main thm 1}}\label{sec: roadmap}
The holomorphic functional calculus for bounded operators allows us to write down the propagator in terms of the Dunford integral \cite[Definition~5.2.11] {Buhler_Functional},
\begin{align*}
    \l[e^{-iHt} \mathscr{P}_+ f \r]_n = \l[ \frac{1}{2\pi i} \oint \limits_C e^{-i\lambda t} R(z) \, dz \, f \r]_n \, ,
\end{align*}
where we take $C$ to be a simple closed loop in $\mathbb{C}$ enclosing the interval $[\g_-, \g_+]$ and no other part of $\sigma(H)$. Taking a limit of smaller and smaller loops, we get 
\begin{align*}
    \l[e^{-iHt} \mathscr{P}_+ f \r]_n &= \lim_{\ep \rightarrow 0^+} \l[ \frac{1}{2\pi i} \int \limits_{\g_-}^{\g_+} e^{-i\lambda t} [R(\lambda + i \ep) - R(\lambda - i \ep)] \, d\lambda \, f \r]_n \\
    &= \lim_{\ep \rightarrow 0^+} \frac{1}{4\pi^2 i} \int \limits_{\g_-}^{\g_+} e^{-i\lambda t} \int \limits_{-\pi}^{\pi} e^{inq} \l( \l[ \tilde{R}(\lambda + i \ep) - \tilde{R}(\lambda - i \ep) \r] \tilde{f} \,\r)(q) \, dq \, d\lambda \, ,
\end{align*}
where $\tilde{R}$ is the Fourier-Laplace transform of the resolvent, see \cref{def: FL transform}, and satisfies 
\begin{align*}
    [R(z) f]_n = \frac{1}{2\pi} \int \limits_{-\pi}^{\pi} e^{inq} \tilde{R}(z) \tilde{f}(q)\, dq \, , \qquad  \text{ for all }~~~ z \notin \sigma(H)  \,.
\end{align*}

We begin by computing an explicit expression for $\tilde{R}(\lambda \pm i\ep)$.
In Section \ref{sec: Computing the Resolvent Part 1}, we show that the Fourier-Laplace transform of \begin{align*}
    (H  - z I)\psi  = f \, , \qquad  \text{ for }~~~ z \notin \sigma(H)  \,,
\end{align*} is given by
\begin{align}\label{Strategy: FL Transform}
    \begin{bmatrix}
        -z & h(q) \\ \overline{h( \overline{q})} & -z
    \end{bmatrix}  (I + UV) \Tilde{\psi}(q) = \tilde{f}(q) \, , 
\end{align}
where $U$ is a rank-one multiplication operator and $V$ is a rank-one integral operator, see \cref{def: U,def: V} respectively. Comparing with the full-line analog, \eqref{eq: DFT of bulk operator}, we see that \eqref{Strategy: FL Transform} is a rank-one perturbation of the bulk problem. Using 
the Sherman-Morrison-Woodbury formula, Theorem \ref{thm: SMW}, we invert \eqref{Strategy: FL Transform} to obtain an expression for the Fourier-Laplace transform of the resolvent. In \cref{thm: FL resolvent} we show that the expression factors into a bulk and edge component, \begin{equation}\label{Strategy: FL resolvent}
\begin{split}
   \tilde{\psi} &= \l[\hat{R}_{bulk}(z) \tilde{f}\,\r](q) -  \l[\tilde{R}_{edge}(z) \tilde{f}\,\r](q) \\
   &:= \l[\hat{R}_{bulk}(z) \tilde{f}\, \r](q) - U(q;z)(I + VU)^{-1} V\l(\l[\hat{R}_{bulk}(z) \tilde{f}\, \r](\cdot)\r) \, .
\end{split}\end{equation}

Furthermore, we show that \eqref{Strategy: FL resolvent} is a multivalued function in $z$. The evaluation of $V\l(\l[\hat{R}_{bulk}(z) \tilde{f}\, \r]\r)$ in \eqref{Strategy: FL resolvent} involves computing the integral of a function with the singular term 
$(k^2(q - \varphi) - z^2 )^{-1}$. A contour deformation argument similar to that of \cite[Section 2]{Komech_2006} relates $V\l(\l[\hat{R}_{bulk}(z) \tilde{f}\, \r](\cdot)\r)$ to $\tilde{f}\l(\q\l(z^2\r) + \varphi \r)$, where $\q\l(z^2\r)$ is the unique pole of $(k^2(q ) - z^2 )^{-1}$ in the lower half-strip $D \cup \Gamma_4$, see \cref{def: q} and \cref{fig: 2 panel}. 
Since $k^2(\cdot)$ is an analytic continuation of the cosine function, up to some scaling and a shift, $\q(\cdot)$
is the analytic continuation of arccosine composed with a quadratic polynomial. 

In \cref{prop: limit of q} we determine some basic properties of $\q(\cdot)$, the most important of which being that for all $\lambda \in \sigma_{ess}(H)$ and $\ep > 0$, 
\begin{align}\label{Strategy: limit of q}
    \q\l((\lambda \pm i 0)^2\r) \equiv  \lim_{\ep \rightarrow 0^+} \q\l((\lambda \pm i \ep)^2\r) = \pm \sgn(\lambda) \cos^{-1}\l(\frac{\lambda^2 - \g_1^2 - |\g_2|^2}{2 \g_1 |\g_2|} \r)\, .
\end{align}
Due to the branch cut of the analytic continuation of arccosine, $\q\l((\lambda \pm i 0)^2\r)$ switches signs when taking the limit from above versus below. 
This property of \eqref{Strategy: limit of q} will introduce a significant challenge to our analysis due to the appearance of an oscillatory integral with a non-integrable singularity, as can be seen in the following subsection.

\subsection{Dunford integral method applied to the bulk Hamiltonian}
In this section we present a toy calculation to demonstrate how one can use a limiting sequence of Dunford integrals and a Fourier representation of the resolvent in order to obtain an expression for the propagator of the bulk Hamiltonian $e^{-i \Hbulk t}$. The same general approach is used in Section \ref{sec: Computing the Propagator} to compute the propagator of the edge Hamiltonian $e^{-i \Hedge t}$. At the end of this toy calculation,  we discuss the main difference with the full calculation in Section \ref{sec: Computing the Propagator} and how the property \eqref{Strategy: limit of q} of $\q(\cdot)$ gives rise to an oscillatory integral with a singular integrand.

To focus on the essential idea behind this calculation, we ignore any constant factors, restrict the domain of integration in the inner integral to $q \in [0, \pi]$ (equivalently, we assume that $\hat{f}(q) = 0$ for all $q < 0$), and assume that we can exchange limits with integration over $\lambda$. Furthermore, we assume that $\g_2 \in \mathbb{R}_+$ and $n = 0$, so that $\varphi = 0$ and $e^{inq} =1$. With such simplifications, we approximate the propagator by
\begin{align*}
    e^{-i \Hbulk t} f 
    &\approx \lim_{\ep \rightarrow 0^+ } \int \limits_{\g_-}^{\g_+} e^{-i \lambda t} 
    \int \limits_{0}^{\pi} \l(\hat{R}_{bulk}(\lambda +  i\ep)\hat{f} - \hat{R}_{bulk}(\lambda -  i\ep)\hat{f}  \,\r) (q) \, dq \, d\lambda \\
    &= \int \limits_{\g_-}^{\g_+} e^{-i \lambda t} \lim_{\ep \rightarrow 0^+ }
    \int \limits_{0}^{\pi} \l(\hat{R}_{bulk}(\lambda +  i\ep)\hat{f} - \hat{R}_{bulk}(\lambda -  i\ep)\hat{f}  \,\r) (q) \, dq \, d\lambda \, .
\end{align*}
Next, we compute each term in the difference, 
\begin{align*}
    &\quad \int \limits_{\g_-}^{\g_+} e^{-i \lambda t} \lim_{\ep \rightarrow 0^+ }
    \int \limits_{0}^{\pi} \l(\hat{R}_{bulk}(\lambda \pm  i\ep)  \hat{f}\,\r)(q) \, dq \, d\lambda \\ 
    &\approx \int \limits_{\g_-}^{\g_+} e^{-i \lambda t} \lim_{\ep \rightarrow 0^+ }
    \int \limits_{0}^{\pi} \frac{1}{k^2(q) - (\lambda \pm i \ep)^2} \begin{bmatrix}
        \lambda & h(q) \\
        \overline{h(q)} & \lambda
    \end{bmatrix}  \hat{f}(q) \, dq \, d \lambda \qquad \qquad \qquad \quad [\text{Equation \eqref{eq: bulk resolvent}}]  \numberthis \label{Strategy: comparison}\\
    &\approx \int \limits_{\g_-}^{\g_+} \frac{e^{-i \lambda t}}{\lambda} \lim_{\ep \rightarrow 0^+ }
    \int \limits_{0}^{\pi} \l(\frac{1}{k(q) - (\lambda \pm i \ep)} - \frac{1}{k(q) + (\lambda \pm i \ep)} \r) \begin{bmatrix}
        \lambda & h(q) \\
        \overline{h(q)} & \lambda
    \end{bmatrix}  \hat{f}(q) \, dq \, d \lambda \, .
\end{align*}
After a change of variables $q \mapsto k(q)$, we have 
\begin{align}
    \int \limits_{\g_-}^{\g_+} e^{-i \lambda t} \lim_{\ep \rightarrow 0^+ }
    \int \limits_{0}^{\pi} \l(\hat{R}_{bulk}(\lambda \pm  i\ep)  \hat{f}\,\r)(q) \, dq \, d\lambda \approx \int \limits_{\g_-}^{\g_+} \frac{e^{-i \lambda t}}{\lambda} \lim_{\ep \rightarrow 0^+ }
    \int \limits_{\g_-}^{\g_+} \l(\frac{1}{k - (\lambda \pm i \ep)} - \frac{1}{k + (\lambda \pm i \ep)} \r)g(k) \, dk \, d \lambda \,, \label{Strategy: pre-SP}\end{align} 
where
\begin{align*}
     g(k) &:= \begin{bmatrix}
    \lambda & h(\qzk \\ \overline{h(\qzk)} & \lambda
    \end{bmatrix} \frac{k \hat{f}(\qzk) }{\g_1 |\g_2| \sqrt{ 1 - \eta^2(k)}}\, , \\ 
    \qzk &:= \q((k + i 0)^2)\, , \qquad \text{ and }  \eta(k) := \frac{k^2 - \g_1^2 - |\g_2|^2}{2 \g_1 |\g_2|} \, . 
\end{align*} 
Applying the Sokhotski-Plemelj theorem, see \cref{thm: Sokhotski-Plemelj}, to \eqref{Strategy: pre-SP} gives
\begin{align}\label{Strategy: applying Sokhotski-Plemelj}
     \int \limits_{\g_-}^{\g_+} e^{-i \lambda t} \lim_{\ep \rightarrow 0^+ }
    \int \limits_{0}^{\pi} \l(\hat{R}_{bulk}(\lambda \pm  i\ep)  \tilde{f}\,\r)(q) \, dq \, d\lambda \approx \int \limits_{\g_-}^{\g_+} \frac{e^{-i \lambda t}}{\lambda}  \l[ \fint \limits_{\g_-}^{\g_+} \frac{g(k)}{k - \lambda} \, dk \pm i \pi g(\lambda) - \int \limits_{\g_-}^{\g_+} \frac{g(k)}{k+\lambda} \, dk \r] \, d\lambda \, . 
\end{align}
When considering the difference of the limiting resolvent above and below the real axis, we get a cancellation of the two double integrals, leaving
\begin{align*}
    \int \limits_{\g_-}^{\g_+} e^{-i \lambda t} \lim_{\ep \rightarrow 0^+}  \int \limits_{-\pi}^{\pi} e^{inq} \l([\hat{R}_{bulk}(\lambda +  i\ep) - \hat{R}_{bulk}(\lambda - i\ep)] \tilde{f}\,\r)(q) \, dq \, d\lambda 
    \approx 2\pi i \int \limits_{\g_-}^{\g_+} \frac{g(\lambda) e^{-i \lambda t}}{\lambda} \, d\lambda \, . 
\end{align*}

\subsubsection{Comparison with edge Hamiltonian}
In Section \ref{sec: Computing the Propagator} we go through the same calculation with the edge component of the resolvent, $\tilde{R}_{edge}(z)$. In contrast with \eqref{Strategy: comparison}, which shows that 
\begin{align*}
    \int \limits_{\g_-}^{\g_+} e^{-i \lambda t} \lim_{\ep \rightarrow 0^+ }
    \int \limits_{0}^{\pi} \l(\hat{R}_{bulk}(\lambda \pm  i\ep)  \hat{f}\,\r)(q) \, dq \, d\lambda
\end{align*}
depends on the data $f$ through $\hat{f}(q)$, 
\cref{prop: handwavy approx} shows that 
\begin{align*}
    \int \limits_{\g_-}^{\g_+} e^{-i \lambda t} \lim_{\ep \rightarrow 0^+}  \int_{-\pi}^{\pi} e^{inq} \l[\tilde{R}_{edge}(\lambda \pm  i\ep)  \tilde{f}\,\r] (q) \, dq \, d\lambda \, .
\end{align*}
depends on the data $f$ through $\tilde{f}(\pm \qzl)$. Crucially, $\tilde{f}(\pm \qzl)$ is a function of $\lambda$, and hence can be pulled out of the inner integral.

After using the change of variables $q \mapsto k(q)$, see \cref{lem: Change of Variables}, and applying the Sokhotski-Plemelj theorem, a rough analog of \eqref{Strategy: applying Sokhotski-Plemelj} for the edge resolvent is
\begin{align}\label{Strategy: edge propagator}
    \int \limits_{\g_-}^{\g_+} e^{-i \lambda t} \lim_{\ep \rightarrow 0^+ }
    \int \limits_{0}^{\pi} \l(\tilde{R}_{edge}(\lambda \pm  i\ep)  \tilde{f}\,\r)(q) \, dq \, d\lambda \approx \int \limits_{\g_-}^{\g_+} \frac{e^{-i \lambda t}}{\lambda}  \l[ \fint \limits_{\g_-}^{\g_+} \frac{G(k)}{k - \lambda} \, dk \pm i \pi G(\lambda) - \int \limits_{\g_-}^{\g_+} \frac{G(k)}{k+\lambda} \, dk \r] \tilde{f}(\pm \qzl) \, d\lambda \, ,
\end{align}
where $G(\cdot)$ is a function which no longer has any dependence on $f$. Now when we compute the difference of the limiting resolvent above and below the real axis, we no longer get any cancellation because generically, 
\begin{align*}
    \tilde{f}(\qzl) - \tilde{f}(- \qzl) \neq 0 \, . 
\end{align*}
This lack of cancellation means that the expression for the propagator of the edge Hamiltonian involves oscillatory single integrals, oscillatory double integrals, and oscillatory double integrals with non-integrable singularities. For the full expression of the propagator, see \cref{thm: propagator of SSH}.

\subsection{Spatially uniform vs spatially weighted decay rate and precise dependence of constants on hopping coefficients $\g_1$ and $\g_2$}

In Theorem \ref{thm: propagator of SSH},  we obtain an explicit representation of $\l[e^{-iHt}\mathscr{P}_+ f\r]_n$ as a sum of oscillatory integrals plus the propagator of the bulk Hamiltonian. Although the expression is quite long, all of the terms in the expression are small variations of one of three types of oscillatory integrals, which can be identified in the right hand side of \eqref{Strategy: edge propagator}. The first type behaves like 
\begin{align}\label{Strategy: Type I}
    \int \limits_0^{\pi} e^{-ik(y)t} e^{iny} \tilde{f}(y) \, dy \, .
\end{align}

We estimate \eqref{Strategy: Type I} using Van der Corput's lemma, see \cref{thm: VDC}. We have two reasonable options for our choice of phase function, 
\begin{align*}
    \phi_1(y) := k(y)  - \frac{ny}{t} 
    \qquad \text{ or } \qquad \phi_2(y) := k(y) \, . 
\end{align*}
In \cref{lem: k(y) facts}, we show that $k'(y)$ and $k'''(y)$ only vanish at the endpoints, while $k''(y)$ only vanishes at a single point in the interior. Hence by partitioning the region of integration into left, middle, and right subintervals, we ensure that some derivative of the phase function is bounded away from zero on each subinterval. Taking the phase function to be $\phi_1$, we get that \eqref{Strategy: Type I} decays like $c(\g_1, |\g_2|) t^{-1/3}$. Using $\phi_2$ as the phase function means that our amplitude function is $e^{iny} \tilde{f}(y)$, which introduces a spatial weight to the decay rate. The benefit of doing so is that the time dependence can be improved to $\log(t)t^{-1/2}$ or $t^{-1/2}$, depending on the rate of spatial decay of $f$. In \cref{prop: application of VDC}, we explicitly determine the dependence of $c(\g_1, |\g_2|)$ on the hopping coefficients by carefully choosing the partition of the region of integration.

\subsection{Estimating the oscillatory double integrals}
The second and third types of oscillatory integrals are 
\begin{align*}
    \int \limits_0^{\pi} \cos(ny) \int \limits_{\g_-}^{\g_+} e^{-i \lambda t} \frac{\tilde{f}(k^{-1}(\lambda))}{\lambda + k(y)}\, d\lambda \, dy  \quad \text{ and } \quad \int \limits_0^{\pi} \cos(ny) \fint \limits_{\g_-}^{\g_+} e^{-i \lambda t} \frac{\tilde{f}(k^{-1}(\lambda))}{\lambda - k(y)}\, d\lambda \, dy \, , 
\end{align*}
which we call Type II and Type III, respectively. 
The analysis of both of these terms starts the same: 
by writing $\tilde{f}(k^{-1}(\lambda)) = \tilde{f}(y) + \l( \tilde{f}(k^{-1}(\lambda)) - \tilde{f}(y) \r) $ we partition each integral into two parts. We define the parts as 
\begin{align*}
    \text{Type IIa} &= \int \limits_0^{\pi} \cos(ny) \tilde{f}(y) \int \limits_{\g_-}^{\g_+}  \frac{e^{-i \lambda t}}{\lambda + k(y)}\, d\lambda \, dy \\
    \text{and Type IIb} &= \int \limits_0^{\pi} \cos(ny) \int \limits_{\g_-}^{\g_+} e^{-i \lambda t} \frac{\tilde{f}(k^{-1}(\lambda)) - \tilde{f}(y)}{\lambda + k(y)}\, d\lambda \, dy\, ,
\end{align*}
with Type IIIa and Type IIIb defined analogously. The benefit of this partition is that for Type IIa and Type IIIa, $\tilde{f}(y)$ can be pulled out of the inner integral, and for Type IIb and Type IIIb, we can use the 1/2 H\"older continuity of $\tilde{f}(k^{-1}(\cdot))$, proven in \cref{prop: properties of F}, to estimate $ \tilde{f}(k^{-1}(\lambda)) - \tilde{f}(y)$. 

From here, the strategy for proving Theorem \ref{thm: Main thm 1} is as follows: 
\begin{itemize}
    \item In \cref{lem: Type IIa estimate,lem: Type IIb estimate} we estimate the Type II term using techniques that are similar to those used for the Type III term (see below). But the technical details are simpler due to the lack of a singularity in the integrand. 
    \item In Section \ref{sec: Type IIIa Decay Rates} we recenter the singularity of the Type IIIa term at the origin with a change of variables, and then use the contour deformation pictured in \cref{fig: Type IIIa contour}. 
    \item In Section \ref{sec: Type IIIb Uniform Decay} we use a change of variables to pull out an oscillatory term that looks like $e^{-iu^{\alpha} t}$ from the Type IIIb term. We then use an extension of Van der Corput's lemma, \cref{Dewez: fractional phase}, to obtain a spatially uniform $t^{-1/\alpha}$ decay rate for all real $\alpha > 2$ in \cref{prop: Type IIIb uniform decay estimate}. By optimizing $\alpha$, we improve this result to a spatially uniform $\log(t) t^{-1/2}$ decay rate in \cref{cor: Type IIIb uniform decay estimate}.
    \item In Section \ref{sec: Type IIIb Local Decay} we assume further regularity of the data, namely $f \in \ell_2^1( \mathbb{N}_0; \mathbb{C}^2)$, in order to apply \cref{prop: application of VDC} directly to the Type IIIb term and obtain a spatially weighted $t^{-1/2}$ decay rate. 
\end{itemize}

\section{Computing the Resolvent }\label{sec: Computing the Resolvent} 
\subsection{Fourier-Laplace Transform and Inverting a Rank-One Perturbation}\label{sec: Computing the Resolvent Part 1}
The $n$'th entry of \begin{align}\label{eq: Resolvent equation}
    (H  - z I)\psi  = f  \text{ for } z \notin \sigma(H) \, ,
\end{align} is given by \begin{align*} 
    \begin{bmatrix}
        \g_1 \psi_n^B + \g_2  \psi_{n-1}^B - z \psi_n^A \\
        \g_1 \psi_n^A + \overline{\g_2} \psi_{n+1}^A - z \psi_n^B
    \end{bmatrix} = 
    \begin{bmatrix}
        f_n^A \\ f_n^B
    \end{bmatrix} \, .
\end{align*}
In order to take the  Fourier-Laplace transform as defined in \eqref{def: FL transform} of \eqref{eq: Resolvent equation},  note that in the first row \begin{align*}
    \sum_{n \geq 0} e^{-inq} \psi_{n-1}^B  = e^{-iq} \sum_{n \geq 0} e^{-i(n-1)q} \psi_{n-1}^B = e^{-iq} \sum_{n \geq -1} e^{-inq} \psi_n^B = e^{-iq} \sum_{n \geq 0}  e^{-inq} \psi_n^B  = e^{-iq} \tilde{\psi}^B \, ,
 \end{align*}
 since the boundary condition we have imposed ensures that $\psi_{-1}^B = 0$. In the second row, \begin{align*}
     \sum_{n \geq 0} e^{-inq} \psi_{n+1}^A = e^{iq} \sum_{n \geq 0} e^{-i (n+1) q} \psi_{n+1}^A = e^{iq} \sum_{n \geq 1} e^{-inq} \psi_n^A = e^{iq} \l( \tilde{\psi}^A - \psi_0^A \r) \, .
 \end{align*}
 Hence taking the Fourier-Laplace transform of \eqref{eq: Resolvent equation} yields \begin{equation}\label{eq: FL transform of resolvent form 1}
\begin{split} 
    (\g_1 + \g_2 e^{-iq}) \Tilde{\psi}^B (q)  - z \Tilde{\psi}^A (q) = \Tilde{f}^A (q) \, , \\
    (\g_1 + \overline{\g_2} e^{iq}) \Tilde{\psi}^A (q)  - \overline{\g_2} e^{iq} \psi_0^A - z \Tilde{\psi}^B (q) = \Tilde{f}^B (q) \, .
\end{split}
\end{equation}
Through some elementary manipulations we now show that $\tilde{f}$ can be written as a rank-one perturbation of the bulk operator \eqref{eq: DFT of bulk operator}.  
Using  equation \eqref{eq: FL transform of resolvent form 1} and the definition  $h(q) := \g_1 + \g_2 e^{-iq}$, we compute 
\begin{subequations}\label{eq: computing FL transform}
\begin{align}
    \Tilde{f} &= \begin{bmatrix}
        -z & h(q) \\ \overline{h( \overline{q})} & -z
    \end{bmatrix} \Tilde{\psi} - \overline{\g_2} e^{iq} \begin{bmatrix}
        0 \\ \psi_0^A
    \end{bmatrix} \label{eq: computing FL transform line a}\\
    &= \begin{bmatrix}
        -z & h(q) \\ \overline{h( \overline{q})} & -z
    \end{bmatrix} \l( \Tilde{\psi} - \overline{\g_2} e^{iq} \begin{bmatrix}
        -z & h(q) \\ \overline{h( \overline{q})} & -z
    \end{bmatrix}^{-1} \begin{bmatrix}
        0 \\ \psi_0^A
    \end{bmatrix} \r) \nonumber \\
    &= \begin{bmatrix}
        -z & h(q) \\ \overline{h( \overline{q})} & -z
    \end{bmatrix} \l(\Tilde{\psi} + \frac{\overline{\g_2} e^{iq}}{z^2 - k^2(q- \varphi)} \begin{bmatrix}
        z & h(q) \\ \overline{h( \overline{q})} & z
    \end{bmatrix} \begin{bmatrix}
        0 \\ \psi_0^A
    \end{bmatrix}  \r) \nonumber \\
    &= \begin{bmatrix}
        -z & h(q) \\ \overline{h( \overline{q})} & -z
    \end{bmatrix} \l(\Tilde{\psi} + \frac{\overline{\g_2} e^{iq}}{z^2 - k^2(q- \varphi)} \begin{bmatrix}
        h(q) & 0 \\ z & 0
    \end{bmatrix} \begin{bmatrix}
        \psi_0^A \\ \psi_0^B
    \end{bmatrix}  \r)  \nonumber \\
    &= \begin{bmatrix}
        -z & h(q) \\ \overline{h( \overline{q})} & -z
    \end{bmatrix}  (I + UV) \Tilde{\psi} \,, 
    \label{eq: computing FL transform line b} 
\end{align}
\end{subequations}
where 
\begin{align*}
    U(q; z) := \frac{\overline{\g_2} e^{iq}}{z^2 -  k^2(q- \varphi)} \begin{bmatrix}
          h(q) & 0 \\ z & 0
    \end{bmatrix}
\end{align*}
and $V$ is the averaging operator $ V \tilde{\psi} = (2\pi)^{-1} \int_{-\pi}^{\pi} \tilde{\psi}(q) \, dq$ (see \cref{def: U,def: V}).

Furthermore, we may write 
\begin{equation}\label{eq: FL transform of resolvent form 2}
    (\Tilde{H} - zI) \Tilde{\psi} = \begin{bmatrix}
        -z & h(q) \\
        \overline{h( \overline{q})} & -z
    \end{bmatrix} (I + UV)   \Tilde{\psi} =
    \Tilde{f} \, ,
\end{equation}
where $\Tilde{H}$ is defined by \begin{equation}\label{def: H tilde}
    \Tilde{H} \Tilde{\psi} := \begin{bmatrix}
        h(q) \Tilde{\psi}^B \\
        \overline{h( \overline{q})} \Tilde{\psi}^A - \g_2 e^{iq} \psi_0^A
    \end{bmatrix} =  \l(\begin{bmatrix}
        0 & h(q) \\ \overline{h( \overline{q})} & 0
    \end{bmatrix} + \begin{bmatrix}
        -z & h(q) \\ \overline{h( \overline{q})} & -z
    \end{bmatrix} UV \r) \Tilde{\psi} \, .
\end{equation}
Note that the first equality in equation (\ref{def: H tilde}) follows from \eqref{eq: computing FL transform line a} and the second equality follows from \eqref{eq: computing FL transform line b}.
Since the matrix $\begin{bmatrix}  -z & h(q) \\ \overline{h( \overline{q})} & -z \end{bmatrix}$ is invertible for $z \notin \sigma_{ess}(H)$, 
we can invert \eqref{eq: FL transform of resolvent form 2} to obtain 
\begin{equation}\label{eq: FL transform of resolvent form 3}
    (I + UV) \begin{bmatrix}
        \Tilde{\psi}^A \\ \Tilde{\psi}^B
    \end{bmatrix} = \frac{1}{ k^2(q- \varphi) - z^2}  \begin{bmatrix}
        z & h(q) \\  \overline{h( \overline{q})} &  z 
    \end{bmatrix} \begin{bmatrix}
        \Tilde{f}^A (q) \\ \Tilde{f}^B (q) 
    \end{bmatrix}\, .
\end{equation}
Note that the right hand side of \eqref{eq: FL transform of resolvent form 3} is just the discrete Fourier transform of the resolvent as computed in \eqref{eq: bulk resolvent}. So
\begin{align}\label{eq: FL transform of resolvent form 4}
    (I + UV) \tilde{\psi} = \l[\hat{R}_{bulk}(z) \tilde{f}\, \r](q)\, .
\end{align}

The operator $(I+UV)$ in \eqref{eq: FL transform of resolvent form 4} is a rank-one perturbation of the identity, which  can be inverted with the SMW formula (Theorem \ref{thm: SMW}) to obtain \begin{equation}\label{eq: FL transform of resolvent SMW}
   \tilde{\psi} = \l[\hat{R}_{bulk}(z) \tilde{f}\, \r](q) - U(I + VU)^{-1} V\l(\l[\hat{R}_{bulk}(z) \tilde{f}\, \r](\cdot)\r) \, .
\end{equation}

\begin{rmk}\label{rmk: soft argument for inversion}
It is not immediately apparent that the right hand side of equation \eqref{eq: FL transform of resolvent SMW} will be in the image of $\ell^2(\mathbb{N}_0)$ under the Fourier-Laplace transform.  
However, since $z \notin \sigma(H)$, we know that there exists a unique solution $\psi \in \ell^2(\mathbb{N}_0)$ to the equation  \begin{align*}
    (H - zI) u = f \, .
\end{align*}
Then we checked in equations \eqref{eq: FL transform of resolvent form 1}, \eqref{eq: computing FL transform}, \eqref{eq: FL transform of resolvent form 2} that $\Tilde{\psi} \in L^2([-\pi, \pi]; \mathbb{C}^2)$ is a solution to \begin{align}\label{eq: FL of resolvent eq}
    (\Tilde{H} - zI) v = \Tilde{f} \, .
\end{align}
By  the SMW formula (Theorem \ref{thm: SMW})  and equation \eqref{eq: FL transform of resolvent form 2}, we know that $\Tilde{H}- zI$ is invertible for $z \notin \sigma(H)$. To see this, note that the matrix $\begin{bmatrix}
    z & h(q) \\ \overline{h( \overline{q})} & z
\end{bmatrix}$ is only singular for $z \in \sigma_{ess}(H)$ and by Proposition \ref{prop: det(I + VU)}, $I + UV $ is  singular if and only if $z = 0$ and $\g_1 > |\g_2|$. 
Thus for $z \notin \sigma(H)$, the equation \eqref{eq: FL of resolvent eq} must have a unique solution, namely, $\Tilde{\psi}$. Since the left hand side of equation \eqref{eq: FL transform of resolvent SMW} is in the image of $\ell^2(\mathbb{N}_0)$ under the Fourier-Laplace transform, the right hand side must be as well. 
\end{rmk}

\begin{rmk}\label{rmk: viewing as a subspace of ell(Z)}
We could alternatively view $\ell^2(\mathbb{N}_0)$ as an embedded subspace of $\ell^2(\mathbb{Z})$ via
\begin{align*}
    \psi \in \ell^2(\mathbb{N}_0) \iff \psi \in \ell^2(\mathbb{Z}) \text{ and } \psi_n = 0 \quad \forall n \leq -1 \, . 
\end{align*}
In such case, as an alternative to \cref{rmk: soft argument for inversion}, we could show that the right hand side of \eqref{eq: FL transform of resolvent SMW} is in the image of $\ell^2(\mathbb{N}_0)$ by using the fact that the resolvent $\tilde{\psi}$ is $2\pi$-periodic and analytic. Roughly, the idea is that one constructs a rectangular contour deformation with the top edge being $[-\pi, \pi]$ and bottom edge $[-\pi - i R, \pi - i R]$. The side edges cancel each other out after integrating, and by sending $R \rightarrow \infty$, one can show that $\psi_n = 0$ for all $n \leq -1$.     
\end{rmk}

\subsection{Evaluation using Contour Integrals}\label{sec: Computing the Resolvent Part 2}
In this subsection we evaluate the right hand side of \eqref{eq: FL transform of resolvent SMW} in order to compute the Fourier-Laplace transform of the resolvent, which is the content of \cref{thm: FL resolvent}. Our method involves frequent use of contour integration (see \cref{fig: 2 panel} for our choice of contour) around the poles of $1/\l(k^2(q) - z^2\r)$. These poles, see \cref{def: q}, are denoted by $\q\l(z^2\r)$ and are studied in the following proposition. We note that parts (1) and (2) of \cref{prop: limit of q} are closely related to \cite[Section 2]{Komech_2006}.

Recall that the essential spectrum of the Hamiltonian $H$ is given by \begin{align*}
    \sigma_{ess}(H) = [-\g_+, - \g_-] \cup [\g_- , \g_+] 
\end{align*}
and that $q \mapsto k(q- \varphi)$ is the dispersion relation of the bulk Hamiltonian $H_{\mathbb{Z}}$.

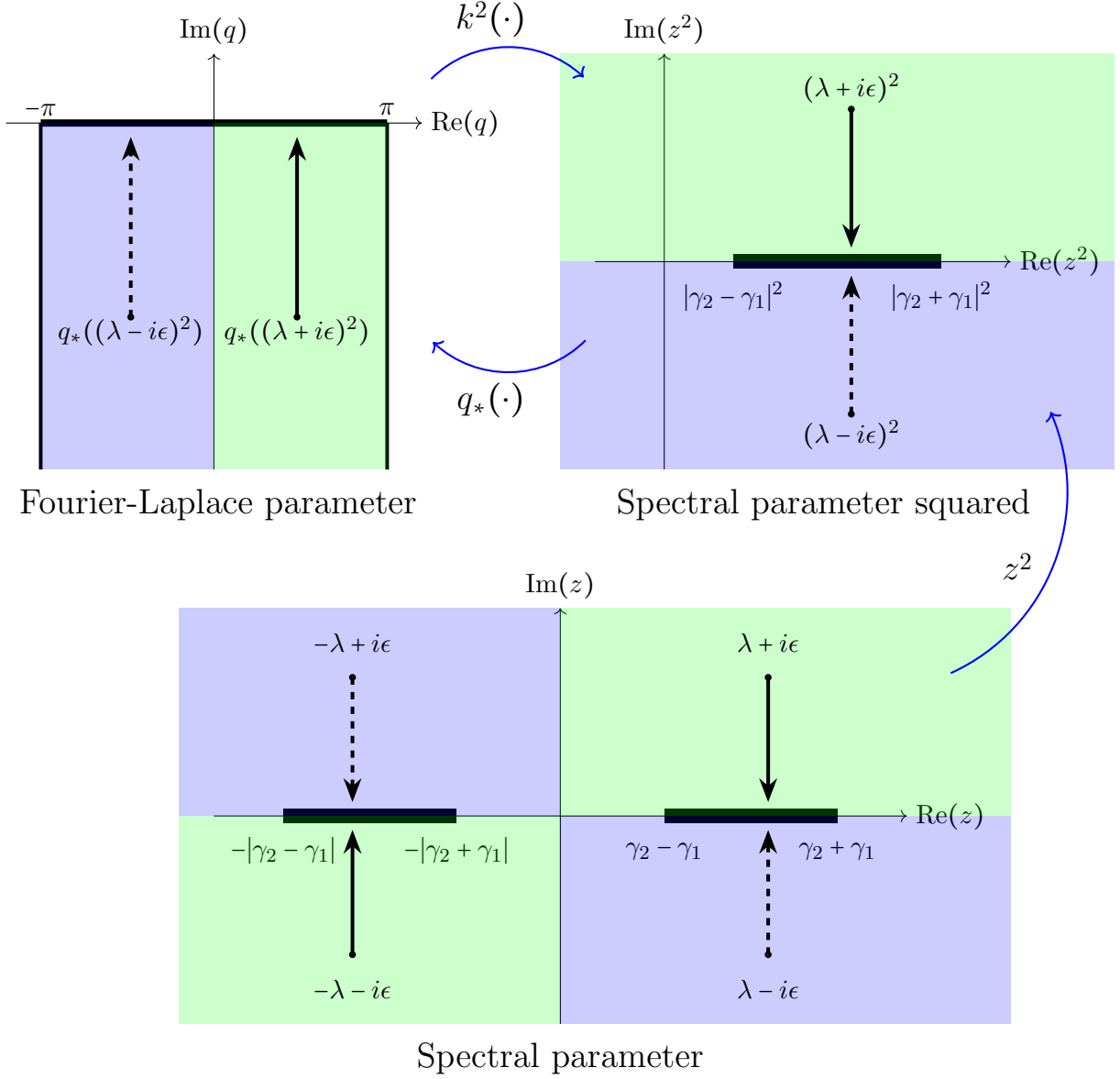
\begin{figure}[ht]
\begin{tikzpicture}
\draw[->] (-3,2) -- (3,2) node[right] {Re$(q)$};
\draw[->] (0,-3) -- (0,3) node[above] {Im$(q)$};
\draw[-, line width=1.5pt] (-2.5,2) -- (-2.5, -3);
\draw[-, line width=1.5pt] (2.5, 2) -- (2.5, -3);
\draw[-, line width=3pt] (-2.5,2) -- (2.5, 2);
\node at (0,-3.5) {{ \Large Fourier-Laplace parameter}};
\node at (-2.5, 2.2) {$-\pi$};
\node at (2.5, 2.2) {$\pi$};
\fill[blue, opacity=0.2] (-2.5,-3) rectangle (0, 2);
\fill[green, opacity=0.2] (0,-3) rectangle (2.5, 2);

\node at (1.2, -1) {$ \q((\lambda + i \ep)^2)$};
\draw[->, >=Stealth, line width=1.5pt] (1.2,-0.8) node[circle, fill=black, inner sep=1pt] {} -- (1.2,1.8);
\node at (-1.2, -1) {$ \q((\lambda - i \ep)^2)$};
\draw[->, dashed, >=Stealth, line width=1.5pt] (-1.2,-0.8) node[circle, fill=black, inner sep=1pt] {} -- (-1.2,1.8);

\draw[->] (5.5,0) -- (11.5,0) node[right] {Re$(z^2)$};
\draw[->] (6.5,-3) -- (6.5,3) node[above] {Im$(z^2)$};
\fill[black!100] (7.5,-0.1) rectangle (10.5,0.1);  
\node at (7.5, -0.5) {$|\g_2 - \g_1|^2$};
\node at (10.5, -0.5) {$|\g_2 + \g_1|^2$};
\node at (8.8, -3.5) {{\Large Spectral parameter squared}}; 
\fill[blue, opacity=0.2] (5,-3) rectangle (13, 0);
\fill[green, opacity=0.2] (5,0) rectangle (13, 3);

\node at (9.2, 2.5)  {$(\lambda + i \ep)^2$};
\draw[->, >=Stealth, line width=1.5pt] (9.2,2.2) node[circle, fill=black, inner sep=1pt] {} -- (9.2,0.2);
\node at (9.2, -2.5)  {$(\lambda - i \ep)^2$};
\draw[->, dashed, >=Stealth, line width=1.5pt] (9.2,-2.2) node[circle, fill=black, inner sep=1pt] {} -- (9.2,-0.2);

\draw[->] (0, -8) -- (10, -8) node[right] {Re$(z)$};
\draw[->] (5, -11) -- (5, -5) node[above] {Im$(z)$};
\fill[black!100] (6.5, -8.1) rectangle (9, -7.9); 
\fill[black!100] (1, -8.1) rectangle (3.5, -7.9); 
\node at (6.5, -8.5) {$\g_2 - \g_1$};
\node at (9, -8.5) {$\g_2 +\g_1$};
\node at (1, -8.5) {$-|\g_2 - \g_1|$};
\node at (3.5, -8.5) {$-|\g_2 +\g_1|$};
\node at (5, -11.5) {{\Large Spectral parameter}};
\fill[blue, opacity=0.2] (5,-8) rectangle (11.5, -11);
\fill[green, opacity=0.2] (5,-8) rectangle (11.5, -5);
\fill[green, opacity=0.2] (5,-8) rectangle (-0.5, -11);
\fill[blue, opacity=0.2] (5,-8) rectangle (-0.5, -5);

\node at (8, -5.5) {$\lambda + i \ep$};
\draw[->, >=Stealth, line width=1.5pt] (8,-6) node[circle, fill=black, inner sep=1pt] {} -- (8,-7.8);
\node at (8, -10.5) {$\lambda - i \ep$};
\draw[->, dashed, >=Stealth, line width=1.5pt] (8,-10) node[circle, fill=black, inner sep=1pt] {} -- (8,-8.2); 
\node at (2, -5.5) {$-\lambda + i \ep$};
\draw[->, dashed, >=Stealth, line width=1.5pt] (2,-6) node[circle, fill=black, inner sep=1pt] {} -- (2,-7.8); 
\node at (2, -10.5) {$-\lambda - i \ep$};
\draw[->, >=Stealth, line width=1.5pt] (2,-10) node[circle, fill=black, inner sep=1pt] {} -- (2,-8.2);

\node at (3,2.5) (A1) {};
\node at (3, -1) (A2) {};
\node at (5.5,2.5) (B1) {};
\node at (5.5, -1) (B2) {};

\draw[thick, blue, -{> [sep=1pt]}] (A1) to [bend left=45] (B1);
\node at (4, 3.5) {{\Large $k^2(\cdot)$}};
\draw[thick, blue, -{> [sep=1pt]}] (B2) to [bend left=45] (A2);
\node at (4, -2) {{\Large $\q(\cdot)$}};

\node at (10.5,-6) (C1) {};
\node at (12, -2) (C2) {};
\draw[thick, blue, -{> [sep=1pt]}] (C1) to [bend right = 45] (C2);
\node at (11.6, -4.4) { {\Large $z^2$} };
\end{tikzpicture}
\caption{Relationship between the Fourier-Laplace parameter $q$ in the top left and the spectral parameter $z = \lambda + i \ep$ on the bottom. The function $k^2(\cdot)$ and its inverse $\q(\cdot)$ map between the Fourier-Laplace parameter and the square of the spectral parameter.  }
\label{fig: 3 panel}
\end{figure}

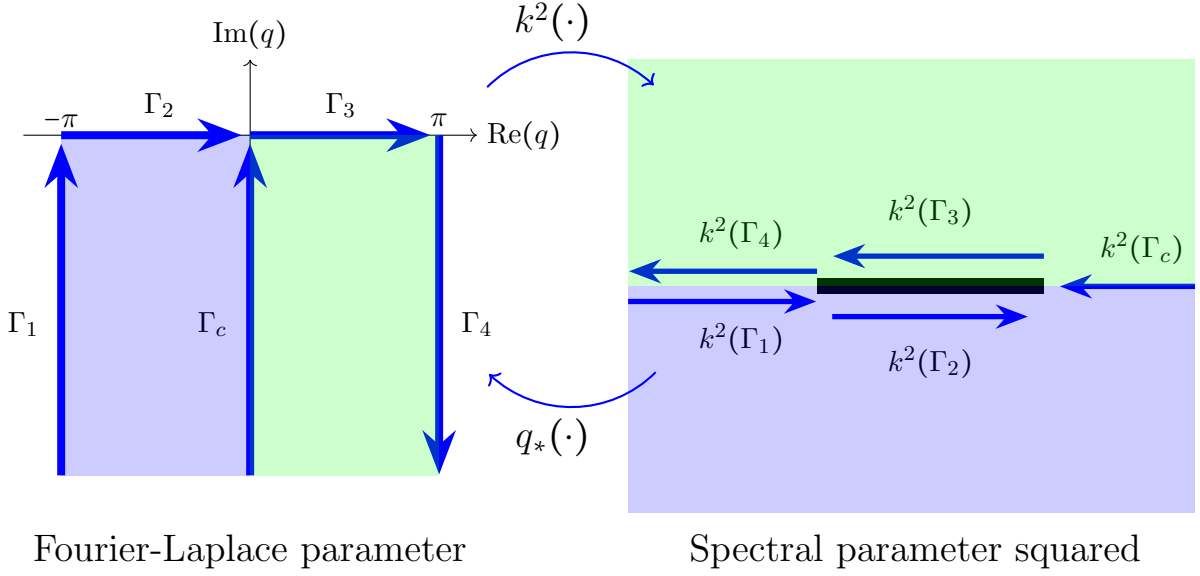
\begin{figure}[ht]
\begin{tikzpicture}
\draw[->] (-3,2) -- (3,2) node[right] {Re$(q)$};   
\draw[->] (0,-2) -- (0,3) node[above] {Im$(q)$};   

\draw[->, thick, blue, >=Stealth, line width=3pt, shorten >=6pt] (-2.5, -2.5) -- (-2.5, 2.1);  
\node at (-3, -0.5) {$\Gamma_1$};
\draw[->, thick, blue, >=Stealth, line width=3pt, shorten >=6pt] (-2.5, 2) -- (0.1, 2); 
\node at (-1.2, 2.4) {$\Gamma_2$};
\draw[->, thick, blue, >=Stealth, line width=3pt, shorten >=6pt] (0, 2) -- (2.6, 2);  
\node at (1.2, 2.4) {$\Gamma_3$};
\draw[->, thick, blue, >=Stealth, line width=3pt, shorten >=6pt] (2.5, 2) -- (2.5, -2.7);  
\node at (3, -0.5) {$\Gamma_4$};
\draw[->, thick, blue, >=Stealth, line width=3pt, shorten >=6pt] (0, -2.5) -- (0, 2.1); 
\node at (-0.5, -0.5) {$\Gamma_c$};

\fill[blue, opacity=0.2] (-2.5,-2.5) rectangle (0, 2);
\fill[green, opacity=0.2] (0,-2.5) rectangle (2.5, 2);

\node at (0,-3.5) {\Large Fourier-Laplace parameter};
\node at (-2.5, 2.2) {$-\pi$};
\node at (2.5, 2.2) {$\pi$};

\fill[black!100] (7.5,-0.1) rectangle (10.5,0.1);  

\draw[->, thick, blue, >=Stealth, line width=2pt, shorten >=6pt] (5, -0.2) -- (7.7, -0.2);  
\node at (6.5, -0.7) {$k^2(\Gamma_1)$};

\draw[->, thick, blue, >=Stealth, line width=2pt, shorten >=6pt]  (7.5, 0.2) -- (4.8, 0.2);  
\node at (6.5, 0.7) {$k^2(\Gamma_4)$};

\draw[->, thick, blue, >=Stealth, line width=2pt, shorten >=6pt] (7.7, -0.4) -- (10.5, -0.4);  
\node at (9, -1) {$k^2(\Gamma_2)$};

\draw[->, thick, blue, >=Stealth, line width=2pt, shorten >=6pt] (10.5, 0.4) -- (7.5, 0.4);  
\node at (9, 1) {$k^2(\Gamma_3)$};

\draw[->, thick, blue, >=Stealth, line width=2pt, shorten >=6pt]  (12.5, 0) -- (10.5, 0);  
\node at (11.8, 0.5) {$k^2(\Gamma_c)$};

\node at (8.8, -3.5) {{\Large Spectral parameter squared}}; 
\fill[blue, opacity=0.2] (5,-3) rectangle (12.5, 0);
\fill[green, opacity=0.2] (5,0) rectangle (12.5, 3);

\node at (3,2.5) (A1) {};
\node at (3, -1) (A2) {};
\node at (5.5,2.5) (B1) {};
\node at (5.5, -1) (B2) {};

\draw[thick, blue, -{> [sep=1pt]}] (A1) to [bend left=45] (B1);
\node at (4, 3.5) {{\Large $k^2(\cdot)$}};
\draw[thick, blue, -{> [sep=1pt]}] (B2) to [bend left=45] (A2);
\node at (4, -2) {{\Large $\q(\cdot)$}};

\end{tikzpicture}
\caption{Edges of the lower half-strip $D$ and their image under $k^2$. }\label{fig: 2 panel}
\end{figure}

\begin{prop}\label{prop: limit of q}
Let $D$ be the lower half-strip and $\Gamma_4$ the right boundary of $D$, as given in \cref{def: q} and \cref{fig: 2 panel}. Then \begin{enumerate}
    \item There exists a unique solution $q \in D \cup \Gamma_4$ to the equation \begin{align*}
    k^2(q) = \omega, \quad \omega \in \mathbb{C}\backslash [\g_-^2 ,\g_+^2 ]\, . 
\end{align*} 
    \item Define this solution to be $\q(\omega)$. As a function, $\q(\cdot):\mathbb{C}\backslash [\g_-^2 ,\g_+^2 ] \rightarrow D \cup \Gamma_4$ is a bijection, and it is biholomorphic onto $D$ when restricted to $\mathbb{C}\backslash (-\infty ,\g_+^2 ]$. See Figure \ref{fig: 3 panel} for a visualization of relationship between $k^2(\cdot)$ and $\q(\cdot)$. 
    \item Define 
\begin{align}\label{eq: qzl}
    \qzl := \text{sgn}(\lambda) \cos^{-1}(\eta(\lambda)),
\end{align}
where $\cos^{-1}(\cdot) : [-1,1] \rightarrow [0, \pi]$ and 
\begin{align}\label{eq: eta var change}
    \eta(\lambda) := \frac{\lambda^2 - \g_1^2 - |\g_2|^2}{2 \g_1 |\g_2|} \, .
\end{align}
Then for all $\lambda \in \sigma_{ess}(H)$, 
\begin{align}\label{eq: limit of q}
    \q((\lambda \pm i0)^2) =  \pm  \qzl. 
\end{align}
\end{enumerate}

\begin{proof}
See Appendix \ref{pf: limit of q} for the proof. 
\end{proof}
\end{prop}

\begin{thm}\label{thm: FL resolvent}
Let $f \in \ell^2(\mathbb{N}_0; \mathbb{C}^2)$ and $z \in \mathbb{C}\backslash \sigma(H)$. Let $\psi$ denote  the unique solution of \begin{align*}
    (H- zI ) \psi = f\ .
\end{align*}
Then,
\begin{equation}\label{eq: resolvent}
        \tilde{\psi}(q) = \l[\hat{R}_{bulk}(z) \tilde{f}\,\r](q) -  \l[\tilde{R}_{edge}(z) \tilde{f}\,\r](q)
\end{equation}
where 
\begin{subequations}
\begin{align}
     \l[\hat{R}_{bulk}(z)  \tilde{f}\,\r](q) &:=   \frac{1}{ k^2(q- \varphi) - z^2}  \begin{bmatrix}
        z & h(q) \\  \overline{h( \overline{q})} &  z 
    \end{bmatrix} \tilde{f}(q) \, , \label{def: Rbulk tilde}  \\
    \l[\tilde{R}_{edge}(z)  \tilde{f}\,\r](q) &:=
    \frac{e^{iq}}{k^2 ( q- \varphi) - z^2}  \begin{bmatrix}
        zh(q)/h(\q(z^2)+ \varphi) & h(q)  \\[5pt] \overline{h\l(\overline{\q(z^2)+ \varphi}\r) } & z
    \end{bmatrix} \widetilde{Sf}(\q(z^2) + \varphi)\, ,\label{def: Redge tilde}
\end{align}
and $S$ is the right-shift operator given in \cref{def: right-shift operator}. 
\end{subequations}
\end{thm}

We have already computed the Fourier-Laplace transform of the resolvent \eqref{eq: FL transform of resolvent SMW} in terms of the $U$ and $V$, see \cref{def: U,def: V}. To prove \cref{thm: FL resolvent}, it remains to compute the second term in the right-hand side of \eqref{eq: FL transform of resolvent SMW}, 
\begin{align*}
    U(I + VU)^{-1} V\l(\l[\hat{R}_{bulk}(z) \tilde{f}\, \r](\cdot)\r). 
\end{align*}
We do this in steps through the following results. First we state a lemma which computes two types of integrals that arise after applying the averaging operator $V$. 
\begin{lem}\label{lem: J integrals}
Let $g \in \ell^2 \l(\mathbb{N}_0; \mathbb{C}\r)$ and define $\q(\cdot)$ as in \cref{def: q}. Then for all $z \in \mathbb{C}\backslash \sigma_{ess}(H)$, 
\begin{align*}
    \frac{1}{2\pi} \int \limits_{-\pi}^{\pi} \frac{\tilde{g}(q)}{k^2 (q - \varphi) - z^2} \, dq &= \frac{i \tilde{g}(\q(z^2) + \varphi)}{2 \g_1 |\g_2| \sin(\q(z^2))} \\
    \text{and} \quad \frac{1}{2\pi} \int \limits_{-\pi}^{\pi} \frac{e^{i q}\tilde{g}(q)}{k^2 (q - \varphi) - z^2} \, dq &= i e^{i \varphi} \l[ \frac{e^{i \q(z^2)} \tilde{g}( \q(z^2) + \varphi)}{2 \g_1 |\g_2| \sin(\q(z^2))} \r] + \frac{g_0 e^{i \varphi}}{\g_1 |\g_2|}\, . 
\end{align*}

\begin{proof}

See Appendix \ref{pf: J integrals} for the proof.  
\end{proof}
\end{lem}

\begin{cor}\label{cor: J integrals}
Let $J$ be given as in \eqref{def: J integral}. Then 
\begin{align*}
    J(0; z) = \frac{i }{2 \g_1 |\g_2| \sin(\q(z^2))} \quad \text{ and } \quad 
    J(1; z) = \frac{ e^{i \varphi}}{2\g_1 |\g_2|} \l[ \frac{ie^{-i \q(z^2)} }{ \sin( \q(z^2)} \r]\, .
\end{align*}

\begin{proof}
Let $g = (1, 0, 0, \dots) \in \ell^2(\mathbb{N}_0; \mathbb{C})$ so that $\tilde{g} \equiv 1$. 
Using Lemma \ref{lem: J integrals},  we compute \begin{align*}
    J(0;z) = \frac{1}{2\pi} \int \limits_{-\pi}^{\pi} \frac{1}{k^2( q- \varphi) - z^2} \, dq = \frac{i }{2 \g_1 |\g_2| \sin(\q(z^2))}
\end{align*}
and 
\begin{align*}
    J(1;z) &= \frac{1}{2\pi} \int \limits_{-\pi}^{\pi} \frac{e^{iq}}{k^2( q- \varphi) - z^2} \, dq = i e^{i \varphi} \l[ \frac{e^{i \q(z^2)} }{2 \g_1 |\g_2| \sin(\q(z^2))} \r] + \frac{ e^{i \varphi}}{\g_1 |\g_2|} \\
    &= \frac{ e^{i \varphi}}{2\g_1 |\g_2|} \l[ \frac{ie^{i \q(z^2)} }{ \sin( \q(z^2)}  + 2\r]  = \frac{ e^{i \varphi}}{2\g_1 |\g_2|} \l[ \frac{ie^{-i \q(z^2)} }{ \sin( \q(z^2)} \r]\, .
\end{align*}

\end{proof}
\end{cor}

\begin{prop}\label{prop: det(I + VU)}
Define $U$ and $V$ as in Definitions \ref{def: U} and \ref{def: V}. For $z \in \mathbb{C} \backslash \sigma_{ess}(H)$, we have \begin{align}\label{eq: I + VU}
     (I + VU)(z)  & = \begin{bmatrix}
        1- \overline{\g_2} K(1; z) & 0 \\
        - \overline{\g_2} z J(1; z)  & 1
    \end{bmatrix}\, 
\end{align}
and
\begin{align}\label{eq: det(I + VU)}
    \det( I + VU)(z) = 
        1 - \overline{\g_2} K(1;z) = 1 - \frac{i}{2 \sin \l( \q(z^2)\r)} \l( e^{-i\q(z^2)
     }+ \frac{|\g_2|}{\g_1}\r) \, .
\end{align}
Furthermore, $ \det ( I + VU)(z) = 0 $ if and only if $z = 0$ and $|\g_2| > \g_1$. 
\begin{proof}
See Appendix \ref{pf: det(I + VU)} for the proof. 
 \end{proof}
\end{prop}

\begin{rmk}
By Proposition \ref{prop: det(I + VU)}, $(I +VU)(z)$ is not invertible if and only if $z =0$ and $|\g_2| > \g_1$. Thus \eqref{eq: FL transform of resolvent form 2} implies $(H - zI) \psi = f$ is not uniquely solvable when 
\begin{align*}
    \det \begin{bmatrix}
        -z & h(q) \\
        \overline{h( \overline{q})} & -z
    \end{bmatrix} &= z^2 - k^2(q - \varphi)  = 0 \text{ for some } q \in [-\pi, \pi] \\
    \iff z &\in \sigma_{ess}(H) = [-\g_+ , - \g_-] \cup [\g_- , \g_+]
\end{align*}
or when $z = 0$ and $|\g_2| > \g_1$. Since every isolated point of $\sigma(H)$ is an eigenvalue of $H$ \cite[Theorem 12.29(c)]{GrandpaRudin},  \cref{prop: det(I + VU)} shows  that there must exist a solution to $H \psi = 0$ without needing to explicitly construct the solution as we did in \eqref{eq: zero energy mode}. 
\end{rmk}

\begin{cor}\label{cor: U(I +VU)^(-1)}
For $z \in \mathbb{C}\backslash \sigma_{ess}(H)$, 
\begin{align*}
     U(q;z) (I + VU)^{-1} (z) = \frac{1}{1- \overline{\g_2} K(1;z)}  U(q;z)\, .
\end{align*}
\begin{proof}

By Proposition \ref{prop: det(I + VU)}, the inverse of the matrix $(I+VU)(z)$ is
\begin{equation}
    (I + VU)^{-1}(z) = \frac{1}{1 - \overline{\g_2} K(1;z)} \begin{bmatrix}
        1 & 0 \\ \overline{\g_2} z J(1;z) & 1 - \overline{\g_2} K(1;z)
    \end{bmatrix} \, .
\end{equation}
Multiplying on the left by $U(q;z)$ (see \cref{def: U}), we compute 
\begin{align*}
    U(q;z) (I + VU)^{-1}(z) &= \frac{\overline{\g_2} e^{iq}}{z^2 - k^2(q- \varphi)} \cdot \frac{1}{1- \overline{\g_2} K(1;z)} \begin{bmatrix}
        h(q) & 0 \\ z & 0
    \end{bmatrix}
    \begin{bmatrix}
        1 & 0 \\ \g_2 z J(1;z) & 1 - \g_2 K(1;z)
    \end{bmatrix} \\
    &= \frac{\overline{\g_2} e^{iq}}{z^2 - k^2(q- \varphi)} \cdot \frac{1}{1- \overline{\g_2} K(1;z)} \begin{bmatrix}
        h(q) & 0 \\ z & 0
    \end{bmatrix} \\
    &= \frac{1}{1- \overline{\g_2} K(1;z)}  U(q;z) \, .
\end{align*}
\end{proof}
\end{cor}

\begin{proof}[\textbf{Proof of Theorem} \ref{thm: FL resolvent}]
We use the results in Section \ref{sec: Computing the Resolvent Part 2} to compute the right hand side of \eqref{eq: FL transform of resolvent SMW}. 
The discrete Fourier transform of the bulk resolvent, computed in \eqref{eq: bulk resolvent}, can be expressed using property \eqref{eq: S shift} of the right-shift operator, see \cref{def: right-shift operator}, as 
\begin{align*}
    \l[\hat{R}_{bulk}(z) \hat{f}\, \r](q) &= \frac{1}{ k^2(q- \varphi) - z^2} \begin{bmatrix}
        z \Tilde{f}^A (q) + (\g_1 + \g_2 e^{-iq}) \Tilde{f}^B (q) \\
        (\g_1 + \overline{\g_2} e^{iq}) \Tilde{f}^A (q) + z \Tilde{f}^B (q)
    \end{bmatrix} \\
    &= \frac{1}{ k^2(q- \varphi) - z^2} \begin{bmatrix}
        z \tilde{f}^A (q) + \g_1 \tilde{f}^B (q) + \g_2 \widetilde{Sf}^B (q)  \\
        \g_1 \tilde{f}^A(q) + \overline{\g_2} e^{iq} \tilde{f}^A (q) + z \tilde{f}^B (q)
    \end{bmatrix} \, .
\end{align*}
 Using Lemma \ref{lem: J integrals} and property \eqref{eq: S shift} of the right-shift operator again, we compute
\begin{align*}
    & V\l(\l[\hat{R}_{bulk}(z) \hat{f}\, \r](\cdot)\r) 
    =  \frac{1}{2\pi} \int \limits_{-\pi}^{\pi} \frac{1}{ k^2(q- \varphi) - z^2} \begin{bmatrix}
        z \tilde{f}^A (q) + \g_1 \tilde{f}^B (q) + \g_2 \widetilde{Sf}^B (q)  \\
        \g_1 \tilde{f}^A(q) + \overline{\g_2} e^{iq} \tilde{f}^A (q) + z \tilde{f}^B (q)
    \end{bmatrix} \, dq  \\ 
    &= \frac{i}{2 \g_1 |\g_2| \sin(\q(z^2))} \begin{bmatrix}
        z \tilde{f}^A(\q(z^2) + \varphi) + \g_1 \tilde{f}^B(\q(z^2) + \varphi) + \g_2 \widetilde{Sf}^B(\q(z^2) + \varphi)  \\
        \g_1 \tilde{f}^A(\q(z^2) + \varphi) + \overline{\g_2} e^{i (\varphi + \q(z^2))} \tilde{f}^A(\q(z^2) + \varphi) + z \tilde{f}^B(\q(z^2) + \varphi)
    \end{bmatrix}  \\
    &\quad +
    \frac{1}{\g_1 |\g_2|}
    \begin{bmatrix}
        0 \\ \overline{\g_2} f_0 e^{i \varphi} 
    \end{bmatrix} \\
    &= \frac{i}{2 \g_1 |\g_2| \sin(\q(z^2))} \begin{bmatrix}
        z \tilde{f}^A(\q(z^2) + \varphi) + (\g_1 + \g_2 e^{-i (\varphi + \q(z^2))} )\tilde{f}^B(\q(z^2) + \varphi)  \\
        (\g_1  + \overline{\g_2} e^{i (\varphi + \q(z^2))} ) \tilde{f}^A(\q(z^2) + \varphi) + z \tilde{f}^B(\q(z^2) + \varphi)
    \end{bmatrix}  +
    \frac{1}{\g_1 |\g_2|}
    \begin{bmatrix}
        0 \\ \overline{\g_2} f_0 e^{i \varphi} 
    \end{bmatrix}  \\
    &= \frac{i}{2 \g_1 |\g_2| \sin(\q(z^2))} \begin{bmatrix}
        z &  h(\q(z^2) + \varphi) \\[5pt]
        \overline{ h \l( \overline{\q(z^2) + \varphi}\r)}  & z
    \end{bmatrix} \tilde{f}(\q(z^2) + \varphi) +
    \frac{1}{\g_1 |\g_2|}
    \begin{bmatrix}
        0 \\ \overline{\g_2} f_0 e^{i \varphi} 
    \end{bmatrix} \, .
\end{align*}

Furthermore, \begin{equation}\label{eq: matrix times V(R_bulk)}
    \begin{bmatrix}
        h(q) & 0 \\ z & 0
    \end{bmatrix} V\l(\l[\hat{R}_{bulk}(z) \tilde{f}\, \r](\cdot)\r) = \frac{i}{2 \g_1 |\g_2| \sin(\q(z^2))} \begin{bmatrix}
        zh(q) & h(q) h(\q(z^2) + \varphi) \\
        z^2 & zh(\q(z^2)+ \varphi)
    \end{bmatrix} \Tilde{f}(\q(z^2) + \varphi) \, .
\end{equation}
Using \cref{cor: U(I +VU)^(-1)} and \eqref{eq: matrix times V(R_bulk)}, we compute the second term in the right hand side of \eqref{eq: FL transform of resolvent SMW},

\begin{align*}
    & U(q;z)(I + VU)^{-1}(z) V\l(\l[\hat{R}_{bulk}(z) \tilde{f}\, \r](\cdot)\r)  \\
    &\quad = \frac{1}{1- \overline{\g_2} K(1;z)}  U(q;z) V\l(\l[\hat{R}_{bulk}(z) \tilde{f}\, \r](\cdot)\r) \\
    &\quad = \frac{1}{1- \overline{\g_2} K(1;z)} \frac{\overline{\g_2} e^{iq}}{z^2 -  k^2(q- \varphi)} \begin{bmatrix}
         h(q) & 0 \\ z & 0
    \end{bmatrix} V\l(\l[\hat{R}_{bulk}(z) \tilde{f}\, \r](\cdot)\r) \\
    &\quad= \frac{1}{1- \overline{\g_2} K(1;z)} \frac{\overline{\g_2} e^{iq}}{z^2 -  k^2(q- \varphi)}  \frac{i}{2 \g_1 |\g_2| \sin(\q(z^2))} \begin{bmatrix}
        zh(q) & h(q) h(\q(z^2) + \varphi) \\
        z^2 & zh(\q(z^2)+ \varphi)
    \end{bmatrix} \tilde{f}(\q(z^2) + \varphi)  \, . \numberthis \label{eq: step 1 in  FL resolvent prop} 
\end{align*}

By \cref{prop: det(I + VU)}, 
\begin{equation}\label{eq: step 2 in  FL resolvent prop}
\begin{split}
    \l( 1- \overline{\g_2} K(1;z) \r) 2\sin(\q(z^2))  &= \l( 1 - \frac{i}{2 \sin \q(z^2)} \l( e^{-i\q(z^2)
     }+ \frac{|\g_2|}{\g_1}\r) \r)2\sin(\q(z^2))  \\
     &= 2\sin(\q(z^2)) - i \l( e^{-i\q(z^2)
     }+ \frac{|\g_2|}{\g_1}\r)  \\
     &= -i \l( e^{i\q(z^2)} + \frac{ |\g_2|}{\g_1}\r) \, . 
\end{split}
\end{equation}
 Using \eqref{eq: e^(i varphi)} and \eqref{eq: step 2 in  FL resolvent prop}, we compute 
\begin{equation}\label{eq: step 3 in FL resolvent prop}
\begin{split}
     \frac{-i\, \overline{\g_2} e^{i \q(z^2)}}{\l( 1- \overline{\g_2} K(1;z) \r) 2 \g_1 |\g_2| \sin(\q(z^2))} &= \frac{\overline{\g_2}}{|\g_2|} \frac{ e^{i \q(z^2)} }{\g_1 e^{i \q(z^2)} + |\g_2| }  = \frac{e^{-i \varphi}}{\g_1 + |\g_2| e^{-i \q(z^2)}} \\
     &= \frac{e^{-i \varphi}}{\g_1 + \g_2 e^{-i (\q(z^2) + \varphi)}} = \frac{e^{-i \varphi}}{ h(\q(z^2) +  \varphi)} \, .
\end{split}
\end{equation}

Multiplying and dividing \eqref{eq: step 1 in  FL resolvent prop} by $-e^{i\q(z^2)}$ and plugging in \eqref{eq: step 3 in  FL resolvent prop}, 
we obtain 

\begin{align*}
    &U(q;z)(I + VU)^{-1} (z) V\l(\l[\hat{R}_{bulk}(z) \tilde{f}\, \r](\cdot)\r) \\
    &\quad = - \frac{e^{iq}}{e^{i\q(z^2)}} \frac{1}{z^2 - k^2 ( q- \varphi)} \frac{e^{-i \varphi}}{h(\q(z^2) + \varphi)} \begin{bmatrix}
        zh(q) & h(q) h(\q(z^2) + \varphi) \\
        z^2 & zh(\q(z^2)+ \varphi)
    \end{bmatrix} \tilde{f}(\q(z^2) + \varphi) \\
    &\quad = \frac{e^{iq}}{e^{i\q(z^2)}} \frac{e^{-i \varphi}}{k^2 ( q- \varphi) - z^2}  \begin{bmatrix}
        zh(q)/h(\q(z^2)+ \varphi) & h(q)  \\
        z^2/h(\q(z^2)+ \varphi) & z
    \end{bmatrix} \tilde{f}(\q(z^2) + \varphi) \\
    &\quad =  \frac{e^{iq}}{k^2 ( q- \varphi) - z^2}  \begin{bmatrix}
        zh(q)/h(\q(z^2)+ \varphi) & h(q)  \\[5pt] \overline{h\l(\overline{\q(z^2)+ \varphi}\r) } & z
    \end{bmatrix} \widetilde{Sf}(\q(z^2) + \varphi)\, , \numberthis \label{eq: step 4 in FL resolvent prop}
\end{align*}
where the last equality follows from the fact that \eqref{eq: k shifted} implies
\begin{align*}
    h\l(\q(z^2)+ \varphi \r)  \overline{h\l(\overline{\q(z^2)+ \varphi}\r) } = k^2(\q(z^2)) = z^2 \, ,
\end{align*}
and that property \eqref{eq: S shift} of the right-shift operator implies
\begin{align*}
    \frac{e^{-i \varphi}}{e^{i\q(z^2)}}  \tilde{f}(\q(z^2) + \varphi) = \widetilde{Sf}(\q(z^2) + \varphi)\, . 
\end{align*}
Plugging \eqref{eq: step 4 in FL resolvent prop} into \eqref{eq: FL transform of resolvent SMW} completes the proof. 
\end{proof}

\section{Computing the Propagator $\exp(-i\Hedge t)$}\label{sec: Computing the Propagator}
The holomorphic functional calculus for bounded operators allows us to write down the propagator in terms of the Dunford integral, 
\begin{align*}
    \l[e^{-iHt} \mathscr{P}_+ f \r]_n = \l[ \frac{1}{2\pi i} \oint \limits_C e^{-i\lambda t} R(z) \, dz \, f \r]_n \, ,
\end{align*}
where we take $C$ to be a simple closed loop in $\mathbb{C}$ enclosing the interval $[\g_-, \g_+]$ and no other part of $\sigma(H)$. Taking a limit of smaller and smaller loops, we get 
\begin{align*}
    \l[e^{-iHt} \mathscr{P}_+ f \r]_n &= \lim_{\ep \rightarrow 0^+} \l[ \frac{1}{2\pi i} \int \limits_{\g_-}^{\g_+} e^{-i\lambda t} [R(\lambda + i \ep) - R(\lambda - i \ep)] \, d\lambda \, f \r]_n \\
    &= \lim_{\ep \rightarrow 0^+} \frac{1}{4\pi^2 i} \int \limits_{\g_-}^{\g_+} e^{-i\lambda t} \int \limits_{-\pi}^{\pi} e^{inq} \l( \l[ \tilde{R}(\lambda + i \ep) - \tilde{R}(\lambda - i \ep) \r] \tilde{f} \,\r)(q) \, dq \, d\lambda \, . 
\end{align*}

Applying \cref{thm: FL resolvent}, we decompose the propagator of $H$ into the difference of the propagator of $H_{\mathbb{Z}}$ and a term that comes from from the boundary condition. More precisely, 

\begin{equation}\label{eq: propagator bulk and edge partition}
\begin{split}
    &\l[e^{-iHt}\mathscr{P}_{+}(H) f\r]_n \\
    &= \l[e^{-iH_{\mathbb{Z}}t} f\r]_n - \frac{1}{4 \pi^2 i} \lim_{\ep \rightarrow 0^+} \int \limits_{\g_-}^{\g_+} e^{-i \lambda t}  \int_{-\pi}^{\pi} e^{inq} \l([\tilde{R}_{edge}(\lambda +  i\ep) - \tilde{R}_{edge} (\lambda - i\ep)] \tilde{f}\,\r)(q) \, dq \, d\lambda  \\
    &= \l[e^{-iH_{\mathbb{Z}}t} f\r]_n + [U_{\rm edge}(t)f]_n \, , 
\end{split}
\end{equation}
where
\begin{align}
    U_{\text{edge}}(t) &= (4\pi^2i)^{-1}(U^-_{\text{edge}}(t)-U^+_{\text{edge}}(t)) \label{eq: U_edge} \\
    \text{and } \quad [U_{\text{edge}}^{\pm}(t)f]_n  &:= \lim_{\ep \rightarrow 0^+} \int \limits_{\g_-}^{\g_+} e^{-i \lambda t}   \int_{-\pi}^{\pi} e^{inq} \l[\tilde{R}_{edge}(\lambda \pm  i\ep)  \tilde{f}\,\r] (q) \, dq \, d\lambda \, . \label{eq: U^pm}
\end{align}

We have already estimated the dispersive decay of $\l[e^{-iH_{\mathbb{Z}}t} f\r]_n$  in
\cref{thm: free dispersive}, so in order to prove \cref{thm: Main thm 1}, it remains to obtain dispersive decay estimates of $U_{\text{edge}}^{\pm}(t)$. Equation \eqref{eq: U^pm} is written in the form of an oscillatory integral, but before we can apply stationary phase techniques, we must write the integrand more explicitly.
In the rest of this section, we obtain an explicit expression for $U_{\text{edge}}^{\pm}$, which completes the proof the following theorem.

\begin{thm}\label{thm: propagator of SSH}
Let $f \in \ell_1^1(\mathbb{N}_0; \mathbb{C}^2)$. The propagator of the edge Hamiltonian $H$ acting on the projection of the positive part of the continuous spectral subspace of $H$ onto $f$ is given by 
\begin{align*}
    \l[e^{-iHt}\mathscr{P}_{+}(H) f\r]_n 
    = \l[e^{-iH_{\mathbb{Z}}t} f\r]_n + \frac{[U_{\text{edge}}^{-}(t) f]_n - [U_{\text{edge}}^{+}(t) f]_n}{4\pi^2 i} \, ,
\end{align*}
where $H_{\mathbb{Z}}$ is the bulk SSH Hamiltonian given in \cref{def: Bulk SSH} and $[U_{\text{edge}}^{\pm}(t)f]_n$ is given explicitly in \eqref{eq: U term 3}. 
\end{thm}

In the remainder of this section, we finish proving \cref{thm: propagator of SSH} by computing $[U_{\text{edge}}^{\pm}(t)f]_n$.

\subsection{Useful Results }

\begin{lem}\label{lem: Change of Variables}
Let $k(q) = \sqrt{\g_1^2 + |\g_2|^2 + 2 \g_1 |\g_2| \cos(q)}$ \eqref{def: k}   and $\eta(\lambda) = \frac{\lambda^2 - \g_1^2 - |\g_2|^2}{2 \g_1 |\g_2|} $ \eqref{eq: eta var change}. 
The change of variables $q \mapsto k(q)$ for $q \in [0, \pi]$ or $q \in [-\pi, 0]$ is given by \begin{align*}
    dq &= \frac{-k }{\g_1 |\g_2| \sin(q)} dk =  \frac{-k \, \sgn(q)}{\g_1 |\g_2| \sqrt{1 - \eta^2(k)}} dk \,.
\end{align*}
Furthermore, \begin{align*}
    q &= \sgn(q) \cos^{-1}( \eta(k)) \\
    \sin(q) &=  \sgn(q) \sqrt{1 - \eta^2(k)}
\end{align*} 
for $q \in [-\pi, \pi] $.
\end{lem}

\begin{prop}\label{prop: no embedded eigenvalues}
$H$ has no nonzero eigenvalues. In particular, $H$ has no embedded eigenvalues in its essential spectrum. 
\begin{proof}
See Appendix \ref{pf: no embedded eigenvalues} for the proof.

\end{proof}
\end{prop}

In the following lemma we show that in the definition of $U_{\text{edge}}^{\pm}(t)f$, we can pass the limit inside the $d\lambda$ integral. 
\begin{lem}\label{lem: DCT for propagator}
If $f \in \ell_1^1 (\mathbb{N}_0; \mathbb{C}^2)$ then for all $n \in \mathbb{N}_0$,
\begin{align*}
    [U_{\text{edge}}^{\pm}(t)f]_n  =  \int \limits_{\g_-}^{\g_+} \lim_{\ep \rightarrow 0^+} e^{-i \lambda t}   \int_{-\pi}^{\pi} e^{inq} \l[\tilde{R}_{edge}(\lambda \pm  i\ep)  \tilde{f}\,\r] (q) \, dq \, d\lambda \, .
\end{align*}
\begin{proof}

See Appendix \ref{pf: DCT for propagator} for the proof.

\end{proof} 
\end{lem}

\subsection{A Remark About the Endpoints of the Essential Spectrum}\label{sec: Remark essential spectrum}
The main point of this subsection is that the projection-valued measure associated with $H$ has no atoms and thus nothing is lost by considering the spectral projection of $H$ onto the interior of $\sigma_{ess}(H) \cap \mathbb{R}_+$ rather than the whole set. This will be important as we compute $\lim_{\ep \rightarrow 0^+}  \int_{-\pi}^{\pi} e^{inq} \l([\tilde{R}(\lambda +  i\ep) - \tilde{R}(\lambda - i\ep)] \tilde{f}\,\r)(q) \, dq$ in the rest of the section. The integrand has singularities at the endpoints, and is manageable everywhere apart from the endpoints themselves through the use of the Sokhotski-Plemelj Theorem, see \cref{thm: Sokhotski-Plemelj}. The following remark allows us to ignore the problematic endpoints.

Let $E$ be the spectral decomposition of $H$. Then note that  
\begin{align*}
    e^{-iHt}\mathscr{P}_{[\g_-, \g_+]}  - e^{-iHt}\mathscr{P}_{(\g_-, \g_+)}  
    &= \int \limits_{[\g_-, \g_+]} e^{-i \lambda t} - e^{- i \lambda t} \mathbbm{1}_{(\g_-, \g_+)}(\lambda) \, dE_{\lambda} \\
    &=  \int \limits_{[\g_-, \g_+]} e^{-i \lambda t} \mathbbm{1}_{\{\g_-\} \cup \{\g_+\}} (\lambda) \, dE_{\lambda}  \\
    &= e^{-i \g_- t} E_{\{\g_-\}} + e^{-i \g_+ t} E_{\{\g_2+\}}  \\
    &= 0 \, .
\end{align*}
The last equality follows from Proposition \ref{prop: no embedded eigenvalues}, which  in particular says $\g_{\pm}$ are not eigenvalues of $H$, and \cite[Theorem 12.29(b)]{GrandpaRudin}, which says that the spectral decomposition of $H$ evaluated at a single point is nonzero if and only if that point is an eigenvalue of $H$.
Thus \begin{align*}
    e^{-iHt}\mathscr{P}_+ f :=  e^{-iHt}\mathscr{P}_{[\g_-, \g_+]} f =  e^{-iHt}\mathscr{P}_{(\g_-, \g_+)} f \, ,
\end{align*}
and the propagator can be written as \begin{align*}
    \l[e^{-iHt}\mathscr{P}_+ f\r]_n &= \frac{1}{2\pi i} \int \limits_{(\g_-, \g_+)} e^{-i \lambda t} \big(\big[R(\lambda +  i0) - R(\lambda - i0)\big]f \big)_n \, d\lambda \\
    &= \frac{1}{2\pi i} \int \limits_{(\g_-, \g_+)} e^{-i \lambda t} \lim_{\ep \rightarrow 0^+} \big(\big[R(\lambda +  i\ep) - R(\lambda - i\ep)\big] f\big)_n \, d\lambda \\
    &= \frac{1}{4\pi^2 i} \int \limits_{(\g_-, \g_+)} e^{-i \lambda t} \lim_{\ep \rightarrow 0^+}  \int_{-\pi}^{\pi} e^{inq} \l([\tilde{R}(\lambda +  i\ep) - \tilde{R}(\lambda - i\ep)] \tilde{f}\,\r)(q) \, dq \, d\lambda \, .
\end{align*}

\subsection{Contribution of Edge Resolvent to the Propagator}
To compute the right hand side of \eqref{eq: propagator bulk and edge partition}, we study  the integrand of $[U_{\text{edge}}^{\pm}(t)f]_n$, see \eqref{eq: U^pm}, in the following proposition. 
Recall equation \eqref{def: Redge tilde}, which says that 
\begin{align*}
 \l[\tilde{R}_{edge}(z)  \tilde{f}\,\r](q) =
    \frac{e^{iq}}{k^2 ( q- \varphi) - z^2}  \begin{bmatrix}
        zh(q)/h(\q(z^2)+ \varphi) & h(q)  \\[5pt] \overline{h\l(\overline{\q(z^2)+ \varphi}\r) } & z
    \end{bmatrix} \widetilde{Sf}(\q(z^2) + \varphi)\, .
\end{align*}
\begin{prop}\label{prop: handwavy approx}
Suppose that $\g_- > 0$ and $f \in \ell_1^1 (\mathbb{N}_0; \mathbb{C}^2)$. There exists $\ep_* =\ep_*(\g_1 ,\g_2) > 0$ such that for all $0 < \ep < \ep_*$ and  $n \in \mathbb{N}_0$, 
\begin{align*}
      &\int \limits_{\g_-}^{\g_+} \int \limits_{-\pi}^{\pi} \l| e^{inq} \l[\tilde{R}_{edge}(\lambda \pm  i\ep)  \tilde{f}\,\r](q) -  \frac{e^{i(n+1)q}}{k^2 ( q- \varphi) - (\lambda \pm i \ep)^2}  \begin{bmatrix}
        \lambda  h(q)/h(\pm \qzl + \varphi) & h(q)  \\[5pt] \overline{h\l( \pm \qzl+ \varphi \r) } & \lambda 
    \end{bmatrix} \widetilde{Sf}(\pm \qzl + \varphi) \r|  \, dq  \, d\lambda \\
    &\quad \leq c(\g_1, |\g_2|) \|f\|_{\ell_1^1} \sqrt{\ep} \, \text{arsinh}\l(\frac{2\g_1|\g_2|}{\ep} \r) \,. 
\end{align*}
Since the right hand side vanishes as $\ep \rightarrow 0^+$, this implies
\begin{align*}
    \l\|\, \int \limits_{-\pi}^{\pi} e^{inq} \l[\tilde{R}_{edge}(\lambda \pm  i\ep)  \tilde{f}\,\r](q) -  \frac{e^{i(n+1)q}}{k^2 ( q- \varphi) - (\lambda \pm i \ep)^2}  \begin{bmatrix}
        \lambda  h(q)/h(\pm \qzl + \varphi) & h(q)  \\[5pt] \overline{h\l( \pm \qzl+ \varphi \r) } & \lambda 
    \end{bmatrix} \widetilde{Sf}(\pm \qzl + \varphi)  \, dq \r\|_{L^1}
\end{align*}
vanishes as $\ep \rightarrow 0^+$. 

\begin{proof}
See Appendix \ref{pf: handwavy approx} for the proof. 
\end{proof}
\end{prop}

Recall that the relation \begin{align*}
    \q((\lambda \pm i0)^2) = \pm \qzl 
\end{align*} is proven in \cref{prop: limit of q}. 
Intuitively, the limit as $\ep \rightarrow 0^+$ of the integral is equal to the limit of the integral when all of the terms in the integrand that are bounded functions of $\ep$ are replaced with their limit. 

\begin{cor}\label{cor: handwavy approx}
Under the same assumptions as \cref{prop: handwavy approx}, 
\begin{align*}
    \l\|\, \int \limits_{-\pi}^{\pi} e^{inq} \l[\tilde{R}_{edge}(\lambda \pm  i\ep)  \tilde{f}\,\r](q) -  \frac{e^{i(n+1)q} (\lambda \pm i \ep)/\lambda}{k^2 ( q- \varphi) - (\lambda \pm i \ep)^2}  \begin{bmatrix}
        \lambda  h(q)/h(\pm \qzl + \varphi) & h(q)  \\[5pt] \overline{h\l( \pm \qzl+ \varphi \r) } & \lambda 
    \end{bmatrix} \widetilde{Sf}(\pm \qzl + \varphi)  \, dq \r\|_{L^1}
\end{align*}
vanishes as $\ep \rightarrow 0^+$. 
\end{cor}

Define 
\begin{align}
    A_n^{\pm}(\ep) &:= \frac{1}{2\lambda} \int \limits_{-\pi}^{\pi}  \l[ \frac{1}{k(q - \varphi) - (\lambda \pm i \ep)} - \frac{1}{k(q - \varphi) + (\lambda \pm i \ep)} \r] e^{i(n+1)q} h(q) \, dq\, , 
    \label{eq: A(ep)}\\
     B_n^{\pm}(\ep) &:= \frac{1}{2\lambda} \int \limits_{-\pi}^{\pi}  \l[ \frac{1}{k(q - \varphi) - (\lambda \pm i \ep)} - \frac{1}{k(q - \varphi) + (\lambda \pm i \ep)} \r] e^{i(n+1)q} \, dq\, ,  \label{eq: B(ep)}
\end{align}
and note that \begin{align}\label{eq: A(ep) in terms of B(ep)}
    A_n^{\pm}(\ep) = \g_1 B_n^{\pm}(\ep) + \g_2 B_{n-1}^{\pm}(\ep)\, .
\end{align}

Since \begin{align*}
   \frac{2 (\lambda \pm i \ep)}{2 \lambda} \frac{1}{k^2 ( q- \varphi) - (\lambda \pm i \ep)^2} = \frac{1}{2\lambda} \l[ \frac{1}{k(q - \varphi) - (\lambda \pm i \ep)} - \frac{1}{k(q - \varphi) + (\lambda \pm i \ep)} \r]\, ,
\end{align*}
\cref{cor: handwavy approx} implies 
\begin{align}\label{eq: U term 1}
    \lim_{\ep \rightarrow 0^+}  \int \limits_{-\pi}^{\pi} e^{inq} \l[\tilde{R}_{edge}(\lambda \pm  i\ep)  \tilde{f}\,\r](q) \, dq = \lim_{\ep \rightarrow 0^+}  \begin{bmatrix}
        \lambda A_n^{\pm}(\ep)  /h(\pm \qzl + \varphi) &   A_n^{\pm}(\ep) \\[5pt]
        \overline{h\l( \pm \qzl+ \varphi \r) } B_n^{\pm}(\ep) & \lambda B_n^{\pm}(\ep) 
    \end{bmatrix} \widetilde{Sf}(\pm \qzl + \varphi) \, .
\end{align}
We now proceed to obtain an explicit representation for $\lim_{\ep \rightarrow 0^+} B_n^{\pm}(\ep) $
through the use of the Sokhotski-Plemelj theorem, see \cref{thm: Sokhotski-Plemelj}.

\begin{lem}\label{lem: computing B(ep)}
Let $B_n^{\pm}(\ep)$ be defined as in \eqref{eq: B(ep)}. Then 
\begin{align*}
    \lim_{\ep \rightarrow 0^+} B_n^{\pm}(\ep) &= \frac{e^{i(n+1)\varphi}}{\lambda} \fint \limits_{\g_-}^{\g_+} \l( \frac{1}{k - \lambda}\r) \frac{kT_{n+1}(\eta(k)) }{\g_1 |\g_2| \sqrt{1 - \eta^2(k)}} \, dk \pm i \pi  \frac{e^{i(n+1)\varphi}}{\lambda} \frac{\lambda T_{n+1}(\eta(\lambda)) }{\g_1 |\g_2| \sqrt{1 - \eta^2(\lambda)}}   \\
    &\quad -
    \frac{e^{i(n+1)\varphi}}{\lambda} \int \limits_{\g_-}^{\g_+} \l( \frac{1}{k + \lambda}\r) \frac{kT_{n+1}(\eta(k)) }{\g_1 |\g_2| \sqrt{1 - \eta^2(k)}} \, dk \, . 
\end{align*}

\begin{proof}
    See Appendix \ref{pf: computing B(ep)} for the proof. 
\end{proof}
\end{lem}

In order to obtain an expression for the matrix in \eqref{eq: U term 1}, we  
use \eqref{eq: A(ep) in terms of B(ep)} to compute
\begin{align*}
    &\lim_{\ep \rightarrow 0^+}  \begin{bmatrix}
        \lambda A_n^{\pm}(\ep) /h(\pm \qzl + \varphi) &   A_n^{\pm}(\ep) \\[5pt]
        \overline{h\l( \pm \qzl+ \varphi \r) } B_n^{\pm}(\ep) & \lambda B_n^{\pm}(\ep) 
    \end{bmatrix} \\
    &= \lim_{\ep \rightarrow 0^+}   \begin{bmatrix}
         \lambda \g_1  B_n^{\pm}(\ep) /h(\pm \qzl + \varphi) &  \g_1 B_n^{\pm}(\ep) \\[5pt]
        \overline{h\l( \pm \qzl+ \varphi \r) } B_n^{\pm}(\ep) & \lambda B_n^{\pm}(\ep) 
    \end{bmatrix} +  \begin{bmatrix}
        \lambda \g_2  B_{n-1}^{\pm}(\ep) /h(\pm \qzl + \varphi) &   \g_2 B_{n-1}^{\pm}(\ep) \\
        0  &  0 
    \end{bmatrix} \\
    &= \lim_{\ep \rightarrow 0^+}    B_n^{\pm}(\ep) \begin{bmatrix}
         \lambda \g_1  /h(\pm \qzl + \varphi) &   \g_1 \\[5pt]
        \overline{h\l( \pm \qzl+ \varphi \r) } & \lambda 
    \end{bmatrix} +  B_{n-1}^{\pm}(\ep)\begin{bmatrix}
         \lambda \g_2  /h(\pm \qzl + \varphi) &   \g_2 \\
        0  &  0 
    \end{bmatrix} \, .
\end{align*}
Plugging in \cref{lem: computing B(ep)},
\begin{equation}\label{eq: U term 2}\begin{split}
    &\lim_{\ep \rightarrow 0^+}  \begin{bmatrix}
        \lambda A_n^{\pm}(\ep) /h(\pm \qzl + \varphi) &   A_n^{\pm}(\ep) \\[5pt]
        \overline{h\l( \pm \qzl+ \varphi \r) } B_n^{\pm}(\ep) & \lambda B_n^{\pm}(\ep) 
    \end{bmatrix} \\ 
    &=  e^{i(n+1)\varphi} \fint \limits_{\g_-}^{\g_+} \l( \frac{1}{k - \lambda}\r) \frac{kT_{n+1}(\eta(k)) }{\g_1 |\g_2| \sqrt{1 - \eta^2(k)}} \, dk  \begin{bmatrix}
        \g_1 /h(\pm \qzl + \varphi) &   \g_1  / \lambda \\[5pt]
        \overline{h\l( \pm \qzl+ \varphi \r) } / \lambda  & 1 
    \end{bmatrix} \\
    &\quad + e^{in\varphi} \fint \limits_{\g_-}^{\g_+} \l( \frac{1}{k - \lambda}\r) \frac{kT_{n}(\eta(k)) }{\g_1 |\g_2| \sqrt{1 - \eta^2(k)}} \, dk \begin{bmatrix}
        \g_2  /h(\pm \qzl + \varphi) &   \g_2/ \lambda \\
        0  &  0 
    \end{bmatrix} \\
    &\quad \pm i \pi e^{i(n+1) \varphi} \frac{\lambda T_{n+1}(\eta(\lambda)) }{\g_1 |\g_2| \sqrt{1 - \eta^2(\lambda)}} \begin{bmatrix}
        \g_1 /h(\pm \qzl + \varphi) &   \g_1  / \lambda \\[5pt]
        \overline{h\l( \pm \qzl+ \varphi \r) } / \lambda  & 1 
    \end{bmatrix} \\
    &\quad \pm i \pi e^{in \varphi} \frac{\lambda T_{n}(\eta(\lambda)) }{\g_1 |\g_2| \sqrt{1 - \eta^2(\lambda)}} \begin{bmatrix}
        \g_2  /h(\pm \qzl + \varphi) &   \g_2/ \lambda \\
        0  &  0 
    \end{bmatrix} \\
    &\quad -e^{i(n+1)\varphi} \fint \limits_{\g_-}^{\g_+} \l( \frac{1}{k + \lambda}\r) \frac{kT_{n+1}(\eta(k)) }{\g_1 |\g_2| \sqrt{1 - \eta^2(k)}} \, dk  \begin{bmatrix}
        \g_1 /h(\pm \qzl + \varphi) &   \g_1  / \lambda \\[5pt]
        \overline{h\l( \pm \qzl+ \varphi \r) } / \lambda  & 1 
    \end{bmatrix} \\
    &\quad - e^{in\varphi} \fint \limits_{\g_-}^{\g_+} \l( \frac{1}{k + \lambda}\r) \frac{kT_{n}(\eta(k)) }{\g_1 |\g_2| \sqrt{1 - \eta^2(k)}} \, dk \begin{bmatrix}
        \g_2  /h(\pm \qzl + \varphi) &   \g_2/ \lambda \\
        0  &  0 
    \end{bmatrix} \, .
\end{split}
\end{equation}

Combining \eqref{eq: U^pm}, \eqref{eq: U term 1}, and \eqref{eq: U term 2} gives
\begin{equation}\label{eq: U term 3}\begin{split}
    &[U_{\text{edge}}^{\pm}(t)f]_n  \\
    &= \int \limits_{\g_-}^{\g_+} e^{-i \lambda t}  e^{i(n+1)\varphi} \fint \limits_{\g_-}^{\g_+} \l( \frac{1}{k - \lambda}\r) \frac{kT_{n+1}(\eta(k)) }{\g_1 |\g_2| \sqrt{1 - \eta^2(k)}} \, dk  \begin{bmatrix}
        \g_1 /h(\pm \qzl + \varphi) &   \g_1  / \lambda \\[5pt]
        \overline{h\l( \pm \qzl+ \varphi \r) } / \lambda  & 1 
    \end{bmatrix} \widetilde{Sf}(\pm \qzl + \varphi) \, d\lambda \\
    &\quad + \int \limits_{\g_-}^{\g_+} e^{-i \lambda t}  e^{in\varphi} \fint \limits_{\g_-}^{\g_+} \l( \frac{1}{k - \lambda}\r) \frac{kT_{n}(\eta(k)) }{\g_1 |\g_2| \sqrt{1 - \eta^2(k)}} \, dk \begin{bmatrix}
        \g_2  /h(\pm \qzl + \varphi) &   \g_2/ \lambda \\
        0  &  0 
    \end{bmatrix} \widetilde{Sf}(\pm \qzl + \varphi) \, d\lambda\\
    &\quad \pm i \pi \int \limits_{\g_-}^{\g_+} e^{-i \lambda t}  e^{i(n+1) \varphi} \frac{\lambda T_{n+1}(\eta(\lambda)) }{\g_1 |\g_2| \sqrt{1 - \eta^2(\lambda)}} \begin{bmatrix}
        \g_1 /h(\pm \qzl + \varphi) &   \g_1  / \lambda \\[5pt]
        \overline{h\l( \pm \qzl+ \varphi \r) } / \lambda  & 1 
    \end{bmatrix} \widetilde{Sf}(\pm \qzl + \varphi) \, d\lambda \\
    &\quad \pm i \pi \int \limits_{\g_-}^{\g_+} e^{-i \lambda t}  e^{in \varphi} \frac{\lambda T_{n}(\eta(\lambda)) }{\g_1 |\g_2| \sqrt{1 - \eta^2(\lambda)}} \begin{bmatrix}
        \g_2  /h(\pm \qzl + \varphi) &   \g_2/ \lambda \\
        0  &  0 
    \end{bmatrix} \widetilde{Sf}(\pm \qzl + \varphi) \, d\lambda\\
    &\quad -\int \limits_{\g_-}^{\g_+} e^{-i \lambda t}  e^{i(n+1)\varphi} \int \limits_{\g_-}^{\g_+} \l( \frac{1}{k + \lambda}\r) \frac{kT_{n+1}(\eta(k)) }{\g_1 |\g_2| \sqrt{1 - \eta^2(k)}} \, dk  \begin{bmatrix}
        \g_1 /h(\pm \qzl + \varphi) &   \g_1  / \lambda \\[5pt]
        \overline{h\l( \pm \qzl+ \varphi \r) } / \lambda  & 1 
    \end{bmatrix} \widetilde{Sf}(\pm \qzl + \varphi) \, d\lambda\\
    &\quad -\int \limits_{\g_-}^{\g_+} e^{-i \lambda t}  e^{in\varphi} \int \limits_{\g_-}^{\g_+} \l( \frac{1}{k + \lambda}\r) \frac{kT_{n}(\eta(k)) }{\g_1 |\g_2| \sqrt{1 - \eta^2(k)}} \, dk \begin{bmatrix}
        \g_2  /h(\pm \qzl + \varphi) &   \g_2/ \lambda \\
        0  &  0 
    \end{bmatrix} \widetilde{Sf}(\pm \qzl + \varphi) \, d\lambda\, .
\end{split}
\end{equation}

\begin{rmk}
Notice that the contribution of the edge resolvent to the propagator, $[U_{\text{edge}}^{\pm}(t)f]_n$, involves evaluating $\widetilde{Sf}(\pm \qzl + \varphi)$. Since the evaluation is at a different point for $[U_{\text{edge}}^{+}(t)f]_n$ versus $[U_{\text{edge}}^{-}(t)f]_n$, when we compute $[U_{\text{edge}}^{-}(t)f]_n - [U_{\text{edge}}^{+}(t)f]_n$, there will be no cancellation of the Cauchy principle values. This will be a source of great difficulty in the later sections when we compute oscillatory integral estimates.
\end{rmk}

\section{Representative Oscillatory Integrals}\label{sec: Representative Oscillatory Integrals}
\subsection{Representative Terms}

We have already obtained a  decay rate estimate of the propagator of the bulk SSH Hamiltonian in \cref{thm: free dispersive}. So by \cref{thm: propagator of SSH}, in order to obtain a decay rate estimate of $e^{-iHt} \mathscr{P}_+ (H)$, it remains to compute oscillatory decay rate estimates for the terms $[U_{\text{edge}}^{\pm}(t)f]_n$ given in \eqref{eq: U term 3}. Since
many of the terms in \eqref{eq: U term 3} are quite similar, we identify three representative terms (in order of difficulty to bound) from the expression and carefully study their decay estimates. These estimates can then be easily applied to all of the terms in the expression for the propagator. 
Recall that $T_n$ are the Chebyshev polynomials of the first kind, see \cref{def: First Chebyshev Polynomial}, and that $\eta(\lambda)$ and $\qzl$ are defined in \cref{prop: limit of q} as
\begin{align*}
    \eta(\lambda) &:= \frac{\lambda^2 - \g_1^2 - |\g_2|^2}{2 \g_1 |\g_2|}  \\
    \qzl &:= \text{sgn}(\lambda) \cos^{-1}(\eta(\lambda)) \, .
\end{align*}

\underline{Type I:}
\begin{align}\label{eq: Type I def}
    \int \limits_{\g_-}^{\g_+} e^{-i \lambda t} \frac{\lambda T_n(\eta(\lambda))}{\g_1 |\g_2| \sqrt{1 - \eta^2(\lambda)}} \tilde{f} (\qzl) \, d\lambda
\end{align}

\underline{Type II:}
\begin{align}\label{eq: Type II def}
    \int \limits_{\g_-}^{\g_+} e^{-i \lambda t}  \int \limits_{\g_-}^{\g_+} \l( \frac{1}{k + \lambda} \r) \frac{k T_n(\eta(k))}{\g_1 |\g_2| \sqrt{1 - \eta^2(k)}} \, dk \,  \tilde{f}(\qzl) \, d \lambda 
\end{align}

\underline{Type III:}
\begin{align}\label{eq: Type III def}
    \int \limits_{\g_-}^{\g_+} e^{-i \lambda t}  \fint \limits_{\g_-}^{\g_+} \l( \frac{1}{k - \lambda} \r) \frac{k T_n(\eta(k))}{\g_1 |\g_2| \sqrt{1 - \eta^2(k)}} \, dk \, \tilde{f}(\qzl) \, d \lambda 
\end{align}

\subsection{Some technical results used in estimation of Type I, II and III terms}

\begin{lem}\label{lem: k(y) facts}
Suppose that $\g_1 > 0 $ and $\g_1 \neq |\g_2|$. Then $k(y) :=  \sqrt{\g_1^2 + |\g_2|^2 +2 \g_1 |\g_2|  \cos(y)}$ restricted to $[0, \pi]$ 
has the following properties. 
\begin{enumerate}
    \item The first four derivatives of $k(y)$ are \begin{align*}
    k'(y) &= \frac{- \g_1 |\g_2| \sin(y)}{k(y)} \\
    k''(y) &=  -\g_1 |\g_2| \l[ \frac{\cos(y)}{k(y)} + \frac{\g_1 |\g_2| \sin^2(y)}{k^3(y)} \r] \\
    k'''(y)
    &= - \g_1 |\g_2| \l[   \frac{3\g_1 |\g_2| \cos(y)\sin(y)}{k^3(y)} + \frac{3 \g_1^2 |\g_2|^2 \sin^3(y)}{k^5(y)} -\frac{\sin(y)}{k(y)} \r]  \\
    k^{(4)}(y) &= -\g_1 |\g_2|\Bigg[ \frac{15 \g_1^3 |\g_2|^3 \sin^4(y)}{k^7(y)} + \frac{18 \g_1^2 |\g_2|^2 \sin^2(y) \cos(y)}{k^5(y)} + \frac{ \g_1 |\g_2| (3\cos^2(y) -4\sin^2(y)) }{k^3(y)} - \frac{\cos(y)}{k(y)} \Bigg]\, .
\end{align*}
    \item $k'(y)$ vanishes only at $y = 0, \pi$. 
    \item $k''(y)$ vanishes only at $y = y_M \in (\pi/2, \pi)$, where  \begin{align}\label{eq: y_M}
        y_M := \begin{cases}
            \cos^{-1} \l( - \g_1 / |\g_2| \r), & \text{if } |\g_2| > \g_1 \\
            \cos^{-1}\l( -|\g_2|/ \g_1 \r), & \text{if } \g_1 > |\g_2| \, .
        \end{cases}
    \end{align}
    \item $k'(y)$ is monotone on $\l[0,  y_M \r]$ and $\l[ y_M, \pi \r]$. Furthermore, \begin{align*}
        \sup_{y\in [0, \pi]} |k'(y)| = |k'(y_M)| = \min\{ \g_1, |\g_2|\}\, .
    \end{align*}
    \item  $k'''\l( y_M  \r) =  \min \{ \g_1, |\g_2| \}$. 
     \item $k'''(y)$ vanishes only at $y = 0, \pi$.
\end{enumerate}
Furthermore, if we define 
\begin{align*}
    A(y) &:= k(y) - \g_- \geq 0 \\
    B(y) &:= \g_+ - k(y) \geq 0
\end{align*}
and let
\begin{align*}
    c_A = c_B = \frac{2 \min \{ \g_1, |\g_2|\}}{\pi^2} \, ,
\end{align*}
then the following two inequalities hold.
\begin{enumerate}
\setcounter{enumi}{6}
    \item $c_A(\pi - y)^2 \leq A(y)$ for all $y \in [0, \pi]$. 
    \item $c_B y^2 \leq B(y)$ for all $y \in [0, \pi]$.
\end{enumerate}
\begin{proof}
Properties (1-5) are straightforward computations. We prove properties (6-8) in Appendix \ref{pf: k(y) facts}. 

\end{proof}
\end{lem}

The following proposition and its corollary are our main tools for obtaining both spatially uniform and spatially weighted decay rates. Frequently, our strategy will be to write one of the three representative terms, in whole or in part, in a form to which we can apply these results.

\begin{prop}\label{prop: application of VDC}
Let $n \in \mathbb{Z}$ and suppose that $g \in W^{1,1}([0,\pi])$. Then the following two decay rate estimates hold,
\begin{align*}
    \sup_{z \in [0, \pi]} \l| \int  \limits_0^{z} e^{-i k(y) t} e^{iny} g(y) \, dy  \r| \lesssim 
    \begin{cases}
    \l(1+ \min\{ \g_1, |\g_2|\}^{-1/2}\r) \langle t \rangle ^{-1/3} \|g\|_{W^{1,1}} \\
    \l(1+\min\{\g_1, |\g_2|\}^{-1}\r) \l[\langle t\rangle^{-1/2}  + n\langle t \rangle^{-1} \r] \|g\|_{W^{1,1}} \, . 
    \end{cases}
\end{align*}

\begin{proof}
 See Appendix \ref{pf: application of VDC}. 
\end{proof}
\end{prop}

\begin{cor}\label{cor: application of VDC}
Suppose that $f \in \ell_1^1(\mathbb{N}_0, \mathbb{C}^2)$. Then the following two decay rate estimates hold: for any $n \in \mathbb{Z}$, 
\begin{align*}
    \sup_{z \in [0, \pi]} \l| \int  \limits_0^{z} e^{-i k(y) t} e^{iny} \tilde{f}(y) \, dy  \r| \lesssim 
    \begin{cases}
    \l(1+ \min\{ \g_1, |\g_2|\}^{-1/2}\r) \langle t \rangle ^{-1/3} \|f\|_{\ell^1} \\
    \l(1+\min\{\g_1, |\g_2|\}^{-1}\r) \l[\langle t\rangle^{-1/2}  + n\langle t \rangle^{-1} \r] \|f\|_{\ell_1^1} \, . 
    \end{cases}
\end{align*}

\begin{proof}
Plugging $g(y) = \tilde{f}(y)$ into \cref{prop: application of VDC} proves the corollary with one exception in the first estimate. Notice that $\|\tilde{f}(y)\|_{W^{1,1}} \lesssim \|f\|_{\ell_1^1}$, but in the first estimate the weight is removed. This improvement can be made using a variant of Van der Corput's lemma discussed in \cite[Lemma~5.1]{Egorova_2015}. We present the essential idea here.

Writing out $\tilde{f}(y) = \sum_{m \geq 0} e^{-imy} f_m$, we compute
\begin{align*}
    \int  \limits_0^{\pi} e^{-i k(y) t} e^{iny} \tilde{f}(y) \, dy = \sum_{m\geq 0} f_m \int  \limits_0^{\pi} e^{-i k(y) t} e^{i(n-m)y}  \, dy  =  \sum_{m\geq 0} f_m \int  \limits_0^{\pi} e^{-i \phi(y) t} \, dy  \, ,
\end{align*}
where $\phi(y) := k(y) - (n-m)y/t$.

Applying Van der Corput's Lemma with the same partition of $[0, \pi]$ used in \cref{prop: application of VDC} yields 
\begin{align*}
    \l|\int  \limits_0^{\pi} e^{-i \phi(y) t} \, dy \r| \lesssim \l( \min\{ \g_1, |\g_2|\} \langle t \rangle \r)^{-1/3} \, .
\end{align*}
Hence \begin{align*}
    \l| \int  \limits_0^{\pi} e^{-i k(y) t} e^{iny} \tilde{f}(y) \, dy \r| \lesssim \sum_{m \geq 0} |f_m| \l( \min\{ \g_1, |\g_2|\} \langle t \rangle \r)^{-1/3} = \l( \min\{ \g_1, |\g_2|\} \langle t \rangle \r)^{-1/3} \|f\|_{\ell^1} \, .
\end{align*}

\end{proof}
\end{cor}

\begin{rmk}
In similar fashion, the first bound of \cref{prop: application of VDC} can be improved to 
\begin{align*}
    \sup_{z \in [0, \pi]} \l| \int  \limits_0^{z} e^{-i k(y) t} e^{iny} g(y) \, dy  \r| \lesssim \l( \min\{ \g_1, |\g_2|\} \langle t \rangle \r)^{-1/3} \|g\|_{L^1}
\end{align*}  if we additionally assume that $g$ is in the Wiener algebra, i.e., $g$ is an absolutely convergent Fourier series. 

\end{rmk}

\begin{prop}\label{prop: properties of F}
Suppose $f \in \ell_1^1(\mathbb{N}_0, \mathbb{C}^2)$. Then $\lambda \mapsto \tilde{f}(\qzl)$ is 1/2-H\"older continuous. More precisely, 
for all $\lambda \in [\g_-, \g_+]$ and $\ep \geq 0$ satisfying $\lambda \pm \ep \in  [\g_-, \g_+]$, we have 
\begin{align*}
    \l|\tilde{f}(q_{*, \lambda \pm \ep}) - \tilde{f}(\qzl) \r| \leq \frac{\l(1 + \sqrt{2}\r)\sqrt{3}}{\sqrt{\min\{ \g_1, |\g_2|\}}} \sqrt{\ep}   \|f \|_{\ell_1^1} \, .
\end{align*}
Furthermore, if  $f \in \ell_2^1(\mathbb{N}_0, \mathbb{C}^2)$, 
then 
\begin{align*}
    \l|\tilde{f}'(q_{*, \lambda \pm \ep}) - \tilde{f}'(\qzl) \r| \leq \frac{\l(1 + \sqrt{2}\r)\sqrt{3}}{\sqrt{\min\{ \g_1, |\g_2|\}}} \sqrt{\ep}  \|f \|_{\ell_2^1}\,  .
\end{align*}

\begin{proof}
 See Appendix \ref{pf: properties of F}. 
\end{proof}
\end{prop}

\subsection{Type I}
Using the change of variables $\lambda = k(y) $, see \cref{lem: Change of Variables}, we compute 
\begin{align*}
    \text{Type I } &:= \int \limits_{\g_-}^{\g_+} e^{-i \lambda t} \frac{\lambda T_n(\eta(\lambda))}{\g_1 |\g_2| \sqrt{1 - \eta^2(\lambda)}} \tilde{f} (\qzl) \, d\lambda = \int \limits_0^{\pi} e^{-ik(y) t} \cos(ny) \tilde{f}(y) \, dy \\
    &= \frac{1}{2} \int \limits_0^{\pi} e^{-ik(y) t} \l(e^{iny} + e^{-iny} \r)  \tilde{f}(y)  \, dy\, .
\end{align*}
By \cref{cor: application of VDC}, we have the following two estimates
\begin{align*}
    \l|\text{Type I}(n,t) \r| \lesssim  \begin{cases}
    \l(1+ \min\{ \g_1, |\g_2|\}^{-1/2}\r) \langle t \rangle ^{-1/3} \|f\|_{\ell^1} \\
    \l(1+\min\{\g_1, |\g_2|\}^{-1}\r) \l[\langle t\rangle^{-1/2}  + n \langle t \rangle^{-1} \r] \|f\|_{\ell_1^1} \, . 
    \end{cases}
\end{align*}

\subsection{Type II}
\begin{prop}\label{prop: Type II estimate}
Suppose  $0 < \g_1, |\g_2|$  and that $\g_1 \neq |\g_2|$. If $f \in \ell_1^1(\mathbb{N}_0; \mathbb{C}^2)$, then the following two decay rate estimates hold,
\begin{align*}
    \l|\text{Type II}(n,t) \r| &\lesssim 
    \begin{cases}
        \l(\g_- \langle t \rangle\r)^{-1/3} \|f\|_{\ell_1^1}\\
        \l(\g_- \langle t \rangle\r)^{-1/2} \|f\|_{\ell_1^1} 
    \end{cases}
\end{align*}
\end{prop}

Using the change of variables $k = k(y)$, see \cref{lem: Change of Variables},  we compute 
\begin{align*}
    \text{Type II} &:= \int \limits_{\g_-}^{\g_+} e^{-i \lambda t}  \int \limits_{\g_-}^{\g_+} \l( \frac{1}{ \lambda + k} \r) \frac{k T_n(\eta(k))}{\g_1 |\g_2| \sqrt{1 - \eta^2(k)}} \, dk \,  \tilde{f}(\qzl) \, d \lambda \\
    &= \int \limits_{\g_-}^{\g_+} e^{-i \lambda t}  \int \limits_{0}^{\pi} \frac{\cos(ny)}{ \lambda + k(y)}   \, dy \,  \tilde{f}(\qzl) \, d \lambda  \\
    &= \int \limits_0^{\pi} \cos(ny) \int \limits_{\g_-}^{\g_+} e^{-i \lambda t} \frac{\tilde{f}(\qzl)}{\lambda + k(y)}\, d\lambda \, dy   \\
    &= \int \limits_0^{\pi} \cos(ny) e^{ik(y)t} \int \limits_{\g_-}^{\g_+} e^{-i (\lambda + k(y)) t} \frac{\tilde{f}(\qzl)}{\lambda + k(y)}\, d\lambda \, dy \\
    &= \text{Type IIa} + \text{Type IIb}\,  ,
\end{align*}
where \begin{align*}
    \text{Type IIa}(n,t) &:= \int \limits_0^{\pi} \cos(ny) \tilde{f}(y) e^{ik(y)t} \int \limits_{\g_-}^{\g_+}  \frac{e^{-i (\lambda + k(y)) t}}{\lambda + k(y)}\, d\lambda \, dy \\
    \text{Type IIb}(n,t) &:= \int \limits_0^{\pi} \cos(ny) e^{ik(y)t} \int \limits_{\g_-}^{\g_+} e^{-i (\lambda + k(y)) t} \frac{\tilde{f}(\qzl) - \tilde{f}(y)}{\lambda + k(y)}\, d\lambda \, dy   \, . 
\end{align*}

The proof of \cref{prop: Type II estimate} follows from \cref{lem: Type IIa estimate,lem: Type IIb estimate} below.

\begin{lem}\label{lem: Type IIa estimate}
Suppose  $0 < \g_1, |\g_2|$  and that $\g_1 \neq |\g_2|$. If $f \in \ell_1^1(\mathbb{N}_0; \mathbb{C}^2)$, then the following two decay rate estimates hold,
\begin{align*}
     |\text{Type IIa}(n,t)| \lesssim  \begin{cases}
         \l(\g_- \langle t \rangle\r)^{-1/3} \|f\|_{\ell^1}\\
         \l(\g_- \langle t \rangle\r)^{-1/2} \|f\|_{\ell^1}\, . 
     \end{cases}
\end{align*}

\begin{proof}
See Appendix \ref{pf: Type IIa estimate}. 
\end{proof}
\end{lem}

\begin{lem}\label{lem: Type IIb estimate}
Suppose  $0 < \g_1, |\g_2|$  and that $\g_1 \neq |\g_2|$. If $f \in \ell_1^1(\mathbb{N}_0; \mathbb{C}^2)$, then the following two decay rate estimates hold,
\begin{align*}
    |\text{Type IIb}(n,t)| \lesssim  \begin{cases}
         \l(\g_- \langle t \rangle\r)^{-1/3} \|f\|_{\ell_1^1} \\
         \l(\g_- \langle t \rangle\r)^{-1/2} \|f\|_{\ell_1^1} \, . 
     \end{cases}
\end{align*}

\begin{proof}
See Appendix \ref{pf: Type IIb estimate}.

\end{proof}
\end{lem}

\subsection{Type III}
Recall that $\g_- = ||\g_2| - \g_1|$ is the radius of the spectral gap. 
\begin{prop}\label{prop: Type III estimate}
Suppose  $0 < \g_1, |\g_2|$  and that $\g_1 \neq |\g_2|$. If $f \in \ell_1^1( \mathbb{N}_0; \mathbb{C}^2)$, then 
\begin{align*}
    \text{Type III}(n,t) := \int \limits_{\g_-}^{\g_+} e^{-i \lambda t}  \fint \limits_{\g_-}^{\g_+} \l( \frac{1}{k - \lambda} \r) \frac{k T_n(\eta(k))}{\g_1 |\g_2| \sqrt{1 - \eta^2(k)}} \, dk \, \tilde{f}(\qzl) \, d \lambda \, , 
\end{align*}
satisfies the following decay rate estimates,
\begin{align*}
    \l|\text{Type III } \r| &\lesssim \l(1 +  \min\{ \g_1, |\g_2| \}^{-2/3} \r) \langle t \rangle^{-1/3} \|f\|_{\ell_1^1} \\ 
    \l|\text{Type III } \r| &\lesssim
    \l(1 +  \min\{ \g_1, |\g_2| \}^{-1}  \r) \l[ \log\l(\sqrt{2+ t^2} \r) \langle t \rangle^{-1/2} + n\langle t \rangle^{-1}\r] \|f\|_{\ell_1^1} \, .
\end{align*}
Moreover, if $f \in \ell_2^1( \mathbb{N}_0; \mathbb{C}^2)$, then
\begin{align*}
    \l|\text{Type III } \r| \lesssim
    \l(1+\min\{\g_1, |\g_2|\}^{-1} + \g_-^{-1/2}\r) \l[\langle t\rangle^{-1/2}  + n\langle t \rangle^{-1} \r] \|f\|_{\ell_2^1} \, .
\end{align*}  
\end{prop}

In the remainder of this section, we split up $\text{Type III } = \text{Type IIIa } + \text{Type IIIb }$, where each term on the right hand side is more amenable to dispersive decay estimates. Under the assumption that $f \in \ell_1^1(\mathbb{N}_0, \mathbb{C}^2)$, we derive spatially uniform and spatially weighted estimates of the Type IIIa term in Section \ref{sec: Type IIIa Decay Rates}, as well as a $\log \l(\sqrt{2 +t} \r) \langle t \rangle^{-1/2}$ spatially uniform decay estimate of the Type IIIb term in Section \ref{sec: Type IIIb Uniform Decay}. Under the stronger assumption that $f\in \ell_2^1(\mathbb{N}_0, \mathbb{C}^2)$, we derive a $ \langle t \rangle^{-1/2}$ spatially weighted decay estimate of the Type IIIb term in Section \ref{sec: Type IIIb Local Decay}. Altogether, this proves \cref{prop: Type III estimate}.

Applying \cref{Gakhov: change of variables} with the change of variables $k = k(y)$ to the Type III term, followed by swapping the order of integration as in \cref{Gakhov: swap order} yields 
\begin{align*}
    \text{Type III } &:= \int \limits_{\g_-}^{\g_+} e^{-i \lambda t}  \fint \limits_{\g_-}^{\g_+} \l( \frac{1}{ \lambda - k} \r) \frac{k T_n(\eta(k))}{\g_1 |\g_2| \sqrt{1 - \eta^2(k)}} \, dk \,  \tilde{f}(\qzl) \, d \lambda \\
    &= \int \limits_{\g_-}^{\g_+} e^{-i \lambda t}  \fint \limits_{0}^{\pi} \l( \frac{1}{ \lambda - k(y)} \r) \cos(ny) \, dy \,  \tilde{f}(\qzl) \, d \lambda \\
    &= \int \limits_{0}^{\pi} \cos(ny)  \fint \limits_{\g_-}^{\g_+} e^{-i \lambda t} \frac{\tilde{f}(\qzl)}{\lambda - k(y)}\, d\lambda \, dy \, .
\end{align*}

\begin{rmk}
The statement of \cref{Gakhov: change of variables} requires that $k'(y)$ vanishes nowhere, but \cref{lem: k(y) facts} shows that $k'(y)$ vanishes at $y= 0, \pi$. In the proof of \cref{Gakhov: change of variables}, the requirement that the derivative of the change of variables transformation vanishes nowhere is only used to ensure the change of variables is invertible. In our case, $k(y)$ is invertible, so this does not cause a problem. 

\cref{Gakhov: swap order} assumes the regions of integration for the inner and outer integrals are both the same smooth contour $L$. In our case, our regions of integration are the intervals $[\g_-, \g_+]$ and $[0, \pi]$. However, this causes virtually no change to the proof of \cref{Gakhov: swap order}. 

\end{rmk}

Continuing with our calculations, 
\begin{align*}
    \text{Type III }
    &= \int \limits_{0}^{\pi} \cos(ny)  \fint \limits_{\g_-}^{\g_+} e^{-i \lambda t} \frac{\tilde{f}(\qzl)}{\lambda - k(y)}\, d\lambda \, dy  \\
    &= \int \limits_{0}^{\pi} \cos(ny) e^{-ik(y)t} \fint \limits_{\g_-}^{\g_+} e^{-i (\lambda-k(y)) t} \frac{\tilde{f}(\qzl)}{\lambda - k(y)}\, d\lambda \, dy \\
    &= \int \limits_{0}^{\pi} \cos(ny) \tilde{f}(y) e^{-ik(y)t} \fint \limits_{\g_-}^{\g_+}  \frac{e^{-i (\lambda-k(y)) t}}{\lambda - k(y)}\, d\lambda \, dy \\
    &\quad + \int \limits_{0}^{\pi} \cos(ny) e^{-ik(y)t} \fint \limits_{\g_-}^{\g_+} e^{-i (\lambda-k(y)) t} \frac{\tilde{f}(\qzl) - \tilde{f}(y)}{\lambda - k(y)}\, d\lambda \, dy \\
    &=: \text{Type IIIa } + \text{Type IIIb.}
\end{align*}
Using the change of variables $u = (\lambda - k(y)) t$ and defining
\begin{align}\label{eq: F(lambda)}
    F(\lambda) = \tilde{f}(\qzl)\, ,
\end{align}
we obtain
\begin{align}
    \text{Type IIIa} &:=  \int \limits_0^{\pi} \cos(ny)
     \tilde{f}(y) e^{-i k(y) t} \fint \limits_{\g_-}^{\g_+} \frac{ e^{-i (\lambda - k(y)) t} }{ \lambda - k(y)} \, d \lambda \, dy  \label{eq: Type IIIa def} \\ 
      &= \int \limits_0^{\pi} \cos(ny)
     \tilde{f}(y) e^{-i k(y) t} \fint \limits_{-A(y) t}^{B(y) t} \frac{ e^{-i u} }{ u} \, d u \, dy \\
    \text{Type IIIb} &:=  \int \limits_0^{\pi} \cos(ny) e^{-ik(y)t} \int \limits_{\g_-}^{\g_+} e^{-i (\lambda - k(y)) t} \frac{F(\lambda) - F(k(y))}{\lambda  - k(y)} \, d \lambda   \, dy  \label{eq: Type IIIb def} \\
     &= \int \limits_0^{\pi} \cos(ny) e^{-ik(y)t} \int \limits_{-A(y) t}^{B(y) t} \frac{e^{-i u} }{u} \l[F\l(k(y) + \frac{u}{t}\r) - F(k(y)) \r] \, du   \, dy\, ,
\end{align}
where $A(y)$ and $B(y)$ are as defined in Lemma \ref{lem: k(y) facts}. The benefit of this partition is that in Type IIIa, the singular integral does not depend on the data, $f$, and in Type IIIb, we can take advantage of the regularity of $F(\lambda)$, studied in \cref{prop: properties of F}, to conclude that the inner integral of the Type IIIb term is actually an integrable singularity.

\section{Estimates of an Oscillatory Cauchy Principal Value Integral}\label{sec: Type IIIa Decay Rates}
We assume throughout this section that $0< \g_1, |\g_2|$ and that $\g_1 \neq |\g_2|$. Recall that $k(y) = \sqrt{\g_1^2 + |\g_2|^2 +2 \g_1 |\g_2|  \cos(q)}$, see \eqref{def: k}. 
\begin{prop}\label{prop: Type IIIa estimate}   
Let $f \in \ell_1^1( \mathbb{N}_0; \mathbb{C}^2)$. Then, 
\begin{align}
    \text{Type IIIa}(n,t) =  \int \limits_0^{\pi} \cos(ny)
     \tilde{f}(y) e^{-i k(y) t} \fint \limits_{-A(y) t}^{B(y) t} \frac{ e^{-i u} }{ u} \, d u \, dy \, , 
\label{eq:TYPE-IIIa}\end{align}
satisfies the following decay rate estimates,
\begin{align*}
    \l|\text{Type IIIa }\r|
     \lesssim \begin{cases}
    \l(1+ \min\{ \g_1, |\g_2|\}^{-2/3}\r) \langle t \rangle ^{-1/3} \|f\|_{\ell^1} \\
    \l(1+\min\{\g_1, |\g_2|\}^{-1}\r) \l[\langle t\rangle^{-1/2}  + n\langle t \rangle^{-1} \r] \|f\|_{\ell_1^1} \, . 
    \end{cases}
\end{align*}
\end{prop}

The remainder of this section is devoted to the proof of \cref{prop: Type IIIa estimate}. 

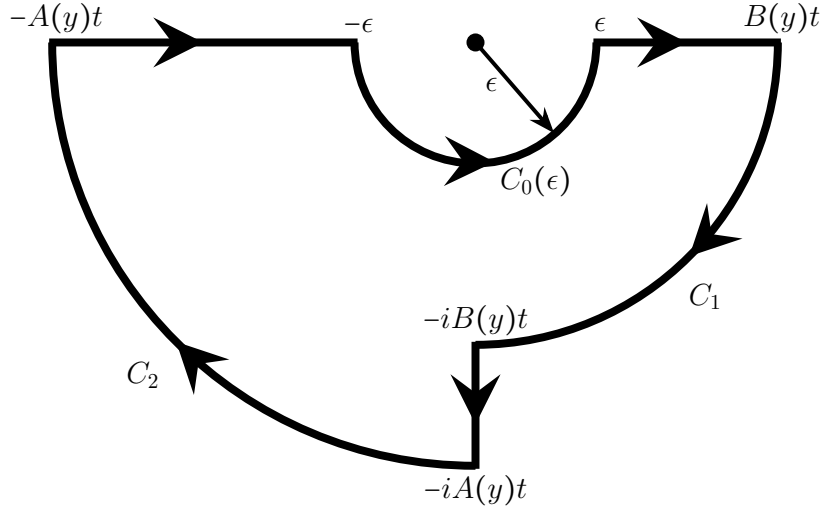
\begin{figure}[h]
\begin{tikzpicture}[scale=0.8, transform shape,
    rightarrow/.style={postaction={decorate,decoration={markings,mark=at position 0.5 with {\arrow[scale=2, #1]{stealth}}}}}, 
    rightarrow2/.style={postaction={decorate,decoration={markings,mark=at position 0.55 with {\arrow[scale=2, rotate=-10, #1]{stealth}}}}},
    rightarrow3/.style={postaction={decorate,decoration={markings,mark=at position 0.65 with {\arrow[scale=2, #1]{stealth}}}}},
    leftarrow/.style={postaction={decorate,decoration={markings,mark=at position 0.5 with {\arrowreversed[scale=2, #1]{stealth}}}}},
    ]
    \draw[rightarrow, line width = 3pt] (-0.05,0) -- (5.05,0);
    \draw[rightarrow2, line width = 3pt] (5,0) arc[start angle = 180, end angle = 360, radius = 2cm];
    \draw[rightarrow, line width = 3pt] (8.95,0) -- (12.05,0);
    \draw[leftarrow,  line width = 3pt] (7,-5) arc[start angle = 270, end angle = 360, radius = 5cm];
    \draw[rightarrow3, line width = 3pt] (7,-4.95) -- (7, -7.05);
    \draw[leftarrow, line width = 3pt] (0,0) arc[start angle = 180, end angle = 270, radius = 7cm];
    \draw[->, >=Stealth, line width=1.5pt] (7,0) node[circle, fill=black, inner sep=3pt] {} -- (8.3,-1.5);
    \node at (7.2, -0.7) {{  \Large$\ep$ }};
    \node at (9, 0.3) {{ \Large $\ep$}};
    \node at (5, 0.3) {{ \Large $-\ep$}};
    \node at (0, 0.4) {{  \Large $-A(y)t$}};
    \node at (12, 0.4) {{  \Large $B(y)t$}};
    \node at (7, -4.6) {{\Large $-iB(y) t$ }};
    \node at (7, -7.4) {{\Large $-iA(y) t$ }};
    \node at (8, -2.3) {{\Large $C_0(\ep)$}};
    \node at (10.8, -4.2) {{\Large $C_1$}};
    \node at (1.5, -5.5) {{\Large $C_2$}};
\end{tikzpicture}
\caption{Contour for the inner integral of the Type IIIa term when $B(y) < A(y)$.}
\label{fig: Type IIIa contour}
\end{figure}

In order to deal with the Cauchy principle value in Type IIIa, we deform the region of integration of the inner integral using the contour shown in Figure \ref{fig: Type IIIa contour}. By  Cauchy's theorem, 
\begin{align*}
    0 = &\int \limits_{-A(y)t}^{-\ep} \frac{e^{-iu}}{u}  \, du + \int \limits_{C_0 (\ep)} \frac{e^{-iu}}{u}  \, du + \int \limits_{\ep}^{B(y) t} \frac{e^{-iu}}{u}  \, du  +  \int \limits_{C_1} \frac{e^{-iu}}{u}  \, du + \int \limits_{-i \min\{ A(y), B(y)\} t}^{-i\max\{A(y), B(y)\} t} \frac{e^{-iu}}{u}  \, du  + \int \limits_{C_2} \frac{e^{-iu}}{u}  \, du \\
    \implies  &\fint \limits_{-A(y) t}^{B(y)t} \frac{e^{-iu}}{u}  \, du = - \lim_{\ep \rightarrow 0} \int \limits_{C_0 (\ep)} \frac{e^{-iu}}{u}  \, du  - \int \limits_{C_1} \frac{e^{-iu}}{u}  \, du - \int \limits_{-i \min\{A(y), B(y)\} t}^{-i\max \{A(y), B(y)\} t} \frac{e^{-iu}}{u}  \, du  - \int \limits_{C_2} \frac{e^{-iu}}{u}  \, du \, . \numberthis \label{eq: toy contour deformation}
\end{align*}

To study the right hand side of \eqref{eq: toy contour deformation}, we examine the integral of $e^{-iu}/u$ over the curves $C_0(\ep), C_1$, and $C_2$. First, letting $u = \ep e^{i \theta}$ for $- \pi \leq \theta \leq 0$, we compute 
\begin{align*}
    \int \limits_{C_0(\ep)} \frac{e^{-iu}}{u} \, du = \int \limits_{-\pi}^{0} \frac{e^{-i\ep e^{i \theta}}}{\ep e^{i \theta}} i \ep e^{i \theta} \, d\theta = i \int \limits_{-\pi}^{0} e^{-i\ep e^{i \theta}} \, d\theta \rightarrow i \pi
\end{align*}
as $\ep \rightarrow 0^+$. 
Next letting $u = B(y) t e^{-i \theta}$ for $0 \leq \theta \leq  \pi/2 $, we compute \begin{align*}
    \int \limits_{C_1} \frac{e^{-iu}}{u}  \, du &= \int \limits_{0}^{\pi/2} \frac{e^{-iB(y) t e^{-i \theta}}}{B(y) t e^{-i \theta}} i B(y) t e^{-i \theta}  \, d\theta = i \int \limits_0^{\pi/2} e^{-iB(y) t \cos(\theta)} e^{-B(y) t \sin(\theta)} \, d\theta \\
    \implies \l| \;\int \limits_{C_1} \frac{e^{-iu}}{u}  \, du \r|  &\leq  \int \limits_0^{\pi/2} e^{-B(y) t \sin(\theta)} \, d\theta \, .
\end{align*}

Using the fact that $\sin(\theta) \geq  2 \theta/ \pi $ for all  $\theta \in [0, \pi/2]$, we compute \begin{align*}
    \int \limits_0^{\pi/2} e^{-B(y) t \sin(\theta)} \, d\theta  \leq \int \limits_0^{\pi/2} e^{-2B(y)t\theta/ \pi} \, d\theta  \leq \frac{\pi}{2 B(y)t} \, ,\quad  \textrm{for any $t>0$}.
\end{align*}
Since we can also bound the integral from above by $\pi/2$, it follows that \begin{align*}
    \l|\,\; \int \limits_{C_1} \frac{e^{-iu}}{u}  \, du \;\r| \leq \frac{\pi}{1 + B(y)t} \, .
\end{align*}
In similar fashion, we obtain\begin{align*}
    \l|\,\; \int \limits_{C_2} \frac{e^{-iu}}{u}  \, du \;\r| \leq \frac{\pi}{1 + A(y)t} \, .
\end{align*}
Lastly note that \begin{align*}
    \int \limits_{-i \min\{A(y), B(y)\} t}^{-i\max\{A(y), B(y)\} t} \frac{e^{-iu}}{u}  \, du = \int \limits
     _{\min\{A(y), B(y)\} t}^{\max\{A(y), B(y)\} t} \frac{e^{-x}}{x} \, dx \, .
\end{align*}
Thus, the Cauchy principal value in \eqref{eq:TYPE-IIIa} is given by \begin{align}
    \fint \limits_{-A(y) t}^{B(y)t} \frac{e^{-iu}}{u}  \, du = 
    -i \pi - \int \limits_{C_1} \frac{e^{-iu}}{u}  \, du - \int \limits_{C_2} \frac{e^{-iu}}{u}  \, du  - \int \limits
     _{\min\{A(y), B(y)\} t}^{\max \{A(y), B(y)\} t} \frac{e^{-x}}{x} \, dx \label{eq: result of contour integral 1} \\
    \text{ where }\quad 
    \l|\; \int \limits_{C_1} \frac{e^{-iu}}{u}  \, du  \;\r|  \leq \frac{\pi}{1+ B(y)t}
    \quad \text{and} \quad  \l|\; \int \limits_{C_2} \frac{e^{-iu}}{u}  \, du \;\r| \leq \frac{\pi}{1+ A(y)t} \label{eq: result of contour integral 2} \, .
\end{align}

Equation \eqref{eq: result of contour integral 1} allows us to decompose the inner integral of the Type IIIa term, see \eqref{eq: Type IIIa def},  into more manageable pieces. In  the following lemmas, we estimate the contribution of each part of the contour in \cref{fig: Type IIIa contour} to the decay rate of the Type IIIa term. \cref{lem: Type IIIa curved parts estimate} corresponds to the first three terms on the right hand side of \eqref{eq: result of contour integral 1} and \cref{lem: Type IIIa vertical part estimate}  corresponds to the last term on the right hand side of \eqref{eq: result of contour integral 1}.

\begin{lem}\label{lem: Type IIIa curved parts estimate}
The contribution of the curved parts of the contour in \cref{fig: Type IIIa contour} to the Type IIIa term satisfies the following estimate,
\begin{align*}
    &\l| \int \limits_0^{\pi} \cos(ny)
     \tilde{f}(y) e^{-i k(y) t} \l[ -i \pi - \int \limits_{C_1} \frac{e^{-iu}}{u}  \, du - \int \limits_{C_2} \frac{e^{-iu}}{u}  \, du \r] \, dy \r| 
     \\
     &\quad \lesssim   
     \begin{cases}
    \l(1+ \min\{ \g_1, |\g_2|\}^{-1/2}\r) \langle t \rangle ^{-1/3} \|f\|_{\ell^1} \\
    \l(1+\min\{\g_1, |\g_2|\}^{-1}\r) \l[\langle t\rangle^{-1/2}  + n\langle t \rangle^{-1} \r] \|f\|_{\ell_1^1} \, . 
    \end{cases}
\end{align*}

\begin{proof}
See Appendix \ref{pf: Type IIIa curved parts estimate}.
\end{proof}
\end{lem}

\begin{lem}\label{lem: Type IIIa vertical part estimate}
The contribution of the vertical part of the contour in \cref{fig: Type IIIa contour} to the Type IIIa term satisfies the following estimate,
\begin{align*}
    \l| \int \limits_{0}^{\pi}  \cos(ny)
     \tilde{f}(y) e^{-i k(y) t}  \int \limits
     _{\min \{A(y), B(y)\} t}^{\max\{A(y), B(y)\} t} \frac{e^{-x}}{x} \, dx \, dy \r|   \lesssim \l(1+ \min \{ \g_1, |\g_2|\}^{-2/3} \r) \|f\|_{\ell^1}  \langle t \rangle^{-1/2} \, .
\end{align*}

\begin{proof}
See  Appendix \ref{pf: Type IIIa vertical part estimate}.
\end{proof}
\end{lem}

The proof of Proposition \ref{prop: Type IIIa estimate} follows from Lemma 
\ref{lem: Type IIIa curved parts estimate} and  Lemma \ref{lem: Type IIIa vertical part estimate}.

\section{Spatially Uniform Estimates of an Oscillatory Integral} \label{sec: Type IIIb Uniform Decay} 

Van der Corput's lemma, see \cref{thm: VDC}, is one of the main tools we have used throughout this paper to obtain oscillatory integral estimates. 
However, when computing a uniform decay rate estimate for the Type IIIb term in \cref{prop: Type IIIb uniform decay estimate}, there arises a degree of freedom, parametrized by $\alpha \in (2, \infty)$, for what natural change of variables one could make. This change of variables leads to an oscillatory integral of the form 
\begin{align*}
    \int \limits_0^b e^{i x^{\alpha} t} \psi(x) \, dx\, .
\end{align*}
An application of the classical Van der Corput's lemma can give at best a $t^{-1/3}$ decay rate, corresponding to the choice of $\alpha =3$, since VDC's lemma requires $\alpha \in \mathbb{N}$. However, intuitively we expect decay rates of $t^{-1/\alpha}$ for any real $\alpha > 2$. This intuition is confirmed by \cref{Dewez: fractional phase}, a consequence of \cite{Dewez_2018}. We use this result to obtain a family of $t^{-1/\alpha}$ decay rate estimates, spatially uniform in the cell site number $n$, for the Type IIIb term in \cref{prop: Type IIIb uniform decay estimate}. The tradeoff for improving the decay rate from $t^{-1/3}$ to $t^{-1/ \alpha}$ for any $\alpha > 2$ is that the constant attached to the decay rate blows up as $\alpha \rightarrow 2^+$. In \cref{cor: Type IIIb uniform decay estimate}, we optimize the choice of $\alpha$ to obtain a $\log \l(\sqrt{2 +t} \r) \langle t \rangle^{-1/2}$ spatially uniform rate of decay.

Recall that the Type IIIb term, see \eqref{eq: Type IIIb def}, is defined to be 
\begin{align*}
    \text{Type IIIb}(n,t) &:=  \int \limits_0^{\pi} \cos(ny) e^{-ik(y)t} \int \limits_{\g_-}^{\g_+} e^{-i (\lambda - k(y)) t} \frac{F(\lambda) - F(k(y))}{\lambda  - k(y)} \, d \lambda   \, dy \, ,
\end{align*}
where $F(\lambda) = \tilde{f}(\qzl)$.

\begin{prop} \label{prop: Type IIIb uniform decay estimate}
Suppose  $0 < \g_1, |\g_2|$  and that $\g_1 \neq |\g_2|$. 
Fix $\alpha \in (2, 3]$. Then for all $n\ge1$ and $t\ge0$, 
\begin{align*}
    |\text{Type IIIb}(n,t)| \lesssim  \frac{\l( \min\{ \g_1, |\g_2|\}\langle t \rangle \r)^{-1/\alpha}}{\alpha -2} \|f\|_{\ell_1^1}  \, .
\end{align*}

\begin{proof}
We compute 
\begin{align*}
    &\text{Type IIIb} =  \int \limits_0^{\pi} \cos(ny)  \int \limits_{\g_-}^{\g_+} e^{-i \lambda t} \frac{F(\lambda) - F(k(y))}{\lambda  - k(y)} \, d \lambda   \, dy  \\\
    &\quad =  \int \limits_0^{\pi} \cos(ny)  \int \limits_{\g_-}^{\g_+} \frac{e^{-i \lambda t}}{\lambda - k(y)}  \int \limits_{k(y)}^{\lambda} F'(s) \, ds \, d \lambda   \, dy \\
    &\quad = \int \limits_0^{\pi} \cos(ny)  \int \limits_{\g_-}^{k(y)} \frac{e^{-i \lambda t}}{k(y) - \lambda }  \int \limits_{\lambda}^{k(y)} F'(s) \, ds \, d \lambda   \, dy + \int \limits_0^{\pi} \cos(ny)  \int \limits_{k(y)}^{\g_+} \frac{e^{-i \lambda t}}{\lambda - k(y)}  \int \limits_{k(y)}^{\lambda} F'(s) \, ds \, d \lambda   \, dy \\
    &\quad =: \text{Type IIIb1} + \text{Type IIIb2}  \, .
\end{align*}

The analyses for $\text{Type IIIb1} $ and $ \text{Type IIIb2}$ are similar, so we show only the estimates for $\text{Type IIIb2}$. Our strategy is to use the change of variables $u = (\lambda - k(y))^{1/\alpha}$ in order to produce the oscillatory term $e^{-iu^{\alpha}t}$, from which we obtain a $t^{-1/\alpha}$ decay rate. Before doing so, we multiply  by $e^{-ik(y)t} e^{ik(y)t}$ and take absolute values, 
\begin{align*}
    \text{Type IIIb2}  &= \int \limits_0^{\pi} \cos(ny) e^{-ik(y) t} \int \limits_{k(y)}^{\g_+} \frac{e^{-i (\lambda - k(y)) t}}{\lambda - k(y)}  \int \limits_{k(y)}^{\lambda} F'(s) \, ds \, d \lambda   \, dy \\
    |\text{Type IIIb2}| &\leq  \int \limits_0^{\pi} \l|\; \int \limits_{k(y)}^{\g_+} \frac{e^{-i (\lambda - k(y)) t}}{\lambda - k(y)}  \int \limits_{k(y)}^{\lambda} F'(s) \, ds \, d \lambda \r| \, dy \\
    &=   \alpha \int \limits_0^{\pi} \l|\; \int \limits_{0}^{\l( \g_+ - k(y)\r)^{1/\alpha}} \frac{e^{-i u^{\alpha} t}}{u}  \int \limits_{k(y)}^{u^{\alpha} + k(y)} F'(s) \, ds \, du \r| \, dy \\
    &=  \alpha \int \limits_0^{\pi} \l|\; \int \limits_{0}^{\l( \g_+ - k(y)\r)^{1/\alpha}} e^{-i u^{\alpha} t}  N_2(u; y) \, du \r| \, dy \, ,
\end{align*}
where \begin{align*}
    N_2(u; y) := \frac{1}{u} \int \limits_{k(y)}^{u^{\alpha} + k(y)} F'(s) \, ds \, . 
\end{align*}

Applying \cref{Dewez: fractional phase} and using our assumption that $\alpha \leq 3$, we compute 
\begin{align}\label{eq: Type IIIb2 uniform est 1}
    |\text{Type IIIb2}| \lesssim t^{-1/\alpha} \int \limits_0^{\pi} \l[ \l|N_2\l( (\g_+ - k(y))^{1/\alpha}; y \r) \r| + \int \limits_0^{\l( \g_+ - k(y)\r)^{1/\alpha}} |\partial_u N_2(u; y)| \, du \r] \, dy \, . 
\end{align}

Since $F(\lambda)$ is 1/2  H\"older continuous by \cref{prop: properties of F}, it follows that
\begin{align*}
    |N_2(u; y)| &\leq \frac{|F(u^{\alpha} + k(y)) - F(k(y))| }{u} \lesssim u^{\alpha/2 -1} \, .
\end{align*}
Intuitively,  $\partial_u |N_2(u; y)|$ should behave like $ u^{\alpha/2 -2}$, and thus we expect that $\partial_u |N_2(u; y)|$ be integrable for all $\alpha > 2$. \cref{lem: bound on N and N'} (see below) makes this intuition rigorous. Applying \cref{lem: bound on N and N'} to \eqref{eq: Type IIIb2 uniform est 1}, we obtain
\begin{align*}
    |\text{Type IIIb2}| \lesssim  \frac{\l( \min\{ \g_1, |\g_2|\}\langle t \rangle \r)^{-1/\alpha}}{\alpha -2} \|f\|_{\ell_1^1}  \, .
\end{align*}

\end{proof}
\end{prop}

\begin{lem}\label{lem: bound on N and N'}
For any $\alpha > 2$, let 
\begin{align*}
    N_2(u; y):=   \frac{1}{u}  \int \limits_{k(y)}^{u^{\alpha} +  k(y)} F'(s) \, ds \, .
\end{align*} 
Then 
\begin{align*}
    &\int \limits_0^{\pi} \l|N_2\l( (\g_+ - k(y))^{1/\alpha}; y \r) \r| \, dy \lesssim   \|f\|_{\ell_1^1} \l( \min\{ \g_1, |\g_2|\} \r)^{-1/\alpha}  \\ \text{and }
    &
    \int \limits_0^{\pi} \int \limits_{0}^{(\g_+-k(y))^{1/\alpha}} \l| \partial_u N_2(u; y) \r| \, du \, dy \lesssim \l(1 +  \frac{1}{\alpha -2}\r) \|f\|_{\ell_1^1} \l( \min\{ \g_1, |\g_2|\} \r)^{-1/\alpha} \, . 
\end{align*}

\begin{proof}
See Appendix \ref{pf: bound on N and N'}. 

\end{proof}
\end{lem}

By optimizing the choice of $\alpha$ for each time $t$ in \cref{prop: Type IIIb uniform decay estimate}, we obtain the following corollary. 

\begin{cor}\label{cor: Type IIIb uniform decay estimate}
Suppose  $0 < \g_1, |\g_2|$  and that $\g_1 \neq |\g_2|$. Then 
\begin{align*}
    |\text{Type IIIb}(n,t)| \lesssim \l(1+ \min\{ \g_1, |\g_2|\}^{-1/2}\r) \log\l(\sqrt{2 + t^2} \r)\langle t \rangle^{-1/2} \|f\|_{\ell_1^1} \, . 
\end{align*}

\begin{proof}
See Appendix \ref{pf: Type IIIb uniform decay estimate corollary}.

\end{proof}
\end{cor}

\section{Spatially Weighted Estimates of an Oscillatory Integral} \label{sec: Type IIIb Local Decay}

\begin{prop}\label{prop: Type IIIb local decay estimate}
Suppose  $0 < \g_1, |\g_2|$  and that $\g_1 \neq |\g_2|$. If $f \in \ell_2^1( \mathbb{N}_0; \mathbb{C}^2)$, then 
\begin{align*}
    |\text{Type IIIb(n,t)}| \lesssim  \l(1+\min\{\g_1, |\g_2|\}^{-1} + \g_-^{-1/2}\r) \l[\langle t\rangle^{-1/2}  + n\langle t \rangle^{-1} \r] \|f\|_{\ell_2^1}
\end{align*}

\end{prop}

As in the proof of \cref{prop: Type IIIb uniform decay estimate}, we write Type IIIb = Type IIIb1 + Type IIIb2, and focus only on the Type IIIb2 term because the analysis for Type IIIb1 is similar. Recall that the change of variables $u = (\lambda - k(y))^{1/\alpha}$ gives
\begin{align*}
    \text{Type IIIb2}  &= \int \limits_0^{\pi} \cos(ny) e^{-ik(y) t} \int \limits_{k(y)}^{\g_+} \frac{e^{-i (\lambda - k(y)) t}}{\lambda - k(y)}  \int \limits_{k(y)}^{\lambda} F'(s) \, ds \, d \lambda   \, dy \\
    &= \alpha \int \limits_0^{\pi} \cos(ny) e^{-ik(y) t}  \int \limits_{0}^{(\g_+-k(y))^{1/\alpha}} \frac{e^{-i u^{\alpha} t}}{u} \int \limits_{k(y)}^{u^{\alpha} + k(y)} F'(s) \, ds \, du \, dy  \, .
\end{align*}
Hence \begin{align}\label{eq: Type IIIb2 with M(y)}
    \text{Type IIIb2} = \alpha \int \limits_0^{\pi} \cos(ny) e^{-ik(y) t} M(y) \, dy \, ,
\end{align}
where 
\begin{align}\label{eq: M(y) def}
    M(y) := \int \limits_{0}^{(\g_+-k(y))^{1/\alpha}} \frac{e^{-i u^{\alpha} t}}{u} \int \limits_{k(y)}^{u^{\alpha} + k(y)} F'(s) \, ds \, du \, . 
\end{align}

Applying the second estimate of \cref{prop: application of VDC} to \eqref{eq: Type IIIb2 with M(y)} yields 
\begin{align*}
    |\text{Type IIIb2}| \lesssim \alpha \l(1+\min\{\g_1, |\g_2|\}^{-1}\r) \l[\langle t\rangle^{-1/2}  + n\langle t \rangle^{-1} \r] \|M\|_{W^{1,1}} \, .
\end{align*}

We proceed to bound $\|M\|_{L^1}$ and $\|M'\|_{L^1}$ in terms of $f$. 

\begin{lem}\label{lem: M bound}
\begin{align*}
    \|M \|_{L^1} \lesssim \alpha^{-1} \| f\|_{\ell_1^1} \, .
\end{align*}

\begin{proof}
See Appendix \ref{pf: M bound}. 
\end{proof}
\end{lem}

Next we compute
\begin{align*}
    &M'(y) = \frac{-k'(y)}{\alpha}   \frac{e^{-i(\g_+ -k (y))t}}{\g_+ - k(y)} \int \limits_{k(y)}^{\g_+} F'(s) \, ds  +  k'(y) \int \limits_0^{(\g_+ - k(y))^{1/\alpha}} \frac{e^{-iu^{\alpha}t}}{u}\l(  F'(u^{\alpha} + k(y)) - F'(k(y)) \r) \, du \\
   &|M'(y)| \leq \frac{| k'(y)|}{\alpha(\g_+ - k(y))}  \int \limits_{k(y)}^{\g_+} |F'(s)| \, ds  +  |k'(y)| \int \limits_0^{(\g_+ - k(y))^{1/\alpha}} \frac{| F'(u^{\alpha} + k(y)) - F'(k(y))| }{u}\, du \\
   &\int \limits_0^{\pi} |M'(y)| \, dy \leq \int \limits_{\g_-}^{\g_+} \frac{1}{\alpha( \g_+ - v)} \int \limits_v^{\g_+} |F'(s)| \, ds \, dv + \int \limits_{\g_-}^{\g_+} \int \limits_0^{(\g_+ - v)^{1/\alpha}} \frac{| F'(u^{\alpha} + v - F'(v) | }{u}\, du \, dv \, ,
\end{align*}
where we used the change of variables $v  = k(y)$ in the last line. Define \begin{align*}
    M_1 &:= \int \limits_{\g_-}^{\g_+} \frac{1}{\alpha( \g_+ - v)} \int \limits_v^{\g_+} |F'(s)| \, ds \, dv \\
    M_2 &:= \int \limits_{\g_-}^{\g_+} \int \limits_0^{(\g_+ - v)^{1/\alpha}} \frac{| F'(u^{\alpha} + v) - F'(v)| }{u}\, du \, dv \, , 
\end{align*}
so that $\|M'\|_{L^1} \leq M_1 + M_2$.

\begin{lem}\label{lem: M1 bound}
\begin{align*}
    M_1 \lesssim \alpha^{-1}  \|f\|_{\ell_1^1}\, .  
\end{align*}
\begin{proof}
See Appendix \ref{pf: M1 bound}. 
\end{proof}
\end{lem}

We now study the integrand of $M_2$. 
Recall from \eqref{eq: F'(lambda)} that \begin{align*}
    F'(\lambda) =  \frac{-2 \lambda}{\sqrt{\g_+ + \lambda}\sqrt{\lambda + \g_-}} \frac{ \tilde{f}'(\qzl)}{\sqrt{\g_+ - \lambda}\sqrt{\lambda - \g_-}} \,. 
\end{align*}

By writing 
$F'(\lambda) = L(\lambda) \mu(\lambda)$, where
\begin{align*}
 L(\lambda) &:= \frac{-2 \lambda}{\sqrt{\g_+ + \lambda}\sqrt{\lambda + \g_-}} \qquad \text{ and } \qquad \mu(\lambda) := \frac{ \tilde{f}'(\qzl)}{\sqrt{\g_+ - \lambda}\sqrt{\lambda - \g_-}} \, ,
\end{align*}
we split up the expression for $F'(\lambda)$ into a Lipschitz continuous part (see \cref{lem: L(lambda) Lipschitz} below) and a remainder term. 

\begin{lem}\label{lem: L(lambda) Lipschitz}
The function $L  : [\g_-, \g_+] \rightarrow \mathbb{R}$ given by 
\begin{align*}
    L(\lambda) := \frac{-2 \lambda}{\sqrt{\g_+ + \lambda}\sqrt{\lambda + \g_-}}
\end{align*}
is Lipschitz. Its  Lipschitz constant $c_L$ satisfies the following bound,  
\begin{align*}
    c_L \leq \frac{3}{4} \frac{1}{\sqrt{\g_- \max\{\g_1, |\g_2| \}}}\, .
\end{align*}

\begin{proof}
See Appendix \ref{pf: L(lambda) Lipschitz}. 
\end{proof}
\end{lem}
With this notation, we compute
\begin{align*}
    &F'(u^{\alpha} + v) - F'(v)  =  L(u^{\alpha} + v) \mu(u^{\alpha} + v) - L(v)\mu(v) \\
    &= L(u^{\alpha} + v) \l[ \mu(u^{\alpha} + v) - \mu(v) \r] + \l[L(u^{\alpha} + v) - L(v) \r] \mu(v)\, .
\end{align*}
Hence by \cref{lem: L(lambda) Lipschitz} and the fact that $L(\lambda)$ is bounded, \begin{equation}\label{eq: F' difference bound} \begin{split}
    |F'(u^{\alpha} + v) - F'(v)| &\leq \|L\|_{\infty} \l| \mu(u^{\alpha} + v) - \mu(v) \r|  + \l|L(u^{\alpha} + v) - L(v) \r|\mu(v) \\
    &\lesssim  \l| \mu(u^{\alpha} + v) - \mu(v) \r|  + \frac{ u^{\alpha} |\mu(v)|}{\sqrt{\g_- \max\{\g_1, |\g_2| \}}}  \, .
\end{split}
\end{equation}
Plugging \eqref{eq: F' difference bound} into the expression for $M_2$, we obtain \begin{equation}\label{eq: M2 bound}\begin{split}
    M_2 &= \int \limits_{\g_-}^{\g_+} \int \limits_0^{(\g_+ - v)^{1/\alpha}} \frac{| F'(u^{\alpha} + v) - F'(v)| }{u}\, du \, dv \\
    &\lesssim \int \limits_{\g_-}^{\g_+} \int \limits_0^{(\g_+ - v)^{1/\alpha}} \frac{\l| \mu(u^{\alpha} + v) - \mu(v) \r| }{u} + \frac{u^{\alpha - 1} \mu(v)}{\sqrt{\g_- \max\{\g_1, |\g_2| \}}}\, du \, dv \\
    &= \int \limits_{\g_-}^{\g_+} \int \limits_0^{(\g_+ - v)^{1/\alpha}} \frac{1}{u} \l| \frac{\tilde{f}'(q_{*, u^{\alpha}+ v}) }{\sqrt{\g_+ - \l(u^{\alpha} +v \r)}\sqrt{u^{\alpha}+v - \g_-}} - \frac{\tilde{f}'(q_{*, v}) }{\sqrt{\g_+ - v} \sqrt{v - \g_-}} \r|  \, du \, dv    \\ &\quad +
    \frac{ 1}{\sqrt{\g_- \max\{\g_1, |\g_2| \}}} \int \limits_{\g_-}^{\g_+} \int \limits_0^{(\g_+ - v)^{1/\alpha}}  \frac{ u^{\alpha -1} |\tilde{f}'(q_{*, v})|}{\sqrt{\g_+ - v}\sqrt{v - \g_-}} \, du \, dv \\
    &:= M_3 + M_4 \, . 
\end{split}\end{equation}

For the integrand of $M_3$, we use \cref{prop: properties of F} and the fact that $\|\tilde{f}'\|_{\infty} \leq \|f\|_{\ell_1^1}$  to compute
\begin{align*}
    & \l| \frac{ \tilde{f}'(q_{*, u^{\alpha} + v}) }{\sqrt{\g_+ - (u^{\alpha} + v)} \sqrt{u^{\alpha} + v  - \g_-}} - \frac{\tilde{f}'(q_{*, v}) }{\sqrt{\g_+ - v} \sqrt{v  - \g_-}} \r| \\
    &\leq \l| \frac{ \tilde{f}'(q_{*, u^{\alpha} + v}) }{\sqrt{\g_+ - (u^{\alpha} + v)} \sqrt{u^{\alpha} + v  -\g_-}} - \frac{\tilde{f}'(q_{*,u^{\alpha} + v}) }{\sqrt{\g_+ - v} \sqrt{v  -\g_-}} \r|
    + \l|  \frac{  \tilde{f}'(q_{*, u^{\alpha} + v}) - \tilde{f}'(q_{*, v}) }{\sqrt{\g_+ - v} \sqrt{v  -\g_-}} \r| \\
    &\lesssim \|f\|_{\ell_1^1} \l|\frac{ 1}{\sqrt{\g_+ - (u^{\alpha} + v)} \sqrt{u^{\alpha} + v - \g_-}} - \frac{1 }{\sqrt{\g_+ - v} \sqrt{v  -\g_-}}  \r| + \frac{1}{\sqrt{\min\{ \g_1, |\g_2|\}}}  \frac{u^{\alpha/2} \|f\|_{\ell_2^1} }{\sqrt{\g_+ - v} \sqrt{v - \g_-}} \,. 
\end{align*}
Hence \begin{equation}\label{eq: M3 bound}\begin{split}
    M_3 &:= \int \limits_{\g_-}^{\g_+} \int \limits_0^{(\g_+ - v)^{1/\alpha}} \frac{1}{u} \l| \frac{\tilde{f}'(q_{*, u^{\alpha}+ v}) }{\sqrt{\g_+ - \l(u^{\alpha} +v \r)}\sqrt{u^{\alpha}+v - \g_-}} - \frac{\tilde{f}'(q_{*, v}) }{\sqrt{\g_+ - v} \sqrt{v - \g_-}} \r|  \, du \, dv \\
    &\lesssim \|f\|_{\ell_1^1} \int \limits_{\g_-}^{\g_+} \int \limits_0^{(\g_+ - v)^{1/\alpha}} \frac{1}{u} \l|\frac{ 1}{\sqrt{\g_+ - (u^{\alpha} + v)} \sqrt{u^{\alpha} + v - \g_-}} - \frac{1 }{\sqrt{\g_+ - v} \sqrt{v  -\g_-}}  \r| \, du \, dv \\
    &\quad + \frac{ \|f\|_{\ell_2^1}}{\sqrt{\min\{ \g_1, |\g_2|\}}}   \int \limits_{\g_-}^{\g_+} \int \limits_0^{(\g_+ - v)^{1/\alpha}}  
    \frac{u^{\alpha/2 -1} }{\sqrt{\g_+ - v} \sqrt{v - \g_-}} \, du \, dv \\
    &:= M_{3a} + M_{3b}\, .
\end{split}\end{equation}

\begin{lem}\label{lem: M3a bound}
\begin{align*}
    M_{3a} &:= \|f\|_{\ell_1^1} \int \limits_{\g_-}^{\g_+} \int \limits_0^{(\g_+ - v)^{1/\alpha}} \frac{1}{u} \l|\frac{ 1}{\sqrt{\g_+ - (u^{\alpha} + v)} \sqrt{u^{\alpha} + v - \g_-}} - \frac{1 }{\sqrt{\g_+ - v} \sqrt{v  -\g_-}}  \r| \, du \, dv \\
    &\lesssim  \alpha^{-1} \|f\|_{\ell_1^1} \, . 
\end{align*}

\begin{proof}
See Appendix \ref{pf: M3a bound}. 

\end{proof} 
\end{lem}

\begin{lem}\label{lem: M3b bound}
\begin{align*}
    M_{3b} := \frac{ \|f\|_{\ell_2^1}}{\sqrt{\min\{ \g_1, |\g_2|\}}} \int \limits_{\g_-}^{\g_+} \int \limits_0^{(\g_+ - v)^{1/\alpha}}  
    \frac{u^{\alpha/2 -1} }{\sqrt{\g_+ - v} \sqrt{v - \g_-}} \, du \, dv \lesssim \alpha^{-1} \|f\|_{\ell_2^1} \, .
\end{align*}
\begin{proof}
\begin{align*}
    \int \limits_{\g_-}^{\g_+} \int \limits_0^{(\g_+ - v)^{1/\alpha}}  
    \frac{u^{\alpha/2 -1} }{\sqrt{\g_+ - v} \sqrt{v - \g_-}} \, du \, dv &= \frac{2}{\alpha}  \int \limits_{\g_-}^{\g_+} \frac{1}{\sqrt{v- \g_-}} \, dv = \frac{4}{\alpha} \sqrt{\g_+ - \g_-} = \frac{4 \sqrt{2 \min\{\g_1, |\g_2| \}}}{\alpha} \,. 
\end{align*}
\end{proof}
\end{lem}

\begin{lem}\label{lem: M4 bound}
\begin{align*}
    M_4 := \frac{ 1}{\sqrt{\g_- \max\{\g_1, |\g_2| \}}} \int \limits_{\g_-}^{\g_+} \int \limits_0^{(\g_+ - v)^{1/\alpha}}  \frac{ u^{\alpha -1} |\tilde{f}'(q_{*, v})|}{\sqrt{\g_+ - v}\sqrt{v - \g_-}} \, du \, dv \lesssim \sqrt{\frac{\min\{ \g_1, |\g_2|\}}{\g_-}} \frac{\|f\|_{\ell_1^1}}{\alpha}  \, . 
\end{align*}
\begin{proof}
See Appendix \ref{pf: M4 bound}.
\end{proof}
\end{lem}

\begin{proof}[Proof of \cref{prop: Type IIIb local decay estimate}]
As mentioned in the discussion above 
\begin{align*}
    |\text{Type IIIb}| \leq |\text{Type IIIb1}| + |\text{Type IIIb2}| \lesssim  |\text{Type IIIb2}|\, ,
\end{align*}
where 
\begin{align*}
    |\text{Type IIIb2}| \lesssim \alpha \l(1+\min\{\g_1, |\g_2|\}^{-1}\r) \l[\langle t\rangle^{-1/2}  + n\langle t \rangle^{-1} \r] \|M\|_{W^{1,1}} \, .
\end{align*}

By \cref{lem: M bound}, $\|M \|_{L^1} \lesssim \alpha^{-1} \| f\|_{\ell_1^1}$, and from the discussion above we have 
\begin{align*}
    \|M'\|_{L^1} &\leq M_1 + M_2 \\
    &\lesssim \alpha^{-1} \| f\|_{\ell_1^1} + M_2  \qquad \qquad \qquad &[\text{\cref{lem: M1 bound}}] \\
    &\leq \alpha^{-1} \| f\|_{\ell_1^1}  + M_3 + M_4 \qquad \qquad \qquad &[\text{Equation } \eqref{eq: M2 bound}] \\
    &\leq \alpha^{-1} \| f\|_{\ell_1^1}  + M_{3a} + M_{3b} + M_4 \qquad \qquad \qquad &[\text{Equation } \eqref{eq: M3 bound}] \\
    &\lesssim \alpha^{-1} \|f\|_{\ell_2^1} + M_4 \qquad \qquad \qquad &[\text{\cref{lem: M3a bound,lem: M3b bound}}] \\
    &\lesssim  \alpha^{-1} \|f\|_{\ell_2^1} + \alpha^{-1} \sqrt{\min\{ \g_1, |\g_2|\}/ \g_-} \|f\|_{\ell_1^1} \,. \qquad \qquad \qquad &[\text{\cref{lem: M4 bound}}]
\end{align*}
Putting this altogether, we obtain 
\begin{align*}
    |\text{Type IIIb}| \lesssim  \l(1+\min\{\g_1, |\g_2|\}^{-1} + \g_-^{-1/2}\r) \l[\langle t\rangle^{-1/2}  + n\langle t \rangle^{-1} \r] \|f\|_{\ell_2^1}
\end{align*}
\end{proof}

\appendix

\addcontentsline{toc}{section}{Appendices}
\addtocontents{toc}{\protect\setcounter{tocdepth}{-10}}

\section{Well-known General Tools}
\begin{definition}\label{def: First Chebyshev Polynomial}
The Chebyshev polynomials of the first kind are defined through the identity \begin{align*}
    T_n (\cos(\theta)) = \cos(n \theta)\, . 
\end{align*}
The Chebyshev polynomials of the first kind can be equivalently be defined through the following equivalence equation \begin{align*}
    T_0(x) &= 1\, , \\
    T_1(x) &= x\, ,\\
    T_{n+1}(x) &= 2x T_n(x) - T_{n-1}(x)\, .
\end{align*}
These polynomials have the important property that they are orthogonal with the weight $\frac{1}{\sqrt{1-x^2}}$ on $[-1,1]$. In particular, \begin{align*}
    \int \limits_{-1}^1 \frac{T_n(x)T_m(x)}{\sqrt{1-x^2}} \,dx = \begin{cases}
        0, & \text{if } n \neq m \\
        \pi, & \text{if } n = m = 0 \\
        \frac{\pi}{2} & \text{if } n = m \neq 0 \, .
    \end{cases}
\end{align*}
\end{definition}

\begin{thm}[Sherman-Morrison-Woodbury Formula]\label{thm: SMW}
Let $\mathcal{H}, \mathcal{K}$ be Hilbert spaces with $U \in B(\mathcal{H}, \mathcal{K})$ and $V \in B(\mathcal{K}, \mathcal{H})$. Then $I + UV$ is invertible if and only if $I + VU$ is invertible. In which case, 
\begin{align*}
    (I + UV)^{-1} = I - U (I + VU)^{-1} V \, .
\end{align*}
\end{thm}

\begin{thm}[Sokhotski-Plemelj Theorem, Lemma 7.2.1 of \cite{Fokas_ComplexVariables}]\label{thm: Sokhotski-Plemelj}

Fix $a,b \in \mathbb{R}$ with $a< b$, $x \in (a,b)$. If $f$ is H\"older continuous on $(a,b)$, then 
\begin{align*}
  \int \limits_a^b \frac{f(y)}{y  - (x \pm i 0)} \, dy  \equiv  \lim_{\ep \rightarrow 0^+} \int \limits_a^b \frac{f(y)}{y  - (x \pm i \ep)} \, dy = PV \int \limits_a^b \frac{f(y)}{y - x} \, dy \pm i \pi f(x) \, .
\end{align*}

\end{thm}

\begin{thm}[Van der Corput's Lemma, Ch. VIII Section 1 of \cite{SteinHarmonic}]\label{thm: VDC}
Suppose that $\phi$ is real-valued and smooth in $(a,b)$, and that $|\phi^{(\ell)}(x)|\geq 1$ for all $x \in (a,b)$. Then for $\psi \in W^{1,1}([a,b])$, we have \begin{align*}
    \l| \int \limits_a^b e^{i t \phi(x)} \psi(x) \, dx\r| \leq \l(5\cdot 2^{\ell-1} -2 \r) t^{-1/\ell} \l[|\psi(b)| + \int \limits_a^b |\psi'(x)| \, dx \r]
\end{align*} 
when $\ell \geq 2$ or $\ell = 1$ and $\phi'(x)$ is monotonic. 
\begin{rmk}
The following alternative ``left-sided" version of van der Corput's lemma follows easily from a slight modification to the proof: 
\begin{align*}
    \l| \int \limits_a^b e^{i t \phi(x)} \psi(x) \, dx\r| \leq \l(5\cdot 2^{\ell-1} -2 \r) t^{-1/\ell} \l[|\psi(a)| + \int \limits_a^b |\psi'(x)| \, dx \r]. 
\end{align*}
\end{rmk}
\end{thm}

One limitation of Van der Corput's Lemma is that it only allows for decay rates of the form $t^{-1/r}$ for $r \in \mathbb{N}$. Throughout this paper, we encounter oscillatory integrals with integrable singularities and points of stationary phase of real, but not necessarily integral, order. For example, see \eqref{eq: Type IIIb def} and Section \ref{sec: Type IIIb Uniform Decay}. To obtain optimal estimates for such integrals, we use an extension of Van der Corput's lemma developed by Dewez in \cite{Dewez_2018}. We state here the particular consequence of the extension that we use.

\begin{lem}[Extension of Van der Corput's Lemma]\label{Dewez: fractional phase}
Let $0 < b \leq \pi$ and let $I$ be an open interval containing $[0,b]$. If $\psi \in C^1(I) \cap C^2(I\backslash \{0\})$ \footnote{It is possible to weaken the assumption on $\psi$ to $ \psi \in W^{1,1}([0,b])$.
}, then for all $\alpha \in [1, \infty)$,
\begin{align*}
    \l| \int \limits_0^b e^{i x^{\alpha} t} \psi(x) \, dx \r| \lesssim \ t^{- 1/ \alpha} \l[ |\psi(b)| + \int \limits_0^b |\psi'(x)| \, dx \r]\, .
\end{align*}
\end{lem}

\begin{thm}[Change of variables for singular integrals, pg. 17 of \cite{Gakhov}]\label{Gakhov: change of variables}
If the function $\tau = \alpha(\zeta)$ has a continuous first derivative $\alpha'(\zeta)$ which does not vanish anywhere, and constitutes an bijective mapping of the contour $L$ onto a contour $L'$, then \begin{align*}
    \fint \limits_L \frac{\varphi(\tau)}{ \tau - t} \, d\tau = \fint \limits_{L'} \frac{\varphi(\alpha (\zeta)) \alpha'(\zeta)}{\alpha(\zeta) - \alpha( \xi)} \, d \zeta \, ,
 \end{align*}
where $t = \alpha(\xi)$. 

\end{thm}

\begin{lem}[Swapping order of integration for singular integrals, pg. 47 of \cite{Gakhov}]\label{Gakhov: swap order}
Suppose that in the repeated integral 
\begin{align*}
    I = \int \limits_L \omega(\tau, z)  \fint \limits_L \frac{\varphi(\tau, \tau_1)}{\tau_1 -\tau} \, d\tau_1 \, d \tau \, ,
\end{align*}
the function $\varphi(\tau, \tau_1)$ is integrable and $\omega(\tau, z)$ is integrable with respect to $\tau$ for all $z$ belonging to a prescribed set of values. Let $L$ be a smooth contour. Then 
\begin{align*}
    \int \limits_L \omega(\tau, z)  \fint \limits_L \frac{\varphi(\tau, \tau_1)}{\tau_1 -\tau} \, d\tau_1 \, d \tau = \int \limits_L   \fint \limits_L \frac{ \omega(\tau, z) \varphi(\tau, \tau_1)}{\tau_1 -\tau} \, d\tau \, d \tau_1, .
\end{align*}
\end{lem}

\section{Proofs from Section \ref{sec: Computing the Resolvent}} \label{app: Computing the Resolvent}

\subsection{\textbf{Proof of Proposition} \ref{prop: limit of q}} \label{pf: limit of q} \,
To prove parts (1) and (2) of \cref{prop: limit of q}, we follow a strategy which closely follows the work done in \cite[Section 2]{Komech_2006} to compute the resolvent of the discrete Laplacian on $\mathbb{Z}$. We first exhibit in \cref{lem: cosine is biholomorphic} a domain and range on which cosine is biholomorphic. Since $k^2(q) = \g_1^2 + |\g_2|^2 + 2 \g_1 |\g_2| \cos(q)$ is just a scaled and shifted cosine, this proves the first two points of \cref{prop: limit of q}.

\begin{lem}\label{lem: cosine is biholomorphic}
Let $D$ be the lower half strip $D := \{ z  = x + iy \mid -\pi < x < \pi, y < 0\}$. Then $\cos(\cdot) : D \rightarrow \mathbb{C}\backslash (-\infty, 1]$ is a biholomorphic map. Moreover, if we let \begin{align*}
    D_- &= \{z = x + iy \mid - \pi < x < 0, y < 0\} \\
    D_+ &= \{z = x + iy \mid 0 < x < \pi, y < 0\}
\end{align*}
then the following maps, \begin{align*}
    \cos(\cdot) : D_- &\rightarrow  \{z \mid \Im(z) < 0\} \\ 
    \cos(\cdot) : D_+ &\rightarrow  \{z \mid \Im(z) > 0\}
\end{align*}
are also biholomorphic.

\begin{proof}
Let \begin{align*}
    E_- &:= \{z = x + iy \mid |z| >1, y < 0\}, & \mathbb{H}_- &:= \{z \mid \Im(z) < 0\}\\
    E_+ &:= \{z = x + iy \mid |z| >1, y >  0\}, & \mathbb{H}_+ &:= \{z \mid \Im(z) > 0\}. 
\end{align*}

The function $f(z) = e^{iz}$ is a biholomorphic map from $D_-$ to $E_-$ and from $D_+$ to $E_+$. We will show that $g(z) = \frac{1}{2}\l(z + \frac{1}{z} \r)$ is biholomorphic as a map from $E_-$ to $\mathbb{H}_-$ and as a map from $E_+$ to $\mathbb{H}_+$. Note that 
\begin{align*}
    g(z) &= \frac{1}{2} \l(z + \frac{1}{z} \r) = \frac{1}{2}\l(x + iy  + \frac{1}{x + iy} \frac{x - iy}{x -iy} \r) = \frac{1}{2} \l( x + iy + \frac{x -iy}{x^2  + y^2} \r) \\
    &= \frac{1}{2}\l(x + \frac{x}{x^2 + y^2}\r) + \frac{iy}{2} \l( 1 - \frac{1}{x^2 + y^2}\r). 
\end{align*}
From this expression for $g(z)$, we can see that $g$ maps $E_-$ into $\mathbb{H}_-$ and $E_+$ into $\mathbb{H}_+$. Moreover, $g$ is clearly holomorphic. Now we check that $g: E_+ \rightarrow \mathbb{H}_+$ is bijective. 

The equation $g(z) = \omega$ reduces to the equation \begin{align*}
    z^2 - 2 \omega z  + 1 = 0, 
\end{align*}
which by the quadratic formula has two distinct solutions in $\mathbb{C}$ when $\omega \neq \pm 1$.  For $\omega \in \mathbb{H}_{\pm}$, let $z_1$ and $z_2$ be these two distinct solutions. Then we compute \begin{align*}
    g(z_1) &= g(z_2) \\
    z_1 + \frac{1}{z_1} &= z_2 + \frac{1}{z_2} \\
    z_1^2 z_2 + z_2 &= z_1 z_2^2 + z_1 \\
     0 &= z_1^2 z_2  - z_1 z_2^2 - z_1 +  z_2 =
     z_1 z_2 ( z_1 - z_2) - z_1 + z_2 = (z_1 z_2 -1) ( z_1 - z_2) 
\end{align*}
Since $z_1 \neq z_2$, the last line implies that $z_1 z_2 = 1$. So $|z_1| = \frac{1}{|z_2|}$. By our calculation of $\Im(g(z))$, in order for $g(z_1) \in \mathbb{H}_+$ to hold, we must have either \begin{align*}
     |z_1|^2 > 1, \Im(z_1) >0  \;\; \text{ or } \;\;   |z_1|^2 < 1, \Im(z_1) < 0. 
\end{align*}
Hence there exists exactly one solution $z \in E_+$ for which $g(z) = \omega \in \mathbb{H}_+$. Thus $g: E_+ \rightarrow \mathbb{H}_+$ is bijective. In similar fashion we may show that $g: E_- \rightarrow \mathbb{H}_-$ is bijective.

Since \begin{align*}
    \cos(z) = \frac{1}{2} \l(e^{iz} + \frac{1}{e^{iz}} \r) = (g \circ f)(z),  
\end{align*}
we see that $\cos(z)$ is biholomorphic as a map from $D_-$ to $\mathbb{H}_-$ and as a map $D_+$ to $\mathbb{H}_+$. 
Recall \begin{align*}
    \Gamma_c := \{ z = x + iy \mid x = 0, y < 0)
\end{align*}
and note that $\cos(\Gamma_c) = (1, \infty)$. Thus, $\cos(z) : D \rightarrow \mathbb{C}\backslash  (-\infty, 1]$ is a bijection. It is clear as well that the map is holomorphic, so $\cos(z) : D \rightarrow \mathbb{C}\backslash  (-\infty, 1]$ is a biholomorphic map. 
\end{proof}
\end{lem}

\begin{proof}[\textbf{Proof of Proposition} \ref{prop: limit of q}]
Recall that $k^2(y) = \g_1^2 + |\g_2|^2 + 2 \g_1 |\g_2| \cos(y)$. By Lemma \ref{lem: cosine is biholomorphic}, $ \cos(\cdot) :D \rightarrow \mathbb{C}\backslash (- \infty, 1]$ is a biholomorphic map, so $k^2(\cdot): D \rightarrow \mathbb{C} \backslash ( - \infty, \g_+^2]$ is also biholomorphic. Additionally, since  $k^2(\Gamma_4) = (- \infty, \g_-^2)$,  the mapping $k^2(\cdot): D \cup \Gamma_4 \rightarrow \mathbb{C} \backslash [-\g_-^2, \g_+^2]$ is a bijection. Hence, the equation 
\begin{align*}
    k^2(q) = \omega, \quad \omega \in \mathbb{C}\backslash [\g_-^2 ,\g_+^2 ]\
\end{align*}
has a unique solution, which we define to be $\q(\omega)$. It is natural to then define the function $\q(\cdot)$ as the inverse of $k^2(\cdot)$. It immediately follows that 
$\q(\cdot) : \mathbb{C} \backslash [- \g_-^2, \g_+^2] \rightarrow D \cup \Gamma_4 $ is a bijection and that $\q(\cdot)$ restricted to  $\mathbb{C} \backslash (- \infty, \g_+^2] $ is a biholomorphism. 

Now we show that 
\begin{align*}
    \q((\lambda \pm i 0)^2) = \pm \qzl 
\end{align*}
for all $\lambda \in \sigma_{ess}(H) = [- \g_+, - \g_-] \cup [\g_- , \g_+]$, where $\qzl$ is given by \eqref{eq: qzl}.  
Borrowing notation from Lemma \ref{lem: cosine is biholomorphic}, we have \begin{align*}
    & k^2: D_- \rightarrow \mathbb{H}_-, \qquad \q : \mathbb{H}_- \rightarrow D_- \\
    &k^2 : D_+ \rightarrow \mathbb{H}_+, \qquad \q : \mathbb{H}_+ \rightarrow D_+ .
\end{align*}

\textbf{Case 1: Limit from above and $\lambda \in [\g_-, \g_+]$}

Suppose that $\lambda \in [\g_-, \g_+]$ and let $\{\ep_n\}_{n \geq 1}$ be any positive sequence converging to zero. Then $(\lambda + i \ep_n)^2 \in \mathbb{H}_+$ and  $\{\q( (\lambda + i \ep_n)^2 )  \}_{n \geq 1}$ is a Cauchy sequence in $D_+$ which must converge in the closure $\overline{D_+} = D_+ \cup \Gamma_c \cup \Gamma_3 \cup \Gamma_4$, see \cref{def: Gamma curve} and Figures \ref{fig: 3 panel} and \ref{fig: 2 panel}. We compute \begin{align*}
    \g_1^2 + |\g_2|^2 + 2 \g_1 |\g_2| \cos( \q((\lambda + i 0)^2) &= k^2( \q((\lambda + i 0)^2) = \lim_{\ep \rightarrow 0^+}  k^2( \q((\lambda + i \ep)^2) = \lim_{\ep \rightarrow 0^+} (\lambda + i \ep)^2 = \lambda^2 \\
    \implies \cos( \q((\lambda + i 0)^2)  &= \frac{\lambda^2 - \g_1^2 - |\g_2|^2}{2 \g_1 |\g_2|} =: \eta(\lambda). 
\end{align*}
Note that that for  $\lambda \in [\g_-, \g_+]$, $\eta(\lambda) \in [-1,1]$. Thus   
\begin{align*}
    \q((\lambda + i 0)^2) & = \cos^{-1} \l( \eta(\lambda) \r)\,, 
\end{align*}
where $\cos^{-1}(\cdot) : [-1,1] \rightarrow [0, \pi]$. 

\textbf{Case 2: Limit from above and $\lambda \in [-\g_+, -\g_-]$}

Suppose that $\lambda \in [-\g_+, -\g_-]$ and let $\{\ep_n\}_{n \geq 1}$ be any positive sequence converging to zero. Then $(\lambda + i \ep_n)^2 \in \mathbb{H}_-$ and  $\{\q( (\lambda + i \ep_n)^2 ) \}_{n \geq 1}$ is a Cauchy sequence in $D_-$ which must converge in the closure $\overline{D_-} = D_- \cup \Gamma_1 \cup \Gamma_2 \cup \Gamma_c$, see \cref{def: Gamma curve} and Figures \ref{fig: 3 panel} and \ref{fig: 2 panel}. As in Case 1, we obtain \begin{align*}
    \cos( \q((\lambda + i 0)^2) = \eta(\lambda) \in [-1,1]. 
\end{align*}
This time, since every element of $\overline{D_-}$ has a nonpositive real component, solving for $ \q((\lambda + i 0)^2)$ yields \begin{align*}
     \q((\lambda + i 0)^2) = - \cos^{-1}(\eta(\lambda)) \, .
\end{align*}

\textbf{Case 3: Limit from below}

The result for $\q((\lambda - i 0)^2)$ follows from the fact that  \begin{align*}
    \q((\lambda - i \ep)^2)  = \q((-\lambda + i \ep)^2). 
\end{align*}

\end{proof}

\subsection{\textbf{Proof of Lemma} \ref{lem: J integrals}}\label{pf: J integrals}

\begin{lem}\label{lem: g0}
Suppose $g \in \ell^2(\mathbb{N}_0)$ and write $ \Tilde{g}(s) = \sum_{n \geq 0} e^{-ins} g_n$. Then
 \begin{align*}
     \lim_{R \rightarrow \infty} \l| \Tilde{g}(s - iR) - g_0 \r| = 0 \, .
 \end{align*}
 \begin{proof}
Using the inequality $ab\leq (a^2 + b^2)/2$ in the second line, we compute 
\begin{align*}
    \Tilde{g}(s - iR) - g_0 &= \l( \sum_{n \geq 0} e^{-in (s - iR)} g_n \r) - g_0  =  \sum_{n \geq 1} e^{-in (s - iR)} g_n  = \sum_{n \geq 1} e^{-ins} e^{-nR} g_n \\
    \l| \Tilde{g}(s - iR) - g_0 \r| &= e^{-R} \l|\sum_{n \geq 1} e^{-ins} e^{-(n-1)R} g_n  \r| \leq \frac{e^{-R}}{2} \sum_{n \geq 1}\l( \l|e^{-(n-1)R} \r|^2 + |g_n|^2 \r) \\
    &\leq \frac{e^{-R}}{2} \l(\|g\|_{\ell^2}^2 + \sum_{m \geq 0} |e^{-mR}|^2 \r) = \frac{e^{-R}}{2} \l(\|g\|_{\ell^2}^2 + \frac{1}{1 - e^{-2R}} \r) \, .
\end{align*}
Taking the limit as $R \rightarrow \infty$ gives our desired result. 
\end{proof}
\end{lem}

\begin{proof}[\textbf{Proof of Lemma} \ref{lem: J integrals}] \,

We compute the two integrals through the use of a contour deformation in the lower half plane. This requires assuming that $\tilde{g}$ is analytic in an open set containing our chosen contour, and in particular, $\tilde{g}$ must analytic in some part of the upper half plane. So we assume initially that $g \in \ell^2(\mathbb{N}_0; \mathbb{C})$ has compact support and then use density to pass to general $g \in \ell^2(\mathbb{N}_0; \mathbb{C})$. 
    
\underline{First Integral:} \newline
Using the fact that $k$ and $\tilde{g}$ are $2\pi$-periodic, \begin{align*}
     \frac{1}{2\pi} \int \limits_{-\pi}^{\pi} \frac{\tilde{g}(q)}{k^2 (q - \varphi) - z^2} \, dq &=  \frac{1}{2\pi} \int \limits_{-\pi}^{\pi} \frac{\tilde{g}(q + \varphi)}{k^2 (q) - z^2} \, dq
\end{align*}
By Proposition \ref{prop: limit of q}, the integrand  has a pole at $\q(z^2)$, so we apply the residue theorem to the rectangular contour with top edge $[-\pi, \pi]$ and bottom edge $[-\pi - i R, \pi - i R]$ for $R >>0$ to obtain \begin{align*}
    \frac{1}{2\pi}\int \limits_{-\pi}^{\pi} \frac{\tilde{g}(q + \varphi)}{k^2 (q) - z^2} \, dq  &+  \frac{1}{2\pi} \int \limits_{\pi -i R}^{-\pi - i R} \frac{\tilde{g}(q + \varphi)}{k^2 (q) - z^2} \, dq  = - i  \text{Res}_{\q(z^2)} \frac{\tilde{g}(q + \varphi)}{k^2 (q) - z^2} 
\end{align*}
Note that the left and right side of the rectangular contour cancel each other out due to periodicity of the integrand. 
Furthermore, \begin{align*}
    \l|\frac{1}{2\pi} \int \limits_{\pi -i R}^{-\pi - i R} \frac{\tilde{g}(q + \varphi)}{k^2 (q) - z^2} \, dq \r|  =  \l|\frac{1}{2\pi} \int \limits_{\pi }^{-\pi} \frac{\tilde{g}(q - i R + \varphi)}{k^2 (q - i R) - z^2} \, dq \r| \rightarrow  0 \quad \text{as } R \rightarrow \infty
\end{align*}
since 
\begin{align*}
    \l| \tilde{g}(q - i R + \varphi) \r| &\leq \sum_{m \geq 0} \l|e^{-mR} g_m \r| \leq \frac{1}{2}  \sum_{m \geq 0}  \l(e^{-2mR} + |g_m|^2 \r) \leq \frac{1}{2} \l(\|g\|_{\ell^2}^2 + \frac{1}{1 - e^{-2R}} \r)
\end{align*}
and \begin{align*}
    \l| k^2(q - iR) \r|  =  \l| \g_1^2 + |\g_2|^2 + 2 \g_1 |\g_2| \cos(q - i R) \r| \rightarrow \infty \quad \text{as } R \rightarrow \infty
\end{align*}
Lastly, using the fact that $z^2 = k^2(\q(z^2))$, we compute 
\begin{align*}
    -i \text{Res}_{\q(z^2)} \frac{\tilde{g}(q + \varphi)}{k^2 (q) - z^2}  = -i  \lim_{ q \rightarrow \q(z^2)} \frac{(q - \q(z^2)) \tilde{g}(q + \varphi)}{2 \g_1 |\g_2| [\cos(q) - \cos(\q(z^2))] } = \frac{i \tilde{g}(\q(z^2) + \varphi)}{2 \g_1 |\g_2| \sin(\q(z^2)}
\end{align*}
Sending $R \rightarrow \infty$ completes the calculation of the first integral. \newline
\underline{Second Integral}\newline
The strategy is the same, except now the contribution from the bottom edge of the contour does not vanish in the limit. In similar fashion to the previous calculation, 
\begin{align*}
     \frac{1}{2\pi} \int \limits_{-\pi}^{\pi} \frac{e^{i q}\tilde{g}(q)}{k^2 (q - \varphi) - z^2} \, dq &=  \frac{e^{i \varphi}}{2\pi} \int \limits_{-\pi}^{\pi} \frac{e^{i q}\tilde{g}(q + \varphi)}{k^2 (q) - z^2} \, dq \\ 
     &=  \frac{e^{i \varphi}}{2\pi} \int \limits_{-\pi }^{\pi} \frac{e^{i (q-iR)}\tilde{g}(q - iR + \varphi)}{k^2 (q-iR) - z^2} \, dq - i e^{i \varphi} \text{Res}_{\q(z^2)} \frac{e^{i q}\tilde{g}(q + \varphi)}{k^2 (q) - z^2} \\
     &=
     \frac{e^{i \varphi}}{2\pi} \int \limits_{-\pi }^{\pi} \frac{e^{i (q-iR)}\tilde{g}(q - iR + \varphi)}{k^2 (q-iR) - z^2} \, dq + i e^{i \varphi} \l[ \frac{e^{i \q(z^2)} \tilde{g}(\q(z^2) + \varphi)}{2 \g_1 |\g_2| \sin(\q(z^2))} \r]
\end{align*}
For the bottom edge, we compute \begin{align}\label{eq: J integral bottom}
    \int \limits_{-\pi }^{\pi} \frac{e^{i (q-iR)}\tilde{g}(q - iR + \varphi)}{k^2 (q-iR) - z^2} \, dq &= \int \limits_{-\pi }^{\pi} \frac{e^{i (q-iR)}[\tilde{g}(q - iR + \varphi) - g_0] }{k^2 (q-iR) - z^2} \, dq + g_0 \int \limits_{-\pi }^{\pi} \frac{e^{i (q-iR)}}{k^2 (q-iR) - z^2} \, dq 
\end{align}
By Lemma \ref{lem: g0}, the first term on the right hand side of \eqref{eq: J integral bottom} vanishes in the limit. For the latter term, we compute
\begin{align*}
    \int \limits_{-\pi }^{\pi} \frac{e^{i (q-iR)}}{k^2 (q-iR) - z^2} \, dq = g_0 \int \limits_{-\pi }^{\pi} \frac{e^{i (q-iR)}}{\g_1^2 + |\g_2|^2 + \g_1 |\g_2| \l[ e^{i(q -iR)} + e^{-i (q- iR)} \r] - z^2} \, dq \rightarrow \int \limits_{-\pi }^{\pi}  \frac{1}{\g_1 |\g_2|} \, dq = \frac{2\pi}{\g_1 |\g_2|}
\end{align*}
as $R \rightarrow \infty$. Thus 
\begin{align*}
    \frac{e^{i \varphi}}{2\pi} \int \limits_{-\pi }^{\pi} \frac{e^{i (q-iR)}\tilde{g}(q - iR + \varphi)}{k^2 (q-iR) - z^2} \, dq = \frac{g_0 e^{i \varphi} }{\g_1 |\g_2|}
\end{align*}

\begin{rmk}
As a consequence of Proposition \ref{prop: limit of q}, when $z^2 \in (-\infty, \g_-^2)$, we have $\q(z^2) \in \Gamma_4$, which is the right boundary of the contour we use to compute the two integrals. To avoid having a pole on the contour, we can use the $2\pi$-periodicity of the integrands to shift over the region of integration by $\pi$. This then places $\q(z^2)$ inside of the new shifted contour. Written more explicitly, we use the fact that 
\begin{align*}
    \frac{1}{2\pi} \int \limits_{-\pi}^{\pi} \frac{\tilde{g}(q + \varphi)}{k^2 (q) - z^2} \, dq &=
    \frac{1}{2\pi} \int \limits_{0}^{2\pi} \frac{\tilde{g}(q + \varphi)}{k^2 (q) - z^2} \, dq \\  
    \frac{e^{i \varphi}}{2\pi} \int \limits_{-\pi}^{\pi} \frac{e^{i q}\tilde{g}(q + \varphi)}{k^2 (q) - z^2} \, dq &= \frac{e^{i \varphi}}{2\pi} \int \limits_{0}^{2\pi} \frac{e^{i q}\tilde{g}(q + \varphi)}{k^2 (q) - z^2} \, dq \, , 
\end{align*}
and choose our new contour to be the rectangle with top edge $[0, 2\pi]$ and bottom edge $[- iR, 2 \pi - i R]$. The rest of the calculations for this shifted contour are nearly identical.  
\end{rmk}

\end{proof}

\subsection{\textbf{Proof of Proposition} \ref{prop: det(I + VU)}}
\begin{proof} \label{pf: det(I + VU)}
    By definition of $U, V, J(n;z), K(n;z)$, it directly follows that
\begin{align*}
    (I + VU)(z)  &= I + \frac{1}{2\pi} \int \limits_{-\pi}^{\pi}  \frac{\overline{\g_2} e^{iq}}{z^2 -  k^2(q- \varphi)} \begin{bmatrix}
         h(q) & 0 \\ z & 0
    \end{bmatrix} \, dq  = \begin{bmatrix}
        1- \overline{\g_2} K(1; z) & 0 \\
        - \overline{\g_2} z J(1; z)  & 1
    \end{bmatrix}
\end{align*}
Applying Corollary $\ref{cor: J integrals}$ and using \eqref{eq: e^(i varphi)}, we compute the determinant
\begin{align*}
    \det(I + VU)(z) &= 1 - \overline{\g_2} K(1;z) = 1 - \overline{\g_2} (\g_1 J(1;z) + \g_2 J(0;z)) = 1-\g_1 \overline{\g_2} J(1;z) - |\g_2|^2 J(0;z) \\
    &= 1- \frac{ e^{i \varphi} \g_1 \overline{\g_2} }{2\g_1 |\g_2|} \l( \frac{ie^{-i \q(z^2)} }{ \sin( \q(z^2)} \r)  -  \frac{i |\g_2|^2 }{2 \g_1 |\g_2| \sin(\q(z^2))}  \\
    &= 1 - \frac{ie^{-i \q(z^2)} }{ 2\sin( \q(z^2)} - \frac{i |\g_2| }{2 \g_1 \sin(\q(z^2))} 
    = 1- \frac{i}{2 \sin(\q(z^2))} \l( e^{-i \q(z^2)} + \frac{|\g_2|}{\g_1} \r)
\end{align*}
Setting $\det(I + VU)(z) = 0$ for some $z \in \mathbb{C}\backslash \sigma_{ess}(H)$ we obtain \begin{align*}
    e^{-i \q(z^2)} + \frac{|\g_2|}{\g_1} &= -2 i \sin(\q(z^2)) = e^{-i \q(z^2)} - e^{i \q(z^2)} \\
    \frac{|\g_2|}{\g_1} &= - e^{i \q(z^2)} 
\end{align*}
Since $\q(z^2)$ takes values in the lower half plane, we see that $\det(I + VU)(z) \neq 0$ whenever $|\g_2| < \g_1$. Now suppose $|\g_2| > \g_1$  and let $ \q(z^2) = a + ib \in D \cup \Gamma_4$. Note that this means $a \in (-\pi, \pi]$ and $b < 0$. We compute
\begin{align*}
    0 = \frac{|\g_2|}{\g_1} + e^{ia} e^{-b} = \frac{|\g_2|}{\g_1} + \cos(a) e^{-b} + i \sin(a) e^{-b}
\end{align*}
Equating real and imaginary components yields \begin{subequations}
\begin{align}
    \frac{|\g_2|}{\g_1} + \cos(a) e^{-b} &= 0 \label{eq: equate real components} \\
    \sin(a) e^{-b}  &= 0 \label{eq: equate imaginary components}
\end{align}
\end{subequations}
Equation \eqref{eq: equate imaginary components} only holds when $\sin(a) = 0$. So $a$ is either 0 or $\pi$. When $a = 0$, equation \eqref{eq: equate real components} has no solution. 
When $a = \pi$, equation \eqref{eq: equate real components} becomes 
\begin{align*}
     \frac{|\g_2|}{\g_1} - e^{-b} = 0  \implies  b = \log(\g_1 / |\g_2|) 
\end{align*}
So $\q(z^2) = \pi + i \log(\g_1 / |\g_2|)$. Then we compute \begin{align*}
    2 \cos(\q(z^2)) &= e^{i\q(z^2)} + e^{-i\q(z^2)} = e^{i \pi} e^{\log( |\g_2|/ \g_1)}  + e^{-i\pi} e^{\log(\g_1/ |\g_2|)} = - \l( \frac{|\g_2|}{\g_1} + \frac{\g_1}{|\g_2|}\r)  \\
    z^2 &= k^2( \q(z^2)) = \g_1^2 + |\g_2|^2  + 2 \g_1 |\g_2| \cos(\q(z^2)) \\
    &= \g_1^2 + |\g_2|^2  - \g_1 |\g_2| \l( \frac{|\g_2|}{\g_1} + \frac{\g_1}{|\g_2|}\r) = 0
\end{align*}
Thus $z = 0$. 
\end{proof}

\section{Proofs from Section \ref{sec: Computing the Propagator}} \label{app: Computing the Propagator}

\subsection{\textbf{Proof of Proposition} \ref{prop: no embedded eigenvalues}}\label{pf: no embedded eigenvalues} 

\begin{proof}[\textbf{Proof of Proposition} \ref{prop: no embedded eigenvalues}] 
The Fourier-Laplace transform of $(H - z I) \psi = f$ is given by \eqref{eq: computing FL transform line b}. Since the matrix \begin{align*}
    \begin{bmatrix}
        -z & h(q) \\
        \overline{h( \overline{q})} & -z
    \end{bmatrix}
\end{align*} 
is singular for $z \in \sigma_{ess}(H)$ and the matrix $(I + VU)(z)$ is singular for $z = 0$ by \cref{prop: det(I + VU)}, it follows that 
\begin{align*}
    \sigma(H) = \{0\} \cup \sigma_{ess}(H). 
\end{align*}
Hence any nonzero eigenvalue, if it exists, must be embedded in the essential spectrum. 

Suppose there does exist a nonzero eigenvalue $z \in \sigma_{ess}(H)$ and let $\psi \in \ell^2( \mathbb{N}_0; \mathbb{C}^2)$ be the associated eigenvector. Taking the Fourier-Laplace transform of $(H - zI) \psi = 0$ as in \eqref{eq: computing FL transform line a}, we obtain 
\begin{align}\label{eq: no embedded evals eq 1}
    \overline{\g_2} e^{iq} \begin{bmatrix}
        0 \\ \psi_0^A
    \end{bmatrix} = \begin{bmatrix}
        -z & h(q) \\ \overline{h( \overline{q})} & -z
    \end{bmatrix} \Tilde{\psi}  \, .
\end{align}
From the first row of \eqref{eq: no embedded evals eq 1}, we obtain 
\begin{align}
    0  = - z \tilde{\psi}^A + h(q) \tilde{\psi}^B \nonumber \\
    \implies \tilde{\psi}^A(q)  = \frac{ h(q) \tilde{\psi}^B (q)}{z}. \label{eq: no embedded evals row 1}
\end{align}
Plugging \eqref{eq: no embedded evals row 1} into the second row of \eqref{eq: no embedded evals eq 1} and using \eqref{eq: k shifted} gives \begin{align*}
    \overline{\g_2} e^{iq} \psi_0^A &= \overline{h( \overline{q})} \tilde{\psi}^A(q) - z \tilde{\psi}^B(q) \\
    \overline{\g_2} e^{iq} \psi_0^A &= \overline{h( \overline{q})} \l( \frac{ h(q) \tilde{\psi}^B (q)}{z} \r)  - z \tilde{\psi}^B(q) = \frac{1}{z} \l(k^2 (q - \varphi) - z^2 \r) \tilde{\psi}^B (q) \\
    \implies \tilde{\psi}^B(q) &= \frac{z \overline{\g_2} e^{iq} \psi_0^A}{k^2 (q - \varphi) - z^2} \, .
\end{align*}

\textbf{Claim:} \textit{If $\psi$ is an eigenvector of $H$, then $\psi_0^A \neq 0$.} 

Assuming the above claim, we now show that $\tilde{\psi}^B$, and thus  $\tilde{\psi}$,  is unbounded in the $L^2$-norm. Since the Fourier-Laplace transform maps sequences in $\ell^2$ to functions in $L^2$, this proves that $\psi$ is not an eigenvector, which further implies that there are no non-zero eigenvalues of $H$. 

Using the $2\pi$-periodicity of $k^2(\cdot)$, we compute 
\begin{align*}
    \|\tilde{\psi}^B\|^2_{L^2} \gtrsim  \int \limits_{- \pi}^{\pi} \frac{1}{|k^2(q - \varphi) - z^2|^2} \, dq  =   \int \limits_{- \pi}^{\pi} \frac{1}{|k^2(q) - z^2|^2} \, dq  \geq  \int \limits_{0}^{\pi} \frac{1}{|k^2(q) - z^2|^2} \, dq \,. 
\end{align*}
Next we make the change of variables $q \mapsto k(q)$ as in \cref{lem: Change of Variables} to obtain 
\begin{align*}
    \int \limits_{0}^{\pi} \frac{1 }{\l|k^2(q) - z^2  \r|^2 } \, dq &= \int \limits_{\g_-}^{\g_+} \frac{1 }{\l|k^2 - z^2  \r|^2 }  \frac{k}{\g_1 |\g_2| \sqrt{1 - \eta^2(k)}} \, dk \gtrsim  \int \limits_{\g_-}^{\g_+} \frac{k }{\l|k^2 - z^2  \r|^2 }  \, dk \, .
\end{align*}
Next using the change of variables $v = k^2 - z^2$, we compute
\begin{align}\label{eq: no embedded evals eq 2}
    \int \limits_{\g_-}^{\g_+} \frac{k }{\l|k^2 - z^2  \r|^2 }  \, dk = \frac{1}{2} \int \limits_{\g_-^2 - z^2}^{\g_+^2 - z^2} \frac{1}{v^2} \, dv \, . 
\end{align}
It is clear to see that if $z \in \sigma_{ess}(H) = [-\g_+, - \g_-] \cup [\g_-, \g_+]$, then \eqref{eq: no embedded evals eq 2} is not integrable. Thus $\tilde{\psi}^B$ is unbounded in the $L^2$-norm, which completes the proof.

\begin{proof}[Proof of Claim]
Let $\psi$ be an eigenvector of $H$ with eigenvalue $z$ and suppose that $\psi_0^A = 0$. Then $(H\psi)_0 = z \psi_0$, or equivalently
\begin{align*}
    \begin{bmatrix}
    \g_1 \psi_0^B \\ \g_2 \psi_1^A
    \end{bmatrix} = z \begin{bmatrix}
        0 \\ \psi_0^B
    \end{bmatrix}\, .
\end{align*}
The first row implies that $\psi_0^B =0$. Plugging this into the second row shows that $\psi_1^A = 0$. By induction, $\psi \equiv 0$. However, this leads to a contradiction. 
\end{proof}
\end{proof}

\subsection{\textbf{Proof of Lemma} \ref{lem: DCT for propagator}}\label{pf: DCT for propagator} 

\begin{proof}[\textbf{Proof of Lemma} \ref{lem: DCT for propagator}] 

Define \begin{align*}
    g_{\ep}(\lambda) &:=  \int \limits_{-\pi}^{\pi} e^{inq} \l[\tilde{R}_{edge}(\lambda \pm  i\ep)  \tilde{f}\,\r](q) \, dq  \,, \\
    g_0(\lambda) &:= \lim_{\ep \rightarrow 0^+} g_{\ep}(\lambda) \,, \\ 
    h_{\ep}(\lambda) &:= \frac{ (\lambda \pm i \ep)}{ \lambda} \int \limits_{-\pi}^{\pi}   \frac{e^{i(n+1)q}}{k^2 ( q- \varphi) - (\lambda \pm i \ep)^2}  \begin{bmatrix}
        \lambda  h(q)/h(\pm \qzl + \varphi) & h(q)  \\[5pt] \overline{h\l( \pm \qzl+ \varphi \r) } & \lambda 
    \end{bmatrix} \widetilde{Sf}(\pm \qzl + \varphi) \, dq \, . 
\end{align*}
Using the Lebesgue dominated convergence theorem, it suffices to find an integrable function $G(\lambda)$  such that $|g_{\ep}(\lambda)| \leq G(\lambda)$. By the triangle inequality,  
\begin{align*}
    |g_{\ep}(\lambda)| \leq |g_{\ep}(\lambda) - h_{\ep}(\lambda)| + |h_{\ep}(\lambda) - g_0(\lambda)| + |g_0(\lambda)|\, . 
\end{align*}
We show that each of these terms is integrable in $\lambda$ for sufficiently small $\ep$. Note that $g_0(\lambda)$ is computed in equations \eqref{eq: U term 1} and \eqref{eq: U term 2}.

It follows from \cref{cor: handwavy approx} that the first term, $|g_{\ep}(\lambda) - h_{\ep}(\lambda)| $, is integrable for all sufficiently small $\ep >0$. By the calculation following \cref{cor: handwavy approx}, showing that the second term, $|h_{\ep}(\lambda) - g_0(\lambda)|$, is integrable reduces to showing that the limit in the Sokhotski-Plemelj formula, see \cref{thm: Sokhotski-Plemelj}, converges in the $L^1$ norm with respect to $\lambda$. In particular, \cref{lem: computing B(ep)} shows that we need
\begin{align}\label{eq: SP for DCT lemma}
    \int \limits_{\g_-}^{\g_+}  \l(\frac{1}{k - (\lambda \pm i \ep)}  \r)  \frac{kT_{n+1}(\eta(k)) }{\g_1 |\g_2| \sqrt{1 - \eta^2(k)}} \, dk - \l( \fint \limits_{\g_-}^{\g_+} \l( \frac{1}{k - \lambda}\r) \frac{kT_{n+1}(\eta(k)) }{\g_1 |\g_2| \sqrt{1 - \eta^2(k)}} \, dk \pm i \pi   \frac{\lambda T_{n+1}(\eta(\lambda)) }{\g_1 |\g_2| \sqrt{1 - \eta^2(\lambda)}}\r)
\end{align}
to converge in the $L^1$ norm as $\ep \rightarrow 0^+$. Let 
\begin{align*}
    \tau(k) = \begin{cases} \frac{kT_{n+1}(\eta(k)) }{\g_1 |\g_2| \sqrt{1 - \eta^2(k)}}\, , \quad &k \in (\g_-, \g_+) \\[6pt]
    0\, , \quad &k \in \mathbb{R}\backslash(\g_-, \g_+) \, , 
    \end{cases}
\end{align*}
and note that $\tau \in L^p([\g_-, \g_+])$ for $p \in [1,2)$ by \eqref{eq: 1 - eta^2}. The first term in \eqref{eq: SP for DCT lemma} can be rewritten as 
\begin{align*}
     \int \limits_{\g_-}^{\g_+}  \frac{\tau(k)}{k - (\lambda \pm i \ep)}  \, dk = \pi \l(Q_{\ep}* \tau \r)(k) \pm  i\pi \l(P_{\ep}* \tau \r)(k) \, , 
\end{align*}
where $P_{\ep}$ is the Poisson kernel on the upper half-space and $Q_{\ep}$ is the harmonic conjugate of $P_{\ep}$. That this converges in the $L^1$ norm to the second term in \eqref{eq: SP for DCT lemma} follows from \cite[Theorem~1.2.19,~Theorem~5.1.5,~and~Remark~5.1.6]{Grafakos_ClassicalFourierAnalysis}.

Lastly we show that $g_0(\lambda)$ is integrable. By the expression for $g_0(\lambda)$ computed in equations
\eqref{eq: U term 1} and \eqref{eq: U term 2} and the discussion in Section \ref{sec: Representative Oscillatory Integrals},  it suffices to show that the three representative oscillatory integrals, Types I-III, with absolute value signs inserted inside of the $d\lambda$ integral, are finite. Showing this for Types I and II is trivial. The Type III term with absolute value signs inserted inside of the $d \lambda$ integral is  
\begin{align*}
    \int \limits_{\g_-}^{\g_+} \l|  e^{-i \lambda t}  \fint \limits_{\g_-}^{\g_+}  \frac{\tau(k)}{k - \lambda} \, dk \, \tilde{f}(\qzl) \r| \, d \lambda   
    \leq \|f\|_{\ell^1}  \int \limits_{\g_-}^{\g_+} \l|\,  \fint \limits_{\g_-}^{\g_+}  \frac{\tau(k)}{k - \lambda}  \, dk \r| \, d \lambda \, . 
\end{align*}

Due to the $L^p$ boundedness of the Hilbert transform for all $p \in (1,\infty)$ \cite[Theorem~5.1.7]{Grafakos_ClassicalFourierAnalysis} and the fact that $\tau \in L^p([\g_-, \g_+])$ for all $p \in [1,2)$, it follows that 
\begin{align*}
    \l\| \fint \limits_{\g_-}^{\g_+}  \frac{\tau(k)}{k - \lambda}  \, dk  \r\|_{L^p([\g_-, \g_+])} 
\end{align*}
is bounded for all $p \in (1,2)$, and thus in $L^1$ since the region of integration is a bounded interval.

\end{proof}

\subsection{\textbf{Proof of Proposition} \ref{prop: handwavy approx}} \label{pf: handwavy approx} \, 

Although the proof of \cref{prop: handwavy approx} requires long computations, the general strategy is fairly simple. 
The expression for the integrand, $ e^{inq} \l[\tilde{R}_{edge}(\lambda \pm  i\ep)  \tilde{f}\,\r](q)$, has a singular term,  $[k^2(q - \varphi) - (\lambda \pm i \ep)^2]^{-1}$, and a bounded term which involves $\exp( -i \q((\lambda \pm i \ep)^2))$. In \cref{lem: handwavy approx}, we show that $\exp( -i \q((\lambda \pm i \ep)^2))$ is $1/2$-H\"older continuous in $\ep$ with a constant that is independent of $\lambda$. 
Part of the challenge of proving \cref{lem: handwavy approx} is that the method by which we show $|\exp( -i \q((\lambda \pm i \ep)^2)) - \exp(\pm i \qzl)| $ is small for $\lambda$ near the endpoints of $[\g_-, \g_+]$ is different than how we show the expression is small for $\lambda$ in the interior of $[\g_-, \g_+]$. Careful attention must be paid to extract a constant that is independent of $\lambda$. 
Having proven \cref{lem: handwavy approx}, showing the equivalence of the two limits in \cref{prop: handwavy approx} can be simplified to showing $\sqrt{\ep} \int_{-\pi}^{\pi} [k^2(q - \varphi) - (\lambda \pm i \ep)^2]^{-1} dq \rightarrow 0  $ as $\ep \rightarrow 0^+$. 

\begin{lem}\label{lem: handwavy approx}
For all $\lambda \in [\g_-, \g_+]$ and $0 < \ep < \min\{1/2, \g_1 |\g_2|\}$, 
\begin{align}\label{eq: handwavy approx lemma}
     \l| e^{-i\q((\lambda \pm i \ep)^2)} - e^{\mp i\qzl} \r| \leq  P_2(\g_1, |\g_2|)\sqrt{\frac{\ep}{\g_1 |\g_2|}}   \, ,
\end{align}
where $P_2(x,y)$ is a polynomial of degree  2.

\begin{proof}
It follows from \cref{def: q} and \cref{prop: limit of q} that \begin{align*}
    \q\l(z^2\r) = \arccos\l(\frac{z^2 - \g_1^2 - |\g_2|^2}{2 \g_1 |\g_2|} \r) = \arccos( e(z)) \quad \forall z \in \mathbb{C}\backslash \sigma_{ess}(H)\, ,
\end{align*}
where $\arccos : \mathbb{C}\backslash (-\infty , 1] \rightarrow D$
is holomorphic. We obtain a more explicit expression for this analytic continuation of arccosine by first noting that $\sqrt{\omega^2 -1}$ is single-valued and holomorphic on $\mathbb{C}\backslash [-1,1]$. 
Recall from \cref{def: Gamma curve} that $\Gamma_c = \{ iy \mid y < 0\}$, and so $\cos(\Gamma_c) = (1,\infty)$. For any $\omega \in (1,\infty)$, we compute \begin{align*}
    \omega &= \cos(iy) = \frac{e^y + e^{-y}}{2} \\
    0 &=   e^{2y} -2\omega e^y + 1 \\
    e^y &= \frac{2\omega \pm \sqrt{ 4\omega^2 -4}}{2} = \omega \pm \sqrt{ \omega^2 -1}\, . 
\end{align*}
Since we require $y < 0$ in order for $iy \in \Gamma_c$, we pick $y = \ln ( \omega - \sqrt{ \omega^2 -1})$. Thus
\begin{align}\label{eq: arccosine}
    \arccos(\omega) &= iy = i \ln ( \omega - \sqrt{ \omega^2 -1})  \quad 
\end{align}
for all $\omega \in (1, \infty)$. We then analytically continue \eqref{eq: arccosine} to all $\omega \in \mathbb{C}\backslash (-\infty, 1]$ and compute \begin{align*}
    e^{-i\q((\lambda \pm i \ep)^2)} = e^{-i \arccos( \eta(\lambda \pm i \ep))} = \eta(\lambda \pm i \ep) - \sqrt{\eta^2(\lambda \pm i \ep) -1}\, .  
\end{align*}
Using the  definition of $\qzl$, see \eqref{eq: qzl}, we have \begin{align*}
    e^{\mp i \qzl} = e^{\mp i \cos^{-1}(\eta(\lambda))} = \eta(\lambda) \mp i \sqrt{1 - \eta^2(\lambda)}\, ,
\end{align*}
where $\cos^{-1}(\cdot) : [-1,1] \rightarrow [0, \pi]$. Thus \begin{align}\label{handwavy: triangle ineq}
\l| e^{-i\q((\lambda \pm i \ep)^2)} - e^{\mp i\qzl} \r| \leq \l| \eta(\lambda \pm i \ep) - \eta(\lambda) \r| + \l|\pm i \sqrt{1 - \eta^2(\lambda)} - \sqrt{\eta^2(\lambda \pm i \ep) - 1} \r| \, .
\end{align}
The first term is easy to easy to estimate. Using the fact that $\lambda < \g_+$ and $\ep < 1/2$, we compute \begin{align*}
    \l| \eta(\lambda \pm i \ep) - \eta(\lambda) \r| &= \l|\frac{(\lambda \pm i \ep)^2 - \g_1^2 - |\g_2|^2}{2 \g_1 |\g_2|} - \frac{\lambda^2 - \g_1^2 -|\g_2|^2}{2\g_1 |\g_2|} \r| = \frac{\l|\pm 2i \ep \lambda - \ep^2 \r|}{2 \g_1 |\g_2|} \\
    &\leq \frac{(2 \g_+  + 1/2) \ep}{2 \g_1 |\g_2|}   \, .
\end{align*}
Estimating the second term on the right hand side of \eqref{handwavy: triangle ineq} will be trickier. First we show that it suffices to consider only one of the two options for the sign.

\textbf{Claim 1:}
For all $\lambda \in [\g_-, \g_+]$ and $\ep > 0$,
\begin{align}\label{handwavy: symmetry}
    \l|i \sqrt{1 - \eta^2(\lambda)} - \sqrt{\eta^2(\lambda + i \ep) - 1} \r| = \l|- i \sqrt{1 - \eta^2(\lambda)} - \sqrt{\eta^2(\lambda - i \ep) - 1} \r| \, .  
\end{align}
\begin{proof}[Proof of Claim 1]
For any $\omega \in \mathbb{C}\backslash [-1,1]$, define
\begin{equation}\label{eq: r and theta def} \begin{split}
    r_1 := |\omega -1 |, \qquad \theta_1 := \arg (\omega- 1)\, , \\
    r_2 := |\omega +1 |, \qquad \theta_2 := \arg (\omega+ 1) \, , 
\end{split}
\end{equation}
as in \cref{fig: sqrt(z^2-1)}, where $\arg(\omega \pm 1) \in (-\pi, \pi] $. 

\begin{figure}[h]
\centering\includegraphics[width=4in,height=4in,keepaspectratio]{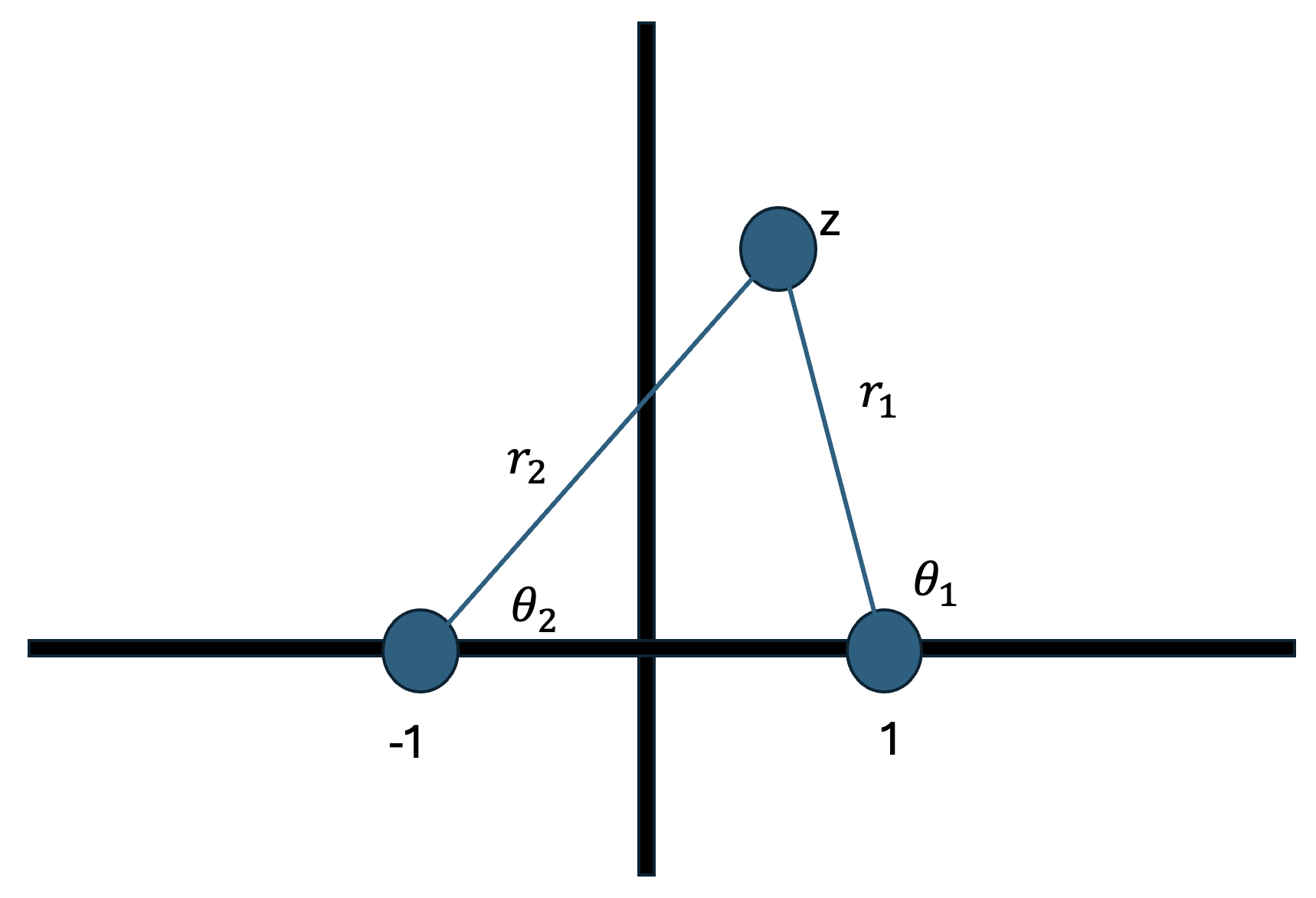}
  \caption{The lengths $r_1, r_2$, and angles $ \theta_1, \theta_2$ corresponding to \eqref{eq: r and theta def}. }
  \label{fig: sqrt(z^2-1)}
\end{figure}
Then 
\begin{align}\label{eq: sqrt(w^2 -1)}
    \sqrt{\omega^2 -1 } = \sqrt{(\omega -1) (\omega + 1)} = \sqrt{ r_1 r_2 e^{i (\theta_1 + \theta_2)}} = \sqrt{r_1 r_2} e^{i (\theta_1 + \theta_2)/2} \, .
\end{align}
Plugging $ \omega = \eta(\lambda + i \ep)$ into \eqref{eq: sqrt(w^2 -1)}, we obtain
\begin{align}\label{handwavy: compute sqrt}
    \sqrt{\eta^2(\lambda + i \ep) -1} &= \sqrt{r_1 r_2} e^{i (\theta_1 + \theta_2) /2} = \sqrt{r_1 r_2} \cos\l(\frac{\theta_1 + \theta_2}{2} \r) + i \sqrt{r_1 r_2} \sin\l(\frac{\theta_1 + \theta_2}{2} \r)\, ,
\end{align}
where $\theta_1, \theta_2 \in (0, \pi)$ since $\ep > 0$. Now let $\omega_- = \eta(\lambda - i \ep)$, and define $r_1^-, r_2^-, \theta_1^-, \theta_2^- $ analogously to \eqref{eq: r and theta def}. It is easy to see geometrically that $r_i = r_i^-$ and $\theta_i = -\theta_i^-$ for $i = 1,2$.
Thus \begin{align*}
    \sqrt{\eta^2(\lambda - i \ep) -1} &= \sqrt{r_1 r_2} e^{i (\theta_1^- + \theta_2^-) /2} = \sqrt{r_1 r_2} \cos\l(\frac{\theta_1^- + \theta_2^-}{2} \r) + i \sqrt{r_1 r_2} \sin\l(\frac{\theta_1^- + \theta_2^-}{2} \r) \\
    &= \sqrt{r_1 r_2} \cos\l(\frac{\theta_1 + \theta_2}{2} \r) - i \sqrt{r_1 r_2} \sin\l(\frac{\theta_1 + \theta_2}{2} \r)\, .
\end{align*}
By plugging in the expressions for $\sqrt{\eta^2(\lambda \pm i \ep) -1}$ in the first and last lines, we compute
\begin{align*}
    \l| i \sqrt{1 - \eta^2(\lambda)} - \sqrt{\eta^2(\lambda + i \ep) - 1} \r|^2 &= \l| i \sqrt{1 - \eta^2(\lambda)} - \sqrt{r_1 r_2} \cos\l(\frac{\theta_1 + \theta_2}{2} \r) - i \sqrt{r_1 r_2} \sin\l(\frac{\theta_1 + \theta_2}{2} \r) \r|^2 \\
    &= \l|\sqrt{r_1 r_2} \cos\l(\frac{\theta_1 + \theta_2}{2} \r) \r|^2 + \l| \sqrt{1 - \eta^2(\lambda)} - \sqrt{r_1 r_2} \sin\l(\frac{\theta_1 + \theta_2}{2} \r) \r|^2 \\
    &= \l| -i \sqrt{1 - \eta^2(\lambda)} - \sqrt{r_1 r_2} \cos\l(\frac{\theta_1 + \theta_2}{2} \r) + i \sqrt{r_1 r_2} \sin\l(\frac{\theta_1 + \theta_2}{2} \r) \r|^2  \\
    &= \l|- i \sqrt{1 - \eta^2(\lambda)} - \sqrt{\eta^2(\lambda - i \ep) - 1} \r|^2\, ,
\end{align*}

which concludes the proof of Claim 1.
\end{proof}

In order to prove the lemma, it now suffices to obtain a bound on the left hand side of  equation \eqref{handwavy: symmetry}. To do so, we continue to use the notation $r_i, \theta_i$ for $i = 1,2$ as in \eqref{eq: r and theta def} with $\omega = \eta(\lambda + i \ep)$. 

Applying the triangle inequality and  \eqref{handwavy: compute sqrt}, we compute \begin{align}\label{handwavy: complex difference}
    \l|\sqrt{\eta^2(\lambda + i \ep) - 1} - i \sqrt{1 - \eta^2(\lambda)}\r| \leq \sqrt{r_1 r_2} \l| \cos\l(\frac{\theta_1 + \theta_2}{2} \r) \r| + \l| \sqrt{r_1 r_2} \sin\l(\frac{\theta_1 + \theta_2}{2} \r)  -  \sqrt{1 - \eta^2(\lambda)}\r| \, ,
\end{align}
where again $\theta_1, \theta_2 \in (0, \pi)$ since $\ep > 0$. 
\begin{rmk}

Using \eqref{eq: r and theta def} with $\omega = \eta(\lambda + i \ep)$, we have \begin{align*}
    r_1 r_2 = |\omega -1| |\omega+1| =  \l|\eta(\lambda + i \ep) - 1 \r| \l|\eta(\lambda + i \ep) + 1 \r|\, .
\end{align*}

Suppose that $\lambda = \g_{\pm}$. As  $\ep \rightarrow 0^+$, either $r_1$ or $r_2$ vanishes, and hence the product $r_1 r_2$ vanishes. This corresponds to $\omega$ approaching $\pm 1$ in Figure \ref{fig: sqrt(z^2-1)}. A direct calculation shows that 
$\sqrt{1 - \eta^2(\g_{\pm})} = 0$. Thus every term on the right hand side of \eqref{handwavy: complex difference}  vanishes as $\ep \rightarrow 0^+$ when $\lambda = \g_{\pm}$. 

Next suppose that $\lambda \in (\g_-, \g_+)$. Sending $\ep \rightarrow 0^+$ corresponds to $\omega$ approaching the interval $(-1,1)$ in Figure \ref{fig: sqrt(z^2-1)}. Hence  $\theta_1 \rightarrow \pi$, $\theta_2 \rightarrow 0$, and the first term on the right hand side of \eqref{handwavy: complex difference} vanishes as $\ep \rightarrow 0^+$. Furthermore, sending $\ep \rightarrow 0^+$, and thus $(\theta_1 + \theta_2)/2 \rightarrow \pi/2$ in \eqref{handwavy: compute sqrt}, implies that \begin{align*}
    i \sqrt{r_1 r_2}  = \sqrt{\eta^2(\lambda) - 1}  = i \sqrt{1 - \eta^2(\lambda)} \,. 
\end{align*}
Thus the second term on the right hand side of \eqref{handwavy: complex difference} vanishes for $\lambda \in (\g_-, \g_+)$ as well. 

The key observation here is that the means by which \eqref{handwavy: complex difference} vanishes depends on how close $\lambda$ is to $\g_{\pm}$. In order to compute a precise rate of decay, it will be advantageous to partition $\lambda \in [\g_-, \g_+]$ into regions where $\lambda$ in close to the boundary and regions where $\lambda$ is far from the boundary. 
\end{rmk}

Our strategy for estimating the right hand side of \eqref{handwavy: complex difference} for all $\lambda \in [\g_-, \g_+] $ and all sufficiently small $\ep > 0$  will be to partition  \begin{align*}
    \l[\g_-, \g_+\r] = A_{\ep} \cup B_{\ep} \cup C_{\ep}\, ,
\end{align*}
where \begin{equation}\label{handwavy: partition}\begin{split}
    A_{\ep} &:= \l\{ \lambda \in\l[\g_-, \g_+\r] : \ep \geq \frac{1}{2} \l|\lambda^2 - \g_-^2 \r| \r\} = \{ \lambda > 0 : \g_-^2 \leq \lambda^2 \leq \g_-^2 + 2\ep  \} \, , \\
    B_{\ep} &:= \l\{ \lambda \in \l[\g_-, \g_+\r] : \ep \leq \frac{1}{2} \l|\lambda^2 - \g_{\pm}^2 \r| \r\}  = \{ \lambda > 0 : \g_-^2  + 2 \ep\leq \lambda^2 \leq \g_+^2 - 2\ep  \}\, , \\
    C_{\ep} &:= \l\{ \lambda \in \l[\g_-, \g_+\r] : \ep \geq \frac{1}{2} \l|\lambda^2 - \g_+^2 \r| \r\}  = \{ \lambda > 0 : \g_+^2 - 2\ep \leq \lambda^2 \leq \g_+^2  \}\, ,\\
\end{split}
\end{equation}
and estimate \eqref{handwavy: complex difference} on each of the sets for each $\ep$. Note that \begin{align*}
    \g_+^2 - \g_-^2 = 4 \g_1 |\g_2|\, ,
\end{align*}
and so in order for $B_{\ep}$ to be nonempty, we must require that $\ep < \g_1 |\g_2|$. In addition, we want higher powers of $\ep$ to be smaller than lower powers of $\ep$. Although $\ep < 1$ would suffice for that goal, it ends up being convenient to impose the stronger condition $\ep < 1/2$ when we consider  $\lambda \in B_{\ep}$. 

We now proceed to obtain a more explicit representation of the terms on the right hand side of \eqref{handwavy: complex difference}. 
Through elementary trigonometric relations, 
we compute 
\begin{align}
    \l| \cos\l(\frac{\theta_1 + \theta_2}{2} \r) \r| &= \sqrt{\frac{1 + \cos(\theta_1) \cos(\theta_2) - \sin(\theta_1) \sin(\theta_2)}{2}} \, \label{eq: cos avg} ,\\
    \sin\l(\frac{\theta_1 + \theta_2}{2} \r) 
    &= \sqrt{\frac{1 - \cos(\theta_1) \cos(\theta_2) + \sin(\theta_1) \sin(\theta_2)}{2}}\, . \label{eq: sin avg}
\end{align}

Next we compute \begin{align}
    \eta(\lambda + i \ep) -1 &= \frac{\lambda^2 - \ep^2 + 2i \ep \lambda - \g_1^2 - |\g_2|^2}{2 \g_1 |\g_2|} - 1 = \frac{\lambda^2 - \ep^2 - \g_+^2  }{2 \g_1 |\g_2|} + i\frac{ 2\ep \lambda}{2\g_1 |\g_2|} \, ,\label{eq: e-1}\\
    \eta(\lambda + i \ep) +1 &= \frac{\lambda^2 - \ep^2 + 2i \ep \lambda - \g_1^2 - |\g_2|^2}{2 \g_1 |\g_2|} + 1 = \frac{\lambda^2 - \ep^2 - \g_-^2}{2 \g_1 |\g_2|} + i\frac{ 2\ep \lambda}{2\g_1 |\g_2|}  \, ,\label{eq: e+1}\\
    r_1 &= \frac{1}{2 \g_1 |\g_2|} \sqrt{|\g_+^2 - (\lambda^2 - \ep^2) |^2 + 4 \ep^2 \lambda^2} \, ,\label{eq: r1}\\
    r_2 &= \frac{1}{2 \g_1 |\g_2|} \sqrt{|\lambda^2 - \ep^2 - \g_-^2 |^2 + 4 \ep^2 \lambda^2} \, .  \label{eq: r2}
\end{align}

Using the geometric definition of the dot product, we have $r_1 \cos \theta_1 = \Re ( \eta(\lambda + i \ep) -1)$ and $r_2 \cos \theta_2 = \Re ( \eta(\lambda + i \ep) +1)$. Thus \begin{align}
    \cos \theta_1 &= \frac{\lambda^2 - \ep^2 - \g_+^2  }{\sqrt{|\g_+^2 - (\lambda^2 - \ep^2)|^2 + 4 \ep^2 \lambda^2}} \, , \label{eq: cos theta1}\\
     \cos \theta_2 &= \frac{\lambda^2 - \ep^2 - \g_-^2  }{\sqrt{|\lambda^2 - \ep^2 - \g_-^2|^2 + 4 \ep^2 \lambda^2}} \, . \label{eq: cos theta2}  
\end{align}
Using $\sin(\theta) = \sqrt{1 - \cos^2(\theta)}$, we compute \begin{align}
    \sin \theta_1 &= \frac{2 \ep \lambda}{\sqrt{|\g_+^2 - (\lambda^2 - \ep^2)|^2 + 4 \ep^2 \lambda^2}} \, , \label{eq: sin theta1} \\
    \sin \theta_2 &= \frac{2 \ep \lambda}{\sqrt{|\lambda^2 - \ep^2 - \g_-^2|^2 + 4 \ep^2 \lambda^2}}  \, . \label{eq: sin theta2}
\end{align} 

Using \eqref{eq: cos avg} and plugging in the expressions for  $r_i, \cos(\theta_i), \sin(\theta_i)$ for $ i = 1,2$ in \eqref{eq: r1}, \eqref{eq: r2}, \eqref{eq: cos theta1}, \eqref{eq: cos theta2}, \eqref{eq: sin theta1}, and \eqref{eq: sin theta2}, we compute \begin{equation}\label{eq: r1 r2 cos^2}\begin{split}
    &r_1 r_2 \cos^2 \l( \frac{\theta_1 + \theta_2}{2}\r) = r_1 r_2\frac{ 1 + \cos(\theta_1) \cos(\theta_2) - \sin( \theta_1) \sin(\theta_2)}{2} \\
    &= \frac{\sqrt{|\g_+^2 - (\lambda^2 - \ep^2) |^2 + 4 \ep^2 \lambda^2} \sqrt{|\lambda^2 - \ep^2 - \g_-^2 |^2 + 4 \ep^2 \lambda^2}  + \l[\lambda^2 - \ep^2 - \g_+^2\r] \l[\lambda^2 - \ep^2 - \g_-^2\r] - 4 \ep^2 \lambda^2}{2(2 \g_1 |\g_2|)^2}  \\
    &= \frac{\sqrt{|\g_+^2 - (\lambda^2 - \ep^2) |^2 + 4 \ep^2 \lambda^2} \sqrt{|\lambda^2 - \ep^2 - \g_-^2 |^2 + 4 \ep^2 \lambda^2}  - \l[\g_+^2 - \l( \lambda^2 - \ep^2  \r) \r] \l[\lambda^2 - \ep^2 - \g_-^2\r] - 4 \ep^2 \lambda^2}{2(2 \g_1 |\g_2|)^2} \, .  
\end{split}
\end{equation}
In similar fashion, we obtain 
\begin{equation}\label{eq: r1 r2 sin^2}
\begin{split}
    &r_1 r_2 \sin^2 \l( \frac{\theta_1 + \theta_2}{2}\r) \\
    &\quad  = \frac{\sqrt{|\g_+^2 - (\lambda^2 - \ep^2) |^2 + 4 \ep^2 \lambda^2} \sqrt{|\lambda^2 - \ep^2 - \g_-^2 |^2 + 4 \ep^2 \lambda^2}  + \l[\g_+^2 - \l( \lambda^2 - \ep^2  \r) \r] 
    \l[\lambda^2 - \ep^2 - \g_-^2\r] + 4 \ep^2 \lambda^2}{2(2 \g_1 |\g_2|)^2} \, .   
\end{split}
\end{equation}
Note that \begin{align*}
    &\l[\g_+^2 - \lambda^2 \r] \l[\lambda^2 - \g_-^2\r] = - \l[ \lambda^2 - ( \g_1^2 + 2 \g_1 |\g_2| + |\g_2|^2) \r] \l[\lambda^2 - ( \g_1^2 - 2 \g_1 |\g_2| + |\g_2|^2) \r] \\
    &= -\l[ (\lambda^2 - \g_1^2 - |\g_2|^2) - 2 \g_1 |\g_2| \r]\l[ (\lambda^2 - \g_1^2 - |\g_2|^2) + 2 \g_1 |\g_2| \r] \\
    &= (2 \g_1 |\g_2|)^2 - \l(\lambda^2 - \g_1^2 - |\g_2|^2\r)^2 \, .
\end{align*}
Hence
\begin{align}\label{eq: 1 - eta^2}
    1 - \eta^2(\lambda) = 1 - \l( \frac{\lambda^2 - \g_1^2 - |\g_2|^2}{2 \g_1 |\g_2|} \r)^2 = \frac{(2 \g_1 |\g_2|)^2 - \l(\lambda^2 - \g_1^2 - |\g_2|^2\r)^2 }{(2 \g_1 |\g_2|)^2} = \frac{\l[\g_+^2 - \lambda^2 \r] \l[\lambda^2 - \g_-^2\r] }{(2 \g_1 |\g_2|)^2}\, .
\end{align}
With these calculations in hand, we proceed to estimate \eqref{handwavy: complex difference} on each of the intervals in \eqref{handwavy: partition}.

\textbf{Case: $ \lambda \in A_{\ep}, \, 0 < \ep < \min \{1/2,  \g_1 |\g_2|\} $}

Using \eqref{eq: r1} and \eqref{eq: r2}, we compute
\begin{align*}
    r_1 r_2 &= \frac{1}{(2 \g_1 |\g_2|)^2} \sqrt{|\g_+^2 - (\lambda^2 - \ep^2) |^2 + 4 \ep^2 \lambda^2}  \sqrt{|\lambda^2 - \ep^2 - \g_-^2 |^2 + 4 \ep^2 \lambda^2}\, .
\end{align*}
Note that for $\lambda \in A_{\ep}$, $0 < \ep < \min \{ 1/2, \g_1 |\g_2| \}$,
\begin{itemize}
    \item $|\lambda^2 - \ep^2 - \g_-^2| \leq |2\ep - \ep^2| \leq 2 \ep$ ,
    \item $|\g_+^2 - (\lambda^2 - \ep^2) | \leq  4\g_1 |\g_2| + \ep^2 \leq 5 \g_1 |\g_2|$ .
\end{itemize}
Putting these facts together, we estimate
\begin{align*}
    r_1 r_2 &\leq \frac{1}{(2 \g_1 |\g_2|)^2} \sqrt{25 \g_1^2 |\g_2|^2 + 4 \g_1^2 |\g_2|^2 \g_+^2 } \sqrt{4\ep^2 + 4 \ep^2 \g_+^2} \\
    &= \frac{ \ep}{2 \g_1 |\g_2|} \sqrt{25  + 4  \g_+^2 } \sqrt{1 + \g_+^2} \, .  
\end{align*}

Furthermore, using \eqref{eq: 1 - eta^2}, we estimate for $\lambda \in A_{\ep}$, 
\begin{align*}
    1- \eta^2(\lambda) = \frac{\l[\g_+^2 - \lambda^2 \r] \l[\lambda^2 - \g_-^2\r] }{(2 \g_1 |\g_2|)^2} \leq  \frac{\l[4 \g_1 |\g_2| - 2 \ep \r] \l[2 \ep\r] }{(2 \g_1 |\g_2|)^2} \leq \frac{8 \ep \g_1 |\g_2|}{(2 \g_1 |\g_2|)^2} = \frac{2 \ep}{\g_1 |\g_2|} \, .  
\end{align*}
Applying the estimates for $r_1 r_2$ and $1- \eta^2(\lambda)$ to \eqref{handwavy: complex difference}, we compute\begin{align*}
    \l|\sqrt{\eta^2(\lambda + i \ep) - 1} - i \sqrt{1 - \eta^2(\lambda)}\r| 
    &\leq \sqrt{r_1 r_2} \l| \cos\l(\frac{\theta_1 + \theta_2}{2} \r) \r| + \l| \sqrt{r_1 r_2} \sin\l(\frac{\theta_1 + \theta_2}{2} \r)  -  \sqrt{1 - \eta^2(\lambda)}\r| \\
    &\leq 2 \sqrt{r_1 r_2} + \sqrt{1 - \eta^2(\lambda)}  \\
    &\leq 2 \sqrt{\frac{\ep}{ 2\g_1 |\g_2|} \sqrt{ 25  + 4 \g_+^2} \sqrt{1 + \g_+^2}} + \sqrt{\frac{2 \ep}{\g_1 |\g_2|}} \\
    &\leq P_1 (\g_1, |\g_2|) \sqrt{\frac{\ep}{\g_1 |\g_2|}} \, , 
\end{align*}
where $P_1(x,y)$ is a polynomial of degree 1 since $\g_+^2 = \g_1^2 + 2 \g_1 |\g_2| + |\g_2|^2$ is of degree 2 in $\g_1, |\g_2|$. The estimate for $\lambda \in C_{\ep}$ is analogous. 

\textbf{Case: $\lambda \in B_{\ep},\, 0 < \ep < \min \l \{  1/2,   \g_1 |\g_2|  \r \}$}

We begin by estimating the square of the first term on the right hand side of \eqref{handwavy: complex difference},
\begin{align*}
    r_1 r_2 \cos^2 \l( \frac{\theta_1 + \theta_2}{2}\r)\, ,
\end{align*}
by using the expression given in \eqref{eq: r1 r2 cos^2}. With the inequality $\sqrt{a^2 + b^2} \leq |a| + |b|$, we compute \begin{equation}\label{handwavy: bound on sqrts 1}\begin{split}
    &\sqrt{\l|\g_+^2 - (\lambda^2 - \ep^2) \r|^2 + 4 \ep^2 \lambda^2} \sqrt{\l|\lambda^2 - \ep^2 - \g_-^2 \r|^2 + 4\ep^2 \lambda^2} 
    \leq \bigg[ \l|\g_+^2 - (\lambda^2 - \ep^2) \r| + 2 \ep \lambda \bigg] \bigg[ \l|\lambda^2 - \ep^2 - \g_-^2 \r| + 2 \ep \lambda  \bigg] \\
    &\quad = |\g_+^2 - (\lambda^2 - \ep^2) | |\lambda^2 - \ep^2 - \g_-^2 | + 2 \ep \lambda \bigg[ |\g_+^2 - (\lambda^2 - \ep^2) | + |\lambda^2 - \ep^2 - \g_-^2 | \bigg] + 4 \ep^2 \lambda^2 \, .   
\end{split}\end{equation}
Note that for all $\lambda \in B_{\ep}$, 
\begin{align}\label{handwavy: remove abs val}
    |\g_+^2 - (\lambda^2 - \ep^2) | = \g_+^2 - (\lambda^2 - \ep^2) \quad \text{and} \quad
    |\lambda^2 - \ep^2 - \g_-^2 |  = \lambda^2 - \ep^2 - \g_-^2 \, .
\end{align}
Thus \eqref{handwavy: bound on sqrts 1} simplifies to \begin{align}\label{handwavy: bound on sqrts 2}
    \sqrt{\l|\g_+^2 - (\lambda^2 - \ep^2) \r|^2 + 4 \ep^2 \lambda^2} \sqrt{\l|\lambda^2 - \ep^2 - \g_-^2 \r|^2 + 4\ep^2 \lambda^2}  \leq \l[\g_+^2 - (\lambda^2 - \ep^2) \r] \l[ \lambda^2 - \ep^2 - \g_-^2 \r] + 8 \g_1 |\g_2| \ep \lambda + 4 \ep^2 \lambda^2 \, .  
\end{align}

Using \eqref{handwavy: bound on sqrts 2} in \eqref{eq: r1 r2 cos^2}, we compute \begin{equation} \label{eq: r1 r2 cos^2 bound}
\begin{split}
    r_1 r_2 \cos^2 \l( \frac{\theta_1 + \theta_2}{2}\r) &\leq \frac{8 \ep  \g_1 |\g_2| \lambda }{2(2 \g_1 |\g_2|)^2} = \frac{\ep \lambda}{\g_1 |\g_2|} \leq \frac{\ep \g_+}{\g_1 |\g_2|} \\
    \implies \sqrt{r_1 r_2} \l| \cos\l(\frac{\theta_1 + \theta_2}{2} \r) \r| &\leq \sqrt{\frac{\ep}{\g_1 |\g_2|}} \sqrt{\g_+} \, .
\end{split}
\end{equation}

Next we estimate the second term on the right hand side of \eqref{handwavy: complex difference}, \begin{align*}
    \l| \sqrt{r_1 r_2} \sin\l(\frac{\theta_1 + \theta_2}{2} \r)  -  \sqrt{1 - \eta^2(\lambda)}\r| \, .  
\end{align*}

\textbf{Claim 2:}
For all $0 < \ep < \min \l \{  1/2,   \g_1 |\g_2|  \r \}$ and  $\lambda \in B_{\ep}$,  
\begin{align}\label{handwavy: claim 2}
    \sqrt{1 - \eta^2(\lambda)} \leq \sqrt{r_1 r_2} \sin\l(\frac{\theta_1 + \theta_2}{2} \r)  \, . 
\end{align}

Equation \eqref{handwavy: claim 2} implies that \begin{align*}
    \l| \sqrt{r_1 r_2} \sin\l(\frac{\theta_1 + \theta_2}{2} \r)  -  \sqrt{1 - \eta^2(\lambda)}\r|  \leq   \sqrt{r_1 r_2}  -  \sqrt{1 - \eta^2(\lambda)} \, ,
\end{align*}
and our strategy for estimating the right hand side is to note that \begin{align}\label{handwavy: above and below trick}
    \l[\sqrt{r_1 r_2}  -  \sqrt{1 - \eta^2(\lambda)} \r] \l[\sqrt{r_1 r_2}  +  \sqrt{1 - \eta^2(\lambda)} \r] = r_1 r_2 - (1 - \eta^2(\lambda)) \, . 
\end{align}
By \eqref{eq: 1 - eta^2}, $1-\eta^2(\lambda)$ is quadratic in $\lambda^2$ and must obtain its minimum at one of the endpoints of $B_{\ep}$. Thus \begin{align*}
    &1 - \eta^2(\lambda) \geq \min_{\lambda^2 \in \{\g_-^2 + 2 \ep, \g_+^2 - 2 \ep \} } \frac{\l[\g_+^2 - \lambda^2 \r] \l[\lambda^2 - \g_-^2\r] }{(2 \g_1 |\g_2|)^2} = \frac{2 \ep ( 4 \g_1 |\g_2| - 2 \ep)}{(2 \g_1 |\g_2|)^2} \geq \frac{2\ep}{\g_1 |\g_2|} \\
    \implies &\sqrt{r_1 r_2} + \sqrt{1 - \eta^2(\lambda)} \geq \sqrt{\frac{2\ep}{\g_1 |\g_2|}}\, . 
\end{align*}
Using this lower bound in \eqref{handwavy: above and below trick}, we have \begin{align}\label{handwavy: upper bound 1}
    \sqrt{r_1 r_2}  -  \sqrt{1 - \eta^2(\lambda)}  \leq  \big[  r_1 r_2 - (1 - \eta^2(\lambda))  \big] \sqrt{\frac{\g_1 |\g_2|}{2\ep}} \, , 
\end{align}
and so it now suffices to obtain an upper bound on $r_1 r_2 - (1 - \eta^2(\lambda))$. 
Using \eqref{eq: r1}, \eqref{eq: r2}, and \eqref{eq: 1 - eta^2}, we compute \begin{align*}
    r_1 r_2 - (1 - \eta^2(\lambda)) = \frac{\sqrt{|\g_+^2 - (\lambda^2 - \ep^2) |^2 + 4 \ep^2 \lambda^2} \sqrt{|\lambda^2 - \ep^2 - \g_-^2 |^2 + 4 \ep^2 \lambda^2}  - [\g_+^2 - \lambda^2] [\lambda^2 - \g_-^2]}{(2 \g_1 |\g_2|)^2}  \, .  
\end{align*}
Then using \eqref{handwavy: bound on sqrts 2}, we compute 
\begin{equation}\label{eq: r1r2 - 1-e2} \begin{split}
    r_1 r_2 - (1 - \eta^2(\lambda)) 
    &\leq \frac{\l [ \g_+^2 - (\lambda^2 - \ep^2) \r] \l [ \lambda^2 - \ep^2 - \g_-^2 \r] + 8 \ep  \g_1 |\g_2| \lambda  + 4 \ep^2 \lambda^2 - [\g_+^2 - \lambda^2] [\lambda^2 - \g_-^2] }{(2 \g_1 |\g_2|)^2}  \\
    &= \frac{\ep^2 (\lambda - \g_-^2) - \ep^2 \l(\g_+^2 - \lambda^2 \r) - \ep^4+ 8 \ep  \g_1 |\g_2| \lambda  + 4 \ep^2 \lambda^2 }{(2 \g_1 |\g_2|)^2} \\
    &= \frac{ 6 \ep^2 \lambda^2 - 2 \ep^2 \l(\g_1^2 + |\g_2|^2 \r) - \ep^4 + 8 \ep  \g_1 |\g_2| \lambda   }{(2 \g_1 |\g_2|)^2} \\
    &\leq \frac{ 6 \g_1 |\g_2| \ep \lambda^2 + 8 \ep  \g_1 |\g_2| \lambda   }{(2 \g_1 |\g_2|)^2}  = \frac{ 3 \ep \lambda^2 + 4 \ep \lambda   }{2 \g_1 |\g_2|} \\
    &\leq \frac{ 3 \ep \g_+^2 + 4 \ep \g_+ }{2 \g_1 |\g_2|} \, .\\
\end{split}
\end{equation}

Combining the above estimate with \eqref{handwavy: upper bound 1}, we compute 
\begin{equation}\label{eq: sqrt r1r2 - sqrt 1-e2} \begin{split}
    \sqrt{r_1 r_2}  -  \sqrt{1 - \eta^2(\lambda)}  
    &\leq  \big[  r_1 r_2 - (1 - \eta^2(\lambda))  \big] \sqrt{\frac{\g_1 |\g_2|}{2\ep}} \\
    &\leq  \frac{ 3 \ep \g_+^2 + 4 \ep \g_+ }{2 \g_1 |\g_2|}  \sqrt{\frac{\g_1 |\g_2|}{2\ep}} \\
    &= (3 \g_+^2 + 4 \g_+) \sqrt{\frac{\ep}{8\g_1 \g_2}} \, .  
\end{split}
\end{equation}
Using \eqref{handwavy: complex difference}, \eqref{handwavy: claim 2}, and plugging in the bounds we computed in \eqref{eq: r1 r2 cos^2 bound} and \eqref{eq: sqrt r1r2 - sqrt 1-e2}, we obtain 
\begin{align*}
     \l|\sqrt{\eta^2(\lambda + i \ep) - 1} - i \sqrt{1 - \eta^2(\lambda)}\r| &\leq \sqrt{r_1 r_2} \l| \cos\l(\frac{\theta_1 + \theta_2}{2} \r) \r| + \l| \sqrt{r_1 r_2} \sin\l(\frac{\theta_1 + \theta_2}{2} \r)  -  \sqrt{1 - \eta^2(\lambda)}\r| \\
     &\leq \sqrt{r_1 r_2} \l| \cos\l(\frac{\theta_1 + \theta_2}{2} \r) \r| +  \sqrt{r_1 r_2}   -  \sqrt{1 - \eta^2(\lambda)}  \\
     &\leq \sqrt{\frac{\ep}{\g_1 |\g_2|}} \sqrt{\g_+} + (3 \g_+^2 + 4 \g_+) \sqrt{\frac{\ep}{8\g_1 \g_2}} \\
     &\leq P_2( \g_1, |\g_2|) \sqrt{\frac{\ep}{\g_1 |\g_2|}}  \, , 
\end{align*}
where $P_2(x,y)$ is a polynomial of degree 2 since $\g_+^2 = \g_1^2 + 2 \g_1 |\g_2| + |\g_2|^2$ is of degree 2 in $\g_1, |\g_2|$.

\begin{proof}[Proof of Claim 2]
In order to prove \eqref{handwavy: claim 2}, it suffices to show that 
\begin{align}
    (2 \g_1 |\g_2|)^2 [ 1 -\eta^2(\lambda) ] \leq  (2 \g_1 |\g_2|)^2  r_1 r_2 \sin^2 \l(\frac{\theta_1 + \theta_2}{2} \r)\, .  \label{handwavy: claim 2 equivalent 1}
\end{align}

Using \eqref{eq: r1 r2 sin^2}, \eqref{handwavy: remove abs val}, and the fact that $\g_+^2 + \g_-^2 = 2 (\g_1^2 + |\g_2|^2)$, we compute \begin{align*}
     &(2 \g_1 |\g_2|)^2  r_1 r_2 \sin^2 \l(\frac{\theta_1 + \theta_2}{2} \r) \\
     &\quad = \sqrt{\l|\g_+^2 - \l(\lambda^2 - \ep^2\r) \r|^2 + 4 \ep^2 \lambda^2} \sqrt{\l|\lambda^2 - \ep^2 - \g_-^2 \r|^2 + 4 \ep^2 \lambda^2} + \l[\g_+^2 - \l( \lambda^2 - \ep^2  \r) \r] 
    \l[\lambda^2 - \ep^2 - \g_-^2\r] + 4 \ep^2 \lambda^2 \\
    &\quad \geq \l|\g_+^2 - \l(\lambda^2 - \ep^2\r) \r| \l|\lambda^2 - \ep^2 - \g_-^2 \r| + \l[\g_+^2 - \l( \lambda^2 - \ep^2  \r) \r] 
    \l[\lambda^2 - \ep^2 - \g_-^2\r] + 4 \ep^2 \lambda^2 \\
    &\quad = 2\l[\g_+^2 - \l( \lambda^2 - \ep^2  \r) \r] 
    \l[\lambda^2 - \ep^2 - \g_-^2\r] + 4 \ep^2 \lambda^2 \\
    &\quad = 2 \Big( \l[\g_+^2 - \lambda^2 \r] \l[ \lambda^2  - \g_-^2\r] + \ep^2  \l[ \lambda^2 - \g_-^2 \r] - \ep^2 \l[ \g_+^2 - \lambda^2  \r]  - \ep^4  \Big) + 4 \ep^2 \lambda^2 \\
    &\quad = 2\l[\g_+^2 - \lambda^2 \r] \l[ \lambda^2  - \g_-^2\r] + 8 \ep^2 \lambda^2 - 4 \ep^2 (\g_1^2 + |\g_2|^2 ) - 2 \ep^4 \, . 
\end{align*}
Combining this lower bound with the fact that by \eqref{eq: 1 - eta^2}, $(2 \g_1 |\g_2|)^2 [ 1 -\eta^2(\lambda) ] = \l[\g_+^2 - \lambda^2 \r] \l[ \lambda^2  - \g_-^2\r] $, we see that in order to prove \eqref{handwavy: claim 2 equivalent 1}, it suffices to prove
\begin{subequations}
\begin{align}
    \l[\g_+^2 - \lambda^2 \r] \l[ \lambda^2  - \g_-^2\r] &\leq 2\l[\g_+^2 - \lambda^2 \r] \l[ \lambda^2  - \g_-^2\r] + 8 \ep^2 \lambda^2 - 4 \ep^2 (\g_1^2 + |\g_2|^2 ) - 2 \ep^4 \\
    \iff  4 \ep^2 (\g_1^2 + |\g_2|^2 ) + 2 \ep^4 &\leq  \l[\g_+^2 - \lambda^2 \r] \l[ \lambda^2  - \g_-^2\r] + 8 \ep^2 \lambda^2 \, .  \label{handwavy: claim 2 equivalent 2}
\end{align}
\end{subequations}

Since the right hand side of \eqref{handwavy: claim 2 equivalent 2} is quadratic in $\lambda^2$, it suffices to check that the inequality holds at both endpoints of the interval $B_{\ep}$. The following calculation shows that we need only to consider the left endpoint: 
\begin{align*}
    \l. [\g_+^2 - \lambda^2] [\lambda^2 - \g_-^2] + 8 \ep^2 \lambda^2 \r|_{\lambda^2 = \g_-^2 + 2 \ep} 
    &= 2 \ep \l( \g_+^2 - \g_-^2 -2 \ep \r) +8 \ep^2 (\g_-^2 + 2 \ep) \\
    &= 2 \ep \l( 4 \g_1 |\g_2| -2 \ep \r) +8 \ep^2 (\g_-^2 + 2 \ep) \\
    &\leq 2 \ep \l( 4 \g_1 |\g_2|-2 \ep \r) +8 \ep^2 (\g_+^2 - 2 \ep)  \\
    &=  \l. [\g_+^2 - \lambda^2] [\lambda^2 - \g_-^2] + 8 \ep^2 \lambda^2 \r|_{\lambda^2 = \g_+^2 - 2 \ep} \, . 
\end{align*}
Hence by taking the minimum over all $\lambda \in B_{\ep}$ of the right hand side of \eqref{handwavy: claim 2 equivalent 2}, it suffices to show that  \begin{align*}
    4 \ep^2 (\g_1^2 + |\g_2|^2 ) + 2 \ep^4 &\leq  2 \ep \l(  4 \g_1 |\g_2| -2 \ep \r) +8 \ep^2 (\g_-^2 + 2 \ep)  \\
    \iff 2 \ep (\g_1^2 + |\g_2|^2 ) + \ep^3 &\leq   4 \g_1 |\g_2| -2 \ep  + 4\ep (\g_-^2 + 2 \ep) \qquad \qquad \qquad \qquad  &[\text{divide by } 2 \ep] \\
    \iff 2 \ep (\g_1^2 + |\g_2|^2 ) + \ep^3 &\leq   4 \g_1 |\g_2| -2 \ep  + 4\ep (\g_1^2 + |\g_2|^2 - 2 \g_1 |\g_2| ) + 8 \ep^2  \qquad  &[\text{expand RHS}]\\
    \iff 8 \ep \g_1 |\g_2| + 2 \ep + \ep^3 &\leq   4 \g_1 |\g_2|  + 2\ep (\g_1^2 + |\g_2|^2 )  + 8 \ep^2 \, .  
\end{align*}

Cauchy's inequality implies $4 \ep \g_1 |\g_2| \leq 2 \ep (\g_1^2 + |\g_2|^2)$ and $\ep < 1/2$ implies that $\ep^3 < 8 \ep^2$, so we may equivalently show  
\begin{align*}
    4 \ep \g_1 |\g_2| + 2 \ep  \leq   4 \g_1 |\g_2| 
    \iff \ep  + \frac{\ep}{2 \g_1 |\g_2|} \leq 1 \, .  
\end{align*}
This last inequality is satisfied for $ \ep \leq \min \l \{ 1/2,  \g_1 |\g_2|  \r \}$ . 
\end{proof}
\end{proof}
\end{lem}

\begin{lem}\label{lem: handwavy h lower bound}
There exists $\ep_* =\ep_*(\g_1 ,\g_2) > 0$ such that for all $0 < \ep < \ep_*$ and any $\lambda \in [\g_-, \g_+] $, \begin{align*}
    \l|e^{-i\q((\lambda \pm i \ep)^2)}  \r| \geq 1/2 \quad \text{ and } \quad 
    |h( \q((\lambda \pm i \ep)^2) + \varphi) | \geq  \g_-/2 \, . 
\end{align*}
\begin{proof}

Applying the reverse triangle inequality to \cref{lem: handwavy approx}, we compute 
\begin{equation}\label{handwavy: reverse triangle}\begin{split}
    P_2(\g_1, |\g_2|)\sqrt{\frac{\ep}{\g_1 |\g_2|}} &\geq  \l| e^{-i\q((\lambda \pm i \ep)^2)} - e^{\mp i\qzl} \r|  \geq 1 - \l|e^{-i\q((\lambda \pm i \ep)^2)} \r| \\
    \implies \l|e^{-i\q((\lambda \pm i \ep)^2)}  \r| &\geq 1-  P_2(\g_1, |\g_2|)\sqrt{\frac{\ep}{\g_1 |\g_2|}} \, .  
\end{split}
\end{equation}
Hence \begin{align*}
    \ep < \frac{\g_1 |\g_2|}{4[P_2(\g_1, |\g_2|)]^2} \implies \l|e^{-i\q((\lambda \pm i \ep)^2)}  \r| \geq 1/2 \, . 
\end{align*}
For the second statement, we compute \begin{align*}
    h( \q((\lambda \pm i \ep)^2) + \varphi) = \g_1 + \g_2 e^{-i ( \q((\lambda \pm i \ep)^2) + \varphi)} = \g_1 + |\g_2| e^{-i \q((\lambda \pm i \ep)^2) } \, . 
\end{align*}
In the case where $\g_1 > |\g_2|$, we have 
\begin{align*}
    |h( \q((\lambda \pm i \ep)^2) + \varphi) | \geq \g_1 - |\g_2| = \g_- \, . 
\end{align*}
Next, consider the case where $|\g_2|  > \g_1$.  Then by \eqref{handwavy: reverse triangle}, for all sufficiently small $\ep >0$, we have \begin{align*}
    \l|e^{-i\q((\lambda \pm i \ep)^2)}  \r| &\geq 1 - \frac{\g_-}{2 |\g_2|}  \\
    \implies |\g_2| \l|e^{-i\q((\lambda \pm i \ep)^2)}  \r| &\geq |\g_2| - \frac{\g_-}{2} = \frac{|\g_2| + \g_1}{2} \\
    \implies |\g_2| \l|e^{-i\q((\lambda \pm i \ep)^2)} \r| - \g_1 &\geq \frac{\g_-}{2} \, . 
\end{align*}
The reverse triangle inequality then implies $|h( \q((\lambda \pm i \ep)^2) + \varphi)| \geq \g_- /2$. 
\end{proof}
\end{lem}

\begin{proof}[\textbf{Proof of Proposition} \ref{prop: handwavy approx}]
Define \begin{align*}
    A(\ep, q , \lambda) &:= 
    \begin{bmatrix}
        (\lambda \pm i \ep) h(q)/h(\q((\lambda \pm i \ep)^2)+ \varphi) & h(q)  \\[5pt] \overline{h\l(\overline{\q((\lambda \pm i \ep)^2)+ \varphi}\r) } & \lambda \pm i \ep
    \end{bmatrix}  \, ,
\end{align*}
and note that by equation \eqref{eq: limit of q} of  \cref{prop: limit of q}, 
\begin{align*}
    A(0^+, q , \lambda) :=  \lim_{\ep \rightarrow 0^+} A(\ep; q , \lambda) =
    \begin{bmatrix}
        \lambda  h(q)/h(\pm \qzl + \varphi) & h(q)  \\[5pt] \overline{h\l(\pm \qzl+ \varphi\r) } & \lambda 
    \end{bmatrix} \, . 
\end{align*}
With this notation, proving \cref{prop: handwavy approx} is equivalent to showing that 
\begin{equation}\begin{split}
    \int \limits_{\g_-}^{\g_+} \int \limits_{- \pi}^{\pi}  \l| \frac{e^{i (n+1) q }  }{k^2(q - \varphi) - (\lambda \pm i \ep)^2} \l(A(\ep; q , \lambda) \widetilde{Sf}(\q((\lambda \pm i \ep)^2) + \varphi) - A(0^+; q , \lambda) \widetilde{Sf}( \pm \qzl + \varphi) \r) \r| \, dq \, d\lambda \\
    \leq c(\g_1, |\g_2|) \|f\|_{\ell_1^1} \sqrt{\ep} \, \text{arsinh}\l(\frac{2\g_1|\g_2|}{\ep} \r) \, .  \qquad \qquad \qquad \qquad \qquad \qquad \qquad \qquad \qquad \qquad \qquad \label{eq: handwavy approx equiv}
\end{split}\end{equation}
In order to prove \eqref{eq: handwavy approx equiv}, it suffices to prove the following two claims. 

\textbf{Claim 1:}
For $\lambda \in [\g_-, \g_+]$  and $0 < \ep < \min\{1/2, \g_1 |\g_2|\}$, 
\begin{align*}
     \| A(\ep, \cdot, \lambda) \widetilde{Sf}(\q((\lambda \pm i \ep)^2) + \varphi) - A(0^+, \cdot, \lambda) \widetilde{Sf}( \pm \qzl + \varphi) \|_{\infty} \leq \frac{P_5(\g_1, |\g_2|) \sqrt{\ep} }{ \g_-^2 \sqrt{\g_1 |\g_2|}} \|f\|_{\ell^1_1} \, .
\end{align*}
where $P_5(x,y)$ is a polynomial of degree 5. 

\textbf{Claim 2:}
There exists a constant $c(\g_1,|\g_2|)$ such that for all $\ep > 0$, 
\begin{align*}
    \int \limits_{\g_-}^{\g_+} \int \limits_{-\pi}^{\pi} \frac{\sqrt{\ep}}{ \l| k^2(q - \varphi) - (\lambda \pm i \ep)^2 \r|} \, dq \, d\lambda  \leq c(\g_1, |\g_2|) \sqrt{\ep} \, \text{arsinh}\l(\frac{2\g_1|\g_2|}{\ep} \r)
\end{align*}

\begin{proof}[Proof of Claim 1]
Recall that $\tilde{f}(q) = \sum_{m \geq 0} e^{-imq} f_m$. Using the identity 
\begin{align*}
    a^n - b^n = (a-b) \sum_{j = 1}^n a^{n-j} b^{j-1} \, ,
\end{align*}
we compute \begin{align*}
    &\l|\Tilde{f}(\q((\lambda \pm i \ep)^2) + \varphi) - \Tilde{f}(\pm \qzl + \varphi) \r| = \l| \sum_{m \geq 0}  \l[ e^{-im\q((\lambda \pm i \ep)^2)) } - e^{\mp im \qzl} \r] e^{-im \varphi} f_m \r| \\
    &\leq  \sum_{m \geq 0}  \l| e^{-im\q((\lambda \pm i \ep)^2)) } - e^{\mp im \qzl} \r| |f_m| \\
    &= \sum_{m \geq 0}  \l| e^{-i\q((\lambda \pm i \ep)^2)) } - e^{\mp i \qzl}\r| \l|\sum_{j = 1}^m e^{-i(m-j)\q((\lambda \pm i \ep)^2)) }  e^{\mp i(j-1) \qzl} \r| |f_m| \, .
\end{align*}

Now using Lemma \ref{lem: handwavy approx} and the fact that $\Im \q((\lambda \pm i \ep)^2)) < 0$ and $\Im \qzl = 0$ imply \begin{align*}
    \l|e^{-im\q((\lambda \pm i \ep)^2)) } \r|, \l| e^{\mp im \qzl} \r| \leq 1\ , 
\end{align*}
we compute \begin{align*}
    \l|\tilde{f}(\q((\lambda \pm i \ep)^2) + \varphi) - \tilde{f}(\pm \qzl + \varphi) \r| &\leq \l| e^{-i\q((\lambda \pm i \ep)^2)) } - e^{\mp i \qzl}\r| \sum_{m \geq 0} m  |f_m|  \\
    &\leq P_2(\g_1, |\g_2|) \|f\|_{\ell_1^1} \sqrt{\frac{\ep}{\g_1 |\g_2|}} \, .
\end{align*}
The upper bound still holds when we replace $f$ with $Sf$, so 
\begin{align}\label{handwavy: Sf}
    \l|\widetilde{Sf}(\q((\lambda \pm i \ep)^2) + \varphi) - \widetilde{Sf}(\pm \qzl + \varphi) \r| \leq P_2(\g_1, |\g_2|) \|f\|_{\ell_1^1} \sqrt{\frac{\ep}{\g_1 |\g_2|}} \, .
\end{align}
Next we bound each entry of $A(\ep, q, \lambda) - A(0^+, q, \lambda)$. We immediately have \begin{align*}
    \l| A_{1,2}(\ep, q, \lambda) - A_{1,2}(0^+, q, \lambda) \r| = 0  \quad \text{ and } \quad
     \l| A_{2,2}(\ep, q, \lambda) - A_{2,2}(0^+, q, \lambda) \r| = \ep \, . 
\end{align*}
For the entries in the first column, recall that $h(q) = \g_1 + \g_2 e^{-iq} $. 
Using \cref{lem: handwavy approx}, we compute \begin{equation}\label{handwavy: A21}\begin{split}
    |A_{2,1}(\ep, q, \lambda) - A_{2,1}(0^+, q, \lambda)| &= \l| \overline{h\l( \overline{\q((\lambda \pm i \ep)^2) + \varphi} \r)} - \overline{h\l(  \pm \qzl + \varphi \r)} \r| \\
    &= \l| \overline{\g_2} e^{i (\q((\lambda \pm i \ep)^2) + \varphi)} -  \overline{\g_2} e^{i ( \varphi \pm \qzl)} \r|  \\
    &\leq |\g_2|  \l| e^{i\q((\lambda \pm i \ep)^2)} - e^{\pm i\qzl} \r|  = |\g_2| \l|\frac{e^{\mp i \qzl} -e^{-i \q((\lambda \pm i \ep)^2))} }{ e^{-i \q((\lambda \pm i \ep)^2))} e^{\mp i \qzl} }  \r| \\
    &\leq  \frac{P_3(\g_1, |\g_2|)\sqrt{\frac{\ep}{\g_1 |\g_2|}} }{\l|e^{-i \q((\lambda \pm i \ep)^2))} \r|}\, . 
\end{split}
\end{equation}
Applying \cref{lem: handwavy h lower bound} to \eqref{handwavy: A21}, we obtain for sufficiently small $\ep > 0$ 
\begin{align*}
     |A_{2,1}(\ep, q, \lambda) - A_{2,1}(0^+, q, \lambda)|
     &\leq \frac{P_3(\g_1, |\g_2|)\sqrt{\frac{\ep}{\g_1 |\g_2|}} }{\l|e^{-i \q((\lambda \pm i \ep)^2))} \r|} 
     \leq P_3(\g_1, |\g_2|)\sqrt{\frac{\ep}{\g_1 |\g_2|}} \, . 
\end{align*}
Next, using the fact that $\lambda, |h(q)| \leq \g_+$, we compute \begin{equation}\label{handwavy: A11}\begin{split}
    |A_{1,1}(\ep, q, \lambda) - A_{1,1}(0^+, q, \lambda)|
     &= \l| \frac{(\lambda \pm i \ep) h(q)}{h(\q((\lambda \pm i \ep)^2)+ \varphi)}  - \frac{\lambda h(q)}{h(\pm \qzl + \varphi)}\r| \\
     &\leq |h(q)| \l| \frac{\lambda \pm i \ep}{h(\q((\lambda \pm i \ep)^2)+ \varphi)} - \frac{\lambda}{h(\q((\lambda \pm i \ep)^2)+ \varphi)} \r| \\
     &\qquad + |h(q)| \l| \frac{\lambda}{h(\q((\lambda \pm i \ep)^2)+ \varphi)} - \frac{\lambda}{h(\pm \qzl + \varphi)} \r| \\
     &\leq \frac{\ep \g_+}{| h(\q((\lambda \pm i \ep)^2)+ \varphi)|} + \g_+^2  \l| \frac{h(\pm \qzl + \varphi) - h(\q((\lambda \pm i \ep)^2)+ \varphi) }{ h(\q((\lambda \pm i \ep)^2)+ \varphi) h(\pm \qzl + \varphi)} \r| \, .
\end{split}
\end{equation}
By \cref{lem: handwavy approx}, 
\begin{equation}\label{eq: diff of h upper bound}\begin{split}
    | h(\q((\lambda \pm i \ep)^2)+ \varphi)  -  h(\pm \qzl + \varphi) | 
    &\leq |\g_2|  \l| e^{-i\q((\lambda \pm i \ep)^2)} - e^{\mp i\qzl} \r|  \leq P_3(\g_1, |\g_2|)\sqrt{\frac{\ep}{\g_1 |\g_2|}} \, .
\end{split}
\end{equation}
Applying \cref{lem: handwavy approx} and the estimate \eqref{eq: diff of h upper bound} to \eqref{handwavy: A11}, we obtain
\begin{align*}
    |A_{1,1}(\ep, q, \lambda) - A_{1,1}(0^+, q, \lambda)| &\leq \frac{2\ep \g_+}{\g_-} + \frac{\g_+^2   P_3(\g_1, |\g_2|) \sqrt{\ep}}{\g_-^2 \sqrt{\g_1 |\g_2|}} \leq \frac{P_5(\g_1, |\g_2|) \sqrt{\ep} }{ \g_-^2 \sqrt{\g_1 |\g_2|}} \, .
\end{align*}
Thus for all $q \in [-\pi, \pi]$, $\lambda \in [\g_-, \g_+]$ and sufficiently small $\ep > 0$, 
\begin{align}\label{handwavy: A}
    |A(\ep, q, \lambda) - A( 0^+, q, \lambda)| \leq \frac{P_5(\g_1, |\g_2|) \sqrt{\ep} }{ \g_-^2 \sqrt{\g_1 |\g_2|}} \, .
\end{align}

Lastly by the triangle inequality and the fact that any two norms on a finite dimensional vector space are equivalent, we  compute 
\begin{align*}
    &\l| A(\ep, q, \lambda) \widetilde{Sf}(\q((\lambda \pm i \ep)^2) + \varphi) - A(0^+, q, \lambda) \widetilde{Sf}( \pm \qzl + \varphi) \r| \\
    \leq \; &\l| A(\ep, q, \lambda) \widetilde{Sf}(\q((\lambda \pm i \ep)^2) + \varphi) - A(0^+, q, \lambda) \widetilde{Sf}(\q((\lambda \pm i \ep)^2) + \varphi)  \r| \\
    &\quad + \l| A(0^+, q, \lambda) \widetilde{Sf}(\q((\lambda \pm i \ep)^2) + \varphi) - A(0^+, q, \lambda) \widetilde{Sf}( \pm \qzl + \varphi) \r| \\
    \lesssim \; & \| A(\ep, \cdot, \lambda) - A(0^+, \cdot, \lambda)\|_{\infty} |\widetilde{Sf}(\q((\lambda \pm i \ep)^2) + \varphi) | \\
    &\quad + \|A(0^+, \cdot, \lambda)\|_{\infty} | \widetilde{Sf}(\q((\lambda \pm i \ep)^2) + \varphi) - \widetilde{Sf}( \pm \qzl + \varphi)|
\end{align*} 

Using \eqref{handwavy: A}, \eqref{handwavy: Sf}, and the fact that $\|A(0^+, \cdot, \lambda)\|_{\infty} \leq \g_+^2 / \g_-$, we conclude 
\begin{align*}
    & \l| A(\ep, q, \lambda) \widetilde{Sf}(\q((\lambda \pm i \ep)^2) + \varphi) - A(0^+, q, \lambda) \widetilde{Sf}( \pm \qzl + \varphi) \r|  \\
    & \lesssim   
     \frac{P_5(\g_1, |\g_2|) \sqrt{\ep} }{ \g_-^2 \sqrt{\g_1 |\g_2|}} \|f\|_{\ell^1} + \frac{\g_+^2}{\g_-} P_2(\g_1, |\g_2|) \|f\|_{\ell_1^1} \sqrt{\frac{\ep}{\g_1 |\g_2|}}  \|f\|_{\ell^1} \\
    & \leq \frac{P_5(\g_1, |\g_2|) \sqrt{\ep} }{ \g_-^2 \sqrt{\g_1 |\g_2|}} \|f\|_{\ell^1_1} \, .
\end{align*}

\end{proof}

\begin{proof}[Proof of Claim 2]
Recall that $k^2(q) = \g_1^2 + |\g_2|^2 + 2 \g_1 |\g_2| \cos(q)$. Using the $2\pi$-periodicity of the integrand and the fact that cosine is even, it suffices to show that \begin{align*}
    \int \limits_{\g_-}^{\g_+} \int \limits_{0}^{\pi} \frac{\sqrt{\ep}}{ \l| k^2(q) - (\lambda + i \ep)^2 \r|} \, dq  \, d\lambda \leq c(\g_1, |\g_2|) \sqrt{\ep} \, \text{arsinh}\l(\frac{2\g_1|\g_2|}{\ep} \r)
\end{align*}

Using the change of variables $q \mapsto k(q)$ as in \cref{lem: Change of Variables}, we compute
\begin{align*}
    \int \limits_{0}^{\pi} \frac{\sqrt{\ep}}{ \l| k^2(q) - (\lambda + i \ep)^2 \r|} \, dq 
    &= \int \limits_{\g_-}^{\g_+} \frac{\sqrt{\ep}}{ \l| k^2 - (\lambda + i \ep)^2 \r|} \cdot \frac{k}{\g_1 |\g_2| \sqrt{1 - \eta^2(k)}} \, dk  \\
    &\leq \frac{\sqrt{\ep} }{\g_1 |\g_2|} \int \limits_{\g_-}^{\g_+} \frac{1}{ \l| k - (\lambda + i \ep) \r| \sqrt{1 - \eta^2(k)} } \, dk \\
    &= \frac{\sqrt{\ep} }{\g_1 |\g_2|} \int \limits_{\g_-}^{\g_+} \frac{1}{ \sqrt{(k - \lambda)^2 + \ep^2} \sqrt{1 - \eta^2(k)} } \, dk \, . 
\end{align*}

Equation \eqref{eq: 1 - eta^2} says 
\begin{align*}\label{}
    1 - \eta^2(\lambda) = \frac{\l[\g_+^2 - \lambda^2 \r] \l[\lambda^2 - \g_-^2\r] }{(2 \g_1 |\g_2|)^2}\, ,
\end{align*}
so 
\begin{align*}
    \int \limits_{0}^{\pi} \frac{\sqrt{\ep}}{ \l| k^2(q) - (\lambda + i \ep)^2 \r|} \, dq 
    &\leq 2\sqrt{\ep} \int \limits_{\g_-}^{\g_+} \frac{1}{ \sqrt{(k - \lambda)^2 + \ep^2} \sqrt{\g_+^2 - \lambda^2} \sqrt{\lambda^2 - \g_-^2} } \, dk  \\
    \implies \int \limits_{\g_-}^{\g_+} \int \limits_{0}^{\pi} \frac{\sqrt{\ep}}{ \l| k^2(q) - (\lambda + i \ep)^2 \r|} \, dq  \, d\lambda
    &\leq 
    2\sqrt{\ep} \int \limits_{\g_-}^{\g_+} \int \limits_{\g_-}^{\g_+} \frac{1}{ \sqrt{(k - \lambda)^2 + \ep^2} \sqrt{\g_+^2 - \lambda^2} \sqrt{\lambda^2 - \g_-^2} } \, dk  \, d\lambda \\
    &= 2\sqrt{\ep} \int \limits_{\g_-}^{\g_+} \frac{1}{\sqrt{\g_+^2 - \lambda^2} \sqrt{\lambda^2 - \g_-^2}} \int \limits_{\g_-}^{\g_+} \frac{1}{ \sqrt{(k - \lambda)^2 + \ep^2}  } \, dk  \, d\lambda
\end{align*}
To complete the proof, we show that 
\begin{align*}
    \int \limits_{\g_-}^{\g_+} \frac{1}{ \sqrt{(k - \lambda)^2 + \ep^2}  } \, dk
\end{align*}
is bounded independently of $\lambda$. Using the change of variables $z = k- \lambda$ and $x = z/\ep$, we compute
\begin{align*}
    \int \limits_{\g_-}^{\g_+} \frac{1}{ \sqrt{(k - \lambda)^2 + \ep^2}  } \, dk &= \int \limits_{\g_- - \lambda}^{\g_+ - \lambda} \frac{1}{ \sqrt{z^2 + \ep^2}  } \, dz = \frac{1}{\ep} \int \limits_{\g_- - \lambda}^{\g_+ - \lambda} \frac{1}{ \sqrt{(z/\ep)^2 + 1}  } \, dz =  \int \limits_{(\g_- - \lambda)/\ep}^{(\g_+ - \lambda)/\ep} \frac{1}{ \sqrt{x^2 + 1}  } \, dz \\
    &\leq  \int \limits_{-2\g_1 |\g_2|/\ep}^{2\g_1 |\g_2|/\ep} \frac{1}{ \sqrt{x^2 + 1}  } \, dz = 2 \int \limits_{0}^{2\g_1 |\g_2|/\ep} \frac{1}{ \sqrt{x^2 + 1}  } \, dz = 2 \,\text{arsinh}\l(\frac{2\g_1 |\g_2|}{\ep} \r) \, .
\end{align*}

\end{proof}
\end{proof}

\subsection{\textbf{Proof of Lemma} \ref{lem: computing B(ep)}} \label{pf: computing B(ep)} 

\begin{proof}[\textbf{Proof of Lemma} \ref{lem: computing B(ep)}] 
First note that by $2\pi$-periodicity of the integrand, \begin{align*}
    B_n^{\pm}(\ep) &= \frac{e^{i(n+1)\varphi}}{2\lambda} \int \limits_{-\pi}^{\pi}  \l[ \frac{1}{k(q ) - (\lambda \pm i \ep)} - \frac{1}{k(q) + (\lambda \pm i \ep)} \r] e^{i(n+1)q} \, dq\, . 
\end{align*}
Next we partition the region of integration into $[-\pi, 0] \cup [0, \pi]$ and use the change of variables  in \cref{lem: Change of Variables} to compute
\begin{align*}
    B_n^{\pm}(\ep) &= \frac{e^{i(n+1)\varphi}}{2\lambda} \int \limits_{0}^{\pi}  \l[ \frac{1}{k(q ) - (\lambda \pm i \ep)} - \frac{1}{k(q) + (\lambda \pm i \ep)} \r] e^{i(n+1)q} \, dq \\
    &\quad +\frac{e^{i(n+1)\varphi}}{2\lambda} \int \limits_{-\pi}^{0}  \l[ \frac{1}{k(q ) - (\lambda \pm i \ep)} - \frac{1}{k(q) + (\lambda \pm i \ep)} \r] e^{i(n+1)q} \, dq\ \\
    &= \frac{e^{i(n+1)\varphi}}{2\lambda} \int \limits_{\g_-}^{\g_+}  \l[ \frac{1}{k - (\lambda \pm i \ep)} - \frac{1}{k + (\lambda \pm i \ep)} \r]  \frac{ke^{i(n+1) \cos^{-1}(\eta(k))}}{\g_1 |\g_2| \sqrt{1 - \eta^2(k)}} \, dk \\
    &\quad +\frac{e^{i(n+1)\varphi}}{2\lambda} \int \limits_{\g_-}^{\g_+}  \l[ \frac{1}{k - (\lambda \pm i \ep)} - \frac{1}{k + (\lambda \pm i \ep)} \r]  \frac{ke^{-i(n+1)\cos^{-1}(\eta(k))}}{\g_1 |\g_2| \sqrt{1 - \eta^2(k)}}\, dk \\
    &= \frac{e^{i(n+1)\varphi}}{\lambda} \int \limits_{\g_-}^{\g_+}  \l[ \frac{1}{k - (\lambda \pm i \ep)} - \frac{1}{k + (\lambda \pm i \ep)} \r]  \frac{kT_{n+1}(\eta(k)) }{\g_1 |\g_2| \sqrt{1 - \eta^2(k)}} \, dk \, ,
\end{align*} where $T_n$ is the $n$'th Chebyshev polynomial of the first kind, see \cref{def: First Chebyshev Polynomial}. Applying \cref{thm: Sokhotski-Plemelj} to the term $[k - (\lambda \pm i \ep)]^{-1}$ completes the proof.

\end{proof}

\section{Proofs from Section \ref{sec: Representative Oscillatory Integrals}} \label{app: Representative Oscillatory Integrals}

\subsection{\textbf{Proof of Lemma} \ref{lem: k(y) facts}} \label{pf: k(y) facts} 

\begin{proof}[\textbf{Proof of Lemma} \ref{lem: k(y) facts}] 
~\\ 
\indent \textbf{Property (6): }
Since \begin{align*}
    k'''(y) &= - \g_1 |\g_2| \l[ \frac{3\g_1 |\g_2| \cos(y)\sin(y)}{k^3(y)} + \frac{3 \g_1^2 |\g_2|^2 \sin^3(y)}{k^5(y)} -\frac{\sin(y)}{k(y)} \r]\, ,
\end{align*}
$k'''(y) $ only vanishes when $\sin(y) = 0$ or when \begin{equation}\label{kfacts: term 1}
\begin{split}
    k^4(y) &= 3 \g_1 |\g_2| \cos(y) k^2(y) + 3 \g_1^2 |\g_2|^2 \sin^2(y) \\
    &= 3 \g_1 |\g_2| \cos(y)\l[\g_1^2 + |\g_2|^2 + 2 \g_1 |\g_2| \cos(y)\r]+ 3 \g_1^2 |\g_2|^2 \sin^2(y) \\
    &= 3 \l(\g_1^2 + |\g_2|^2 \r) \g_1 |\g_2| \cos(y)  + 6 \g_1^2 |\g_2|^2 \cos^2(y) + 3 \g_1^2 |\g_2|^2 \sin^2(y) \\
    &= 3\g_1^2 |\g_2|^2 \cos^2(y) + 3 \l(\g_1^2 + |\g_2|^2 \r) \g_1 |\g_2| \cos(y) + 3\g_1^2 |\g_2|^2  \,. 
\end{split}
\end{equation}
Note that \begin{equation}\label{kfacts: term 2}
\begin{split}
    k^4(y) &= \l[\g_1^2 + |\g_2|^2 + 2 \g_1 |\g_2| \cos(y)\r]^2 \\
    &= 4 \g_1^2 |\g_2|^2 \cos^2(y) + 4\l( \g_1^2 + |\g_2|^2\r) \g_1 |\g_2| \cos(y) + \g_1^4 + |\g_2|^4 + 2 \g_1^2 |\g_2|^2 \, . 
\end{split}
\end{equation}
Subtracting the last line of \eqref{kfacts: term 1} from the right hand side of \eqref{kfacts: term 2}, we see that $k'''(y)$ only vanishes when  $\sin(y) = 0$ or \begin{align*}
    0 &= \g_1^2 |\g_2|^2 \cos^2(y) + \l( \g_1^2 + |\g_2|^2\r) \g_1 |\g_2| \cos(y) + \g_1^4 + |\g_2|^4 - \g_1^2 |\g_2|^2 \, .
\end{align*}
Since we restrict $k$ to act on $[0, \pi]$, in order to prove the lemma it suffices to show that the polynomial,  \begin{align}\label{kfacts: term 3}
    \g_1^2 |\g_2|^2 z^2 + \l( \g_1^2 + |\g_2|^2\r) \g_1 |\g_2| z + \g_1^4 + |\g_2|^4 - \g_1^2 |\g_2|^2 \, ,
\end{align}
has no zeros in the interval $[-1,1]$. In fact, \eqref{kfacts: term 3} has two imaginary roots. We show this by noting that the discriminant of \eqref{kfacts: term 3} is negative, \begin{align*}
    &\l( \g_1^2 + |\g_2|^2\r)^2 \g_1^2 |\g_2|^2 - 4 \g_1^2 |\g_2|^2 \l( \g_1^4 + |\g_2|^4 - \g_1^2 |\g_2|^2 \r)  \\
    &= \g_1^2 |\g_2|^2 \l[ \l( \g_1^2 + |\g_2|^2\r)^2  - 4  \l( \g_1^4 + |\g_2|^4 - \g_1^2 |\g_2|^2 \r) \r] \\
    &= \g_1^2 |\g_2|^2 \l[  \g_1^4 + |\g_2|^4 + 2 \g_1^2 |\g_2|^2  - 4  \l( \g_1^4 + |\g_2|^4 \r) + 4\g_1^2 |\g_2|^2  \r] \\
    &= \g_1^2 |\g_2|^2 \l[   6 \g_1^2 |\g_2|^2  - 3  \l( \g_1^4 + |\g_2|^4 \r)  \r] \\
    &= -3 \g_1^2 |\g_2|^2 \l[ \g_1^4 + |\g_2|^4 - 2 \g_1^2 |\g_2|^2  \r] \\
    &= -3 \g_1^2 |\g_2|^2 \l( \g_1^2 - |\g_2|^2 \r)^2 < 0 \, . 
\end{align*}

\textbf{Property (7):}
The following inequality holds for all $y \in [0, \pi]$, 
\begin{align*}
    \cos(y) \geq -1 + \frac{2(\pi-y)^2}{\pi^2}\, .
\end{align*}
Hence for all $y \in [0, \pi]$, 
\begin{align*}
    A(y)  &=  \sqrt{\g_1^2 + |\g_2|^2 + 2 \g_1 |\g_2| \cos(y)} - \g_- \\
    &\geq \sqrt{\g_1^2 + |\g_2|^2 + 2 \g_1 |\g_2| \l(-1 + \frac{2(\pi-y)^2}{\pi^2} \r)} - \g_- \\
    &\geq \sqrt{\g_-^2 +  \frac{4 \g_1 |\g_2| (\pi - y)^2}{\pi^2} } - \g_- \,.  \\
\end{align*}
We seek the largest constant $c_A > 0$ such that for all $y \in [0, \pi]$, $\sqrt{\g_-^2 + 4\g_1|\g_2|(\pi-y)^2 /\pi^2} - \g_- \geq c_A (\pi - y)^2$ by computing
\begin{align*}
    \sqrt{\g_-^2 +  \frac{4 \g_1 |\g_2| (\pi - y)^2}{\pi^2} } - \g_-  &\geq c_A (\pi - y)^2 \\
    \iff \sqrt{\g_-^2 +  \frac{4 \g_1 |\g_2| (\pi - y)^2}{\pi^2} } &\geq c_A (\pi - y)^2 + \g_- \\
    \iff \g_-^2 +  \frac{4 \g_1 |\g_2| (\pi - y)^2}{\pi^2} &\geq c_A^2 (\pi - y)^4 + \g_-^2 + 2 \g_- c_A ( \pi -y)^2 \\
    \iff \frac{4 \g_1 |\g_2|}{\pi^2} &\geq c_A^2 (\pi -y)^2 + 2 \g_- c_A \, .
\end{align*}
In order for this last inequality to hold for all $y \in [0, \pi]$, it suffices to check that it holds for $y =0$. Doing so, we require
\begin{align*}
    0 \geq \pi^2 c_A^2  + 2 \g_- c_A  - \frac{4 \g_1 |\g_2|}{\pi^2}\, . 
\end{align*}
Setting this inequality to be an equality and solving for $c_A > 0$, we obtain \begin{align*}
    c_A = \frac{-\g_- + \sqrt{\g_-^2 + 4\g_1 |\g_2}}{ \pi^2} = \frac{\g_+ - \g_-}{\pi^2} = \frac{2 \min\{ \g_1, |\g_2|\}}{\pi^2}\, .
\end{align*}

\textbf{Property (8):}
\begin{rmk}
An argument analogous to the proof of Property (7) would show that \begin{align*}
     B(y) \geq  \frac{2\g_1 |\g_2| y^2}{\pi^2 \g_+} \, . 
\end{align*}
However, we can obtain a better (larger) constant through the following method instead.
\end{rmk}
We seek a constant $c_B > 0$ such that $B(y) \geq cy^2$ for all $y \in [0, \pi]$, which is equivalent to 
\begin{align*}
    \g_+ - c_B y^2 &\geq k(y) \\
    \iff \g_+^2 + c_B^2 y^4 - 2 c_B \g_+ y^2 &\geq \g_1^2 + |\g_2|^2 + 2 \g_1 |\g_2| \cos(y) \\
    \iff c_B^2 y^4  - 2 c_B \l(\g_1 + |\g_2| \r) y^2 + 2 \g_1 |\g_2| (1 - \cos(y)) &\geq 0 \, . 
\end{align*}
We now check that \begin{align*}
    c_B = \frac{2 \min\{ \g_1, |\g_2|\}}{\pi^2} 
\end{align*}
satisfies the last inequality. 
Without loss of generality, suppose that $|\g_2| > \g_1$. Then proving the last inequality is equivalent to showing that $f(y; \g_1, |\g_2|) \geq 0$, where 
\begin{align*}
    f(y; \g_1, |\g_2|) &:= \frac{4 \g_1^2}{\pi^4} y^4 - \frac{4 \g_1}{\pi^2} \l(\g_1 + |\g_2| \r)y^2 + 2 \g_1 |\g_2| y^2 + 2 \g_1 |\g_2| (1 - \cos(y)) \\
    &= \frac{4 \g_1^2}{\pi^4} y^4 - \frac{4 \g_1^2 y^2}{\pi^2} + 2 \g_1 |\g_2| \l[1 - \cos(y) - \frac{2y^2}{\pi^2} \r] \, . 
\end{align*}
Since the term in brackets is nonnegative for all $y \in [0,\pi]$, it suffices to prove that $f(y; \g_1, \g_1) \geq 0$. We compute \begin{align*}
    f(y; \g_1, \g_1) &= \frac{4 \g_1^2}{\pi^4} y^4 - \frac{4 \g_1^2 y^2}{\pi^2} + 2 \g_1^2 \l[1 - \cos(y) - \frac{2y^2}{\pi^2} \r] \\
    &= 2 \g_1^2 \l[ \frac{2y^4}{\pi^4} - \frac{4 y^2}{\pi^2} +1 - \cos(y) \r] \, . 
\end{align*} 
One can then check that the term in brackets is nonnegative for all $y \in [0,\pi]$, completing the proof.

\end{proof}

\subsection{\textbf{Proof of Proposition} \ref{prop: application of VDC}} \label{pf: application of VDC} 

Van der Corput's lemma, see \cref{thm: VDC}, is the main tool used to obtain the decay rate estimate in \cref{prop: application of VDC}. Since $k^{(\ell)}(y)$ has zeros in the interval $[0, \pi]$ for $\ell = 1,2,$ and 3, we partition the interval into regions where either the first or the second derivative (resp. second or third derivative) is bounded away from zero to obtain a $\langle t \rangle^{-1/2}$ decay rate (resp. $\langle t \rangle^{-1/3}$ decay rate) over the entire interval.

The constant in the decay rate given by Van der Corput's Lemma depends on the magnitude of $\inf |k^{(\ell)}(y)|$. Thus to obtain a good
constant, it is natural to choose the points at which $|k'(y)|$ and $|k''(y)|$  (resp. $|k''(y)|$ and $|k'''(y)|$)  intersect as our partition points.  In Lemmas \ref{lem: k 1 and 2 lower bound} and \ref{lem: k 2 and 3 lower bound}, we compute lower bounds on $|k''(y)|$ at these intersection points. This provides a lower bound on either $|k'(y)|$ or $|k''(y)|$  (resp. $|k''(y)|$ or $|k'''(y)|$) over the entire interval $[0, \pi]$. With this lower bound in hand, we apply Van der Corput's Lemma to obtain a decay rate estimate with a constant that has an explicit dependence on the hopping coefficients.

Recall that $k(y) := \sqrt{\g_1^2 + |\g_2|^2 + 2 \g_1 |\g_2| \cos(y)}$. It is useful to make the following change of basis in the parameter space 
\begin{align*}
    \{\g_1, |\g_2| \} \mapsto \l\{ \g_1 |\g_2|, G:= \frac{|\g_2|}{\g_1} + \frac{\g_1}{|\g_2|} \r\}\, .
\end{align*}
Note that $G \in (2, \infty)$ and $k(y) = \sqrt{\g_1 |\g_2|} \sqrt{G + 2 \cos(y)}$.

We can view $\g_1 |\g_2|$ as the overall size of the hopping coefficients and $G$ as a measure of the gap between them.  
Part of the motivation for using this change of basis is that in the expression for $k(y)$, the parameters $\g_1$ and $|\g_2|$ are interchangeable, and so it makes sense to pick a basis where this symmetry persists. A more concrete reason is that the first four derivatives of $k$, see \cref{lem: k(y) facts} and \eqref{eq: k derivatives G notation}, can all be written as a product of $\sqrt{\g_1 |\g_2| (G + 2 \cos(y))^{-1}}$ and a term that only depends on $y$ and $G$. This property will make it easier to study the points at which these derivatives intersect or vanish. We compute the first four derivatives of $k$ in the basis $\{\g_1 |\g_2|, G\}$ below. 
\begin{equation}\label{eq: k derivatives G notation} \begin{split}
    k'(y) &=  - \sqrt{\frac{\g_1 |\g_2|}{G + 2 \cos(y)}} \l[\sin(y) \r] \, ,\\
    k''(y) &= - \sqrt{\frac{\g_1 |\g_2|}{G + 2 \cos(y)}} \l[ \cos(y) + \frac{\sin^2(y)}{G + 2 \cos(y)} \r] \, , \\
    k'''(y)
    &=  - \sqrt{\frac{\g_1 |\g_2|}{G + 2 \cos(y)}} \l[ \frac{3 \cos(y) \sin(y)}{G + 2 \cos(y)} + \frac{3 \sin^3(y)}{[G + 2 \cos(y)]^2}  - \sin(y)\r] \, ,  \\
    k^{(4)}(y) &= - \sqrt{\frac{\g_1 |\g_2|}{G + 2 \cos(y)}} \l[ \frac{15 \sin^4(y)}{[G+2\cos(y)]^3} + \frac{18 \sin^2(y) \cos(y)}{[G+2\cos(y)]^2}+ \frac{3 \cos^2(y) - 4 \sin^2(y)}{G + 2 \cos(y)}  - \cos(y)\r]\, .
\end{split}\end{equation}

\begin{lem}\label{lem: g1 g2 / G}
For all $y \in [0, \pi]$, 
\begin{align*}
    \sqrt{ \frac{\g_1 |\g_2|}{G+2 \cos(y)}} \geq \frac{\min(\g_1, |\g_2|)}{2}\,. 
\end{align*}

\begin{proof}
Suppose that $\g_1 < |\g_2|$ and write $|\g_2| = \alpha \g_1$ with $\alpha > 1$.
Then \begin{align*}
    \frac{\g_1 |\g_2|}{G + 2 \cos(y)} \geq  \frac{\g_1 |\g_2|}{G + 2 } = \frac{\g_1 |\g_2|}{\frac{|\g_2|}{\g_1} + \frac{\g_1}{|\g_2|} + 2  } = \frac{\alpha \g_1^2}{\alpha + \alpha^{-1} + 2 } = \frac{\g_1^2}{1+ \alpha^{-2} +\frac{2 }{\alpha}} \geq \frac{\g_1^2}{4} \,.
\end{align*}
Allowing for $\g_1 < |\g_2|$ or  $\g_1 > |\g_2|$, we obtain 
\begin{align*}
    \frac{\g_1 |\g_2|}{G + 2 \cos(y)} \geq \frac{\min\{\g_1^2, |\g_2|^2\}}{4} \,.
\end{align*}

\end{proof}
\end{lem}

\begin{lem}\label{lem: k 1 and 2 lower bound}
Suppose $\g_1, |\g_2| > 0$ and $\g_1 \neq |\g_2|$. There exist exactly two points in the interval $[0, \pi]$ for which $|k'(y)| = |k''(y)|$, which we call $y_L$ and $y_R$. These points satisfy
\begin{align*}
    0 < y_L(G) < y_M <  y_R(G) < \pi \quad
    \text{ where } G:= \frac{|\g_2|}{\g_1} + \frac{\g_1}{|\g_2|}
\end{align*}
and $y_M$ is the zero of $k''(y)$,  given explicitly in \cref{lem: k(y) facts}. 
At these points, there exist positive constants $c_L$ and $c_R$, independent of $\g_1$ and $\g_2$, for which \begin{align*}
     \quad k'(y_L) = k''(y_L) \geq c_L \min\{ \g_1, |\g_2|\} \quad \text{ and } \quad -k'(y_R) = k''(y_R) \geq c_R \min\{ \g_1, |\g_2|\} \, . 
\end{align*}
\begin{proof}
Using the expressions for the derivatives of $k$ given in \eqref{eq: k derivatives G notation}, we compute

 \begin{align*}
     &k'(y) = k''(y) \\
     \iff &\sin(y) = \cos(y) + \frac{\sin^2(y)}{G + 2 \cos(y)} \\
     \iff & G [\sin(y) - \cos(y)]  + 2 \cos(y) \sin(y) - 2 \cos^2(y) = \sin^2(y) \\
     \iff & G =  \frac{1 + \cos^2(y) - \sin(2y)}{\sin(y) - \cos(y)} \numberthis \quad \quad [\text{equation for } y_L(G)] \label{eq: k' k'' yL} 
 \end{align*}
 and
 \begin{align*}
    &-k'(y) = k''(y) \\
    \iff &-\sin(y) = \cos(y) + \frac{\sin^2(y)}{G + 2 \cos(y)} \\
    \iff & \cos(y) + \sin(y) = -\frac{\sin^2(y)}{G + 2 \cos(y)} \\
    \iff &G [\cos(y) + \sin(y)] + 2 \cos^2(y) + 2 \cos(y) \sin(y) = - \sin^2(y) \\
    \iff &  G = - \frac{1 + \cos^2(y) + \sin(2y)}{\sin(y) + \cos(y)} \numberthis \quad \quad [\text{equation for } y_R(G)] \, . \label{eq: k' k'' yR}
\end{align*}
The following facts about $y_L(G)$ and $y_R(G)$ can immediately be seen from a plot of \eqref{eq: k' k'' yL} and \eqref{eq: k' k'' yR}: \begin{itemize}
    \item $y_L(G)$ and $y_R(G)$ exist and are unique,
    \item $ \frac{\pi}{4} < y_L(G) < \frac{\pi}{2} $,
    \item $\frac{3\pi}{4} < y_R(G) < \pi$,
    \item $y_R(G)$ is monotonically increasing,
    \item $\lim_{G \rightarrow 2} y_R(G) = \pi$,
    \item $\lim_{G \rightarrow \infty} y_R(G) = \frac{3\pi}{4}$.
\end{itemize} 
Using \cref{lem: g1 g2 / G} and the fact that $\sin(y_L(G))$ is bounded away from zero, we compute 
\begin{align*}
    |k'(y_L(G))| \geq  \frac{\sin(y_L(G)) \min\{\g_1, |\g_2|\}}{2} \geq \frac{\sqrt{2} \min\{\g_1, |\g_2|\} }{4}\, .
\end{align*}
A Taylor expansion of \eqref{eq: k' k'' yR} near $y = \pi$, which corresponds to $G$ near 2, yields \begin{align*}
    - \frac{1 + \cos^2(y) + \sin(2y)}{\sin(y) + \cos(y)} \approx 2 + (\pi - y)^3 \, . 
\end{align*}
In fact, $2 + (\pi - y)^3$ is a lower bound of the left hand side over the interval $y \in (3\pi/4, \pi)$.

Let $\ep_y \in (0, \pi/4)$ be the solution\footnote{Numerically, one can check that $\ep_y \approx 0.4$.} to
\begin{align*}
    - \frac{1 + \cos^2(y) + \sin(2y)}{\sin(y) + \cos(y)} = 2 + 2(\pi - y)^3 \, . 
\end{align*}
Then for $y \in (\pi - \ep_y, \pi)$, we have 
\begin{align*}
    - \frac{1 + \cos^2(y) + \sin(2y)}{\sin(y) + \cos(y)} \leq 2 + 2(\pi - y)^3 \, .
\end{align*}
Define $\ep_G > 0$ so that\footnote{$\ep_G \approx 0.13$.}
\begin{align*}
    y_R(2 + \ep_G) = \pi - \ep_y\, .
\end{align*}
Then for $G \in (2, 2 + \ep_G)$,
\begin{align*}
   G \leq 2 + 2(\pi -y_R(G))^3 \, . 
\end{align*}
Using this fact, we compute
\begin{align*}
    |k'(y_R(G))|  &=  \frac{\sqrt{\g_1 |\g_2|} \sin(y_R(G))}{\sqrt{G + 2 \cos(y_R(G))}}  \geq  \frac{\sqrt{\g_1 |\g_2|} \sin(y_R(G))}{\sqrt{2 + 2(\pi - y_R(G))^3 + 2 \cos(y_R(G))}} 
\end{align*}
Since $y \mapsto \sin(y) [ 2+ 2(\pi - y)^3 + 2 \cos(y)]^{-1/2} $ is monotonically increasing for $y \in (3\pi/4, \pi)$, we have
\begin{align*}
    |k'(y_R(G))| \geq  \frac{\sqrt{\g_1 |\g_2|} \sin(3\pi/4)}{ \sqrt{2 + 2(\pi/4)^3 + 2 \cos(3\pi/4)}}  \geq 0.567 \sqrt{\g_1 |\g_2|}  \geq 0.567 \min\{ \g_1, |\g_2|\} \, .
\end{align*}

The estimate for $G \in [2 + \ep_G, \infty)$ is much simpler. First note that when $G \in [2 + \ep_G, \infty)$, we have $y_R(G) \in (3\pi/4,  \pi-\ep_y]$.  Then by \cref{lem: g1 g2 / G}, we compute \begin{align*}
    |k'(y_R(G))| \geq \frac{\sin(y_R(G)) \min\{ \g_1, |\g_2|\}}{2} \geq \min\{ \g_1, |\g_2|\} \frac{\sin(\pi - \ep_y) }{2} \, .
\end{align*}
Lastly, by \cref{lem: k(y) facts}, we know $k'(y) < 0$ for $y \in (0, \pi)$, $k''(y) < 0$ for $y < y_M$, and $k''(y) > 0$ for $y > y_M$. Thus \begin{align*}
    y_L(G) < y_M < y_R(G)\,. 
\end{align*}

\end{proof}
\end{lem}

\begin{lem}\label{lem: k 2 and 3 lower bound}
Suppose $\g_1, |\g_2| > 0$ and $\g_1 \neq |\g_2|$. There exist exactly two points in the interval $[0, \pi]$ for which $|k''(y)| = |k'''(y)|$, which we call $y_L$ and $y_R$. These points satisfy
\begin{align*}
    &0 < y_L(G) < y_M < y_R(G) < \pi \quad
    \text{ where } \quad G:= \frac{|\g_2|}{\g_1} + \frac{\g_1}{|\g_2|}\, 
\end{align*}
and $y_M$ is the zero of $k''(y)$,  given explicitly in \cref{lem: k(y) facts}. 
At these points, there exist positive constants $c_L$ and $c_R$, independent of the hopping coefficients, for which \begin{align*}
     \quad -k''(y_L) = k'''(y_L) \geq c_L \min\{ \g_1, |\g_2|\} \quad \text{ and } \quad k''(y_R) = k'''(y_R) \geq c_R \min\{ \g_1, |\g_2|\} \, . 
\end{align*}
\begin{proof}
Using the expressions for the derivatives of $k$ in \eqref{eq: k derivatives G notation}, we compute  
\begin{align*}
    &-k''(y) = k'''(y) \\
    \iff & - \cos(y) - \frac{\sin^2(y)}{G + 2 \cos(y)} = \frac{3 \cos(y) \sin(y)}{G + 2 \cos(y)} + \frac{3 \sin^3(y)}{[G + 2 \cos(y)]^2}  - \sin(y) \\
    \iff & \sin(y) -  \cos(y) - \frac{\sin^2(y)  + 3 \cos(y) \sin(y) }{G  + 2 \cos(y)} - \frac{3 \sin^3(y)}{[G + 2 \cos(y)]^2} = 0 \\
    \iff & (\sin(y) -  \cos(y)) [G + 2 \cos(y)]^2 - \sin(y) ( \sin(y) + 3 \cos(y)) [G + 2 \cos(y)] - 3 \sin^3(y) =0 \, .
\end{align*}
Using the quadratic formula and some trigonometry, we compute 
\begin{align*}
    G + 2 \cos(y) &= \frac{\sin(y) ( \sin(y) + 3 \cos(y)) \pm \sqrt{\sin^2(y) (\sin(y) + 3 \cos(y))^2 +12 \sin^3(y) ( \sin(y) - \cos(y))}}{2 (\sin(y) - \cos(y))} \\
    G &= \sin(y) \frac{\sin(y) + 3 \cos(y) \pm \sqrt{9+4 \sin^2(y) - 3 \sin(2y)}}{2 (\sin(y) - \cos(y))} - 2 \cos(y)\, .
\end{align*}
Since $G\in (2, \infty)$, in order for the above equation to have any solutions for $y \in (0, \pi)$, we pick the plus sign on the right hand side. A plot then shows that for all $G\in 2, \infty)$, there exists a unique solution to
\begin{align}
    G = \sin(y)\frac{ \sin(y) + 3\cos(y) + \sqrt{9 + 4 \sin^2(y) - 3 \sin(2y) } }{2 (\sin(y) - \cos(y))} - 2 \cos(y), \quad y \in (0, \pi)\label{eq: k'' k''' yL}
\end{align}
for $y$ in a subset of $(\pi/4, 3\pi/4)$.
We define this solution to be $y_L(G)$. Note that by \cref{lem: g1 g2 / G}, we have  \begin{align*}
    |k''(y)| &=  \sqrt{\frac{\g_1 |\g_2|}{G + 2 \cos(y)}} \l| \cos(y) + \frac{\sin^2(y)}{G + 2 \cos(y)} \r| \\
    &\geq \frac{\min\{ \g_1, |\g_2|\}}{2} \l| \cos(y) + \frac{\sin^2(y)}{G + 2 \cos(y)} \r|\, . 
\end{align*}
Thus 
\begin{align*}
    &|k''(y_L(G))| \geq \frac{\min\{ \g_1, |\g_2|\}}{2} \l| \cos(y_L(G)) + \frac{\sin^2(y_L(G))}{G + 2 \cos(y_L(G))} \r| \\
    &= \frac{\min\{ \g_1, |\g_2|\}}{2}  \l| \cos(y_L(G)) + \frac{2 \sin(y_L(G)) ( \sin(y_L(G)) - \cos(y_L(G)))}{\sin(y_L(G)) + 3 \cos(y_L(G)) + \sqrt{9+4 \sin^2(y_L(G)) - 3 \sin(2y_L(G)) }}  \r| \, .
\end{align*} 
Note that over the interval $y \in [0, 3\pi/4]$, 
\begin{align*}
    y \mapsto \cos(y) + \frac{2 \sin(y) (\sin(y) - \cos(y))}{ \sin(y) + 3 \cos(y) + \sqrt{9 + 4 \sin^2(y) - 4 \sin(2y)}}
\end{align*}
is nonnegative and monotonically decreasing. Since $y_L(G) < 3\pi/4$, 
this implies that $|k''(y_L(G))| \geq c_L \min\{ \g_1, |\g_2|\}$ for some constant $c_L  >0$.

Using \eqref{eq: k derivatives G notation} again, we compute
\begin{align*}
    &\qquad k''(y) = k'''(y) \\
    &\iff  \cos(y) + \frac{\sin^2(y)}{G + 2 \cos(y)} = \frac{3 \cos(y) \sin(y)}{G + 2 \cos(y)} + \frac{3 \sin^3(y)}{[G + 2 \cos(y)]^2}  - \sin(y)   \\
    &\iff  \sin(y) + \cos(y) + \frac{ \sin^2(y) - 3 \cos(y) \sin(y) }{G + 2 \cos(y)}  - \frac{3 \sin^3(y)}{[G + 2 \cos(y)]^2}  = 0 \\
    &\iff (\sin(y) + \cos(y) ) [G + 2 \cos(y)]^2 + \sin(y) ( \sin(y) - 3 \cos(y)) [G + 2 \cos(y)] - 3  \sin^3(y) = 0  \,. 
\end{align*}
Using the quadratic formula and some trigonometry, we compute \begin{align*}
    G + 2 \cos(y) &= \frac{\sin(y) (3 \cos(y) - \sin(y)) \pm \sqrt{ \sin^2(y)  ( \sin(y) - 3 \cos(y))^2 + 12 \sin^3(y)(\sin(y) + \cos(y) )  }}{2 (\sin(y) + \cos(y))} \\
    G &= \sin(y)\frac{3 \cos(y) - \sin(y) \pm \sqrt{9 + 4 \sin^2(y) + 3 \sin(2y) } }{2 (\sin(y) + \cos(y))} - 2 \cos(y)\, . 
\end{align*}
Since $G \in (2, \infty)$, this equation only has a solution for $y \in (3\pi/4, \pi)$ when we consider the minus sign on the right hand side. A plot of the right hand side shows that there is a unique solution $y$ of 
\begin{align}
    G = \sin(y)\frac{3 \cos(y) - \sin(y) - \sqrt{9 + 4 \sin^2(y) + 3 \sin(2y) } }{2 (\sin(y) + \cos(y))} - 2 \cos(y), \quad y \in (0, \pi) \label{eq: k'' k''' yR}
\end{align}
for all $G \in (2, \infty)$. We define this solution to be $y_R(G)$. The following facts about $y_R(G)$ are immediate from a plot of \eqref{eq: k'' k''' yR}: \begin{itemize}
    \item $\frac{3\pi}{4} < y_R(G) < \pi$,
    \item $y_R(G)$ is monotonically increasing,
    \item $\lim_{G \rightarrow 2} y_R(G) = \pi$,
    \item $\lim_{G \rightarrow \infty} y_R(G) = \frac{3\pi}{4}$.
\end{itemize}
A Taylor expansion of the right hand side of \eqref{eq: k'' k''' yR} near $y = \pi$ yields 
\begin{align*}
    \sin(y)\frac{3 \cos(y) - \sin(y) - \sqrt{9 + 4 \sin^2(y) + 3 \sin(2y) } }{2 (\sin(y) + \cos(y))} - 2 \cos(y)  \approx 2 + 3(\pi -  y) \,. 
\end{align*} 
In fact, $2 + 3 (\pi -y)$ is a lower bound for all $y \in (3\pi/4, \pi)$. Hence 
\begin{align*}
    G  \geq 2 + 3 (\pi - y_R(G))\, . 
\end{align*}
Since $k'''(y) \geq 0 $ by \cref{lem: k(y) facts}, we have $k''(y_R(G)) = k'''(y_R(G)) \geq 0$. Then the expression for $k''(y)$ in \eqref{eq: k derivatives G notation} implies that \begin{align*}
    \cos(y_R(G)) + \frac{\sin^2(y_R(G))}{G + 2 \cos(y_R(G))} \leq 0\, .
\end{align*}
The lower bound on $G$ implies
\begin{align*}
      \cos(y_R(G)) + \frac{\sin^2(y_R(G))}{G + 2 \cos(y_R(G))} \leq \cos(y_R(G)) + \frac{\sin^2(y_R(G))}{2 + 3 (\pi - y_R(G)) + 2 \cos(y_R(G))} \, . 
\end{align*}
The function $y \mapsto \cos(y) + \frac{\sin^2(y)}{2 + 3(\pi - y) + 2 \cos(y)}$ attains its supremum over $[3\pi/4, \pi]$ at $y = 3\pi/4$. Hence 
\begin{align*}
    \cos(y_R(G)) + \frac{\sin^2(y_R(G))}{G + 2 \cos(y_R(G))} \leq \cos(3\pi/4) + \frac{\sin^2(3\pi/4)}{2 + 3\pi/4 + 2 \cos(3\pi/4)} \approx -0.537\, . 
\end{align*}
Combining the above estimate with \cref{lem: g1 g2 / G}, we obtain  \begin{align*}
    |k''(y_R(G))| = \sqrt{\frac{\g_1 |\g_2|}{G + 2 \cos(y)}} \l| \cos(y) + \frac{\sin^2(y)}{G + 2 \cos(y)} \r| \geq \frac{\min\{\g_1, |\g_2| \}}{4}\, .
\end{align*}
Lastly, by \cref{lem: k(y) facts}, we know $k'''(y) \geq 0$, $k''(y) < 0$ for $y < y_M$, and $k''(y) > 0$ for $y > y_M$. Thus \begin{align*}
    y_L(G) < y_M < y_R(G)\,. 
\end{align*}

\end{proof}
\end{lem}

Recall that $k'''(y) \geq 0$ by \cref{lem: k(y) facts}. 
\begin{lem}\label{lem: inf of k'''}
For any $0\leq a < b \leq \pi$, we have 
\begin{align*}
    \inf_{y \in [a,b]} k'''(y)  = \min\{k'''(a), k'''(b)   \}\, . 
\end{align*}

\begin{proof}
It suffices to show that $k^{(4)}(y)$ has exactly one zero in $[0, \pi]$. To see why this suffices, suppose that $k^{(4)}(y)$ goes from positive to negative as $y$ increases. Then the infimum must be attained either before $k'''(y)$ begins increasing (i.e., at $k'''(a)$) or after it has decreased as much as possible (i.e., at $k'''(b))$. 

Setting $k^{(4)}(y) = 0 $ and using the expression for $k^{(4)}(y)$ in \eqref{eq: k derivatives G notation}, it suffices to count the number of solutions $y \in [0, \pi]$ to 
\begin{align*}
    \frac{15 \sin^4(y)}{[G+2\cos(y)]^3} + \frac{18 \sin^2(y) \cos(y)}{[G+2\cos(y)]^2}+ \frac{3 \cos^2(y) - 4 \sin^2(y)}{G + 2 \cos(y)}  - \cos(y) = 0 \,.  
\end{align*}
Since $\cos : [0, \pi] \rightarrow [-1,1]$ is a bijection, we can multiply the above equation by $[G+2\cos(y)]^3$, set $c = \cos(y)$, and count the number of solutions $c \in [-1,1]$ to \begin{align*}
    15 (1-c^2)^2 + 18 (1-c^2) c [G+ 2c] + (7c^2 -4) [G+2c]^2 - c [G+2c]^3 = 0\, . 
\end{align*}
The left hand side is equal to $-P(c)$ where \begin{align*}
    P(c) := c^4 + 2G c^3 + (10-G^2) c^2 + (G^3 - 2G) c + 4G^2 - 15\, . 
\end{align*}
Note that $P(-1) = -G^3 + 3G^2 -4 < 0$ and $P(1) = G^3 + 3G^2 -4$ for all $G > 2$. So showing that $P(c)$ is strictly increasing will prove that $k^{(4)}(y)$ has exactly one zero. We compute 
\begin{align*}
    P'(c) &= 4c^3 + 6Gc^2 + 2(10-G^2) c + G^3 - 2G \\
    &= 4(c+1)^3 + (G-2) \l[6 \l(c - \frac{G+2}{6} \r)^2 + \frac{5G^2 + 8G + 8}{6} \r] \\
    &\geq (G-2) \frac{5G^2 + 8G + 8}{6}  > 0 \, . 
\end{align*}

\end{proof}
\end{lem}

\begin{proof}[\textbf{Proof of Proposition} \ref{prop: application of VDC}]
Since $\l| \int  \limits_0^{z} e^{-i k(y) t} e^{iny} g(y) \, dy  \r|$ is continuous in $z$, its supremum must be attained at some $z_* \in [0,\pi]$. We will go over the proof in the case where $z_* = \pi$, but the same argument works with minimal adjustments for $z_* < \pi$. 

\textbf{$t^{-1/3}$ decay rate:}\newline
Let $0 < y_L(G)< y_R(G) < \pi$ be the points at which $|k''(y)|$ and $|k'''(y)|$ intersect, see \cref{lem: k 2 and 3 lower bound}, and partition 
\begin{align*}
    [0, \pi] = \l[0, y_L(G)\r] \cup \l[y_L(G), y_R(G)\r] \cup \l[ y_R(G), \pi\r] \, . 
\end{align*}
Define the phase function $\phi(y) := k(y) - ny/t$ and note that \begin{align*}
    \phi'(y) &= k'(y) - n/t \\
    \phi''(y) &= k''(y)\\
    \phi'''(y) &= k'''(y) \, .
\end{align*}
Then by the triangle inequality, 
\begin{align*}
    \l| \int  \limits_0^{\pi} e^{-i k(y) t} e^{iny} g(y) \, dy  \r| = \l| \int  \limits_0^{\pi} e^{-i \phi(y) t} g(y) \, dy  \r| \qquad \qquad \qquad \qquad \qquad \qquad \qquad \qquad  \\
    \leq \l| \int  \limits_0^{y_L(G)}  e^{-i \phi(y) t} g(y) \, dy  \r| + \l|\; \int  \limits_{y_L(G)}^{y_R(G)}  e^{-i \phi(y) t} g(y) \, dy  \r| + \l|\; \int  \limits_{y_R(G)}^{\pi}  e^{-i \phi(y) t} g(y) \, dy  \r|\, . 
\end{align*}

We proceed to estimate these three integrals, starting with the first one. By \cref{lem: k(y) facts}, $k'''(y) \geq 0$, so $k''(y)$ is monotonically increasing and by \cref{lem: k 2 and 3 lower bound}, $y_L(G) < y_M$, where $y_M$ is the zero of $k''(y)$.  So $|k''(y)|$ is monotonically decreasing and attains its minimum over $[0, y_L(G)]$ at $y_L(G)$. By \cref{lem: k 2 and 3 lower bound} again, we have \begin{align*}
    \inf_{y \in [0, y_L(G)]} |k''(y)| = |k''(y_L(G))| \geq c_L \min\{ \g_1, |\g_2|\}\, .
\end{align*}

Applying Van der Corput's Lemma, see \cref{thm: VDC}, with $\ell = 2$, we compute 
\begin{align*}
    \l| \int  \limits_0^{y_L(G)} e^{-i \phi(y) t} g(y) \, dy  \r| &\lesssim \frac{t^{-1/2}}{\sqrt{\min\{\g_1, |\g_2|\}}} \l[ |g(y_L(G))| + \int \limits_0^{y_L(G)} \l|g'(y) \r| dy \r] \\
    &\leq \frac{ t^{-1/2} \|g\|_{W^{1,1}}}{\sqrt{\min\{\g_1, |\g_2|\}}}\, .
\end{align*}
In similar fashion, we obtain the same bound on the integral over $[y_R(G), \pi]$. For the  integral over $[y_L(G), y_R(G)]$, we use  \cref{lem: k 2 and 3 lower bound,lem: inf of k'''} to compute 
\begin{align*}
    \inf_{y\in[y_L(G), y_R(G)]} \phi'''(y)  &= \inf_{y\in[y_L(G), y_R(G)]} k'''(y) \\
    &= \min \{ k'''(y_L(G)), k'''(y_R(G)) \} \\
    &\geq \min\{c_L, c_R\} \min\{ \g_1, |\g_2|\}\, . 
\end{align*} 
Thus applying Van der Corput's Lemma with $\ell =3$, yields
\begin{align*}
    \l|\; \int  \limits_{y_L(G)}^{y_R(G)} e^{-i \phi(y) t} g(y) \, dy  \r| &\lesssim \frac{ t^{-1/3}}{\l(\min\{\g_1, |\g_2|\}\r)^{1/3} } \l[ |g(y_L(G))| + \int \limits_{y_L(G)}^{y_R(G)} \l|g'(y) \r| dy \r] \\
    &\leq \frac{ t^{-1/3} \|g\|_{W^{1,1}}}{\l(\min\{\g_1, |\g_2|\}\r)^{1/3}}\, .
\end{align*}
Putting these estimates together, we get 
\begin{align*}
    \l| \int  \limits_0^{\pi} e^{-i k(y) t} e^{iny} g(y) \, dy  \r| &\lesssim  \frac{ t^{-1/2} \|g\|_{W^{1,1}}}{\sqrt{\min\{\g_1, |\g_2|\}}} + 
    \frac{ t^{-1/3} \|g\|_{W^{1,1}}}{\l(\min\{\g_1, |\g_2|\}\r)^{1/3}} \\
    &\lesssim \l(1+ \min\{ \g_1, |\g_2|\}^{-1/2}\r) \langle t \rangle ^{-1/3} \|g\|_{W^{1,1}}
\end{align*}

\textbf{$t^{-1/2}$ decay rate:}\newline
The strategy is largely the same except for two differences. First, we let $0 < y_L(G)< y_R(G) < \pi$ be the points at which $|k'(y)|$ and $|k''(y)|$ intersect, see \cref{lem: k 1 and 2 lower bound}. Second, in the middle integral, 
\begin{align*}
    \l| \; \int  \limits_{y_L(G)}^{y_R(G)} e^{-i k(y) t} e^{iny} g(y) \, dy  \r|
\end{align*}
we take the phase function to be $k(y)$ and the amplitude to be $e^{iny} g(y)$. 
Since $k''(y)$ has exactly one zero in $[0, \pi]$ , by \cref{lem: k 1 and 2 lower bound} and the same reasoning as in \cref{lem: inf of k'''}, we compute \begin{align*}
    \inf_{y\in[y_L(G), y_R(G)]} |k'(y)| = \min \{ |k'(y_L(G))|, |k'(y_R(G))| \} \geq \min\{c_L, c_R\} \min\{ \g_1, |\g_2|\}\, .
\end{align*}
Integrating by parts\footnote{Equivalently, applying Van der Corput's lemma with $\ell = 1$.}, we compute 
\begin{align*}
    \l|\; \int  \limits_{y_L(G)}^{y_R(G)} e^{-i k(y) t} e^{iny} g(y) \, dy  \r|  &\lesssim \frac{ t^{-1}}{\min\{\g_1, |\g_2|\} } \l[| g(y_L(G)) | + \int  \limits_{y_L(G)}^{y_R(G)} \l| \frac{d}{dy} e^{iny} g(y) \r|  \, dy \r] \\
    &= \frac{ t^{-1}}{\min\{\g_1, |\g_2|\} } \l[| g(y_L(G)) | + \int  \limits_{y_L(G)}^{y_R(G)} \l| in e^{iny} g(y) + e^{iny}  g'(y) \r|  \, dy \r] \\
    &\lesssim t^{-1} \frac{(1+n)\|g\|_{L^{1}} + \|g'\|_{L^1}}{\min\{\g_1, |\g_2|\} } \, .
 \end{align*}
Combining this with the estimate
\begin{align*}
    \l| \int  \limits_0^{y_L(G)} e^{-i \phi(y) t} g(y) \, dy  \r| + \l|\; \int  \limits_{y_L(G)}^{\pi} e^{-i \phi(y) t} g(y) \, dy  \r|
    \lesssim \frac{ t^{-1/2} \|g\|_{W^{1,1}}}{\sqrt{\min\{\g_1, |\g_2|\}}}
\end{align*}
yields
\begin{align*}
    \l| \int  \limits_{0}^{\pi} e^{-i k(y) t} e^{iny} g(y) \, dy  \r| \lesssim \l(1+\min\{\g_1, |\g_2|\}^{-1}\r) \l[\langle t\rangle^{-1/2}  + n\langle t \rangle^{-1} \r] \|g\|_{W^{1,1}} \,. 
\end{align*}

\end{proof}

\subsection{\textbf{Proof of Proposition} \ref{prop: properties of F}} \label{pf: properties of F} 

\begin{proof}[\textbf{Proof of Proposition} \ref{prop: properties of F}]
 By the definition of $\qzl$, see \eqref{eq: qzl}, we have \begin{align*}
    e^{\mp i \qzl} = e^{\mp i \cos^{-1}(\eta(\lambda))} = \eta(\lambda) \mp i \sqrt{1 - \eta^2(\lambda)} \, .
\end{align*}
Furthermore, using the identity \begin{align*}
    a^n - b^n = (a-b) \sum_{j = 1}^n a^{n-j}b^{j-1}\, ,
\end{align*}
we compute
\begin{align*}
    \l|\tilde{f}(q_{*, \lambda + \ep}) -\tilde{f}(\qzl) \r| &= \l| \sum_{m \geq 0} \l(e^{-im q_{*, \lambda + \ep}} - e^{-im \qzl} \r) f_m \r| \leq \sum_{m \geq 0} \l|e^{-im q_{*, \lambda + \ep}} - e^{-im \qzl} \r| |f_m| \\
    &\leq \l|e^{-i q_{*, \lambda + \ep}} - e^{-i \qzl} \r| \sum_{m \geq 0}  \sum_{j = 1}^m \l|e^{-i(m-j)q_{*, \lambda + \ep} } e^{-i(j-1) \qzl} \r| |f_m| \\
    &\leq \l|e^{-i q_{*, \lambda + \ep}} - e^{-i \qzl} \r| \sum_{m \geq 0} m |f_m| \\
    &\leq \l|e^{-i q_{*, \lambda + \ep}} - e^{-i \qzl} \r| \|f\|_{\ell_1^1} \, . 
\end{align*}
So it remains to bound \begin{align*}
    \l|e^{-i q_{*, \lambda + \ep}} - e^{-i \qzl} \r| \leq \l| \eta\l(\lambda + \ep \r) - \eta(\lambda) \r| +  \l| \sqrt{1 - \eta^2 \l(\lambda + \ep \r) } - \sqrt{ 1- \eta^2(\lambda)}  \r|
\end{align*} 
uniformly in $\lambda \in [\g_-, \g_+]$. 
Since \begin{align*}
    k : [0, \pi] &\rightarrow [\g_-, \g_+] \\
    y &\mapsto \sqrt{\g_1^2 + |\g_2|^2 + 2 \g_1 |\g_2| \cos(y)} 
\end{align*}
is a bijection, there exists a unique $y \in [0, \pi]$ such that $k(y) = \lambda$. 
Thus we may equivalently bound 
\begin{align}\label{eq: properties of F eqn 1}
    \l| \eta\l(k(y) + \ep \r) - \eta(k(y)) \r| +  \l| \sqrt{1 - \eta^2 \l(k(y) + \ep \r) } - \sqrt{ 1- \eta^2(k(y))}  \r|
\end{align} 
uniformly in $y \in [0, \pi]$. 

For the first term in \eqref{eq: properties of F eqn 1}, we compute
\begin{align*}
    \eta(k(y) + \ep) - \eta(k(y)) &= \frac{(k(y) + \ep)^2 - \g_1^2 - |\g_2|^2}{2 \g_1 |\g_2|} - \frac{k(y)^2 - \g_1^2 - |\g_2|^2}{2\g_1 |\g_2|} = \frac{2 k(y) \ep + \ep^2}{2 \g_1 |\g_2|} =: \delta(\ep, y)
\end{align*}
and since $\ep \leq \g_+$, \begin{align*}
    \delta( \ep, y) \leq \frac{2 \g_+ \ep + \ep^2}{2 \g_1 |\g_2|} \leq \frac{3 \g_+ }{2 \g_1 |\g_2|}\ep  \leq \frac{3 \ep}{2  \min\{ \g_1, |\g_2|\}}\, .
\end{align*}
Moreover since $\eta: [\g_- , \g_+] \rightarrow [-1, 1]$, we also have $\delta(\ep, y) \leq 2$. Thus
\begin{align}\label{eq: delta(ep,y) bound}
    \delta( \ep, y) \leq \min\l\{  \frac{3 \ep}{2  \min\{ \g_1, |\g_2|\}} , 2 \r\} \, .
\end{align}

For the second term in \eqref{eq: properties of F eqn 1}, we use the fact that $\eta(k(y) + \ep) = \delta( \ep, y) + \eta(k(y)) = \delta(\ep, y) + \cos(y)$ to compute
\begin{align*}
    \sqrt{1 - \eta^2\l(k(y) +\ep \r)} = \sqrt{1 - \l(\cos(y) + \delta \r)^2} =  \sqrt{\sin^2(y) - [2 \delta \cos(y) + \delta^2]} \, .
\end{align*}

\underline{Claim 1:}
For all $a \geq 0$ and $x \leq a^2$,  \begin{align*}
    \l|a - \sqrt{a^2 - x} \r| \leq \sqrt{|x|}
\end{align*}
Note that  $\sqrt{1 - \eta^2(k(y))} = \sin(y)$. Using Claim 1 and \eqref{eq: delta(ep,y) bound}, we compute for all $y \in [0, \pi]$,  \begin{align*}
    \l| \sqrt{1 - \eta^2(k(y))} - \sqrt{1 - \eta^2 \l(k(y) +\ep \r) } \r| &= \l| \sin(y) - \sqrt{\sin^2(y) - [2 \delta \cos(y) + \delta^2]}\r| \leq \sqrt{| 2 \delta \cos(y) + \delta^2| }  \\
    &\leq 2\sqrt{\delta}\, .
\end{align*}
Using \eqref{eq: delta(ep,y) bound} again, we estimate \eqref{eq: properties of F eqn 1},
\begin{align*}
    \l| \eta\l(k(y) + \ep \r) - \eta(k(y)) \r| +  \l| \sqrt{1 - \eta^2 \l(k(y) + \ep \r) } - \sqrt{ 1- \eta^2(k(y))}  \r| 
    &\leq \delta + 2 \sqrt{\delta} = (2 + \sqrt{2}) \sqrt{\delta} \\
    &\leq (2+ \sqrt{2}) \sqrt{\frac{3}{2}} \frac{\sqrt{\ep}}{\sqrt{\min\{ \g_1, |\g_2|\}}} \\
    &= \frac{\l(1 + \sqrt{2}\r)\sqrt{3}}{\sqrt{\min\{ \g_1, |\g_2|\}}} \sqrt{\ep} \, .
\end{align*}

\begin{proof}[Proof of Claim 1]
If $x \leq 0$, then \begin{align*}
    \l|a - \sqrt{a^2 - x} \r| = \sqrt{a^2 + |x|} -a  \leq a + \sqrt{|x|} -a =  \sqrt{|x|}
\end{align*}
If $0 \leq x \leq a^2$, then the claim reduces to showing that \begin{align*}
    a - \sqrt{a^2 - x} &\leq \sqrt{x} \\
    \iff \l( a - \sqrt{a^2 - x} \r)^2 &\leq x \\
    \iff a^2 + (a^2 -x) - 2a\sqrt{a^2 -x} &\leq x \\
    \iff 2a^2 - 2x &\leq 2 a\sqrt{a^2 -x} \\
    \iff \sqrt{a^2 -x}  &\leq a
\end{align*}

This completes the proof of Claim 1. 
\end{proof}
In similar fashion we can estimate $ \l|\tilde{f}(q_{*, \lambda - \ep}) -\tilde{f}(\qzl) \r|$. The proof to show that $\lambda \mapsto \tilde{f}'(\qzl)$ is 1/2-H\"older continuous is the same apart from requiring $f\in \ell_2^1(\mathbb{N}_0, \mathbb{C}^2)$. 
\end{proof}

\subsection{\textbf{Proof of Lemma} \ref{lem: Type IIa estimate}} \label{pf: Type IIa estimate} 

\begin{proof}[\textbf{Proof of Lemma} \ref{lem: Type IIa estimate}]
Using the change of variables $u^{\alpha} = \lambda + k(y)$, we compute
\begin{align*}
    \text{Type IIa} &:= \int \limits_0^{\pi} \cos(ny) \tilde{f}(y) e^{ik(y)t} \int \limits_{\g_-}^{\g_+}  \frac{e^{-i (\lambda + k(y)) t}}{\lambda + k(y)}\, d\lambda \, dy \\
    &=  \alpha \int \limits_0^{\pi} \cos(ny) \tilde{f}(y) e^{ik(y)t}  \int \limits_{(\g_- + k(y))^{1/\alpha}}^{(\g_+ + k(y))^{1/\alpha}} \frac{e^{-iu^{\alpha}t}}{u} \, du \, dy \, . 
\end{align*}
Taking absolute value signs and applying Van der Corput's lemma to the inner integral, we compute
\begin{align*}
    |\text{Type IIa}| &\leq \alpha \|f\|_{\ell^1} \int \limits_0^{\pi}   \l| \; \int \limits_{(\g_- + k(y))^{1/\alpha}}^{(\g_+ + k(y))^{1/\alpha}} \frac{e^{-iu^{\alpha} t}}{u} \, du \r| \, dy \\
    &\lesssim \|f\|_{\ell^1} t^{-1/\alpha} 
    \int \limits_0^{\pi} \l[ (\g_+ + k(y))^{-1/\alpha} + \int \limits_{(\g_- + k(y))^{1/\alpha}}^{(\g_+ + k(y))^{1/\alpha}} u^{-2} \, du \r] \, dy \\
    &=\|f\|_{\ell^1}  t^{-1/\alpha} \int \limits_0^{\pi} \frac{1}{(\g_- + k(y))^{1/\alpha}}  \, dy \\
    &\leq \|f\|_{\ell^1}  (\g_- t)^{-1/\alpha} \, .
\end{align*}
Setting $\alpha  =2,3$ completes the proof.

\end{proof}

\subsection{\textbf{Proof of Lemma} \ref{lem: Type IIb estimate}} \label{pf: Type IIb estimate} 

\begin{proof}[\textbf{Proof of Lemma} \ref{lem: Type IIb estimate}]
Let $F(\lambda) = \tilde{f}(\qzl) $ as in \eqref{eq: F(lambda)}. Using the change of variables $u^{\alpha} = \lambda + k(y)$, we compute
\begin{align*}
    \text{Type IIb} &:= \int \limits_0^{\pi} \cos(ny) e^{ik(y)t} \int \limits_{\g_-}^{\g_+} e^{-i (\lambda + k(y)) t} \frac{\tilde{f}(\qzl) - \tilde{f}(y)}{\lambda + k(y)}\, d\lambda \, dy  \\
    &= \int \limits_0^{\pi} \cos(ny) e^{ik(y)t} \int \limits_{\g_-}^{\g_+} e^{-i (\lambda + k(y)) t} \frac{F(\lambda) - F(k(y))}{\lambda + k(y)}\, d\lambda \, dy  \\
    &= \alpha \int \limits_0^{\pi} \cos(ny) e^{ik(y)t}  \int \limits_{(\g_- + k(y))^{1/\alpha}}^{(\g_+ + k(y))^{1/\alpha}} \frac{e^{-iu^{\alpha}t}}{u} \int \limits_{k(y)}^{u^{\alpha}- k(y)} F'(s) \, ds \, du \, dy \, .
\end{align*}
In order to have the innermost integral always be positively oriented, we write \begin{align*}
    \text{Type IIb} = \text{Type IIb2} - \text{Type IIb1} \, ,
\end{align*}
where \begin{align*}
    \text{Type IIb1} &: = \alpha \int \limits_0^{\pi} \cos(ny) e^{ik(y)t}  \int \limits_{(\g_- + k(y))^{1/\alpha}}^{(2k(y))^{1/\alpha}} \frac{e^{-iu^{\alpha}t}}{u} \int \limits_{u^{\alpha}- k(y)}^{k(y)} F'(s) \, ds \, du \, dy \\
    \text{Type IIb2} &: = \alpha \int \limits_0^{\pi} \cos(ny) e^{ik(y)t}  \int \limits_{(2k(y))^{1/\alpha}}^{(\g_+ + k(y))^{1/\alpha}} \frac{e^{-iu^{\alpha}t}}{u} \int \limits_{k(y)}^{u^{\alpha}- k(y)} F'(s) \, ds \, du \, dy \, .\\
\end{align*}
We proceed to estimate only the Type IIb2 term since the analysis for the Type IIb1 term is similar. Taking absolute value signs and applying van der Corput's lemma, see \cref{thm: VDC}, we compute
\begin{align*}
    |\text{Type IIb2}| &\lesssim \int \limits_0^{\pi} \l|   \int \limits_{(2k(y))^{1/\alpha}}^{(\g_+ + k(y))^{1/\alpha}} \frac{e^{-iu^{\alpha}t}}{u} \int \limits_{k(y)}^{u^{\alpha}- k(y)} F'(s) \, ds \, du  \r| \, dy \\
    &\lesssim t^{-1/\alpha} \int \limits_0^{\pi} \int \limits_{(2k(y))^{1/\alpha}}^{(\g_+ + k(y))^{1/\alpha}} \l| \frac{\partial}{\partial u} \l( \frac{1}{u} \int \limits_{k(y)}^{u^{\alpha}- k(y)} F'(s) \, ds\r) \r| \, du \, dy \\
    &\leq t^{-1/\alpha } \int \limits_0^{\pi} \int \limits_{(2k(y))^{1/\alpha}}^{(\g_+ + k(y))^{1/\alpha}} \alpha u^{\alpha-2} |F'(u^{\alpha } - k(y))| \, du \, dy \\
    &\quad+ t^{-1/\alpha } \int \limits_0^{\pi} \int \limits_{(2k(y))^{1/\alpha}}^{(\g_+ + k(y))^{1/\alpha}} \frac{|F(u^{\alpha} - k(y)) - F(k(y))|}{u^2}  \, du \, dy \\
    &:= t^{-1/\alpha} P_1 + t^{-1/\alpha} P_2
\end{align*}

Using the change of variables $\lambda = u^{\alpha} - k(y)$, we compute
\begin{align*}
    P_1 &= \int \limits_0^{\pi} \int \limits_{(2k(y))^{1/\alpha}}^{(\g_+ + k(y))^{1/\alpha}} \alpha u^{\alpha-2} |F'(u^{\alpha } - k(y))| \, du \, dy \\
    &= \int \limits_0^{\pi} \int \limits_{k(y)}^{\g_+} \frac{|F'(\lambda)|}{(\lambda + k(y))^{1/\alpha}} \, d\lambda \, dy \\
    &\lesssim \frac{\|f\|_{\ell_1^1}}{\sqrt{\min\{\g_1, |\g_2| \}}} \int \limits_0^{\pi} \int \limits_{k(y)}^{\g_+} \frac{1}{(\lambda + k(y))^{1/\alpha}} \l(\frac{1}{\sqrt{\lambda - \g_-}} - \frac{1}{\sqrt{\g_+ - \lambda}} \r) \, d\lambda \, dy \qquad \qquad &[\text{\cref{lem: F' bound}}] \\
    &\lesssim \frac{\|f\|_{\ell_1^1}}{\g_-^{1/\alpha} \sqrt{\min\{\g_1, |\g_2| \}}} \int \limits_0^{\pi} \int \limits_{k(y)}^{\g_+} \frac{1}{\sqrt{\lambda - \g_-}} - \frac{1}{\sqrt{\g_+ - \lambda}} \, d\lambda \, dy\qquad \qquad &[k(y) \geq \g_-] \\
    &= \frac{2 \|f\|_{\ell_1^1}}{\g_-^{1/\alpha} \sqrt{\min\{\g_1, |\g_2| \}}}
    \int \limits_0^{\pi} \sqrt{\g_+ - \g_-} + \sqrt{k(y) - \g_-} + \sqrt{\g_+ - k(y)} \, dy \\
    &\leq \frac{6 \|f\|_{\ell_1^1} \sqrt{\g_+ - \g_-}}{\g_-^{1/\alpha} \sqrt{\min\{\g_1, |\g_2| \}}}  \\
    &\lesssim \g_-^{-1/\alpha} \|f\|_{\ell_1^1} \, . 
\end{align*} 
Next we compute
\begin{align*}
    P_2 &= \int \limits_0^{\pi} \int \limits_{(2k(y))^{1/\alpha}}^{(\g_+ + k(y))^{1/\alpha}} \frac{|F(u^{\alpha}-k(y)) - F(k(y))|}{u^2} \, du \, dy \\
    &\lesssim \frac{\|f\|_{\ell_1^1}}{\sqrt{\min\{\g_1, |\g_2|\}}} \int \limits_0^{\pi} \int \limits_{(2k(y))^{1/\alpha}}^{(\g_+ + k(y))^{1/\alpha}} \frac{\sqrt{u^{\alpha}- 2k(y)}}{u^2} \, du \, dy \qquad &[\text{\cref{prop: properties of F}}] \\
    &\leq \frac{\|f\|_{\ell_1^1}}{\sqrt{\min\{\g_1, |\g_2|\}}} \int \limits_0^{\pi}  \sqrt{\g_+ - k(y)} \int \limits_{(2k(y))^{1/\alpha}}^{(\g_+ + k(y))^{1/\alpha}}  u^{-2} \, du \, dy  \\
    &\lesssim \|f\|_{\ell_1^1} \int \limits_0^{\pi} \int \limits_{(2k(y))^{1/\alpha}}^{(\g_+ + k(y))^{1/\alpha}}  u^{-2} \, du \, dy \qquad \qquad &[\g_+ - k(y) \geq \g_+ - \g_- = 2 \min\{\g_1, |\g_2|\}] \\
    &=  \|f\|_{\ell_1^1} \int \limits_0^{\pi} (2k(y))^{-1/\alpha} - (\g_+ + k(y))^{-1/\alpha} \, dy \\
    &\lesssim \g_-^{-1/\alpha} \|f\|_{\ell_1^1} \, . 
\end{align*}
Setting $\alpha  =2,3$ completes the proof.

\end{proof}

\section{Proofs from Section \ref{sec: Type IIIa Decay Rates}} \label{app: Type IIIa Decay Rates}

\subsection{\textbf{Proof of Lemma} \ref{lem: Type IIIa curved parts estimate}}\label{pf: Type IIIa curved parts estimate} 
\begin{proof}[\textbf{Proof of Lemma} \ref{lem: Type IIIa curved parts estimate}] 
We treat each term in the bracketed expression in the integral separately.  \cref{cor: application of VDC} implies 
\begin{align*}
     \l| \int \limits_0^{\pi} \cos(ny) \tilde{f}(y) e^{-i k(y) t} i \pi   \, dy \r| \lesssim 
    \begin{cases}
    \l(1+ \min\{ \g_1, |\g_2|\}^{-1/2}\r) \langle t \rangle ^{-1/3} \|f\|_{\ell^1} \\
    \l(1+\min\{\g_1, |\g_2|\}^{-1}\r) \l[\langle t\rangle^{-1/2}  + n\langle t \rangle^{-1} \r] \|f\|_{\ell_1^1} \, . 
    \end{cases}
\end{align*}
Next using \eqref{eq: result of contour integral 2}, we compute
\begin{align*}
    \l| \int \limits_0^{\pi} \cos(ny)
     \tilde{f}(y) e^{-i k(y) t} \int \limits_{C_1} \frac{e^{-iu}}{u}  \, du  \, dy \r| \leq \|\tilde{f}\|_{\infty}  \int \limits_0^{\pi} \l| \;  \int \limits_{C_1} \frac{e^{-iu}}{u}  \, du \r|   \, dy \leq \pi \|\tilde{f}\|_{\infty} \int \limits_0^{\pi} \frac{1}{1+ B(y)t} \, dy \, .
\end{align*}
Using property (8) of \cref{lem: k(y) facts} and the substitution $u = y \sqrt{c_B t}$, we compute \begin{align*}
    \int \limits_0^{\pi} \frac{1}{1+ B(y) t} \, dy &\leq \int \limits_0^{\pi} \frac{1}{1+ c_B y^2 t} \, dy \leq \int \limits_0^{\infty} \frac{1}{1+ c_By^2 t} \, dy  = \frac{1}{\sqrt{c_Bt}}   \int \limits_0^{\infty} \frac{1}{1 + u^2}\, du \\ 
    &\lesssim \l( \min\{ \g_1, |\g_2|\} \langle t \rangle \r)^{-1/2}\, .
\end{align*}
Thus \begin{align}
   \l| \int \limits_0^{\pi} \cos(ny)
     \tilde{f}(y) e^{-i k(y) t} \int \limits_{C_1} \frac{e^{-iu}}{u}  \, du  \, dy \r| \lesssim \l( \min\{ \g_1, |\g_2|\} \langle t \rangle \r)^{-1/2}\|f\|_{\ell^1}  \, , 
\end{align}
and in similar fashion, 
\begin{align}
    \l| \int \limits_0^{\pi} \cos(ny)
     \tilde{f}(y) e^{-i k(y) t} \int \limits_{C_2} \frac{e^{-iu}}{u}  \, du  \, dy \r|  \lesssim \l( \min\{ \g_1, |\g_2|\} \langle t \rangle \r)^{-1/2} \|f\|_{\ell^1}   \, . 
\end{align}

\end{proof}

\subsection{\textbf{Proof of Lemma} \ref{lem: Type IIIa vertical part estimate}}\label{pf: Type IIIa vertical part estimate} 
\begin{proof}[\textbf{Proof of Lemma} \ref{lem: Type IIIa vertical part estimate}] 

Notice that $A(y)$ and $B(y)$, see \cref{lem: k(y) facts}, intersect at $y = M \in (\pi/2, 2 \pi/3)$, where 
\begin{align*}
    M:= \begin{cases}
        \cos^{-1}\l( \frac{- \g_1}{2 |\g_2|} \r), & \text{if } |\g_2| > \g_1 \\[5pt]
        \cos^{-1}\l( \frac{- |\g_2|}{2 \g_1} \r), & \text{if } \g_1 > |\g_2| \, . 
    \end{cases} 
\end{align*}
Since $B(y) \leq A(y)$ on $[0, M]$ and $A(y) \leq B(y)$ on $[M, \pi]$, we compute 
\begin{align*}
    \l| \int \limits_{0}^{\pi}  \cos(ny)
     \tilde{f}(y) e^{-i k(y) t}  \int \limits
     _{\min \{A(y), B(y)\} t}^{\max\{A(y), B(y)\} t} \frac{e^{-x}}{x} \, dx \, dy \r|  
     \leq \|f\|_{\ell^1}  \int \limits_{0}^{\pi}    \int \limits
     _{\min \{A(y), B(y)\} t}^{\max\{A(y), B(y)\} t} \frac{e^{-x}}{x} \, dx \, dy  \\
     = \|f\|_{\ell^1}  \int \limits_{0}^{M}    \int \limits
     _{ B(y) t}^{A(y) t} \frac{e^{-x}}{x} \, dx \, dy +
     \|f\|_{\ell^1}  \int \limits_{M}^{\pi}    \int \limits
     _{A(y) t}^{B(y) t} \frac{e^{-x}}{x} \, dx \, dy \ . 
\end{align*}

The analysis of the two terms on the right hand side are similar, so we only consider the first term. 
Since   \cref{lem: k(y) facts} implies $B(y) \geq c_By^2$, we have
\begin{align*}
    \int \limits_{0}^{M}    \int \limits
     _{ B(y) t}^{A(y) t} \frac{e^{-x}}{x} \, dx \, dy \leq \int \limits_{0}^{M}    \int \limits
     _{ c_B y^2 t}^{\infty} \frac{e^{-x}}{x} \, dx \, dy \, . 
\end{align*}
We isolate the singularity of the integrand by splitting up the interval $[0, M]$ into a ``near"
and ``far" part, $[0,M] = \l[0, t^{-1/2}\r] \cup \l[ t^{-1/2}, M \r]$. Using this, we compute  
\begin{align*}
    \int \limits_{0}^{M}    \int \limits
     _{ c_B y^2 t}^{\infty} \frac{e^{-x}}{x} \, dx \, dy  &= \int \limits_{0}^{t^{-1/2}}    \int \limits
     _{ c_B y^2 t}^{\infty} \frac{e^{-x}}{x} \, dx \, dy + \int \limits_{t^{-1/2}}^{M}    \int \limits
     _{ c_B y^2 t}^{\infty} \frac{e^{-x}}{x} \, dx \, dy \\
     &\leq \int \limits_{0}^{t^{-1/2}} \frac{1}{\l(c_B y^2 t\r)^{1/3}}   \int \limits
     _{ c_B y^2 t}^{\infty} \frac{e^{-x}}{x^{2/3}} \, dx \, dy + \int \limits_{t^{-1/2}}^{M}  \frac{1}{\l(c_B y^2 t\r)^{2/3}}    \int \limits
     _{ c_B y^2 t}^{\infty} \frac{e^{-x}}{x^{1/3}} \, dx \, dy \\
     &\lesssim c_{B}^{-1/3} t^{-1/3} \int \limits_{0}^{t^{-1/2}} y^{-2/3} \, dy + c_B^{-2/3} t^{-2/3} \int \limits_{t^{-1/2}}^M y^{-4/3} \, dy \\
     &= 3 c_B^{-1/3} t^{-1/3} t^{-1/6} + 3 c_B^{-2/3} t^{-2/3} \l[t^{1/6} - M^{-1/3} \r] \\
     &\lesssim \l(1+ c_B^{-2/3} \r) \langle t \rangle^{-1/2}\, . 
\end{align*}

\end{proof}

\section{Proofs from Section \ref{sec: Type IIIb Uniform Decay}} \label{app: Type IIIb Uniform Decay}

\subsection{\textbf{Proof of Lemma} \ref{lem: bound on N and N'}}\label{pf: bound on N and N'}

\begin{lem}\label{lem: sqrt inequalities}

For all $y \in [0, \pi]$ and $u \in [0, (\g_+ - k(y))^{1/\alpha}]$, the following inequalities hold.  
\begin{align*}
    \frac{ \sqrt{\g_+ - \g_-} - \sqrt{k(y) -\g_-} }{\sqrt{\g_+- k(y)}} \leq 1 \, ,\\
     \frac{\sqrt{u^{\alpha} + k(y) -\g_-} - \sqrt{k(y)- \g_-}}{u^{\alpha/2} }\leq 1 \, , \\
     \frac{ \sqrt{\g_+ - (u^{\alpha} + k(y))} - \sqrt{\g_+ -  k(y)}}{u^{\alpha/2} } \leq 1\, . 
\end{align*}

\begin{proof}
The inequalities follow from squaring both sides and some simple arithmetic. 
\end{proof} 
\end{lem}

\begin{lem}\label{lem: F' bound}
For all $\lambda \in [\g_-,\g_+]$, 
\begin{align*}
    \l| F'(\lambda) \r| \leq \frac{2 \|f\|_{\ell_1^1}}{\sqrt{ \min\{ \g_1, |\g_2|\} }}  \l( \frac{1}{\sqrt{\lambda - \g_-}} +  \frac{1}{\sqrt{\g_+ - \lambda}} \r)\, . 
\end{align*}
\begin{proof}
Recall that $F(\lambda) = \tilde{f}(\qzl) = \tilde{f}(\cos^{-1}(\eta(\lambda)))$. Using \eqref{eq: 1 - eta^2}, we compute
\begin{align}\label{eq: F'(lambda)}
    F'(\lambda) = \frac{-\lambda \tilde{f}'(\qzl)}{\g_1 |\g_2| \sqrt{1 - \eta^2(\lambda)}} = \frac{-2 \lambda \tilde{f}'(\qzl)}{\sqrt{\g_+^2 - \lambda^2}\sqrt{\lambda^2 - \g_-^2}} =  \frac{-2 \lambda}{\sqrt{\g_+ + \lambda}\sqrt{\lambda + \g_-}} \frac{ \tilde{f}'(\qzl)}{\sqrt{\g_+ - \lambda}\sqrt{\lambda - \g_-}} \, . 
\end{align}
Note that for all $\lambda \in [\g_-, \g_+]$, 
\begin{align*}
    \frac{2 \lambda}{\sqrt{\g_+ + \lambda}\sqrt{\lambda + \g_-}}  \leq 2\, . 
\end{align*}

Furthermore, by looking at $\lambda \in [\g_-, (\g_- + \g_+)/2]$ and $\lambda \in [(\g_- + \g_+)/2, \g_+]$ separately, we see that \begin{align*}
    \frac{1}{\sqrt{\g_+ - \lambda}\sqrt{\lambda - \g_-}} &\leq \frac{\sqrt{2}}{\sqrt{\g_+ - \g_-}} \max \l\{ \frac{1}{\sqrt{\lambda - \g_-}}, \frac{1}{\sqrt{\g_+ - \lambda}} \r\} \\
    &\leq \frac{1}{\sqrt{ \min\{ \g_1, |\g_2|\} }} \l( \frac{1}{\sqrt{\lambda - \g_-}} +  \frac{1}{\sqrt{\g_+ - \lambda}} \r)  \, . 
\end{align*} 
Hence \begin{align*}
    |F'(\lambda)| \leq \frac{ 2|\tilde{f}'(\qzl)|}{\sqrt{ \min\{ \g_1, |\g_2|\} }}  \l( \frac{1}{\sqrt{\lambda - \g_-}} +  \frac{1}{\sqrt{\g_+ - \lambda}} \r)\, .
\end{align*}
\end{proof}
\end{lem}

\begin{proof}[\textbf{Proof of Lemma} \ref{lem: bound on N and N'}]
~\\
\indent \textbf{Bound on $N_2$:} 
Using \cref{lem: F' bound} to bound $|F'(s)|$, we compute 
\begin{align*}
    \l|N_2\l( (\g_+ - k(y))^{1/\alpha}; y \r) \r| &\leq \frac{1}{(\g_+ - k(y))^{1/\alpha}} \int \limits_{k(y)}^{\g_+} \l| F'(s) \r| \, ds  \\
    &\lesssim \frac{\|f\|_{\ell_1^1}}{\sqrt{ \min\{ \g_1, |\g_2|\} }}  \frac{1}{(\g_+ - k(y))^{1/\alpha}} \int \limits_{k(y)}^{\g_+} \frac{1}{\sqrt{s - \g_-}} + \frac{1}{\sqrt{\g_+ - s}} \, ds \\
    &= \frac{\|f\|_{\ell_1^1}}{\sqrt{ \min\{ \g_1, |\g_2|\} }}  \frac{2}{(\g_+ - k(y))^{1/\alpha}} \l( \sqrt{\g_+- \g_-} - \sqrt{k(y) - \g_-} + \sqrt{\g_+ - k(y)} \r) \\
    &= \frac{\|f\|_{\ell_1^1}}{\sqrt{ \min\{ \g_1, |\g_2|\} }}  \frac{2 \sqrt{\g_+ - k(y)}}{(\g_+ - k(y))^{1/\alpha}} \l( \frac{\sqrt{\g_+- \g_-} - \sqrt{k(y) - \g_-}}{\sqrt{\g_+ - k(y)}} + 1 \r) \\
    &\leq  \frac{4\|f\|_{\ell_1^1}}{\sqrt{ \min\{ \g_1, |\g_2|\} }} (\g_+ - k(y))^{\frac{1}{2} - \frac{1}{\alpha}} \qquad \quad [\cref{lem: sqrt inequalities}]\\
    &\leq \frac{4 \|f\|_{\ell_1^1}}{\sqrt{ \min\{ \g_1, |\g_2|\} }}\l( \g_+ - \g_- \r)^{\frac{1}{2} - \frac{1}{\alpha}} \\
    &= \frac{4 \|f\|_{\ell_1^1}}{\sqrt{ \min\{ \g_1, |\g_2|\} }}\l(2 \min\{ \g_1, |\g_2|\} \r)^{\frac{1}{2} - \frac{1}{\alpha}} \, .
\end{align*}

Hence \begin{align*}
    \int \limits_0^{\pi} \l|N_2\l( (\g_+ - k(y))^{1/\alpha}; y \r) \r| \, dy \lesssim   \|f\|_{\ell_1^1} \l( \min\{ \g_1, |\g_2|\} \r)^{-1/\alpha} \, . 
\end{align*}

\textbf{Bound on $\partial_u N_2$:} We compute
\begin{align*}
    N_2(u; y) &= \frac{1}{u} \int \limits_{k(y)}^{u^{\alpha} + k(y)} F'(s) \, ds \\
   \partial_u N_2(u; y) &= \alpha  u^{\alpha -2}F'(u^{\alpha} + k(y)) - \frac{1}{u^2 } \int \limits_{k(y)}^{u^{\alpha } + k(y)} F'(s) \, ds \, . 
\end{align*}
Hence 
\begin{equation}\label{eq: partial wrt u of N} 
\begin{split}
    &\int \limits_0^{\pi} \int \limits_{0}^{(\g_+-k(y))^{1/\alpha}} \l| \partial_u N_2(u; y) \r| \, du \, dy
    \\
      \leq  &\int \limits_0^{\pi}  \int \limits_{0}^{(\g_+-k(y))^{1/\alpha}} \alpha  u^{\alpha -2} \l| F'(u^{\alpha} + k(y))\r| \, du \, dy +  \int \limits_0^{\pi} \int \limits_{0}^{(\g_+-k(y))^{1/\alpha}} \frac{1}{u^2 } \int \limits_{k(y)}^{u^{\alpha} + k(y)} \l|F'(s)\r| \, ds \, du \, dy \,  .  
\end{split}\end{equation}
For the first term on the right hand side of \eqref{eq: partial wrt u of N}, we use the change of variables $\lambda = u^{\alpha} + k(y)$  and \cref{lem: F' bound}  to compute
\begin{align*}
    \int \limits_0^{\pi} \int \limits_0^{(\g_+ - k(y))^{1/\alpha }} \alpha u^{\alpha -2} &\l| F'(u^{\alpha } + k(y))\r| \, du \, dy  
    =   \int \limits_0^{\pi} \int \limits_{k(y)}^{\g_+} \frac{|F'(\lambda)|}{(\lambda - k(y))^{1/\alpha}}\, d\lambda \, dy \\
    & \leq \frac{\|f\|_{\ell_1^1}}{\sqrt{ \min\{ \g_1, |\g_2|\} }}   \int \limits_0^{\pi} \int \limits_{k(y)}^{\g_+} \l( \frac{1}{\sqrt{\lambda - \g_-}} + \frac{1}{\sqrt{\g_+ - \lambda}} \r) \frac{1}{( \lambda - k(y))^{1/\alpha}}\, d\lambda \, . \numberthis \label{eq: partial wrt u of N part 1} 
\end{align*}
We study the two terms in \eqref{eq: partial wrt u of N part 1} separately. 
Using the fact that $\g_- \leq k(y) \leq \lambda$ implies $(\lambda -\g_-)^{-1/2} \leq (\lambda - k(y))^{-1/2}$, we obtain the following estimate for the first term in \eqref{eq: partial wrt u of N part 1}, 
\begin{equation} \label{eq: partial wrt u of N part 1a}
\begin{split}
     \int \limits_0^{\pi}  \int \limits_{k(y)}^{\g_+} \frac{1}{\sqrt{\lambda - \g_-}}  \frac{1}{(\lambda - k(y))^{1/\alpha }} \, d\lambda \, dy 
     &\leq \int \limits_0^{\pi}  \int \limits_{k(y)}^{\g_+} \frac{1}{(\lambda - k(y))^{\frac{1}{2} + \frac{1}{\alpha}}} \, d \lambda \, dy  
     \\
     &= \frac{1}{\frac{1}{2} - \frac{1}{\alpha}} \int \limits_0^{\pi} (\g_+ - k(y))^{\frac{1}{2} - \frac{1}{\alpha}} \, dy \\
     &= \frac{2 \alpha}{\alpha - 2}  \int \limits_0^{\pi} (\g_+ - k(y))^{\frac{1}{2} - \frac{1}{\alpha}} \, dy \, . 
\end{split}\end{equation} 
 
To estimate the second term in \eqref{eq: partial wrt u of N part 1}, we first note that \begin{align*}
    \g_+ - \lambda \geq  \frac{\g_+ - k(y)}{2}, \quad & \text{if } \lambda \in \l[ k(y), \frac{\g_+ + k(y)}{2} \r] \\
    \lambda - k(y) \geq  \frac{\g_+ - k(y)}{2}, \quad & \text{if } \lambda \in \l[ \frac{\g_+ + k(y)}{2} , \g_+ \r]\, .
\end{align*}
Using these lower bounds, we compute \begin{equation} \label{eq: partial wrt u of N part 1b}
\begin{split}
     &\int \limits_0^{\pi}  \int \limits_{k(y)}^{\g_+}  \frac{1}{\sqrt{\g_+ - \lambda}} \frac{1}{(\lambda - k(y))^{1/\alpha}} \, d\lambda \, dy \\
     &=  \int \limits_0^{\pi}  \int \limits_{k(y)}^{(\g_++k(y))/2}  \frac{1}{\sqrt{\g_+ - \lambda}} \frac{1}{(\lambda - k(y))^{1/\alpha}} \, d\lambda \, dy  + 
     \int \limits_0^{\pi}  \int \limits_{(\g_++k(y))/2}^{\g_+}  \frac{1}{\sqrt{\g_+ - \lambda}} \frac{1}{(\lambda - k(y))^{1/\alpha}} \, d\lambda \, dy \\
     &\lesssim  \int \limits_0^{\pi} \frac{1}{\sqrt{\g_+ - k(y)}} \int \limits_{k(y)}^{(\g_++k(y))/2}   \frac{1}{(\lambda - k(y))^{1/\alpha}} \, d\lambda \, dy  + 
     \int \limits_0^{\pi}  \frac{1}{(\g_+ - k(y))^{1/\alpha}} \int \limits_{(\g_+ +k(y))/2}^{\g_+}  \frac{1}{\sqrt{\g_+ - \lambda}} \, d\lambda \, dy \\
      &= \l(\frac{1}{1 - \frac{1}{\alpha}} \r) \int \limits_0^{\pi} \frac{ \l( \frac{\g_+ - k(y)}{2}\r)^{1 - 1/\alpha}}{\sqrt{\g_+ - k(y)}} \, dy + \sqrt{2} \int \limits_0^{\pi}  \frac{\sqrt{ \g_+ - k(y)}}{(\g_+ - k(y))^{1/\alpha}}  \, dy \\
      &\lesssim \int \limits_0^{\pi} (\g_+ - k(y))^{\frac{1}{2} - \frac{1}{\alpha}} \, dy \, .
\end{split}\end{equation}
In the last line we used the fact that $\alpha > 2$ implies \begin{align*}
    \frac{1}{1 - \frac{1}{\alpha}}=  1 + \frac{1}{\alpha -1} < 2\, .
\end{align*}

Combining \eqref{eq: partial wrt u of N part 1}, \eqref{eq: partial wrt u of N part 1a}, and \eqref{eq: partial wrt u of N part 1b}, we obtain the following estimate of the first term on the right hand side of \eqref{eq: partial wrt u of N}
\begin{align}\label{eq: partial wrt u of N part 1 estimate}
    \int \limits_0^{\pi} \int \limits_0^{(\g_+ - k(y))^{1/\alpha }} \alpha u^{\alpha -2} \l| F'(u^{\alpha } + k(y))\r| \, du \, dy  \lesssim \l(1 + \frac{1}{\alpha -2} \r)\frac{\|f\|_{\ell_1^1}}{\sqrt{ \min\{ \g_1, |\g_2|\} }}   \int \limits_0^{\pi} (\g_+ - k(y))^{\frac{1}{2} - \frac{1}{\alpha}} \, dy \, . 
\end{align}

For the second term on the right hand side of \eqref{eq: partial wrt u of N},  we use \cref{lem: F' bound} to compute
\begin{equation}\label{eq: partial wrt u of N part 2 estimate} \begin{split}
     &\int \limits_0^{\pi} \int \limits_{0}^{(\g_+-k(y))^{1/\alpha}} \frac{1}{u^2 } \int \limits_{k(y)}^{u^{\alpha} + k(y)} \l|F'(s)\r| \, ds  \, du \, dy \\
     &\lesssim \frac{\|f\|_{\ell_1^1}}{\sqrt{ \min\{ \g_1, |\g_2|\} }}   \int \limits_0^{\pi} \int \limits_{0}^{(\g_+-k(y))^{1/\alpha}} \frac{1}{u^2 } \int \limits_{k(y)}^{u^{\alpha}+ k(y)} \frac{1}{\sqrt{s- \g_-}} + \frac{1}{\sqrt{\g_+ -s}} \, ds \, du \, dy \\
     &= \frac{2\|f\|_{\ell_1^1}}{\sqrt{ \min\{ \g_1, |\g_2|\} }} \int \limits_0^{\pi} \int \limits_{0}^{(\g_+ - k(y))^{1/\alpha}} \frac{\sqrt{u^{\alpha} + k(y) -\g_-} - \sqrt{k(y)- \g_- } + \sqrt{\g_+- k(y)} - \sqrt{\g_+ - (u^{\alpha} + k(y))}}{u^2 } \, du \, dy \\
     &\lesssim  \frac{\|f\|_{\ell_1^1}}{\sqrt{ \min\{ \g_1, |\g_2|\} }}  \int \limits_0^{\pi} \int \limits_{0}^{(\g_+-k(y))^{1/\alpha}}  u^{\alpha/2 -2} \, du \, dy \qquad \qquad \qquad [\cref{lem: sqrt inequalities}]\\
     &=  \frac{1}{\frac{\alpha}{2} -1} \frac{\|f\|_{\ell_1^1}}{\sqrt{ \min\{ \g_1, |\g_2|\} }} \int \limits_0^{\pi} (\g_+ -k(y))^{\frac{1}{2}-\frac{1}{\alpha}} \, dy \\
     &\lesssim \frac{1}{\alpha -2} \frac{\|f\|_{\ell_1^1}}{\sqrt{ \min\{ \g_1, |\g_2|\} }}  \int \limits_0^{\pi} (\g_+ - k(y))^{\frac{1}{2} - \frac{1}{\alpha}} \, dy \, .
\end{split}\end{equation} 

Applying the estimates \eqref{eq: partial wrt u of N part 1 estimate} and \eqref{eq: partial wrt u of N part 2 estimate} to \eqref{eq: partial wrt u of N} yields
\begin{align*}
    \int \limits_0^{\pi} \int \limits_{0}^{(\g_+-k(y))^{1/\alpha}} \l| \partial_u N_2(u^{\alpha} + k(y), y) \r| \, du \, dy \lesssim \l(1 +  \frac{1}{\alpha -2}\r) \frac{\|f\|_{\ell_1^1}}{\sqrt{ \min\{ \g_1, |\g_2|\} }}   \int \limits_0^{\pi} (\g_+ - k(y))^{\frac{1}{2} - \frac{1}{\alpha}} \, dy \, .
\end{align*} 
To complete the proof, we note that 
\begin{align*}
    &\g_+ - k(y) \leq \g_+ - \g_- = 2 \min\{ \g_1, |\g_2|\} \\
    \implies &\frac{1}{\sqrt{ \min\{ \g_1, |\g_2|\} }}   \int \limits_0^{\pi} (\g_+ - k(y))^{\frac{1}{2} - \frac{1}{\alpha}} \, dy \lesssim  \l( \min \{ \g_1, |\g_2|\}\r)^{-1/\alpha} \, . 
\end{align*}
\end{proof}

\subsection{\textbf{Proof of Corollary} \ref{cor: Type IIIb uniform decay estimate}}\label{pf: Type IIIb uniform decay estimate corollary}

\begin{proof}[\textbf{Proof of Corollary} \ref{cor: Type IIIb uniform decay estimate}]

For each time $t \geq 0$, we seek to find  $\alpha(t) > 2$ which optimizes the decay rate estimate from \cref{prop: Type IIIb uniform decay estimate},
\begin{align*}
    |\text{Type IIIb}(n,t)| \lesssim  \frac{\l( \min\{ \g_1, |\g_2|\}\langle t \rangle \r)^{-1/\alpha}}{\alpha -2} \|f\|_{\ell_1^1}  \, .
\end{align*}
Letting $\ep = \alpha -2$, we compute
\begin{align}\label{eq: optimizing decay in alpha}
     \frac{\langle t \rangle^{-1/\alpha}  }{\alpha -2} = \frac{\langle t \rangle ^{\frac{-1}{2+\ep}}}{\ep}  \leq \langle t \rangle^{-1/2} \frac{\langle t \rangle^{\ep}}{\ep} = \langle t \rangle^{-1/2} \exp \l[ \ep \log(\langle t \rangle ) - \log(\ep) \r] \, . 
\end{align}
To optimize in $\ep$, we set the first derivative of the bracketed term equal to zero, 
\begin{align*}
    \frac{\partial }{\partial \ep} \ep \log(\langle t \rangle ) - \log(\ep) =  \log(\langle t \rangle ) - \frac{1}{\ep} = 0 \, . 
\end{align*}
Hence the optimal choice of $\ep$ is 
\begin{align*}
    \ep_* = \frac{1}{\log(\langle t \rangle )} \, ,
\end{align*}
and plugging this choice into the right hand side of \eqref{eq: optimizing decay in alpha} gives a $\log(\langle t \rangle ) \langle t\rangle^{-1/2}$ decay rate. Lastly, in order for the decay rate to hold for small $t$, we replace $\log(\langle t \rangle )$ with $\log\l( \sqrt{2 + t^2}\r)$.

\end{proof}

\section{Proofs from Section \ref{sec: Type IIIb Local Decay}} \label{app: Type IIIb Local Decay}

\subsection{\textbf{Proof of Lemma} \ref{lem: M bound}}\label{pf: M bound}

\begin{proof}[\textbf{Proof of Lemma} \ref{lem: M bound}]
Using \cref{lem: F' bound} to bound $|F'(s)|$, we compute
\begin{align*}
    |M(y)| &\leq  \int \limits_{0}^{(\g_+-k(y))^{1/\alpha}} \frac{1}{u} \int \limits_{k(y)}^{u^{\alpha} + k(y)} |F'(s)| \, ds \, du  \\
    &\lesssim  \frac{\|f\|_{\ell_1^1}}{\sqrt{ \min\{ \g_1, |\g_2|\} }} \int \limits_{0}^{(\g_+-k(y))^{1/\alpha}} \frac{1}{u} \int \limits_{k(y)}^{u^{\alpha} + k(y)}   \l( \frac{1}{\sqrt{s - \g_-}} +  \frac{1}{\sqrt{\g_+ - s}} \r) \, ds \, du \\
    &= \frac{2 \|f\|_{\ell_1^1}}{\sqrt{ \min\{ \g_1, |\g_2|\} }} \int \limits_{0}^{(\g_+-k(y))^{1/\alpha}} \frac{\sqrt{u^{\alpha} + k(y) - \g_-} - \sqrt{k(y) - \g_-}  + \sqrt{\g_+ - k(y)} - \sqrt{\g_+ - (u^{\alpha} +k(y))}}{u} \, du \\
    &\leq \frac{4 \|f\|_{\ell_1^1}}{\sqrt{ \min\{ \g_1, |\g_2|\} }} \int \limits_{0}^{(\g_+-k(y))^{1/\alpha}} u^{\alpha/2 -1}\, du \qquad \qquad \qquad [\cref{lem: sqrt inequalities} ] \\
    &= \frac{8 \|f\|_{\ell_1^1}}{ \alpha \sqrt{ \min\{ \g_1, |\g_2|\} }} \sqrt{\g_+ - k(y)} \, . 
\end{align*}
The result follows from the fact that \begin{align*}
    \sqrt{\g_+ - k(y)} \leq \sqrt{\g_+ - \g_-} = \sqrt{2 \min\{ \g_1, |\g_2| \}} \, . 
\end{align*}

\end{proof}

\subsection{\textbf{Proof of Lemma} \ref{lem: M1 bound}}\label{pf: M1 bound}

\begin{proof}[\textbf{Proof of Lemma} \ref{lem: M1 bound}]

Recall, \begin{align*}
    M_1 &:= \int \limits_{\g_-}^{\g_+} \frac{1}{\alpha( \g_+ - v)} \int \limits_v^{\g_+} |F'(s)| \, ds \, dv \, . 
\end{align*}
By \cref{lem: F' bound}, 
\begin{align*}
    \int \limits_{v}^{\g_+} |F'(s)| \, ds &\lesssim \frac{\|f\|_{\ell_1^1}}{\sqrt{ \min\{\g_1, |\g_2| \}}} \int \limits_{v}^{\g_+} \frac{1}{\sqrt{s - \g_-}} + \frac{1}{\sqrt{\g_+ -s}} \, ds \\
     &= \frac{2\|f\|_{\ell_1^1}}{\sqrt{ \min\{\g_1, |\g_2| \}}} \l[\sqrt{\g_+ - \g_-} -  \sqrt{v - \g_-} + \sqrt{\g_+ - v}\r] \\
    &= \frac{2\|f\|_{\ell_1^1} \sqrt{\g_+ - v}}{\sqrt{ \min\{\g_1, |\g_2| \}}} \l[ \frac{\sqrt{\g_+ - \g_-} - \sqrt{v - \g_-}}{\sqrt{\g_+ - v}} + 1 \r] \\
    &\leq \frac{4\|f\|_{\ell_1^1} \sqrt{\g_+ - v}}{\sqrt{ \min\{\g_1, |\g_2| \}}}\, . \qquad \qquad [\cref{lem: sqrt inequalities}]
\end{align*}
Hence \begin{align*}
    M_1 \leq \frac{4 \|f\|_{\ell_1^1}}{\alpha \sqrt{ \min\{\g_1, |\g_2| \}} } \int \limits_{\g_-}^{\g_+} \frac{1}{\sqrt{\g_+ - v}} \, dv = \frac{8 \|f\|_{\ell_1^1} \sqrt{\g_+ - \g_-}}{\alpha \sqrt{ \min\{\g_1, |\g_2| \}} } = \frac{8 \sqrt{2} \|f\|_{\ell_1^1}}{\alpha} \, . 
\end{align*}

\end{proof}

\subsection{\textbf{Proof of Lemma} \ref{lem: L(lambda) Lipschitz}}\label{pf: L(lambda) Lipschitz}

\begin{proof}[\textbf{Proof of Lemma} \ref{lem: L(lambda) Lipschitz}]

We compute \begin{align*}
    -\frac{1}{2} L'(\lambda) &= \frac{1}{\sqrt{\g_+ + \lambda}{\sqrt{\lambda + \g_-}}} - \frac{\lambda( \g_- + \g_+ + 2 \lambda) }{2(\g_+ + \lambda)^{3/2} (\lambda + \g_-)^{3/2}} \\
     -\frac{1}{2} L''(\lambda) &=  \frac{ \frac{3\lambda (\g_- + \g_+ + 2 \lambda)^2}{4(\g_- + \lambda)(\g_+ + \lambda)} - (\g_- + \g_+ + 3 \lambda)   }{(\g_+ + \lambda)^{3/2} (\g_- + \lambda)^{3/2}}\, . 
\end{align*}
Hence $-L''(\lambda)/2 < 0$ if and only if \begin{align*}
   \frac{3\lambda (\g_- + \g_+ + 2 \lambda)^2}{4(\g_- + \lambda)(\g_+ + \lambda)} < \g_- + \g_+ + 3 \lambda \, 
\end{align*}
After multiplying both sides by $4(\g_- + \lambda)(\g_+ + \lambda)$ and expanding, it is straightforward to see that this last inequality is true. Hence $-L'(\lambda)/2$ is always decreasing. Furthermore, 
\begin{align*}
    -\frac{1}{2}L'(\lambda) > 0 \iff 2 (\g_+ + \lambda)(\lambda + \g_-) > \lambda ( \g_- + \g_+ + 2 \lambda) \, ,
\end{align*}
and the right hand side is easily seen to be true. Hence 
\begin{align*}
    \sup_{\lambda \in [\g_-, \g_+]} \l|- \frac{1}{2}L'(\lambda) \r| &= - \frac{1}{2} L'(\g_-) 
    = \frac{1}{\sqrt{\g_+ + \g_-}{\sqrt{2\g_- }}} - \frac{\g_-( 3\g_- + \g_+ ) }{2(\g_+ + \g_-)^{3/2} (2\g_-)^{3/2}}  \\
    &= \frac{1}{\sqrt{\g_+ + \g_-}{\sqrt{2\g_- }}} \l[ 1 - \frac{3 \g_- + \g_+}{4(\g_+ + \g_-)} \r] \\
    &=  \frac{1}{\sqrt{2 \max\{\g_1, |\g_2|\}}{\sqrt{2\g_- }}} \frac{\g_- + 3 \g_+}{4(\g_+ + \g_-)} \\
    &\leq \frac{3}{8 \sqrt{\g_- \max\{\g_1, |\g_2|\}}} \, . 
\end{align*}

\end{proof}

\subsection{\textbf{Proof of Lemma} \ref{lem: M4 bound}}\label{pf: M4 bound}

\begin{proof}[\textbf{Proof of Lemma} \ref{lem: M4 bound}]
Recall
\begin{align*}
    M_4 := \frac{ 1}{\sqrt{\g_- \max\{\g_1, |\g_2| \}}} \int \limits_{\g_-}^{\g_+} \int \limits_0^{(\g_+ - v)^{1/\alpha}}  \frac{ u^{\alpha -1} |\tilde{f}'(q_{*, v})|}{\sqrt{\g_+ - v}\sqrt{v - \g_-}} \, du \, dv \, . 
\end{align*}

We compute\begin{align*}
    \int \limits_{\g_-}^{\g_+} \int \limits_0^{(\g_+ - v)^{1/\alpha}}  \frac{ u^{\alpha -1} |\tilde{f}'(q_{*, v})|}{\sqrt{\g_+ - v}\sqrt{v - \g_-}} \, du \, dv 
    &\leq \|f\|_{\ell_1^1} \int \limits_{\g_-}^{\g_+}  \frac{ 1}{\sqrt{\g_+ - v}\sqrt{v - \g_-}} \int \limits_0^{(\g_+ - v)^{1/\alpha}} u^{\alpha -1} \, du \, dv  \\
    &= \frac{\|f\|_{\ell_1^1}}{\alpha} \int \limits_{\g_-}^{\g_+} \frac{\sqrt{\g_+ - v}}{\sqrt{v - \g_-}}\, dv \leq \frac{\sqrt{\g_+ - \g_-}\|f\|_{\ell_1^1}}{\alpha} \int \limits_{\g_-}^{\g_+} \frac{1}{\sqrt{v - \g_-}} \, dv \\
    &= \frac{2(\g_+ - \g_-) \|f\|_{\ell_1^1} }{\alpha}  \\
    &= \frac{2\sqrt{2} \min\{ \g_1, |\g_2|\} \|f\|_{\ell_1^1} }{\alpha} \, . 
\end{align*}
So \begin{align*}
    M_4 \lesssim  \frac{\min\{ \g_1, |\g_2|\} \|f\|_{\ell_1^1} }{\alpha \sqrt{\g_- \max\{\g_1, |\g_2| \}}} \leq \sqrt{\frac{\min\{ \g_1, |\g_2|\}}{\g_-}} \frac{\|f\|_{\ell_1^1}}{\alpha} \, . 
\end{align*}

\end{proof}

\subsection{\textbf{Proof of Lemma} \ref{lem: M3a bound}}\label{pf: M3a bound}

\begin{proof}[\textbf{Proof of Lemma} \ref{lem: M3a bound}]
Recall
\begin{align*}
    M_{3a} &:= \|f\|_{\ell_1^1} \int \limits_{\g_-}^{\g_+} \int \limits_0^{(\g_+ - v)^{1/\alpha}} \frac{1}{u} \l|\frac{ 1}{\sqrt{\g_+ - (u^{\alpha} + v)} \sqrt{u^{\alpha} + v - \g_-}} - \frac{1 }{\sqrt{\g_+ - v} \sqrt{v  -\g_-}}  \r| \, du \, dv  \, . 
\end{align*}

Using the triangle inequality, we compute
\begin{align*}
    &\int \limits_{\g_-}^{\g_+} \int \limits_0^{(\g_+ - v)^{1/\alpha}} \frac{1}{u} \l|\frac{ 1}{\sqrt{\g_+ - (u^{\alpha} + v)} \sqrt{u^{\alpha} + v - \g_-}} - \frac{1 }{\sqrt{\g_+ - v} \sqrt{v  -\g_-}}  \r| \, du \, dv\\
    &\leq \int \limits_{\g_-}^{\g_+} \int \limits_0^{(\g_+ - v)^{1/\alpha}} \frac{1}{u \sqrt{u^{\alpha} + v - \g_-}} \l|\frac{ 1}{\sqrt{\g_+ - (u^{\alpha} + v)} } - \frac{1 }{\sqrt{\g_+ - v} }  \r| \, du \, dv \\
    &\quad + \int \limits_{\g_-}^{\g_+} \int \limits_0^{(\g_+ - v)^{1/\alpha}} \frac{1}{u \sqrt{\g_+ - v}} \l|\frac{ 1}{ \sqrt{u^{\alpha} + v - \g_-}} - \frac{1 }{\sqrt{v  -\g_-}}  \r| \, du \, dv \\ 
    &:= P + Q\, . 
\end{align*}
We prove below that  $P, Q \lesssim \alpha^{-1}$ through the same general method. We split up the region of integration into several parts so that we may isolate singularities in the integrand. From  there, we can use tricks like change of variables and swapping order of integration in order to obtain a bound.

\textbf{Estimating P:}
In the integrand of $P$, when $u$ is far away from $0$, the only term that can blow up is $1/\sqrt{\g_+ - (u^{\alpha }+ v})$. This suggests that we partition
\begin{align*}
    u \in \l[0, (\g_+ - v)^{1/\alpha} \r] = \l[0, \l(\frac{\g_+ - v}{2}\r)^{1/\alpha} \r] \cup \l[ \l(\frac{\g_+ - v}{2}\r)^{1/\alpha}, (\g_+ - v)^{1/\alpha} \r] \, .
\end{align*}

When $u$ is near zero, we would like to distinguish between blow up due to the  term $1/\sqrt{\g_+ - (u^{\alpha} + v)}$  when $v$ is near $\g_+$ and blow up from the term $1/\sqrt{u^{\alpha} + v - \g_-}$ when $v$ is near $\g_-$. This suggests that we partition \begin{align*}
    v \in [\g_-, \g_+] = \l[\g_-, \frac{\g_+ + \g_-}{2}\r] \cup \l[\frac{\g_+ + \g_-}{2}, \g_+ \r] \, .
\end{align*}
Altogether, this motivates our choice to partition the region of integration into three regions as follows, 
\begin{align*}
    P &= \int \limits_{\g_-}^{\g_+} \int \limits_0^{(\g_+ - v)^{1/\alpha}} \frac{1}{u \sqrt{u^{\alpha} + v - \g_-}} \l|\frac{ 1}{\sqrt{\g_+ - (u^{\alpha} + v)} } - \frac{1 }{\sqrt{\g_+ - v} }  \r| \, du \, dv \\
    &= 
    \int \limits_{\g_-}^{\g_+}  \int \limits_{\l(\frac{\g_+ - v}{2}\r)^{1/\alpha}}^{(\g_+ - v)^{1/\alpha}}  \dots \, du \, dv  
    + 
    \int \limits_{\g_-}^{\frac{\g_+ + \g_-}{2}}  \int \limits_0^{\l(\frac{\g_+ - v}{2}\r)^{1/\alpha}}  \dots \, du \, dv
    + 
    \int \limits_{\frac{\g_+ + \g_-}{2}}^{\g_+}  \int \limits_0^{\l(\frac{\g_+ - v}{2}\r)^{1/\alpha}}  \dots \, du \, dv  \\
    &:= P_1 + P_2 + P_3 \, . 
\end{align*}

\textit{Bounding the $P_1$ component:}
For $P_1$, note that 
\begin{align*}
    \frac{1}{\sqrt{u^{\alpha} + v - \g_-}} \leq \frac{1}{\sqrt{ \frac{\g_+ + v}{2} -\g_-}} \leq \frac{1}{\sqrt{ \frac{\g_+ - \g_-}{2}}}  = \frac{1}{\sqrt{\min\{\g_1, |\g_2|\}}}\, . 
\end{align*}
Hence \begin{align*}
    P_1 &= \int \limits_{\g_-}^{\g_+}  \int \limits_{\l(\frac{\g_+ - v}{2}\r)^{1/\alpha}}^{(\g_+ - v)^{1/\alpha}}  \frac{1}{u \sqrt{u^{\alpha} + v - \g_-}} \l|\frac{ 1}{\sqrt{\g_+ - (u^{\alpha} + v)} } - \frac{1 }{\sqrt{\g_+ - v} }  \r| \, du \, dv  \\
    &\leq \frac{1}{\sqrt{\min\{\g_1, |\g_2|\}}}  \int \limits_{\g_-}^{\g_+}  \int \limits_{\l(\frac{\g_+ - v}{2}\r)^{1/\alpha}}^{(\g_+ - v)^{1/\alpha}}  \frac{1}{u\sqrt{\g_+ - (u^{\alpha} + v)}} \, du \, dv\,. 
\end{align*}
Using the change of variables $z = \g_+ -v$ and swapping the order of integration, we compute \begin{align*}
    &\int \limits_{\g_-}^{\g_+}  \int \limits_{\l(\frac{\g_+ - v}{2}\r)^{1/\alpha}}^{(\g_+ - v)^{1/\alpha}}  \frac{1}{u\sqrt{\g_+ - (u^{\alpha} + v)}} \, du \, dv 
    = \int \limits_{0}^{2 \min\{\g_1, |\g_2|\}} \int \limits_{(z/2)^{1/\alpha}}^{z^{1/\alpha}} \frac{1}{u \sqrt{z - u^{\alpha}}} \, du \, dz \\
    &= \int \limits_0^{(\min\{\g_1, |\g_2|\})^{1/\alpha}} \frac{1}{u} \int \limits_{u^{\alpha}}^{2u^{\alpha}} \frac{1}{\sqrt{z-u^{\alpha}}} \, dz \, du  +  \int \limits_{(\min\{\g_1, |\g_2|\})^{1/\alpha}}^{(2\min\{\g_1, |\g_2|\})^{1/\alpha}} \frac{1}{u} \int \limits_{u^{\alpha}}^{2 \min\{\g_1, |\g_2|\}} \frac{1}{\sqrt{z-u^{\alpha}}} \, dz \, du \\
    &= 2 \int \limits_0^{(\min\{\g_1, |\g_2|\})^{1/\alpha}}  u^{\alpha/2 -1} \, du + 2 \int \limits_{(\min\{\g_1, |\g_2|\})^{1/\alpha}}^{(2\min\{\g_1, |\g_2|\})^{1/\alpha}}  \frac{\sqrt{2\min\{\g_1, |\g_2|\} - u^{\alpha}}}{u} \, du \\
    &\leq \frac{4\sqrt{\min\{\g_1, |\g_2|\}}}{\alpha} + \sqrt{\min\{\g_1, |\g_2|\}} \int \limits_{(\min\{\g_1, |\g_2|\})^{1/\alpha}}^{(2\min\{\g_1, |\g_2|\})^{1/\alpha}} \frac{1}{u} \, du \\
    &= \frac{\l(4 + \log(2) \r) \sqrt{\min\{\g_1, |\g_2|\}}}{\alpha}  \, . 
\end{align*}
Thus \begin{align*}
    P_1 \lesssim \alpha^{-1}\, . 
\end{align*}

\textit{Bounding the $P_2$ component:}
 \cref{lem: sqrt inequalities} implies \begin{align*}
    \l|\frac{ 1}{\sqrt{\g_+ - (u^{\alpha} + v)} } - \frac{1 }{\sqrt{\g_+ - v} }  \r| \leq \frac{\l|\sqrt{\g_+ -v} - \sqrt{\g_+ - (u^{\alpha} + v)}\r|}{\sqrt{\g_+ - (u^{\alpha} + v)} \sqrt{\g_+ -v}}  \leq \frac{u^{\alpha/2}}{\sqrt{\g_+ - (u^{\alpha} + v)} \sqrt{\g_+ -v}}\, . 
\end{align*}
Additionally, since $u^{\alpha} \leq (\g_+ - v)/2$ and $v \leq (\g_+ + \g_-)/2$, it follows that 
\begin{align*}
    \g_+ - \l( u^{\alpha} + v\r) \geq \g_+ - \l(\frac{\g_+ - v}{2} +v\r)  = \frac{\g_+ - v}{2} \geq \frac{
    \g_+ - \g_-}{4} \, . 
\end{align*}

Hence 
\begin{align*}
    P_2 &= \int \limits_{\g_-}^{\frac{\g_+ + \g_-}{2}}  \int \limits_0^{\l(\frac{\g_+ - v}{2}\r)^{1/\alpha}} \frac{1}{u \sqrt{u^{\alpha} + v - \g_-}} \l|\frac{ 1}{\sqrt{\g_+ - (u^{\alpha} + v)} } - \frac{1 }{\sqrt{\g_+ - v} }  \r| \, du \, dv  \\
    &\leq  \int \limits_{\g_-}^{\frac{\g_+ + \g_-}{2}}  \int \limits_0^{\l(\frac{\g_+ - v}{2}\r)^{1/\alpha}}  
    \frac{1}{\sqrt{u^{\alpha} + v - \g_-}} 
    \frac{u^{\alpha/2 -1}}{\sqrt{\g_+ - (u^{\alpha} + v)} \sqrt{\g_+ -v}} \, du \, dv \qquad \qquad [\cref{lem: sqrt inequalities}]\\
    &\leq \frac{2}{\sqrt{\g_+ - \g_-}}\int \limits_{\g_-}^{\frac{\g_+ + \g_-}{2}}  \int \limits_0^{\l(\frac{\g_+ - v}{2}\r)^{1/\alpha}}  
    \frac{1}{\sqrt{u^{\alpha} + v - \g_-}} \frac{u^{\alpha/2 -1}}{\sqrt{\g_+ -v}} \, du \, dv \qquad \qquad [u^{\alpha} \leq (\g_+ - v)/2] \\
    &\leq \frac{2}{\sqrt{\g_+ - \g_-}} \int \limits_{\g_-}^{\frac{\g_+ + \g_-}{2}} \frac{1}{\sqrt{ v - \g_-} \sqrt{\g_+ -v}} \int \limits_0^{\l(\frac{\g_+ - v}{2}\r)^{1/\alpha}}  u^{\alpha/2 -1}
     \, du \, dv \\
    &= \frac{2 \sqrt{2}}{ \alpha \sqrt{\g_+ - \g_-}} \int \limits_{\g_-}^{\frac{\g_+ + \g_-}{2}}  \frac{1}{\sqrt{v - \g_-}} \, dv \\
     &= \frac{4}{\alpha}  \, . 
\end{align*}

\textit{Bounding the $P_3$ component:}
Using the estimate $\sqrt{u^{\alpha} + v - \g_-} \geq \sqrt{(\g_+- \g_-)/2}$ and the change of variables $z = \g_+-v$, we compute \begin{align*}
    P_3 &= \int \limits_{\frac{\g_+ + \g_-}{2}}^{\g_+}  \int \limits_0^{\l(\frac{\g_+ - v}{2}\r)^{1/\alpha}} \frac{1}{u \sqrt{u^{\alpha} + v - \g_-}} \l|\frac{ 1}{\sqrt{\g_+ - (u^{\alpha} + v)} } - \frac{1 }{\sqrt{\g_+ - v} }  \r| \, du \, dv  \\
    &\leq \frac{\sqrt{2}}{\sqrt{\g_+ - \g_-}} \int \limits_{\frac{\g_+ + \g_-}{2}}^{\g_+}  \int \limits_0^{\l(\frac{\g_+ - v}{2}\r)^{1/\alpha}} \frac{1}{u} \l|\frac{ 1}{\sqrt{\g_+ - (u^{\alpha} + v)} } - \frac{1 }{\sqrt{\g_+ - v} }  \r| \, du \, dv \\
    &= \frac{\sqrt{2}}{\sqrt{\g_+ - \g_-}} \int \limits_0^{ \frac{\g_+ - \g_-}{2}} \int \limits_0^{(z/2)^{1/\alpha}} \frac{1}{u} \l| \frac{1}{\sqrt{z - u^{\alpha}}} - \frac{1}{\sqrt{z}} \r| \, du \, dz \\
    &= \frac{\sqrt{2}}{\sqrt{\g_+ - \g_-}} \int \limits_0^{\l(\frac{\g_+ - \g_-}{4}\r)^{1/\alpha}} \frac{1}{u} \int \limits_{2u^{\alpha}}^{\frac{\g_+ - \g_-}{2} }  \frac{1}{\sqrt{z - u^{\alpha}}} - \frac{1}{\sqrt{z}} \, dz \, du \\
    &= \frac{2\sqrt{2}}{\sqrt{\g_+ - \g_-}} \int \limits_0^{\l(\frac{\g_+ - \g_-}{4}\r)^{1/\alpha}} \frac{\sqrt{\frac{\g_+ - \g_-}{2} - u^{\alpha}} - u^{\alpha/2} + \sqrt{2} u^{\alpha/2} - \sqrt{\frac{\g_+ - \g_-}{2}} }{u} \, du\, . 
\end{align*}
In similar fashion to \cref{lem: sqrt inequalities},  we have 
\begin{align*}
    \l|\sqrt{\frac{\g_+ - \g_-}{2} - u^{\alpha}} - \sqrt{\frac{\g_+ - \g_-}{2}}  \r| \leq u^{\alpha/2}\, .
\end{align*}

Hence \begin{align*}
    P_3 &\leq \frac{4}{\sqrt{\g_+ - \g_-}} \int \limits_0^{\l(\frac{\g_+ - \g_-}{4}\r)^{1/\alpha}} u^{\alpha/2 -1} \, du  =  \frac{4}{\alpha} \, . 
\end{align*}

\textbf{Estimating Q:}
In the integrand of $Q$, the term $1/\sqrt{\g_+ - v}$ blows up when $v$ is near $\g_+$ while the term $1/\sqrt{v- \g_- }$ blows  up when $v$ is near $\g_-$. To separate this behavior, we partition the region of integration into two regions as follows, 
\begin{align*}
    Q &=  \int \limits_{\g_-}^{\g_+} \int \limits_0^{(\g_+ - v)^{1/\alpha}} \frac{1}{u \sqrt{\g_+ - v}} \l|\frac{ 1}{ \sqrt{u^{\alpha} + v - \g_-}} - \frac{1 }{\sqrt{v  -\g_-}}  \r| \, du \, dv  \\
    &=  \int \limits_{\g_-}^{\frac{\g_+ + \g_-}{2}} \int \limits_0^{(\g_+ - v)^{1/\alpha}} \dots \, du \, dv  +  
    \int \limits_{\frac{\g_+ + \g_-}{2}}^{\g_+} \int \limits_0^{(\g_+ - v)^{1/\alpha}} \dots \, du \, dv  \\
    &:= Q_1 + Q_2 \, .
\end{align*}

\textit{Bounding the $Q_1$ component:}
Using the change of variables $z = v - \g_-$, we compute
\begin{align*}
    Q_1 &=  \int \limits_{\g_-}^{\frac{\g_+ + \g_-}{2}} \int \limits_0^{(\g_+ - v)^{1/\alpha}} \frac{1}{u \sqrt{\g_+ - v}} \l|\frac{ 1}{ \sqrt{u^{\alpha} + v - \g_-}} - \frac{1 }{\sqrt{v  -\g_-}}  \r| \, du \, dv  \\
    &\leq \frac{\sqrt{2}}{\sqrt{\g_+ - \g_-}} \int \limits_{\g_-}^{\frac{\g_+ + \g_-}{2}} \int \limits_0^{(\g_+ - v)^{1/\alpha}} \frac{1}{u} \l|\frac{ 1}{ \sqrt{u^{\alpha} + v - \g_-}} - \frac{1 }{\sqrt{v  -\g_-}}  \r| \, du \, dv \qquad \qquad [v \leq (\g_+ + \g_-)/2]\\
    &= \frac{\sqrt{2}}{\sqrt{\g_+ - \g_-}} \int \limits_0^{\frac{\g_+ - \g_-}{2}} \int \limits_0^{(\g_+ - \g_- -z)^{1/\alpha}} \frac{1}{u} \l( \frac{1}{\sqrt{z}} - \frac{1}{\sqrt{z + u^{\alpha}}} \r) \, du \, dz \, .
\end{align*}
Swapping the order of integration, and using the fact that $|\sqrt{x} - \sqrt{x \pm u^{\alpha}}| \leq u^{\alpha/2}$, we compute
\begin{align*}
    &\int \limits_0^{\frac{\g_+ - \g_-}{2}} \int \limits_0^{(\g_+ - \g_- -z)^{1/\alpha}} \frac{1}{u} \l( \frac{1}{\sqrt{z}} - \frac{1}{\sqrt{z + u^{\alpha}}} \r) \, du \, dz \\ 
    &= \int \limits_0^{\l(\frac{\g_+ - \g_-}{2}\r)^{1/\alpha}} \frac{1}{u}   \int \limits_0^{\frac{\g_+ - \g_-}{2}}  \frac{1}{\sqrt{z}} - \frac{1}{\sqrt{z + u^{\alpha}}}  \, dz \, du + \int \limits_{\l(\frac{\g_+ - \g_-}{2}\r)^{1/\alpha}}^{(\g_+ - \g_-)^{1/\alpha}} \frac{1}{u}
    \int \limits_0^{\g_+ - \g_- - u^{\alpha}}   \frac{1}{\sqrt{z}} - \frac{1}{\sqrt{z + u^{\alpha}}} \, dz \, du \\
    &= 2 \int \limits_0^{\l(\frac{\g_+ - \g_-}{2}\r)^{1/\alpha}} \frac{\sqrt{\frac{\g_+ - \g_-}{2}} - \sqrt{\frac{\g_+ - \g_-}{2} + u^{\alpha}} + u^{\alpha/2} }{u} \, du  + 2  \int \limits_{\l(\frac{\g_+ - \g_-}{2}\r)^{1/\alpha}}^{(\g_+ - \g_-)^{1/\alpha}} \frac{\sqrt{\g_+ - \g_- - u^{\alpha}} - \sqrt{\g_+ - \g_-} - u^{\alpha/2} }{u} \, du \\
    &\leq 2 \int \limits_0^{\l(\frac{\g_+ - \g_-}{2}\r)^{1/\alpha}} \frac{u^{\alpha/2} + u^{\alpha/2} }{u} \, du  + 2  \int \limits_{\l(\frac{\g_+ - \g_-}{2}\r)^{1/\alpha}}^{(\g_+ - \g_-)^{1/\alpha}} \frac{u^{\alpha/2} + u^{\alpha/2} }{u} \, du \qquad \qquad \qquad [|\sqrt{x} - \sqrt{x \pm u^{\alpha}}| \leq u^{\alpha/2}]\\
    &\leq 4 \int \limits_0^{(\g_+ - \g_-)^{1/\alpha}}  u^{\alpha/2 -1} \, du  = \frac{8 \sqrt{\g_+ - \g_-}}{\alpha} \, . 
\end{align*}
Hence \begin{align*}
    Q_1 \leq \frac{8 \sqrt{2}}{\alpha} \, . 
\end{align*}

\textit{Bounding the $Q_2$ component:}
By \cref{lem: sqrt inequalities}. 
\begin{align*}
    Q_2 &= \int \limits_{\frac{\g_+ + \g_-}{2}}^{\g_+} \int \limits_0^{(\g_+ - v)^{1/\alpha}} \frac{1}{u \sqrt{\g_+ - v}} \l|\frac{ 1}{ \sqrt{u^{\alpha} + v - \g_-}} - \frac{1 }{\sqrt{v  -\g_-}}  \r| \, du \, dv \\
    &=  \int \limits_{\frac{\g_+ + \g_-}{2}}^{\g_+} \int \limits_0^{(\g_+ - v)^{1/\alpha}} \frac{u^{\alpha/2 -1}}{\sqrt{\g_+ - v}\sqrt{u^{\alpha} + v - \g_-} \sqrt{v- \g_-} } \, du \, dv \\
    &\leq \int \limits_{\frac{\g_+ + \g_-}{2}}^{\g_+} \int \limits_0^{(\g_+ - v)^{1/\alpha}} \frac{u^{\alpha/2 -1}}{\sqrt{\g_+ - v} (v- \g_-) } \, du \, dv \, . 
\end{align*}
Since $v \geq (\g_+ + \g_-)/2$, it follows that $1/ (v - \g_-) \leq 2/(\g_+ - \g_-)$. Hence \begin{align*}
    Q_2 \leq \frac{2}{\g_+ - \g_-} \int \limits_{\frac{\g_+ + \g_-}{2}}^{\g_+} \frac{1}{\sqrt{\g_+ - v}} \int \limits_0^{(\g_+ -v )^{1/\alpha}} u^{\alpha/2 -1} \, du \, dv = \frac{2}{\g_+ - \g_-} \int \limits_{\frac{\g_+ + \g_-}{2}}^{\g_+}\frac{2 }{\alpha} \, dv = \frac{2}{\alpha} \, . 
\end{align*}
\end{proof}

\printbibliography

\end{document}